\documentclass[acmsmall,screen]{acmart} 
\setcopyright{rightsretained}
\acmPrice{}
\acmDOI{10.1145/3371089}
\acmYear{2020}
\copyrightyear{2020}
\acmJournal{PACMPL}
\acmVolume{4}
\acmNumber{POPL}
\acmArticle{21}
\acmMonth{1}
\usepackage{booktabs}   %% For formal tables:
\usepackage{subcaption} %% For complex figures with subfigures/subcaptions
\usepackage{color}
\usepackage{amsmath}
\usepackage{amssymb}
\usepackage{graphicx}
\usepackage{amsthm}
\usepackage{stmaryrd}
\usepackage[all]{xy}
\usepackage{multirow}
\usepackage{bm}
\usepackage[linesnumbered,ruled,vlined]{algorithm2e}
\SetKwInOut{Input}{Input}

\def\>{\ensuremath{\rangle}}
\def\<{\ensuremath{\langle}}

\def\h{\ensuremath{\mathcal{H}}}

\newcommand {\supp } {{\rm supp}}
\newcommand {\spans } {{\rm span}}
\newcommand {\proj } {{\rm proj}}
\newcommand {\B } {{\mathcal{B}}}
\newcommand {\E } {{\mathcal{E}}}
\newcommand {\F } {{\mathcal{F}}}

\newcommand {\D } {{\mathcal{D}}}
\newcommand{\hs}{\mathcal{H}}

\newcommand {\cM} {{\mathcal{M}}}
\newcommand {\cN} {{\mathcal{N}}}

\newcommand {\tr} {{\mathrm{tr}}}

\newcommand{\sm}[1]{\ensuremath{\llbracket #1\rrbracket}}
\newcommand{\cp}[2]{\ensuremath{\left\langle #1, #2\right\rangle}}

\newtheorem{thm}{Theorem}[section]

\newtheorem{lem}{Lemma}[section]
\newtheorem{defn}{Definition}[section]
\newtheorem{prop}{Proposition}[section]
\newtheorem{exam}{Example}[section]

\newtheorem{rem}{Remark}[section]
\newtheorem{fact}{Fact}[section]

\begin{document}
	
	\title{Relational Proofs for Quantum Programs}         %% [Short Title] is optional;
	
	\author{Gilles Barthe}
	
	\affiliation{
		\institution{Max Planck Institute for Security and Privacy}            %% \institution is required
		\country{Germany}
	}
	\affiliation{
		\institution{IMDEA Software Institute}            %% \institution is required
		\country{Spain}
	}
	\author{Justin Hsu}
	\affiliation{
		%\position{Position2a}
		\department{Department of Computer Sciences}             %% \department is recommended
		\institution{University of Wisconsin--Madison}           %% \institution is required
		\country{USA}
	}
	%\email{first2.last2@inst2a.com}         %% \email is recommended
	\author{Mingsheng Ying}
	\affiliation{
		\department{Centre for Quantum Software and Information}
		\institution{University of Technology Sydney}           %% \institution is required
		\country{Australia}
	}
	\affiliation{
		\department{State Key Laboratory of Computer Science}
		\institution{Institute of Software, Chinese Academy of Sciences}           %% \institution is required
		\country{China}
	}
	\affiliation{
		\institution{Tsinghua University}           %% \institution is required
		\country{China}
	}
	%\email{first2.last2@inst2b.org}         %% \email is recommended
	\author{Nengkun Yu}
	\affiliation{
		\department{Centre for Quantum Software and Information}             %% \department is recommended
		\institution{University of Technology Sydney}           %% \institution is required
		\country{Australia}
	}
	%\email{first2.last2@inst2b.org}         %% \email is recommended
	\author{Li Zhou}
	\affiliation{
		\institution{Max Planck Institute for Security and Privacy}            %% \institution is required
		\country{Germany}
	}
	\affiliation{
		\department{Department of Computer Science and Technology}             %% \department is recommended
		\institution{Tsinghua University}           %% \institution is required
		\country{China}
	}
	%\authorsaddresses{}
	\begin{abstract}
		Relational verification of quantum programs has many potential
		applications in quantum and post-quantum security and other domains. We propose a relational
		program logic for quantum programs. The interpretation of our logic is
		based on a quantum analogue of probabilistic couplings. We use our
		logic to verify non-trivial relational properties of
		quantum programs, including uniformity for samples generated by the
		quantum Bernoulli factory, reliability of quantum teleportation against noise
		(bit and phase flip), security of quantum one-time pad and equivalence of quantum walks. 
	\end{abstract}
	

	
	  %% \keywords are mandatory in final camera-ready submission
	
	\maketitle
	
	\section{Introduction}\label{sec:intro}
	Program verification is traditionally focused on proving properties of
	a single program execution. In contrast, relational verification aims
	to prove properties about two program executions. In some cases, such
	as program refinement and program equivalence, the goal is to relate
	executions of two different programs on equal or related
	inputs. However, some properties consider two executions of
	the same program (with related inputs); examples include
	information flow policies (\emph{non-interference}: two runs of a
	program on states that only differ in their secret have equal visible
	effects) and robustness ($k$-\emph{Lipschitz continuity}:
	running a program on two initial states at distance $d$
	yields two final states at distance at most $k \cdot d$). In the probabilistic
	setting, relational verification can also show that a
	program outputs a uniform distribution, or that two programs yield
	\lq\lq approximately equal\rq\rq\ distributions. By taking suitable
	instantiations of approximate equality, relational verification has
	found success in
	cryptography~\cite{BartheGZ09}, machine learning~\cite{BartheEGHS18} and
	differential privacy~\cite{BartheKOZ12,BartheGGHS16}.
	
	This paper develops a relational program logic, called rqPD, for a
	core quantum programming language. Our logic is based on the
	interpretation of predicates as physical observables, mathematically
	modelled as Hermitian operators~\cite{DP06}, and is inspired by the
	qPD program logic \cite{Ying11,Ying16} for quantum
	programs. Concretely, our judgments have the form:
	\begin{equation*}
		P_1\sim P_2: A\Rightarrow B
	\end{equation*}
	where $P_1$ and $P_2$ are quantum programs, and precondition $A$ and
	postcondition $B$ are Hermitian operators over the tensor product Hilbert
	spaces of $P_1$ and $P_2$. We define an interpretation of these judgments, develop a rich
	set of sound proof rules, and show how these rules can be used to
	verify relational properties for quantum programs.
	
	\paragraph*{Technical challenges and solutions.} The central challenge for building a useful relational logic is to
	find an interpretation of judgments that captures properties of
	interest, while guaranteeing soundness of a convenient set of proof
	rules.  This challenge is not unique to quantum programs. In the probabilistic setting~\cite{BartheGZ09,BartheEGHSS15},
	one solution is to interpret judgments in terms of
	probabilistic couplings, a standard abstraction from probability
	theory~\cite{Lindvall02,Thorisson00,Villani08}.  The connection with
	probabilistic couplings has many advantages: (i) it builds the logic
	on an abstraction that has proven to be useful for probabilistic
	reasoning; (ii) it identifies natural extensions of the logic; and
	(iii) it suggests other applications and properties that can be
	handled by similar techniques.  Unfortunately, the quantum setting
	raises additional challenges. Notably, we may need to reason about
	entangled quantum states. 
	There are some existing proposals of analogue of probabilistic couplings in
	the quantum setting (see \cite{KS16, Winter16}). In particular, 
	\citet{quantumstrassen} addressed this issue by developing a
	notion of quantum coupling, and validated their
	definition by showing an analogue of Strassen's
	theorem~\cite{strassen1965existence}.\footnote{Informally, Strassen's theorem states that there exists a
		$B$-coupling between two distributions $\mu$ and $\mu'$ over sets
		$X$ and $X'$ respectively iff for every subset $Y$ in the support of
		$\mu$, $\mu(Y)\leq\mu(B(Y))$, where $B(Y)$ is the set-theoretic
		image of $Y$ under $B$.}  In this work, we base our notion of
	valid judgment on this definition of quantum coupling.
	
	Once the interpretation of the logic is fixed, the next challenge is
	to define a useful set of proof rules. As in other relational
	logics, we need structural rules and three kinds of construct-specific rules. \emph{Synchronous} rules apply when
	$P_1$ and $P_2$ have the same top-level construct and operates on both
	programs, whereas the \emph{left} and \emph{right} rules only operate on one of the
	two programs. In the quantum setting, the main difficulties are:
	\begin{itemize}
		\item \textit{Structural rules}: many useful rules are not sound in the quantum
		setting or require further hypotheses; in particular, in the presence of entanglement --- an indispensable resource in quantum computation and communication.
		We generalise the core judgment to
		track and enforce the hypotheses required to preserve soundness.
		\item \textit{Construct-specific rules}: all proof rules of quantum Hoare logic qPD can be directly generalised into quantum relational logic rqPD (see Subsection \ref{subsec-basic}). Although these directly generalised rules are useful, they do not fully capture the essence of quantum relational reasoning (see Subsection \ref{subsec-mea}). Synchronous proof rules for
		classical control-flow constructs, i.e.\, conditionals and loops,
		generally require that the two programs follow the same control flow
		path, so that they execute in lockstep. In order to retain soundness
		in our setting, we introduce a measurement condition ensuring that corresponding branches in the control flow are taken with equal probabilities.
	\end{itemize}
	
\paragraph*{Simplifying side conditions.} Checking measurement conditions is
often challenging. To make this step easier, we introduce a simplified version
of rqPD where assertions are modelled as projective predicates, or
equivalently, (closed) subspaces of the state Hilbert space --- a special case
of Hermitian operators. The restriction to projective predicates leads to
simpler inference rules and easier program verification, at the cost of
expressiveness \cite{ZYY19}. In particular, checking measurement conditions reduces to
showing that a program condition lies in a subspace (with probability $1$), a
task that is often simpler. We provide a formal comparison between the
original logic system rqPD and its simplified version, and leverage this
comparison to relate our work with a recent proposal for a quantum relational
Hoare logic with projective predicates~\cite{Unruh18}.
	
	\paragraph*{Applications.}
	To test its effectiveness, we apply rqPD to verify non-trivial relational properties of several
	quantum programs, including uniformity for samples generated by the
	quantum Bernoulli factory, reliability of quantum teleportation against noise
	(bit and phase flip), equivalence of quantum walks and security of quantum one-time pad.
	Using our simplified rqPD, we are able to verify the relational properties of
some more sophisticated quantum programs, for example, equivalence of quantum walks with different coin tossing operators.
	
	\paragraph*{Working Example.} We will use the following pair of simple quantum programs as our working example to illustrate our basic ideas along the way:   
	\begin{exam}\label{exam-1}Let $q$ be a qubit (quantum bit). Consider two programs:
		$$P_1\equiv q:=|0\rangle;q:=H[q];Q_1,\qquad\qquad P_2\equiv q:=|0\rangle; Q_2; q:=H[q].$$ 
		In both, $q$ is initialised in a basis state $|0\rangle$. $P_1$
applies the Hadamard gate $H$ to $q$ and executes $Q_1$, while $P_2$
first executes $Q_2$ and then applies $H$ to $q$.
The subprograms $Q_1, Q_2$ are as follows:
		\begin{align*}
			Q_1&\equiv\mathbf{if}\ (\cM[q]=0\rightarrow q:=X[q]\ \square\ 1\rightarrow q:=H[q])\ \mathbf{fi}\ \\
			Q_2&\equiv \mathbf{if}\ (\cM^\prime [q]=0\rightarrow q:=Z[q]\ \square\ 1\rightarrow q:=H[q])\ \mathbf{fi}
		\end{align*} where $\cM, \cM^\prime$ are the measurement in the computational basis $|0\rangle, |1\rangle$ and the measurement in basis $|\pm\rangle=\frac{1}{\sqrt{2}}(|0\rangle\pm |1\rangle)$ respectively.  
		Intuitively, $Q_1$ first performs $\cM$ on $q$, then applies either the Pauli gate $X$ or Hadamard gate $H$, depending on whether the measurement outcome is $0$ or $1$. But $Q_2$ uses the outcomes of a different measurement $\cM$ to choose between the Pauli gate $Z$ and Hadamard gate $H$.\end{exam}
	Obviously, programs $P_1,P_2$ have similar structures. The logic rqPD developed in this paper will enable us to specify and prove some interesting symmetry between them. 
	
	\section{Mathematical Preliminaries}
	We assume basic familiarity with Hilbert spaces, see~\citet{NC00} for an introduction.
\paragraph*{Quantum states.}
	The state space of a quantum system is a Hilbert space
	$\mathcal{H}$. In this paper, we only consider finite-dimensional
	$\mathcal{H}$. A \emph{pure state} of the quantum system is modelled by a (column) vector in
	$\mathcal{H}$ of length $1$; we use the Dirac notation ($|\varphi\rangle, |\psi\rangle$) to denote pure states. For example, qubit $q$ in Example \ref{exam-1} 
	has the $2$-dimensional Hilbert space as its state space; it can be in basis states $|0\rangle, |1\rangle$ as well as in their superpositions  $|+\rangle, |-\rangle=\frac{1}{\sqrt{2}}(|0\rangle\pm |1\rangle)$. 
	An operator $A$ in an $d$-dimensional Hilbert space $\mathcal{H}$ is represented as an $d\times d$ matrix. Its \emph{trace} is defined as
	$\mathit{tr}(A)=\sum_i\langle i|A|i\rangle$, where $\{|i\rangle\}$ is
	an orthonormal basis of $\mathcal{H}$. A positive operator $\rho$ in
	$\mathcal{H}$ is called a \emph{partial density operator} if its trace
	satisfies $\mathit{tr}(\rho)\leq 1$; if $tr(\rho)=1$,
	then $\rho$ is called a \emph{density operator}. A \emph{mixed state} of a
	quantum system is a distribution over pure states. If state $|\psi_i\rangle$ has probability $p_i$, the mixed state can be represented by a density operator $\rho=\sum_ip_i|\psi_i\rangle\langle\psi_i|$, where row vector $\langle\psi_i|$ stands for the conjugate transpose of $|\psi_i\rangle$. 
	For example, if qubit $q$ is in state $|0\rangle$ with probability $\frac{2}{3}$ and in $|+\rangle$ with probability $\frac{1}{3}$, then its state can be described by density operator
	\begin{equation}\label{ex-mix}\rho=\frac{2}{3}|0\rangle\langle 0|+\frac{1}{3}|+\rangle\langle +|=\frac{1}{6}\left (\begin{array}{cc}5 & 1\\ 1&1\end{array}\right ).\end{equation} We write $\mathcal{D}^\leq(\mathcal{H})$ and
	$\mathcal{D}(\mathcal{H})$ for the set of partial density operators
	and the set of density operators in $\mathcal{H}$, respectively. For
	any $\rho\in \mathcal{D}^\leq(\mathcal{H})$, the \emph{support} $\supp(\rho)$
	of $\rho$ is defined as the span of the eigenvectors of $\rho$ with
	nonzero eigenvalues.
	
\paragraph*{Operations on states.}
	A basic operation on a (closed) quantum system is modelled as a \textit{unitary operator} $U$ such that $U^{\dag}U=I_{\mathcal{H}}$, where $^\dag$ stands for conjugate and
	transpose. For example, the Pauli and Hadamard gates used in Example \ref{exam-1} are: 
$$X=\left(\begin{array}{cc}0&1\\1&0\end{array}\right),\qquad Z=\left(\begin{array}{cc}1&0\\0&-1\end{array}\right),\qquad H=\frac{1}{\sqrt{2}}
\left(\begin{array}{cc}1&1\\1&-1\end{array}\right),$$ and $X, Z, H$ transform states $|0\rangle, |1\rangle$ to $X|0\rangle =|1\rangle,\ X|1\rangle=|0\rangle;\ Z|0\rangle =|0\rangle,\ Z|1\rangle=-|1\rangle;\ H|0\rangle =|+\rangle,\ H|1\rangle=|-\rangle$, respectively. 
	
	Another basic operation is \emph{measurement}.
	A \emph{physical observable} is modelled by an operator $A$ in
	$\mathcal{H}$ that is Hermitian, i.e., $A^\dag =A$. An operator $P$ is a \emph{projection} onto a (closed)
	subspace of $\mathcal{H}$ if and only if it is Hermitian (i.e. $P^\dag=P$) and
	idempotent (i.e. $P^2=P$). Quantum measurements are constructed
	from observables $A$. An \emph{eigenvector} of $A$ is a non-zero vector
	$|\psi\rangle\in\mathcal{H}$ such that $A|\psi\rangle=\lambda |\psi\rangle$ for
	some complex number $\lambda$ (indeed, $\lambda$ must be real when $A$ is
	Hermitian). In this case, $\lambda$ is called an \emph{eigenvalue} of $A$. For
	each eigenvalue $\lambda$, the set
	$\{|\psi\rangle:A|\psi\rangle=\lambda|\psi\rangle\}$ of eigenvectors
	corresponding to $\lambda$ and zero vector is a (closed) subspace of
	$\mathcal{H}$. We write $P_\lambda$ for the projection onto this subspace. By
	the spectral decomposition \citep[Theorem 2.1]{NC00}, $A$ can be decomposed as a sum $A=\sum_\lambda\lambda P_\lambda$ where
	$\lambda$ ranges over all eigenvalues of $A$. Moreover,
	$\cM=\{P_\lambda\}_\lambda$ is a (projective) measurement.

If we perform $\cM$ on
	the quantum system in state $\rho$, then outcome $\lambda$ is obtained with
	probability $p_\lambda=\mathit{tr}(P_\lambda^\dag
	P_\lambda\rho)=\mathit{tr}(P_\lambda\rho),$ and after that, the system will be in state $(P_\lambda\rho P_\lambda)/p_\lambda$. Therefore, the expectation of $A$ in
	state $\rho$ is $\llbracket A\rrbracket_\rho=\sum_\lambda
	p_\lambda\cdot\lambda=\sum_\lambda\lambda\mathit{tr}(P_\lambda\rho)=\mathit{tr}(A\rho).$
	For instance, the measurements in Example \ref{exam-1} are defined as $\cM=\{M_0,M_1\}, \cM^\prime=\{M_0^\prime, M_1^\prime\}$ 
	with \begin{align*}&M_0=|0\rangle\langle 0|=\left(\begin{array}{cc}1&0\\0&0\end{array}\right),\qquad\qquad\qquad M_1=|1\rangle\langle 1|=\left(\begin{array}{cc}0&0\\0&1\end{array}\right),\\ &M_0^\prime =|+\rangle\langle +|=\frac{1}{2}\left(\begin{array}{cc}1&1\\1&1\end{array}\right),\qquad\qquad\ \ \ M_1^\prime=|-\rangle\langle -|=\frac{1}{2}\left(\begin{array}{cc}1&-1\\-1&1\end{array}\right).\end{align*}
If we perform $\cM^\prime$ on a qubit in (mixed) state $\rho$ given in equation~(\ref{ex-mix}), then the probability that we get outcome \textquotedblleft$1$\textquotedblright\ is $$p(1)=\mathit{tr}(M_1^\prime\rho)=tr\left[\frac{1}{2}\left(\begin{array}{cc}1&-1\\-1&1\end{array}\right)
\cdot\frac{1}{6}\left(\begin{array}{cc}5& 1\\ 1&1\end{array}\right)
\right]=\frac{1}{12}\cdot tr\left(\begin{array}{cc}4&0\\-4&0\end{array}\right)=\frac{1}{3}$$ and after that, the qubit's state will change to 
$$M_1^\prime\rho M_1^\prime/p(1) = \frac{1}{2}\left(\begin{array}{cc}1&-1\\-1&1\end{array}\right)\cdot\frac{1}{6}\left(\begin{array}{cc}5& 1\\ 1&1\end{array}\right)\cdot\frac{1}{2}\left(\begin{array}{cc}1&-1\\-1&1\end{array}\right)\div\frac{1}{3} = \frac{1}{2}\left(\begin{array}{cc}1&-1\\-1&1\end{array}\right).$$
Similarly, the probability of outcome \textquotedblleft$0$\textquotedblright\ is $p(0)=\frac{2}{3}$, and then the state changes to $|+\>\<+|$. 
	
	We will use observables as predicates in our logic.  To compare two operators
$A$ and $B$ in a Hilbert space $\mathcal{H}$, we will use the \emph{L\"{o}wner
order} between operators defined as follows: $A\sqsubseteq B$ if and only if
$B-A$ is positive. A quantum predicate \cite{DP06} (or an \emph{effect}) in a
Hilbert space $\mathcal{H}$ is an observable (a Hermitian operator) $A$ in
$\mathcal{H}$ with $0\sqsubseteq A\sqsubseteq I$, where $0$ and $I$ are the
zero operator and the identity operator in $\mathcal{H}$, respectively.
	
\paragraph*{Tensor Products of quantum states.}
	Let $\mathcal{H}_1, \mathcal{H}_2$ be the state Hilbert spaces of two
	quantum systems considered in isolation. Then the composite system has state
	space modeled by the tensor product
	$\mathcal{H}_1\otimes\mathcal{H}_2$. The notion of \emph{partial trace} is
	needed to extract the state of a subsystem. Formally, the partial trace over
	$\mathcal{H}_1$ is a mapping $\mathit{tr}_1(\cdot)$ from operators on
	$\mathcal{H}_1\otimes\mathcal{H}_2$ to operators in $\mathcal{H}_2$ defined by the following equation:
	$\mathit{tr}_1(|\varphi_1\rangle\langle\psi_1|\otimes
	|\varphi_2\rangle\langle\psi_2|)=\langle\psi_1|\varphi_1\rangle\cdot
	|\varphi_2\rangle\langle\psi_2|$ for all
	$|\varphi_1\rangle,|\psi_1\rangle\in \mathcal{H}_1$ and
	$|\varphi_2\rangle,|\psi_2\rangle\in\mathcal{H}_2$ together with
	linearity.  The partial trace $\mathit{tr}_2(\cdot)$ over
	$\mathcal{H}_2$ can be defined symmetrically. Suppose that we have a composite system of two subsystems with state spaces $\mathcal{H}_1,\mathcal{H}_2$, respectively, and it is in (mixed) state $\rho$. Then the states of the first and second subsystems can be described by $\mathit{tr}_2(\rho), \mathit{tr}_1(\rho)$, respectively. For example, if the subsystems are both qubits, and they are maximally entangled; i.e. in state $
	|\Phi\rangle=\frac{1}{\sqrt{2}}(|00\rangle+|11\rangle)$ or equivalently \begin{equation}\label{max-ent}|\Phi\>\<\Phi| = \frac{1}{2}(|0\>\<0|\otimes|0\>\<0|+ |0\>\<1|\otimes|0\>\<1|+|1\>\<0|\otimes|1\>\<0|+|1\>\<1|\otimes|1\>\<1|)\end{equation} then the partial traces $\tr_1(|\Phi\>\<\Phi|) = \frac{1}{2}(|0\>\<0|+|1\>\<1|) \text{\ and\ } \tr_2(|\Phi\>\<\Phi|) = \frac{1}{2}(|0\>\<0|+|1\>\<1|)$
 describe states of the second and first subsystems, respectively.
	
	\section{Quantum Couplings and Liftings}\label{Coupling}
	
	\subsection{Quantum Couplings}
	To relate pairs of quantum programs, our program logic will rely on a 
	quantum version of probabilistic coupling.
	In the probabilistic world, a
	coupling for two discrete distributions $\mu_1$ and
	$\mu_2$ over sets $A_1$ and $A_2$ is a discrete distribution $\mu$
	over $A_1\times A_2$ such that the first and second marginals of $\mu$
	are equal to $\mu_1$ and $\mu_2$ respectively. A coupling $\mu$ is an $R$-lifting for $\mu_1$
	and $\mu_2$ if additionally its support is included in $R$, i.e.\,
	every element outside $R$ has probability zero.

	In order to define the quantum analogue of couplings, we apply a
	correspondence between the probabilistic and quantum worlds~\cite{NC00}:
	\begin{align*}
		\mathit{probability\ distributions} \Leftrightarrow \mathit{density\ operators}
\qquad\qquad
		\mathit{marginal\ distributions} \Leftrightarrow \mathit{partial\ traces}
	\end{align*}
	This leads to the following definition of quantum coupling.
	\begin{defn}[Coupling] Let $\rho_1\in\mathcal{D}^\le(\mathcal{H}_1)$ and
		$\rho_2\in\mathcal{D}^\le(\mathcal{H}_2)$. Then
		$\rho\in\mathcal{D}^\le(\mathcal{H}_1\otimes\mathcal{H}_2)$ is
		called a coupling for $\cp{\rho_1}{\rho_2}$ if
		$\mathit{tr}_1(\rho)=\rho_2$ and $\mathit{tr}_2(\rho)=\rho_1$.
	\end{defn}
	
	\begin{prop}[Trace equivalence]\label{prop-treq}
		If $\rho$ is a coupling for $\cp{\rho_1}{\rho_2}$, then they have the same trace:
		$\tr(\rho) = \tr(\rho_1) = \tr(\rho_2)$.
	\end{prop}
	
	The following are examples of quantum couplings. They are quantum
generalisations of several typical examples of (discrete) probabilistic
couplings (see \cite{fullversion}). From these simple examples, we can see a
close and natural correspondence as well as some essential differences between
probabilistic coupling and their quantum counterparts.  Our first example shows
that couplings always exist.
	
	\begin{exam}\label{exam-qtriv} Let $\rho_1\in\mathcal{D}(\mathcal{H}_1)$ and $\rho_2\in\mathcal{D}(\mathcal{H}_2)$ be density operators. The
		tensor product $\rho_\otimes =\rho_1\otimes\rho_2$ is a coupling for
		$\cp{\rho_1}{\rho_2}$.
	\end{exam}
	
Just like the case for probabilistic couplings, there
can be more than one quantum coupling between two operators.
	
	\begin{exam}\label{exam-qunif} Let $\mathcal{H}$ be a $d$-dimensional Hilbert space. Let
		$\mathcal{B}=\{|i\rangle\}$ be an orthonormal basis of
		$\mathcal{H}$. Then the uniform density operator over $\mathcal{H}$
		is $\mathbf{Unif}_\mathcal{H}=\frac{1}{d}\sum_i|i\rangle\langle i|.$
		For each unitary operator $U$ in $\mathcal{H}$, we write
		$U(\mathcal{B})=\{U|i\rangle\}$, which is also an orthonormal basis of
		$\mathcal{H}$. Then $\rho_U=\frac{1}{d}\sum_i(|i\rangle
		U|i\rangle)(\langle i|\langle i|U^\dag)$ is a coupling for
		$\cp{\mathbf{Unif}_\mathcal{H}}{\mathbf{Unif}_{\mathcal{H}}}$. Indeed, the arbitrariness of $U$ shows that there are (uncountably) infinitely many couplings for $\cp{\mathbf{Unif}_\mathcal{H}}{\mathbf{Unif}_{\mathcal{H}}}$.
		For instance, the maximally entangled state $|\Phi\>\<\Phi|$ in equation (\ref{max-ent}) is such a coupling for
		\vspace{-0.1cm}
		$$\cp{\mathbf{Unif}_{\mathcal{H}_1} = \frac{1}{2}(|0\>\<0|+|1\>\<1|)}{\mathbf{Unif}_{\mathcal{H}_2} = \frac{1}{2}(|0\>\<0|+|1\>\<1|)}.$$
	\end{exam}
	
	\begin{exam}\label{exam-qid} Let $\rho$ be a partial density operator in $\mathcal{H}$. Then by the spectral decomposition theorem \citep[Theorem 2.1]{NC00}, $\rho$ can be written as $\rho=\sum_i p_i|i\rangle\langle i|$ for some orthonormal basis $\mathcal{B}=\{|i\rangle\}$ and $p_i\geq 0$ with $\sum_i p_i\leq 1$. We define $\rho_{{\rm id}(\mathcal{B})}=\sum_i p_i|ii\rangle\langle ii|.$ Then it is easy to see that $\rho_{{\rm id}(\mathcal{B})}$ is a coupling for $\cp{\rho}{\rho}$. An essential difference between this example and its classical counterpart (see \cite{Hsu17} Example 2.1.5) is that $\rho$ might be decomposed with other orthonormal bases, say $\mathcal{D}=\{|j\rangle\}$: $\rho=\sum_j q_j|j\rangle\langle j|.$ In general, $\rho_{{\rm id}(\mathcal{B})}\neq\rho_{{\rm id}(\mathcal{D})}$, and we can define a different coupling for $\cp{\rho}{\rho}$: $\rho_{{\rm id}(\mathcal{D})}=\sum_jq_j|jj\rangle\langle jj|.$
		
	\end{exam}
	
	\subsection{Quantum Lifting}
	
	Although there can be many couplings for two operators, it is usually not simple to find one suited to our application. As said at the beginning of this section, lifting can help for this purpose.
	The definition of liftings smoothly generalises to the quantum case.
	\begin{defn}\label{lift-def} Let $\rho_1\in\mathcal{D}^\le(\mathcal{H}_1)$ and $\rho_2\in\mathcal{D}^\le(\mathcal{H}_2)$, and let $\mathcal{X}$ be (the projection onto) a (closed) subspace of $\mathcal{H}_1\otimes\mathcal{H}_2$. Then $\rho\in\mathcal{D}^\le(\mathcal{H}_1\otimes\mathcal{H}_2)$ is called a witness of the lifting $\rho_1\mathcal{X}^\#\rho_2$ if: \begin{enumerate}\item $\rho$ is a coupling for $\cp{\rho_1}{\rho_2}$;
			\item $\supp(\rho)\subseteq\mathcal{X}$.
		\end{enumerate}
	\end{defn}
	\begin{exam}\label{exam-lifting}
	The following are examples of quantum liftings.
	\begin{enumerate}\item The coupling $\rho_U$ for the uniform density operator and itself in
			Example \ref{exam-qunif} is a witness for the lifting $\mathbf{Unif}_\mathcal{H} \mathrel{\mathcal{X}(\mathcal{B},U)^\#} \mathbf{Unif}_{\mathcal{H}}$, 
			where $\mathcal{X}(\mathcal{B},U)=\spans\{|i\rangle U|i\rangle\}$ is a subspace of $\mathcal{H}\otimes\mathcal{H}.$
			\item The coupling $\rho_{{\rm id}(\mathcal{B})}$ in Example
			\ref{exam-qid} is a witness of the lifting $\rho (=_\mathcal{B})^\#
			\rho$, where $(=_\mathcal{B}) \equiv\spans\{|ii\rangle\}$ defined by
			the orthonormal basis $\mathcal{B}=\{|i\rangle\}$ is a subspace of
			$\mathcal{H}\otimes\mathcal{H}$.  It is interesting to note that the
			maximally entangled state $|\Psi\rangle=\frac{1}{\sqrt{d}}\sum_i
			|ii\rangle$ is in $=_\mathcal{B}$.
			
			\item The coupling $\rho_{{\rm id}(\mathcal{B})}$ in Example \ref{exam-qid} is a
			witness of the lifting $\rho (=_\mathit{sym})^\# \rho$, defining
			$=_\mathit{sym}$ to be the symmetrisation operator, i.e. $(=_\mathit{sym}) \equiv \frac{1}{2}(I\otimes I+S),$ where $S$ is the SWAP
			operator defined by $S|\varphi,\psi\rangle=|\psi,\varphi\rangle$ for any
			$|\varphi\rangle, |\psi\rangle\in\mathcal{H}$ together with linearity.
			Operator $S$ is independent of the basis, and given any orthonormal basis
			$\{|i\>\}$ of $\mathcal{H}$, $S$ has the following form:
			$\mathrm{S}= \sum_{ij}|i\>\<j|\otimes|j\>\<i|.$
			\item The coupling $\rho_1\otimes\rho_2$ in Example \ref{exam-qtriv} is a
			witness of the lifting $\rho_1
			\mathrel{(\mathcal{H}_1\otimes\mathcal{H}_2)^\#} \rho_2$.
	\end{enumerate}\end{exam}
	
	The two operators $=_\mathcal{B}$ and $=_\mathit{sym}$ in the above example represents two different kind of symmetry between two quantum systems with the same state Hilbert space $\mathcal{H}$. They will be used to describe relational properties of the two quantum programs $P_1,P_2$ in Example \ref{exam-1}.   
	Liftings of equality are especially interesting for verification, since
	they can be interpreted as relating equivalent quantum systems. The following
	proposition characterizes these liftings.
	\begin{prop}\label{prop-equal} Let $\rho_1,\rho_2\in
		\mathcal{D}^\le(\mathcal{H})$. The following statements are equivalent:
		\begin{enumerate}
			\item[1.] $\rho_1=\rho_2$;
			\item[2.] there exists an orthonormal basis $\mathcal{B}$ s.t.
			$\rho_1 (=_\mathcal{B})^\# \rho_2$;
			\item[3.] $\rho_1 (=_\mathit{sym})^\# \rho_2$.
		\end{enumerate}
	\end{prop}
	
	We see from Example \ref{exam-lifting} and Proposition \ref{prop-equal} that in the quantum world, equality has different generalisations $=_\mathcal{B}$ and $=_\mathit{sym}$.
	Our logic will establish the existence of a lifting of equality, which then
	implies equality of density operators, i.e., equivalence of quantum states.
	
	The notion of quantum lifting can be further generalised to a quantitative version, which will be more convenient in defining the semantics of our logic. We first recall a notation introduced in \citet{DP06}: $\rho \models_\lambda A\ {\rm means}\ \mathit{tr}(A\rho)\geq\lambda$. It can be understood as a quantitative satisfaction relation between a state $\rho$ and an observable $A$ with a real number $\lambda>0$ as a threshold.
	
	\begin{defn} Let $\rho_1\in\mathcal{D}^\le(\mathcal{H}_1)$ and $\rho_2\in\mathcal{D}^\le(\mathcal{H}_2)$, let $A$ be an observable in $\mathcal{H}_1\otimes\mathcal{H}_2$, and let $\lambda>0$. Then $\rho\in\mathcal{D}^\le(\mathcal{H}_1\otimes\mathcal{H}_2)$ is called a witness of the $\lambda$-lifting $\rho_1 A^\#\rho_2$ if: \begin{enumerate}\item $\rho$ is a coupling for $\cp{\rho_1}{\rho_2}$;
			\item $\rho\models_\lambda A$.
		\end{enumerate}
	\end{defn}
	
	It is obvious that whenever $A$ is the projection onto subspace $\mathcal{X}$ and $\lambda =\mathit{tr}(\rho)$, then the above definition degenerates to Definition \ref{lift-def}.
	
	\subsection{Separable versus Entangled Liftings}\label{sec-entangle}
	
	Entanglement presents a major difference between classical and quantum systems and is responsible for most of the advantages of quantum computing and communication over their classical counterparts.
	A partial density operator $\rho\in\mathcal{D}^\leq(\mathcal{H}_1\otimes\mathcal{H}_2)$ is said to be
	\emph{separable} if there exist $\rho_{mi}\in\mathcal{D}^\le(\mathcal{H}_i)$ $(i=1,2)$ such that $\rho=\sum_m  \left(\rho_{m1}\otimes\rho_{m2}\right).$
	A (mixed) state $\rho$ in $\mathcal{H}_1\otimes\mathcal{H}_2$ is said to be
	\emph{entangled} if it is not separable.
	Indeed, the notions of separability and entanglement can be defined for a general positive operator (rather than density operator).
	
	The following proposition shows that entanglement can provide a stronger witness of lifting even with respect to a \textit{separable} observable $A$; that is, sometimes an entangled witness is possible but separable witness does not exist (the proof is given in \cite{fullversion}). 
	
	\begin{prop}\label{entangled-witness} There are states $\rho_i$ in $\mathcal{H}_i$ $(i=1,2)$, separable observable $A$ over $\mathcal{H}_1\otimes\mathcal{H}_2$, entangled state $\rho$ in $\mathcal{H}_1\otimes\mathcal{H}_2$, and $\lambda>0$ such that:
		\begin{enumerate}
			\item $\rho$ is a witness of $\lambda$-lifting $\rho_1 A^\# \rho_2$; and
			\item any separable state $\sigma$ in $\mathcal{H}_1\otimes\mathcal{H}_2$ is not a witness of $\lambda$-lifting $\rho_1 A^\# \rho_2$.
		\end{enumerate}
	\end{prop}
	
	\section{Quantum Programming Language}\label{synt}
	We recall the syntax and semantics for a quantum programming language given in \cite{Ying11, Ying16}. Let
	$\mathit{Var}$ be a set of quantum variables. For each $q\in\mathit{Var}$, we write 
	$\mathcal{H}_q$ for its state Hilbert space.
	
	\begin{defn}[Syntax]\label{syntax} Quantum programs are defined by
		the following syntax:
		\begin{align*}
			P ::={} \mathbf{skip}
			&\mid P_1;P_2
			\mid q:=|0\rangle
			\mid \overline{q}:=U[\overline{q}]
			\mid \mathbf{if}\ \left(\square m\cdot \cM[\overline{q}] =m\rightarrow P_m\right)\ \mathbf{fi} \mid \mathbf{while}\ \cM[\overline{q}]=1\ \mathbf{do}\ P\ \mathbf{od}
		\end{align*}
	\end{defn}
	
	The initialisation $q:=|0\rangle$ sets quantum variable $q$ to a basis state
	$|0\rangle$.  The statement $\overline{q}:=U[\overline{q}]$ means that unitary
	transformation $U$ is applied to register $\overline{q}$.  The
	$\mathbf{if}$-statement is a quantum generalisation of a classical case
	statement. In executing it, measurement $\cM=\{M_{m}\}$ is performed on
	$\overline{q}$, and then a subprogram $P_m$ is selected to be executed next
	according to the outcome $m$ of measurement. The $\mathbf{while}$-statement is a
	quantum generalisation of the classical while loop. The measurement in it has
	only two possible outcomes $0$, $1$; if the outcome $0$ is observed then the
	program terminates, otherwise the program executes the subprogram $P$ and
	continues.
	
	We write $\mathit{var}(P)$ for the set of quantum variables occurring in a
	quantum program $P$. Then tensor product $\mathcal{H}_P=\bigotimes_{q\in
		\mathit{var}(P)}\mathcal{H}_q$ is the state Hilbert space of $P$. A
	\emph{configuration} is a pair $C=\langle P,\rho\rangle,$ where $P$ is a program
	or the termination symbol $\downarrow$, and
	$\rho\in\mathcal{D}^\leq(\mathcal{H}_P)$ is a partial density operator modeling
	the state of quantum variables.
	
	\begin{defn}[Operational Semantics]\label{def-op-sem} The operational semantics of quantum programs is defined as a transition relation $\rightarrow$ by the transition rules in Fig. \ref{fig 3.1}. \begin{figure}[h]\centering
			\begin{equation*}
				\begin{split}&({\rm Sk})\ \langle\mathbf{skip},\rho\rangle\rightarrow\langle \downarrow,\rho\rangle\ \ \ \ \ \ \ \ \ \ \ \ \ \ \ \ \ \ \ \ \ \ \ \ \ \ \ \ \ \ \ \ \ ({\rm In})\ \langle
					q:=|0\rangle,\rho\rangle\rightarrow\langle \downarrow,\rho^{q}_0\rangle\\
					&({\rm UT})\ \ \langle\overline{q}:=U[\overline{q}],\rho\rangle\rightarrow\langle
					\downarrow,U\rho U^{\dag}\rangle\ \ \ \ \ \ \ \ \ \ \ \ \ \ \ ({\rm SC})\ \frac{\langle P_1,\rho\rangle\rightarrow\langle
						P_1^{\prime},\rho^{\prime}\rangle} {\langle
						P_1;P_2,\rho\rangle\rightarrow\langle
						P_1^{\prime};P_2,\rho^\prime\rangle}\\
					&({\rm IF})\ \langle\mathbf{if}\ (\square m\cdot
					\cM[\overline{q}]=m\rightarrow P_m)\ \mathbf{fi},\rho\rangle\rightarrow\langle
					P_m,M_m\rho M_m^{\dag}\rangle\\
					&({\rm L}0)\ \langle\mathbf{while}\
					\cM[\overline{q}]=1\ \mathbf{do}\
					P\ \mathbf{od},\rho\rangle\rightarrow\langle \downarrow, M_0\rho M_0^{\dag}\rangle\\
					&({\rm L}1)\ \langle\mathbf{while}\
					\cM[\overline{q}]=1\ \mathbf{do}\ P\ \mathbf{od},\rho\rangle  \rightarrow \langle
					P;\mathbf{while}\ \cM[\overline{q}]=1\ \mathbf{do}\ P\ \mathbf{od}, M_1\rho
					M_1^{\dag}\rangle\end{split}\end{equation*}
			\caption{Transition Rules. Symbol $\downarrow$ stands for termination. In rule (In), $\rho^{q}_0=\sum_{i}|0\rangle_q\langle i|\rho|i\rangle_q\langle
				0|$ for a given orthonormal basis $\{|i\rangle\}$ of $\mathcal{H}_q$.
				In (IF), $m$ ranges over all possible outcomes of measurement $\cM.$}\label{fig 3.1}
	\end{figure}\end{defn}
	
	The transitions in rules (IF), (L0) and (L1) are essentially
	probabilistic. In both {\bf if} and {\bf while} statements, a measurement is performed at the beginning, and then the program enters different branches based on the measurement outcome.
	For each outcome $m$, the transition in (IF)
	happens with probability $p_m=\mathit{tr}(M_m^\dag M_m\rho),$ and the
	program state $\rho$ is changed to $\rho_m=M_m\rho M_m^\dag /p_m.$ 
	In rule (L0) and (L1) the outcome ``$0$''  occurs with the probability $p_0=\tr(M_0\rho M_0^\dag)$, and the program terminates in state $M_0\rho M_0^\dag/p_0$; otherwise, with the probability $p_1=\tr(M_1\rho M_1^\dag)$, the outcome ``$1$'' occurs, the program state is changed to $M_1\rho M_1^\dag/p_1$, and then the program executes the loop body $P$ and goes back to the beginning of the loop.
	We follow a convention suggested by \citet{Se04} to combine
	probability $p_m$ and density operator $\rho_m$ into a partial density
	operator $M_m\rho M_m^\dag=p_m\rho_m.$ This convention is useful for
	presenting the operational semantics as a non-probabilistic transition
	system, simplifying the presentation.
	
	\begin{defn}[Denotational Semantics]\label{den-sem-def} For any quantum program $P$, its semantic function is the mapping $\llbracket P\rrbracket:\mathcal{D}^\leq(\mathcal{H}_P)\rightarrow \mathcal{D}^\leq(\mathcal{H}_P)$ defined as follows: for every $\rho\in\mathcal{D}^\leq(\mathcal{H}_P)$, \begin{equation}\llbracket P\rrbracket(\rho)=\sum\left\{\!|\rho^\prime: \langle P,\rho\rangle\rightarrow^\ast\langle \downarrow,\rho^\prime\rangle|\!\right\},\end{equation} where $\rightarrow^\ast$ is the reflexive and transitive closure of $\rightarrow$, and $\left\{\!|\cdot|\!\right\}$ denotes a multi-set.
	\end{defn}
	
	For instance, let us consider program $Q_2$ in our working example \ref{exam-1} with input $\rho$ given in equation (\ref{ex-mix}). According to Definition \ref{def-op-sem}, it has two transitions:
	\begin{equation*}\langle Q_2,\rho\rangle\rightarrow\langle q:=Z[q],\rho_0\rangle\rightarrow\langle\downarrow,\rho_0^\prime\rangle,\qquad \langle Q_2,\rho\rangle\rightarrow\langle q:=H[q],\rho_1\rangle\rightarrow\langle\downarrow,\rho_1^\prime\rangle, 
	\end{equation*} where: 
	\begin{align*}&\rho_0=M_0^\prime\rho M_0^\prime=\frac{1}{3}\left(\begin{array}{cc}1&1\\1&1\end{array}\right),\qquad &\rho_0^\prime=Z\rho_0 Z=\frac{1}{3}\left(\begin{array}{cc}1&-1\\-1&1\end{array}\right),\\ &\rho_1=M_1^\prime\rho M_1^\prime=\frac{1}{6}\left(\begin{array}{cc}1&-1\\-1&1\end{array}\right),\qquad &\rho_1^\prime=H\rho_1 H=\frac{1}{3}\left(\begin{array}{cc}0&0\\0&1\end{array}\right).\end{align*} According to Definition \ref{den-sem-def}, the output is
$\llbracket Q_2\rrbracket(\rho)=\rho_0^\prime+\rho_1^\prime=\frac{1}{3}\left(\begin{array}{cc}1&-1\\-1&2\end{array}\right).$
	Furthermore, one can show that for any possible input $\rho$ with trace one, programs $P_1,P_2$ in Example \ref{exam-1} have the same output: $\sm{P_1}(\rho) = \sm{P_2}(\rho) = \frac{1}{4}\left(\begin{array}{cc}1 & -1\\ -1 & 3\end{array}\right).$

	The soundness of some of the proof rules in probabilistic relational Hoare logic requires programs to terminate \cite{BartheGZ09}. The same is true in the quantum setting.
	\begin{defn} A quantum program $P$ is called \emph{lossless}, written $\models
		P\ {\rm lossless}$, if its semantics function $\llbracket P\rrbracket$
		is \emph{trace-preserving}; that is, $\mathit{tr}(\llbracket
		P\rrbracket(\rho))=\mathit{tr}(\rho)$ for all
		$\rho\in\mathcal{D}^\leq(\mathcal{H}_P)$.\end{defn}
	For example, programs $P_1$ and $P_2$ in Example \ref{exam-1} are both lossless.
	\begin{rem} The lossless property of quantum loop $\mathbf{while}\
		\cM[\overline{q}]=1\ \mathbf{do}\ P\ \mathbf{od}$ was previously studied \cite{YYFD13}.
		Let the semantic function of loop body $P$ be given in Kraus
		operator-sum form: $\llbracket P\rrbracket (\rho)=\sum_iE_i\rho
		E_i^\dag.$ We define (super-)operator $\mathcal{E}$ by $\mathcal{E}
		(\rho)=\sum_i(M_1^{\dag}E_i^\dag)\rho (E_iM_1)$ for every $\rho$. A square matrix $X$
		is called an eigenvector of $\mathcal{E}$ corresponding to
		an eigenvalue $\lambda$ if $\mathcal{E}(X)=\lambda X$. It
		was shown that the loop is lossless if and only if
		any eigenvector of $\mathcal{E}$ corresponding to an
		eigenvalue with modulus $1$ is traceless.
	\end{rem}
	
	\section{Relational program logic}\label{sec:qrhl}
	We adopt standard conventions and notations for relational program
	logics. For each quantum variable $q\in\mathit{Var}$, we assume two
	tagged copies $q\langle 1\rangle$ and $q\langle 2\rangle$, and their
	state Hilbert spaces are the same as that of $q$:
	$\mathcal{H}_{q\langle 1\rangle}=\mathcal{H}_{q\langle
		2\rangle}=\mathcal{H}_q$. For $i=1,2$, if $X\subseteq\mathit{Var}$,
	then we write $X\langle i\rangle=\{q\langle i\rangle \mid q\in
	X\}$. Furthermore, for every quantum program $P$ with
	$\mathit{var}(P)\subseteq\mathit{Var}$, we write $P\langle i\rangle$
	for the program obtained by replacing each quantum variable $q$ in $P$
	with $q\langle i\rangle$. Also, for each operator $A$ in
	$\mathcal{H}_X=\bigotimes_{q\in X}\mathcal{H}_q$, we write $A\langle
	i\rangle$ for the corresponding operator of $A$ in
	$\mathcal{H}_{X\langle i\rangle}=\bigotimes_{q\in
		X}\mathcal{H}_{q\langle i\rangle}$. For simplicity, we will drop the
	tags whenever they can be understood from the context; for example, we
	often simply write $A\otimes B$ instead of $A\langle 1\rangle\otimes B\langle 2\rangle$.
	
	\subsection{Judgments and Satisfaction}
	Judgments in our logic are of the form \begin{equation}\label{syntax-judge}\Gamma\vdash P_1\sim P_2:A \Rightarrow B\end{equation}
	where $P_1$ and $P_2$ are quantum programs, $A$ and $B$ are quantum predicates in $\mathcal{H}_{P_1\langle 1\rangle}\otimes
	\mathcal{H}_{P_2\langle 2\rangle}$, and $\Gamma$ is a set of
	measurement or separability conditions.
	If $\Gamma=\{\Sigma_1,...,\Sigma_n\}$, then for any $\rho\in\mathcal{D}^\le
	\left(\mathcal{H}_{P_1\langle 1\rangle}\otimes \mathcal{H}_{P_2\langle
		2\rangle}\right)$, $\rho\models\Gamma$ means $\rho\models\Sigma_i$ for all $i=1,...,n$.
	We defer the definition of measurement or separability condition
	$\Sigma_i$ for now, simply assuming a given notion of satisfaction
	$\rho\models \Sigma_i$. In particular, if $\Gamma=\emptyset$, then we simply write $\vdash P_1\sim P_2:A \Rightarrow B$ for $\Gamma\vdash P_1\sim P_2:A \Rightarrow B$. 
	
	\begin{defn}\label{valid-ju}
		The judgment $\Gamma\vdash P_1\sim P_2:A \Rightarrow B$ is valid,
		written: $$\Gamma \models P_1\sim P_2: A\Rightarrow B$$ if for every
		$\rho\in\mathcal{D}^\le\left(\mathcal{H}_{P_1\langle 1\rangle}\otimes
		\mathcal{H}_{P_2\langle 2\rangle}\right)$ such that $\rho\models
		\Gamma$, there exists a quantum coupling $\sigma$ for $\cp{\sm{P_1}(\mathit{tr}_{\langle 2\rangle}(\rho))}{\sm{P_2}(\mathit{tr}_{\langle 1\rangle}(\rho))}$ such
		that \begin{equation}\label{eqe-judge}
			\mathit{tr}(A\rho)\leq\mathit{tr}(B\sigma)+\mathit{tr}(\rho)-\mathit{tr}(\sigma).
		\end{equation}
		We will often use $\rho\models P_1\sim P_2: A\Rightarrow B$ as shorthand.
	\end{defn}
	The above definition differs from validity in probabilistic relational
	Hoare logic in several ways. Except the set $\Gamma$ of measurability and separability
	conditions (explained below), lifting does not appear explicitly.
	However, the existence of a lifting can be established from inequality
	(\ref{eqe-judge}) under mild conditions, as we now explain.  First, we
	note that $\mathit{tr}(\rho)- \mathit{tr}(\sigma)$ captures the
	non-termination probability of the programs, as in the
	(non-relational) quantum program logic qPD. To see a clearer
	probabilistic-quantum correspondence, let us consider the simple case
	where both $P_1$ and $P_2$ are lossless. Then
	$\mathit{tr}(\rho)-\mathit{tr}(\sigma)=0$ and inequality
	(\ref{eqe-judge}) is simplified to $\mathit{tr}(A\rho)\leq
	\mathit{tr}(B\sigma)$, or equivalently: ${\rm for\ any}\ \lambda>0,\ \rho\models_\lambda A\Rightarrow\sigma\models_\lambda B.$ This is a real number-valued analogue of
	boolean-valued proposition \textquotedblleft$\rho\in
	A\Rightarrow\sigma\in B$\textquotedblright. Therefore, for any $\lambda>0$, if $\rho\models_\lambda A$, then $\sigma\models_\lambda B$ and combined with
	the assumption that $\sigma$ is a coupling for $\cp{\sm{P_1}(\mathit{tr}_{\langle 2\rangle}(\rho))}{\sm{P_2}(\mathit{tr}_{\langle 1\rangle}(\rho))}$, we see
	that $\sigma$ is a witness for $\lambda$-lifting $\llbracket
	P_1\rrbracket (\mathit{tr}_{\langle 2\rangle}(\rho)) B^\#\llbracket
	P_2\rrbracket (\mathit{tr}_{\langle 1\rangle}(\rho))$.
	
	An interesting symmetry between programs $P_1,P_2$ in our working example \ref{exam-1} can be expressed as the following judgment: 
	\begin{equation}\label{ex-1.0}\vdash P_1\sim P_2:\ 
	(=_\mathcal{B}) \Rightarrow (=_\mathit{sym}).\end{equation}
	where precondition $=_\mathcal{B}$ is the equality defined by
	the computational basis $\mathcal{B} =\{|0\rangle, |1\rangle\}$ of a
	qubit; i.e. $(=_\mathcal{B}) =\mathit{span}\{|00\rangle,
	|11\rangle\}=|00\rangle\langle 00|+|11\rangle\langle 11|$ [see Example
	\ref{exam-lifting} 2)], and postcondition $=_\mathit{sym}$ is the
	projector onto the symmetric space [see Example \ref{exam-lifting}
	3)]. 
	The validity of this judgment can be checked by the denotational semantics of $P_1,P_2$. 
	We first observe that for any $\rho\in\mathcal{D}^\le\left(\mathcal{H}_{P_1}\otimes
	\mathcal{H}_{P_2}\right)$, $\tr(=_\mathcal{B}\rho)\le\tr(\rho)$. Moreover, we have:
	$$\sm{P_1}(\mathit{tr}_{\langle 2\rangle}(\rho)) = \sm{P_2}(\mathit{tr}_{\langle 1\rangle}(\rho)) = \frac{1}{4}\left(\begin{array}{cc}1 & -1\\ -1 & 3\end{array}\right)\times \tr(\rho)$$
	by noting that $\tr(\mathit{tr}_{\langle 2\rangle}(\rho)) = \tr(\mathit{tr}_{\langle 1\rangle}(\rho)) = \tr(\rho)$. As shown in Example \ref{exam-lifting} (3), lifting $\sm{P_1}(\mathit{tr}_{\langle 2\rangle}(\rho))$ $(=_\mathit{sym})^\#\sm{P_2}(\mathit{tr}_{\langle 1\rangle}(\rho))$ holds and, suppose $\sigma$ is a witness, then we have $\tr(=_{sym}\sigma) = \tr(\sigma)$ and therefore, 
	$\tr(=_\mathcal{B}\rho) \le \tr(=_{sym}\sigma)$ as $\tr(\rho) = \tr(\sigma)$
	according to Proposition \ref{prop-treq}, which actually implies the validity of judgment (\ref{ex-1.0}).
	
	In the remainder of this section, we gradually develop the proof system for our logic rqPD. At the same time, we will see how the proof rules in rqPD can be used to verify judgment (\ref{ex-1.0}). For readability, we first give a proof outline of judgment (\ref{ex-1.0}) in Fig. \ref{proofoutline-workingexample}, where a judgment $\Gamma\vdash P_1\sim P_2:A\Rightarrow B$ derived by an inference rule $R$ in rqPD is displayed as
	\begin{equation*}\begin{split}
	&\{A\}\ \{{\rm SC:}\ \Gamma\} \\
	&\qquad\qquad P_1\sim P_2 \qquad\qquad (R)\\
	&\{B\}\end{split}\end{equation*}
	
	\begin{figure}
		\small
		\centering
		\begin{align*}
		&\{=_{\B}\}\{I\otimes I\} &&\hspace{-1cm}{\rm(Conseq)}&&\text{Derivation\ of}\ Q_1\sim Q_2\qquad\quad {\rm(IF1)}\\
		&\qquad\qquad q:=|0\>; \sim q:=|0\>;  &&\hspace{-1cm}{\rm(Init)} &&\cM\approx \cM^\prime \models I\otimes I\Rightarrow \{A_{00},A_{11}\}\\
		&\{I\otimes I\}\ \ \{{\rm SC:}\ \cM^\prime\approx \cM^\prime\}&& &&\overline{\left\{A_{00}=\frac{1}{2}[I\otimes I+(X\otimes ZH)S(X\otimes HZ)]\right\}}\\
		&\qquad\qquad q:=H[q]; \sim \mathbf{skip}; &&\hspace{-1cm}{\rm(UT\text{-}L)} &&\qquad\qquad q:=X[q]; \sim q:=Z[q];\quad{\rm(UT)}\\
		&\{I\otimes I\}\ \ \{{\rm SC:}\ \cM\approx \cM^\prime\} && && \left\{B=\frac{1}{2}[I\otimes I+(I\otimes H)S(I\otimes H)]\right\} \\
		&\qquad\qquad Q_1; \sim Q_2; &&\hspace{-1cm}{\rm(IF1)} && \overline{\left\{A_{11}=\frac{1}{2}[I\otimes I+(H\otimes HH)S(H\otimes HH)]\right\}}\\
		& \left\{B = \frac{1}{2}[I\otimes I+(I\otimes H)S(I\otimes H)]\right\}\!\!\!\!\!\!\!\!&& &&\qquad\qquad q:=H[q]; \sim q:=H[q];\quad{\rm(UT)}\\
		&\qquad\qquad \mathbf{skip};\sim q:=H[q]; &&\hspace{-1cm}{\rm(UT\text{-}R)}&&\left\{B=\frac{1}{2}[I\otimes I+(I\otimes H)S(I\otimes H)]\right\} \\
		& \left\{(=_{sym}) = \frac{1}{2}[I\otimes I+S]\right\}&& &&
		\end{align*}
		\caption{Verification of working example \ref{exam-1} in rqPD: ${\rm P_1}\sim{\rm P_2}$. The proof outline is shown in the left column with side-condition labeled by SC, and the derivation of $Q_1\sim Q_2$ is displayed in the right column with measurement condition.
		}
		\label{proofoutline-workingexample}
	\end{figure}
	
	\subsection{Basic Construct-Specific Rules}\label{subsec-basic}
	As usual, the proof system consists of two categories of rules: \textit{construct-specific rules} and \textit{structural rules}. Let us start from a set of construct-specific rules that can be directly adapted from quantum Hoare logic qPD~\cite{Ying11,Ying16}. They include two-side rules and one-side rules given in Figs. \ref{fig 4.4-0} and \ref{fig 4.5-0}, respectively. It is worth noting that all of these rules do not introduce any measurement or separability condition. 
	These rules are easy to understand if compared with the corresponding rules of qPD, which are displayed in \cite{fullversion}, and those of probabilistic logic pRHL.  
	For the working example \ref{exam-1}, these rules are used to prove the following judgments in Fig. \ref{proofoutline-workingexample}: 
	\begin{enumerate}
		\item $\vdash  q:=|0\> \sim q:=|0\>:  (=_{\B}) \Rightarrow I\otimes I$ by rule (Conseq) and (Init);
		\item $\vdash q:=H[q] \sim \mathbf{skip}: I\otimes I \Rightarrow I\otimes I$ by rule (UT-L);
		\item $\vdash \mathbf{skip}\sim q:=H[q]: B\Rightarrow (=_{sym})$ by rule (UT-R); and 
		\item $\vdash q:=X[q] \sim q:=Z[q]: A_{00} \Rightarrow B$ and $\vdash q:=H[q] \sim q:=H[q]: A_{11} \Rightarrow B$ by rule (UT).
	\end{enumerate}
As we will see in Section \ref{sec-examples}, these basic rules are already enough to verify interesting relational properties of quantum programs, including security of quantum one-time pad. 
	\begin{figure*}[!h]
		\centering
		\begin{equation*}\begin{split}
				&({\rm Skip})\ \ \vdash \mathbf{Skip}\sim \mathbf{Skip}: A\Rightarrow A
				\\ &({\rm Init}) \ \ \vdash q_1:=|0\rangle\sim q_2:=|0\rangle:\sum_{i,j}(|i\rangle_{q_1\langle 1\rangle}\langle 0| \otimes |j\rangle_{q_2\langle 2\rangle}\langle 0|) A( |0\rangle_{q_1\langle 1\rangle}\langle
				i|\otimes |0\rangle_{q_2\langle 2\rangle}\langle
				j|) \Rightarrow A\\
				&({\rm UT})\ \ \vdash \overline{q_1}:=U_1\left[\overline{q_1}\right]\sim \overline{q_2}:=U_2\left[\overline{q_2}\right]: \left(U_1^\dag\otimes U_2^\dag\right) A\left( U_1 \otimes U_2\right)\Rightarrow A
				\\ &({\rm SC})\ \ \ \frac{\vdash P_1\sim P_2:A\Rightarrow B\qquad \vdash P_1^\prime\sim P_2^\prime:B\Rightarrow C}{\vdash P_1;P_1^\prime\sim P_2;P_2^\prime:A\Rightarrow C}
				\\ &({\rm IF})\ \ \
				\frac{\vdash P_{1m}\sim P_{2n}: B_{mn}\Rightarrow C\ {\rm for\ every}\ (m,n)\in S\qquad \forall\ m,n:\ \ \models P_{1m}, P_{2n}\ {\rm lossless}}{\begin{array}{cc}\vdash \mathbf{if}\ (\square m\cdot
						\mathcal{M}_1[\overline{q}]=m\rightarrow P_{1m})\ \mathbf{fi}\sim \mathbf{if}\ (\square n\cdot
						\mathcal{M}_2[\overline{q}]=n\rightarrow P_{2n})\ \mathbf{fi}:\\
						\sum_{(m,n)\in S}
						\left(M_{1m}^{\dag}\otimes M_{2n}^\dag\right)B_{mn}\left(M_{1m}\otimes M_{2n}\right)\Rightarrow C\end{array}}\\
				&({\rm LP}) \
				\frac{
					\begin{array}{c}
						\models \mathbf{while}\ \cM_i[\overline{q}]=1\ \mathbf{do}\ P_i\ \mathbf{od} \text{\ lossless}\ (i = 1,2) \\
						\vdash P_1\sim P_2: B\Rightarrow (M_{10}\otimes M_{20})^\dag A(M_{10}\otimes M_{20})
						+ (M_{11}\otimes M_{21})^\dag B(M_{11}\otimes M_{21})
				\end{array}}{
					\begin{array}{l}
						\vdash \mathbf{while}\
						\cM_1[\overline{q}]=1\ \mathbf{do}\ P_1\ \mathbf{od}\sim    \mathbf{while}\
						\cM_2[\overline{q}]=1\ \mathbf{do}\ P_2\ \mathbf{od}:  \\
						\qquad\qquad\qquad(M_{10}\otimes M_{20})^\dag A(M_{10}\otimes M_{20})
						+ (M_{11}\otimes M_{21})^\dag B(M_{11}\otimes M_{21})\Rightarrow A
				\end{array}}
		\end{split}\end{equation*}
		\caption{Two-sided rqPD rules. The set $S$ in rule (IF) is a subset of the Cartesian product of the possible outcomes of measurements $\mathcal{M}_1$ and $\mathcal{M}_2$.}\label{fig 4.4-0}
	\end{figure*}
	
	\begin{figure*}\centering
		\begin{equation*}\begin{split}
				&(\mathrm{Init\text{-}L}) \ \ \ \vdash q_1:=|0\rangle\sim \mathbf{skip}:
				\sum_{i}\left(|i\rangle_{q_1\langle 1\rangle}\langle 0|\right) A\left
				( |0\rangle_{q_1\langle 1\rangle}\langle i|\right) \Rightarrow A\\
				&({\rm UT\text{-}L})\ \ \ \vdash \overline{q_1}:=U_1\left[\overline{q_1}
				\right] \sim \mathbf{skip}: U_1^\dag A  U_1\Rightarrow A\\
				&({\rm IF\text{-}L})\ \ \
				\frac{\vdash P_{1m}\sim P: B_m\Rightarrow C\ {\rm for\ every}\ m}{\vdash \mathbf{if}\ (\square m\cdot
					M_1[\overline{q}]=m\rightarrow P_{1m})\ \mathbf{fi}\sim P:
					\sum_m M_{1m}^{\dag}B_mM_{1m}\Rightarrow C}\\
				&({\rm LP\text{-}L}) \
				\frac{
					\begin{array}{c}
						\models \mathbf{while}\ \cM_1[\overline{q}]=1\ \mathbf{do}\ P_1\ \mathbf{od} \text{\ lossless}\quad
						\vdash P_1\sim \mathbf{skip}: B\Rightarrow M_{10}^\dag AM_{10} + M_{11}^\dag BM_{11}
				\end{array}}{
					\vdash \mathbf{while}\
					\cM_1[\overline{q}]=1\ \mathbf{do}\ P_1 \mathbf{od}\sim  \mathbf{skip}:
					M_{10}^\dag AM_{10} + M_{11}^\dag BM_{11}\Rightarrow A
				}
		\end{split}\end{equation*}
		\caption{One-sided rqPD rules. We omitted the right-sides rules, which are symmetric to the ones here.}\label{fig 4.5-0}
	\end{figure*}
	
	\begin{rem} Note that in rule (IF) the branches of two case statements are not required to match exactly. Whenever there is an one-to-one correspondence between the outcomes of measurement $\mathcal{M}_1$ and $\mathcal{M}_2$, then (IF) can be simplified to (IF-w) in Fig. \ref{fig 4.4-00}.
		\begin{figure*}[!h]
			\centering
			\begin{equation*}\begin{split}
					({\rm IF\text{-}w})\ \ \
					\frac{\vdash P_{1m}\sim P_{2m}: B_m\Rightarrow C\ {\rm for\ every}\ m\qquad \forall\ m:\ \ \models P_{1m}, P_{2m}\ {\rm lossless}}{\begin{array}{cc}\vdash \mathbf{if}\ (\square m\cdot
							M_1[\overline{q}]=m\rightarrow P_{1m})\ \mathbf{fi}\sim \mathbf{if}\ (\square m\cdot
							M_2[\overline{q}]=m\rightarrow P_{2m})\ \mathbf{fi}:\\
							\sum_m
							\left(M_{1m}^{\dag}\otimes M_{2m}^\dag\right)B_m\left(M_{1m}\otimes M_{2m}\right)\Rightarrow C\end{array}}
			\end{split}\end{equation*}
			\caption{A weak rule for case statements.}\label{fig 4.4-00}
		\end{figure*}
	\end{rem}
	
	\subsection{Measurement Conditions}\label{subsec-mea}
	
	The straightforward generalisations of the proof rules for case statements and loops in qPD given in the above subsection are not strong enough for more complicated applications of rqPD. In particular, they do not reveal the subtle differences between the relational and non-relational properties of quantum programs. 
	To understand this point, let us take a closer look at derivation of the judgment about two {\bf if} statements $Q_1$ and $Q_2$ in Fig. \ref{proofoutline-workingexample}. Let us first list all derivable judgments of possible combinations of branches as follows:
	\begin{enumerate}
		\item $\vdash q:=X[q] \sim q:=Z[q]: A_{00} \Rightarrow B$ and $\vdash q:=H[q] \sim q:=H[q]: A_{11} \Rightarrow B$;
		\item $\vdash q:=X[q]\sim q:=H[q]: A_{01}\Rightarrow B$ and $\vdash q:=H[q]\sim q:=Z[q]: A_{10}\Rightarrow B$
	\end{enumerate}
	where $A_{00}, A_{11}$ and $B$ are given as in Fig. \ref{proofoutline-workingexample} and 
	\begin{align*}
	A_{01}=\frac{1}{2}[I\otimes I+(X\otimes HH)S(X\otimes HH)],\quad \ A_{10}=\frac{1}{2}[I\otimes I+(H\otimes ZH)S(H\otimes HZ)].
	\end{align*}
	Applying rule (IF) directly we obtain:
	$
	\vdash Q_1 \sim Q_2:\ A \Rightarrow B,
	$
	where $$A = \sum_{i,j=0}^1(M_i\otimes M_j^\prime)^\dag A_{ij}(M_i\otimes M_j^\prime) = \left(
	\begin{array}{cccc}
	7/8 & 1/8 & 0 & 0 \\
	1/8 & 7/8 & 0 & 0 \\
	0 & 0 & 7/8 & -1/8 \\
	0 & 0 & -1/8 & 7/8 \\
	\end{array}
	\right).$$
	Then, using rule (UT-L), (UT-R) and (Init) for the rest parts of the programs, we are only able to derive 
	$\vdash P_1\sim P_2:\ \frac{7}{8}I\otimes I\Rightarrow\ =_{sym}.$
	However, $=_\B\ \not\sqsubseteq \frac{7}{8}I\otimes I$, so the rule (IF)
is too weak to derive judgment (\ref{ex-1.0}) as we desire. A similar argument
shows that rules (IF-L) and (IF-R) are also too weak.
%	let us consider a very simple example. 
%	
%	\begin{exam} Let $S$ be a trivial ${\bf while}$-loop defined as 
%		\begin{align*}
%			S \equiv {\bf while}\ M[q] = 1\ {\bf do\ skip\ od}
%		\end{align*}
%		where the measurement $M = \{M_0 = \frac{1}{\sqrt{2}}I,\ M_1 = \frac{1}{\sqrt{2}}I\}$. 
%		It is easy to show that $S$ is lossless, i.e., $S$ almost surely terminates with all possible input. Moreover, the output is exactly the same as the input, and therefore, $\models S\ \sim\ S: \ =_{sym} \Rightarrow\ =_{sym},$ where $=_\mathit{sym}$ is the projector onto the symmetric space [see Example \ref{exam-lifting} 3)]. However, such a simple relational property cannot be verified using the rules presented in the above subsection. 
%		In fact, if we use rule (LP) to prove the following judgment: 
%		$\vdash S\ \sim\ S: \ A \Rightarrow\ =_{sym},$ then it must hold that 
%		\begin{equation}\label{impossible}
%			A = (M_0\otimes M_0)^\dag =_{sym} (M_0\otimes M_0) + (M_1\otimes M_1)^\dag B (M_1\otimes M_1) 
%			\sqsubseteq  \frac{1}{4} =_{sym} + \frac{1}{4} I
%		\end{equation}
%		where $B\sqsubseteq I$ is some predicate used in the rule (LP). This implies that $A\neq\ =_{sym}$ because $=_{sym}\not\sqsubseteq\frac{1}{4} =_{sym} + \frac{1}{4} I.$ A similar argument shows that using rule (LP-L) and (LP-R) is also unable to prove (\ref{impossible}).
%	\end{exam}
	
	We have more examples (e.g., Example \ref{Bernoulli}) showing that some important relational properties cannot be verified simply using rules (IF) and (LP). The reason can be seen from the soundness proof of (IF) and (LP) \cite{fullversion}, where we only use a part of the output states to construct the coupling, so for a given postcondition, the derivable preconditions are sometimes too weak. To resolve this issue, we need to capture more relational information between two programs.  A crucial issue in developing inference rules for relational reasoning is to guarantee that two programs $P_1$ and $P_2$ under comparison execute in lockstep. In probabilistic relational Hoare logic, a side-condition $\Theta \Rightarrow e_1 = e_2$ is introduced for this purpose, where $\Theta$ is the precondition, $e_1$ and $e_2$ are the guards in $P_1$ and $P_2$, respectively. In the quantum case, branching (control flow) is determined by the measurement outcomes.
	So, more sophisticated rules for case analysis, loops, and conditionals in rqPD involve
	\emph{measurement conditions}.
	
	\begin{defn}\label{meas-eq} Let $\cM_1 = \{M_{1m}\}$ and $\cM_2 = \{M_{2m}\}$ be
		two measurements with the same set $\{m\}$ of possible outcomes in $\hs_{P_1}$
		and $\hs_{P_2}$, respectively, and let
		$\rho\in\mathcal{D}^\le\left(\mathcal{H}_{P_1\langle 1\rangle}\otimes
		\mathcal{H}_{P_2\langle 2\rangle}\right)$. Then we say that $\rho$ satisfies
		$\cM_1\approx \cM_2$, written $\rho\models \cM_1\approx \cM_2$, if
		$\cM_1$ and $\cM_2$ produce equal probability distributions
		when applied to $\rho$. That is,  for all $m$, we have: $\mathit{tr}(M_{1m}\mathit{tr}_{\langle 2\rangle}(\rho) M_{1m}^\dag) = \mathit{tr}(M_{2m}\mathit{tr}_{\langle 1\rangle}(\rho) M_{2m}^\dag).$
	\end{defn}
	
	Intuitively, the above measurement conditions mean that $P_1$ and $P_2$ enter the corresponding branches with the same probability (and thus execute in lockstep).
	
	The above definition is enough for relating measurements in case
	statements. But when dealing with loops, we have to consider the
	measurements in the loop guards together with the loop bodies. To address this issue, we further introduce the following definition:
	
	\begin{defn}\label{meas-p-eq} Let $P_1$ and $P_2$ be two programs, and let $\cM_1=\{M_{10},M_{11}\}$, $\cM_2=\{M_{20},M_{21}\}$ be boolean-valued measurements in $\hs_{P_1},\hs_{P_2}$, respectively. Then for any $\rho\in\mathcal{D}^\le\left(\mathcal{H}_{P_1\langle 1\rangle}\otimes \mathcal{H}_{P_2\langle 2\rangle}\right)$, we say that $\rho$ satisfies $(\cM_1,P_1)\approx (\cM_2,P_2)$, written $$\rho\models (\cM_1,P_1)\approx (\cM_2,P_2),$$
		if $\cM_1$ and $\cM_2$ produce equal probability distributions in iterations of $P_1$ and $P_2$, respectively; that is, for all $n\geq 0$:
		\begin{equation}\label{eq-prob}
			\mathit{tr}[\mathcal{E}_{10}\circ(\llbracket P_1\rrbracket\circ\mathcal{E}_{11})^n(\mathit{tr}_{\langle 2\rangle}(\rho))] = \mathit{tr}[\mathcal{E}_{20}\circ(\llbracket P_2\rrbracket\circ\mathcal{E}_{21})^n(\mathit{tr}_{\langle 1\rangle}(\rho))]
		\end{equation}
		where $\mathcal{E}_{ij}(\cdot) = M_{ij}(\cdot) M_{ij}^\dag$ for $i = 1,2$ and $j = 0,1$.
	\end{defn}
	
	In the above definition, equation (\ref{eq-prob}) is required to hold
	for all $n\ge0$ (and thus, for infinitely many $n$). But the next
	lemma shows that it can be actually checked within a finite number of
	steps when the state Hilbert spaces are finite-dimensional, as in our
	setting.  Therefore, an algorithm for checking the
	measurement condition $(\cM_1,P_1)\approx (\cM_2,P_2)$ can be
	designed and incorporated into the tools (e.g. theorem prover)
	implementing our logic in the future.
	
	\begin{lem}\label{prop-finite}
		\label{eq-fi} Let $d_i=\dim \mathcal{H}_{P_i}$ $(i=1,2)$. If (\ref{eq-prob}) holds for any $0\le n\le d_1^2+d_2^2-1$, then it holds for all $n\ge0$.
	\end{lem}
	
	Note that a branching structure appears after a measurement is
	performed. To describe it, we introduce the following:
	
	\begin{defn} \label{judgment-measurment} Let $\cM_1 = \{M_{1m}\}$ and $\cM_2 = \{M_{2m}\}$ be as in Definition \ref{meas-eq}, and let $A$ and $B_m$ be quantum predicates in $\mathcal{H}_{P_1\langle 1\rangle}\otimes \mathcal{H}_{P_2\langle 2\rangle}.$
		We define:
		$$\cM_1\approx \cM_2\models  A\Rightarrow \{B_m\}$$
		if for any $\rho\models \cM_1\approx \cM_2$, and for each $m$, there exists a coupling $\sigma_m$ for
		$\big\langle M_{1m}\mathit{tr}_{\langle 2\rangle}(\rho)M_{1m}^\dag,$ $M_{2m}\mathit{tr}_{\langle 1\rangle}(\rho)M_{2m}^\dag\big\rangle$
		such that
		\begin{equation}\label{eqem-judge}\mathit{tr}(A\rho)\le \mathit{tr}\Big(\sum_mB_m\sigma_m\Big).\end{equation}
	\end{defn}

	For the working example \ref{exam-1}, one may check $\cM\approx \cM^\prime \models I\otimes I\Rightarrow \{A_{00},A_{11}\}$ as shown in Fig. \ref{proofoutline-workingexample}. To see this, suppose $\rho\models\cM\approx \cM^\prime$. For $m=0$, $M_0\mathit{tr}_{\langle 2\rangle}(\rho)M_0 = p_0|0\>\<0|$ and $M_0^\prime\mathit{tr}_{\langle 1\rangle}(\rho)M_0^\prime = p_0|+\>\<+|$ with parameter $p_0 = \tr(M_0\mathit{tr}_{\langle 2\rangle}(\rho)M_0)$, and it is straightforward to check $\sigma_0 = p_0|0\>\<0|\otimes|+\>\<+|$ is a witness of lifting $(M_0\mathit{tr}_{\langle 2\rangle}(\rho)M_0){A_{00}}^\#(M_0^\prime\mathit{tr}_{\langle 1\rangle}(\rho)M_0^\prime)$, which leads to $\tr(A_{00}\sigma_0) = p_0$. Similar arguments hold for $m=1$, with $p_1 = \tr(M_1\mathit{tr}_{\langle 2\rangle}(\rho)M_1)$, witness $\sigma_1 = p_1|1\>\<1|\otimes|-\>\<-|$ and $\tr(A_{11}\sigma_1) = p_1$. Observe that $\tr(I\otimes I\rho) = \tr(\rho)$ and $\tr(\rho) = p_0+p_1$. Thus, we conclude that $\tr(I\otimes I\rho) = \tr(A_{00}\sigma_0) + \tr(A_{11}\sigma_1)$.

	A one-side variant of the above definition will also be useful.
	
	\begin{defn}\label{judgment-one-side-measurment}
		Let $\cM_1 = \{M_{1m}\}$, $A$ and $B_m$ be as in Definition \ref{judgment-measurment}. We define $$\cM_1\approx I_2\models  A\Rightarrow \{B_m\},$$ where $I_2$ stands for the identity operator in $\hs_{P_2}$,
		if for any $\rho\in\mathcal{D}^\le\left(\mathcal{H}_{P_1\langle 1\rangle}\otimes \mathcal{H}_{P_2\langle 2\rangle}\right)$,
		and for each $m$, there exist $\rho_{2m}\in\mathcal{D}^\le\left(\mathcal{H}_{P_2\langle 2\rangle}\right)$
		and a coupling $\sigma_m$ for
		$\cp{ M_{1m}\mathit{tr}_{\langle 2\rangle}(\rho)M_{1m}^\dag}{\rho_{2m}}$
		such that $\sum_m\rho_{2m} = \mathit{tr}_{\langle 1\rangle}(\rho)$ and
		\begin{equation}\mathit{tr}(A\rho)\le \mathit{tr}\Big(\sum_mB_m\sigma_m\Big).\end{equation}
	\end{defn}
	
	Similarly, we can define $I_1\approx \cM_2 \models A\Rightarrow \{B_m\}$, where $I_1$ is the identity operator in $\hs_{P_1}$.
	
	Now we are ready to present our new rules for case statements and loops in Fig. \ref{fig 4.5}. As pointed out in the Introduction,
	synchronous rules in non-probabilistic relational Hoare logic RHL
	and probabilistic logic pRHL for control-flow constructs
	(conditionals and loops) require that the two programs under
	comparison execute in lockstep. The control flows of quantum
	programs studied in this paper are determined by the outcome of
	measurements. Thus, measurement conditions $\cM_1\approx \cM_2$ and $(\cM_1,P_1)\approx
	(\cM_2,P_2)$ in our rules (IF1) and (LP1) and their one-side variants are introduced to warrant that the two programs
	execute in lockstep; more precisely, they enter the same branch in
	their control flows with equal probabilities. 
	
	Using rule (IF1), we are able to derive $\cM\approx\cM^\prime\vdash Q_1\sim Q_2: I\otimes I\Rightarrow B$ for our working example, shown in Fig. \ref{proofoutline-workingexample}. 
	Also in Example \ref{Bernoulli}, correctness of Quantum Bernoulli Factory is verified using (LP1) while (LP) is too weak to derive the desired judgment.
	
	\begin{figure*}\centering
		\begin{equation*}\begin{split}
				&({\rm IF1})\ \
				\frac{\cM_1\approx \cM_2\models A\Rightarrow\{B_m\}\quad  \vdash P_{1m}\sim P_{2m}: B_m\Rightarrow C\ {\rm for\ every}\ m}{\begin{array}{cc}\cM_1\approx \cM_2\vdash \mathbf{if}\ (\square m\cdot
						\cM_1[\overline{q}]=m\rightarrow P_{1m})\ \mathbf{fi}  \sim \mathbf{if}\ (\square m\cdot
						\cM_2[\overline{q}]=m\rightarrow P_{2m})\ \mathbf{fi}:
						A\Rightarrow C\end{array}}\\
				&({\rm LP1}) \
				\frac{\cM_1\approx \cM_2\models A\Rightarrow\{B_0,B_1\}\quad \vdash P_{1}\sim P_{2}: B_1\Rightarrow A}{(\cM_1,P_1)\approx (\cM_2,P_2)\vdash \mathbf{while}\
					\cM_1[\overline{q}]=1\ \mathbf{do}\ P_1 \mathbf{od}\sim    \mathbf{while}\
					\cM_2[\overline{q}]=1\ \mathbf{do}\ P_2\ \mathbf{od}:  A\Rightarrow B_0}\\
				&({\rm IF1\text{-}L})\ \ \
				\frac{\cM_1\approx I_2 \models A\Rightarrow\{B_m\} \quad  \vdash P_{1m}\sim P: B_m\Rightarrow C\ {\rm for\ every}\ m}{\vdash \mathbf{if}\ (\square m\cdot
					\cM_1[\overline{q}]=m\rightarrow P_{1m})\ \mathbf{fi} \sim P: A\Rightarrow C}\\
				&({\rm LP1\text{-}L})\ \ \
				\frac{\begin{array}{c}
						\models \mathbf{while}\ \cM_1[\overline{q}]=1\ \mathbf{do}\ P_1\ \mathbf{od} \text{\ lossless} \\
						\cM_1\approx I_2\models A\Rightarrow\{B_0,B_1\} \quad \vdash P_1\sim \mathbf{skip}: B_1\Rightarrow A
				\end{array}}{\vdash \mathbf{while}\
					\cM_1[\overline{q}]=1\ \mathbf{do}\ P_1\ \mathbf{od} \sim \mathbf{skip}:  A\Rightarrow B_0}\\
		\end{split}\end{equation*}
		\caption{More rules for case statements and loops. We omitted the right-sides rules, which are symmetric to the ones displayed here. In (LP1-L), $M_{10}$ and $M_{11}$ only apply on the Hilbert space of the left program, that is, e.g., $M_{10}$ is an abbreviation of $M_{10}\otimes I_2$.}\label{fig 4.5}
	\end{figure*}
	
	{\vskip 3pt}
	
	\textbf{Comparison between Rules (IF), (LP) and (IF1), (LP1)}: A careful comparison between the rules without and with measurement conditions is helpful for us to determine where the rules with measurement conditions are needed. \begin{enumerate}\item First, we notice that the appearance of the special case (IF-w) of (IF) is similar to (IF1). Indeed, whenever the measurement conditions are true and each branch is terminating, then (IF1) degenerates to (IF-w) provided we set: 
		$A =\sum_m \left(M_{1m}^{\dag}\otimes M_{2m}^\dag\right)B_m\left(M_{1m}\otimes M_{2m}\right).$
		However, this choice of $A$ is much weaker than the best possible choice. To see this, suppose $\rho$ is a coupling of inputs that satisfy the premises, and let $\rho_{1m} = M_{1m}\tr_2(\rho)M_{1m}^\dag$ 
		and $\rho_{2m} = M_{2m}\tr_1(\rho)M_{2m}^\dag$
		for all $m$. Actually, (IF-w) uses $\sum_m\rho_{1m}\otimes \rho_{2m}$ as part of the coupling of the output states to derive the precondition. However, this state represents only $1/d$ of the output in general, where $d$ is the dimension of the quantum register being measured. More precisely, the set $\left\{\left(M_{1m}\otimes M_{2m}\right )\right\}_m$ is a part of quantum measurement $\mathcal{M}_1\otimes \mathcal{M}_2 = \left\{\left(M_{1m}\otimes M_{2n}\right)\right\}_{m, n}$ and only contains about $1/d$ measurement operators of $\mathcal{M}_1\otimes\mathcal{M}_2$. Consequently, the trace of the coupling state (probability of occurrence) is smaller than possible, which leads to a weaker precondition.
		
		\item The above defect was remedied in the general rule (IF) by allowing all possible combinations $(m,n)$ rather than only diagonal $(m,m)$.  
		But there is another issue that sometimes prevents (IF) to derive relational properties as strong as those by (IF1). As can be seen in its soundness proof, (IF) simply relates two programs in a manner of tensor product, which does not captures the possible  correlation between these programs. Recall that in probabilistic logic pRHL, coupling was introduced to warrant that two programs be executed in a lockstep manner so that sharing randomness can be achieved. The rule (IF1) is proposed for the same purpose and can be used to reason about stronger relational properties of quantum programs, as shown in Example \ref{exam-1} as well as Example  \ref{Bernoulli}. On the other hand, whenever a strong correlation between two programs does not exist or is unnecessary for our purpose (see for instance, Example \ref{exam-01} - quantum one-time pad), we prefer to use (IF) because it is simpler. 
		
		\item The same argument applies to the rules (LP) and (LP1) for loops.
	\end{enumerate}
	
	\subsection{Separability Conditions}\label{subsec-sep}
	
	We now turn to the structural rules for our logic rqPD. The rules (Conseq), (Weaken) and (Case) of probabilistic relation Hoare logic (pRHL) can be straightforwardly generalised to the quantum setting and are shown in Fig. \ref{fig 4.6}. In (Conseq), we use the L\"owner order
	between quantum predicates (Hermitian operators) in place of boolean implication. The meanings of rules (Weaken) and (Case) are obvious.
	\begin{figure}[h]\centering
		\begin{equation*}\begin{split}
				&({\rm Conseq})\
				\frac{\Gamma\vdash P_1\sim P_2: A^\prime\Rightarrow B^\prime\ \ \ \ A\sqsubseteq A^\prime\ \ \ \ B^\prime\sqsubseteq B}{\Gamma\vdash P_1\sim P_2: A\Rightarrow B}\\
				&({\rm Weaken})\
				\frac{\Gamma\subseteq\Gamma^\prime\quad \Gamma\vdash P_1\sim P_2: A\Rightarrow B}{\Gamma^\prime\vdash P_1\sim P_2: A\Rightarrow B}\\
				&({\rm Case})\
				\frac{\Gamma\vdash P_{1}\sim P_{2}: A_i\Rightarrow B\ (i=1,...,n)\quad \{p_i\}\ {\rm is\ a\ probability\ distribution}}{\Gamma\vdash P_1\sim P_2: \sum_{i=1}^n p_iA_i\Rightarrow B}\\
				&({\rm Frame})\
				\frac{\Gamma\vdash P_1\sim P_2: A\Rightarrow B}{\Gamma\cup\left\{[V, \mathit{var}(P_1,P_2)]\right\} \vdash P_1\sim P_2:A\otimes C\Rightarrow B\otimes C}
		\end{split}\end{equation*}
		\caption{Structural rqPD rules.\ \ \ \ In (Conseq), $\sqsubseteq$ stands for the L\"{o}wner order between operators. In (Frame), $V\bigcap\mathit{var}(P_1,P_2)=\emptyset$ and $C$ is a quantum predicate in $\mathcal{H}_V$.}\label{fig 4.6}
	\end{figure} However, the (Frame) rule
	requires special care.
	
	A typical difficulty in reasoning about a quantum system is entanglement between
	its subsystems. The notions of bipartite separability and entanglement considered in Subsection \ref{sec-entangle} can be generalised to the case of more than two subsystems. A partial density operator $\rho$ in
	$\bigotimes_{i=1}^n\mathcal{H}_i$ is separable between $\mathcal{H}_i$ $(i=1,...,n)$ if there exist partial density operators $\rho_{mi}\in\mathcal{D}^\le(\mathcal{H}_i)$ such that $\rho=\sum_m  \left(\bigotimes_{i=1}^n \rho_{mi}\right)$.
	The following separability condition can be introduced to specify that certain entanglement is not provided (as a resource) or not allowed (e.g. between an adversary and a storage containing sensitive information).
	
	\begin{defn} Let $P_1, P_2$ be two programs and $\Sigma =[X_1,...,X_n]$ a partition of $\mathit{var}(P_1\langle 1\rangle)\cup \mathit{var}(P_2\langle 2\rangle)$. Then we say that a state $\rho\in\mathcal{D}^\le\left(\mathcal{H}_{P_1\langle 1\rangle}\otimes \mathcal{H}_{P_2\langle 2\rangle}\right)$ satisfies separability condition $\Sigma$, written $\rho\models\Sigma$, if $\rho$ is separable between $\mathcal{H}_{X_i}$ $(i=1,...,n)$.
	\end{defn} 
	
	With the above definition, we can define the (Frame) rule for quantum programs in Fig. \ref{fig 4.6} where a separability condition between the programs $P_1,P_2$ and the new predicate $C$.
	Recall that in probabilistic logic pRHL, the frame rule allows an assertion $C$
	to be carried from the precondition through to the postcondition. The validity
	of the frame rule is based on the assumption that the two programs $P_1$ and
	$P_2$ cannot modify the (free) variables in $C$; or mathematically speaking,
	$\mathit{var}(P_1,P_2)\cap\mathit{var}(C)=\emptyset$. In the quantum setting,
	however, the syntactic disjointness between $\mathit{var}(P_1,P_2)$ and $\mathit{var}(C)$ is not enough. Indeed, an entanglement can occur between them even if they are disjoint, and some properties of the subsystem denoted by the variables in $C$ can be changed by certain actions, say measurements of $P_1$ or $P_2$.
	So, the separability condition $\Gamma =[V,\mathit{var}(P_1,P_2)]$ is introduced in the conclusion part of the frame rule to exclude such an entanglement between $P_1,P_2$ and $C$, where $V$ is the set of quantum variables appearing in $C$.
	
	\subsection{Entailment between Side-Conditions}
	
	In the above two subsections, measurement and separability conditions are introduced into our logic rqPD. But the (SC) rule for sequential composition in Fig. \ref{fig 4.4-0} does not contain these conditions. 
	So, it must be accordingly strengthened to accommodate the propagation of measurement and separability conditions. To this end, we need the following:
	
	\begin{defn}\label{def-entail} Let $\Gamma,\Delta$ be two sets of measurement or separability conditions, and $P_1,P_2$ two programs. We say that $\Delta$ is couple-entailed by $\Gamma$ through $(P_1,P_2)$, written \begin{equation}\label{entail-form}\Gamma\stackrel{(P_1,P_2)}{\models}\Delta,\end{equation} if for any $\rho\models\Gamma$, whenever $\sigma$ is a coupling for $\cp{\sm{ P_1}\left(\mathit{tr}_{\<2\>}(\rho)\right)}{\sm{ P_2}\left(\mathit{tr}_{\<1\>}(\rho)\right)}$, then $\sigma\models\Delta$.
	\end{defn}
	
	Using the above definition, a strengthened version of (SC) is presented as rule (SC+) in Fig. \ref{fig 4.7}. 
	\begin{figure}\begin{align*}
			&({\rm SC+})\quad \frac{\Gamma\vdash P_1\sim P_2:A\Rightarrow B\qquad \Delta\vdash P_1^\prime\sim P_2^\prime:B\Rightarrow C\qquad \Gamma\stackrel{(P_1,P_2)}{\models}\Delta}{\Gamma\vdash P_1;P_1^\prime\sim P_2;P_2^\prime:A\Rightarrow C}
		\end{align*}
		\caption{Strong sequential rule}\label{fig 4.7}
	\end{figure}

	Let us go back to the working example \ref{exam-1}. 
	After a direct calculation, we are able to show
	$$\emptyset\stackrel{(q:=|0\>,q:=|0\>)}{\models}\cM^\prime\approx\cM^\prime\quad\text{and}\quad \cM^\prime\approx\cM^\prime\stackrel{(q:=H[q],{\bf skip})}{\models}\cM\approx\cM^\prime.$$
	Now, using rule (SC+) we can combine all the segment judgments shown in last few sections together to reason about judgment (\ref{ex-1.0}), as shown in Fig. \ref{proofoutline-workingexample}. 

	In general, it is not easy to use this rule because it is hard to check a side condition of the form (\ref{entail-form}).  
	However, in most of the applications, we do not need the full power of (SC+) because most of conditionals and loops can be dealt with by using (IF) and (LP) where no measurement conditions are introduced and thus $\Gamma=\Delta=\emptyset$, and 
	(IF1) and (LP1) are only employed for a few times. In particular, if we only need (IF1) or (LP1) to reason about a single conditional or loop, then $\Delta=\emptyset$ and side condition (\ref{entail-form}) is trivially valid; for instance, Example \ref{Bernoulli} is actually this case. The difficulty of applying (SC+) will arise only when (IF1) and (LP1) are needed to reason about many conditionals and loops. In the case of finite-dimensional state Hilbert spaces, for a large class of quantum programs, 
	this difficulty can be significantly eased by a back-tracking strategy for collecting a set of measurement or separability conditions at the beginning of a sequence of consecutive judgments in order to warrant that all side-conditions of the form (\ref{entail-form})  are valid. An elaboration of this strategy is given in \cite{fullversion}. 
	
	\subsection{Auxiliary Rules for General Quantum Operations}
	
	Our programming language only contains two simple kinds of quantum operations: unitary transformations and quantum measurements. Also, the post-measurement states are recorded in the semantics so that the dimension of the state Hilbert space is fixed. In applications, however, it is often convenient to apply general quantum operations; for example, quantum noises and communication channels.
	Principally,  a general quantum operation can be implemented by unitary transformations and measurements through introducing ancilla systems and discarding post-measurement states (see \cite{NC00}, Section 8.2.2 for the system-environment model). But the implementations are usually quite complicated. So, for convenience, we choose to expand the programming language by explicitly adding program constructs of the form:
	\begin{equation}\label{super-pro}P::=\overline{q}:=\mathcal{E}[\overline{q}]
	\end{equation} where $\mathcal{E}$ is a general quantum operation. Mathematically, $\mathcal{E}$ is modelled as trace-preserving super-operator with Hilbert space $\mathcal{H}_{\overline{q}}$ as its domain, but its codomain can be a different Hilbert space, even with a different dimension, e.g.  $\mathcal{H}_{\overline{q}\setminus\overline{q^\prime}}$ or $\mathcal{H}_{\overline{q}\cup\overline{q^\prime}}$. It is well-known that a super-operator $\mathcal{E}$ can be represented by a set of operators $\{E_i\}$ (Kraus operator-sum representation):
	$ \mathcal{E}(\rho)=\sum_i E_i\rho E_i^\dag$ for every state $\rho$ in $\mathcal{H}_{\overline{q}}$.
	
	The operational semantics of program $\overline{q}:=\mathcal{E}[\overline{q}]$ is defined by the following transition rule:
	$$\langle \overline{q}:=\mathcal{E}[\overline{q}],\rho\rangle\rightarrow\langle\downarrow,\mathcal{E}(\rho)\rangle.$$ Based on this, the denotational semantics of quantum programs containing super-operators can be defined in the same way as Definition \ref{den-sem-def}, provided allowing that the domain and codomain of the semantic function of a quantum program can be different.
	
	We present three inference rules for general quantum operations in Fig. \ref{fig 5.1}, generalising rules (UT), (UT-L) and (UT-R), respectively.
	
	\begin{figure}[h]\centering
		\begin{equation*}\begin{split}
				&({\rm SO})\ \vdash \overline{q_1}:=\mathcal{E}_1\left[\overline{q_1}\right]\sim \overline{q_2}:=\mathcal{E}_2\left[\overline{q_2}\right]: \left(\mathcal{E}_1^\ast\otimes \mathcal{E}_2^\ast\right) (A)\Rightarrow A\\
				&({\rm SO\text{-}L})\ \ \ \vdash \overline{q_1}:=\mathcal{E}_1\left[\overline{q_1}\right]\sim \mathbf{skip}: \mathcal{E}_1^\ast (A)\Rightarrow A
		\end{split}\end{equation*}
		\caption{Rules for trace-preserving super-operators (quantum operations). We use $\mathcal{E}^\ast$ to denote the dual of super-operator $\mathcal{E}$; that is, $\mathcal{E}^\ast(A)=\sum_i E_i^\dag A E_i$ if $\mathcal{E}(\rho)=\sum_i E_i\rho E_i^\dag$. Rule (SO-R) is symmetric to (SO-L).}\label{fig 5.1}
	\end{figure}
	
	\begin{rem}It is worth mentioning that allowing different dimensions of the domain and codomain of $\mathcal{E}$ in (\ref{super-pro}) has a benefit; that is, it enables us to introduce auxiliary quantum variables or discard a quantum variable.
		The construct of introducing auxiliary variables can be defined as cylinder extension, i.e. tensor product with the identity operator of the state Hilbert space of the auxiliary variables (divided by its dimension for normalisation), and the construct of discarding a quantum variable $q\in\overline{q}$ is indeed included in Selinger's quantum programming language QPL \cite{Selinger04}. It can be defined as a partial trace $\mathbf{Tr}[q]$, with its semantics described as a special super-operator: $\mathcal{E}(\rho)=\sum_i \langle i|\rho|i\rangle\in\mathcal{D}(\mathcal{H}_{\overline{q}\setminus\{q\}})$ for any $\rho\in\mathcal{D}(\mathcal{H}_{\overline{q}})$, where $\{|i\rangle\}$ is an orthonormal basis of $\mathcal{H}_q$. Then rules (SO) and its one-side variants (SO-L), (SO-R) warrant that introducing auxiliary variables and discarding a variable can be safely done in relational reasoning.   
	\end{rem}
	
	\subsection{Soundness Theorem}
	
	We can prove that our proof system is sound with respect to validity
	of judgments. The proof of soundness is given in the complete version of this paper \cite{fullversion}.
	\begin{thm}[Soundness]\label{soundness-thm}
		Derivable judgments are valid:
		$$\Gamma\vdash P_1\sim P_2:A
		\Rightarrow B\ \Longrightarrow \ \Gamma\models P_1\sim P_2:A
		\Rightarrow B$$
	\end{thm}
	
	Completeness of relational logics is a challenging problem. In the deterministic setting, relational Hoare logic can be shown to be relatively complete for terminating programs provided it includes sufficiently many one-sided rules. Relative completeness fails for probabilistic relational Hoare logic; a further potential complication is that the coupling method --- upon which probabilistic relational Hoare logic builds --- is itself not complete for proving convergence Markov chains \cite{AKR01}.
	
	\section{Examples}\label{sec-examples}
	
	In this section, we give several examples to illustrate the power of rqPD. We
	mainly show how their relational properties can be formally specified
	as rqPD judgments. Due to the limited space, only a proof outline of the
first example is given, and the formal derivations of other judgments are
deferred to \cite{fullversion}.
	
	\subsection{Symmetry between Simple Programs}
	Let us start from our working example \ref{exam-1}. 
	In last section, we already proved judgment (\ref{ex-1.0}) in our logic rqPD. 
	A symmetry between programs
	$P_1$ and $P_2$ modelled by judgment (\ref{ex-1.0}) is more interesting than their similarity we can observe at the first glance. It worth noting that two different kinds
	of \textquotedblleft quantum equality\textquotedblright\ are used in
	the precondition and postcondition. To understand this judgment better, let recall Proposition \ref{prop-equal}, which give us an intuition of what $=_\mathcal{B}$ and $=_\mathit{sym}$ mean. The judgment tell us that, if the inputs of $P_1$ and $P_2$ are the same, then the outputs are also same, or in other words, program $P_1$ and $P_2$ are actually equal.
	
	\begin{rem}
		As discussed before, rule (IF1) is necessary to derive judgment (\ref{ex-1.0}) while using rule (IF-w) or more general (IF) is impossible to prove it. A strong correlation between two programs can be detected only if we run them in a lockstep manner. This is why rule (IF1) works. However, (IF) only requires that two programs run simultaneously while lockstep is not guaranteed. Therefore, it is not surprising to see (IF) fails here.
	\end{rem}

	\subsection{Uniformity}\label{sec-uniform}
	
	An elegant characterisation of uniform probability distribution with
	coupling was given in \cite{BartheEGHS17}. Unfortunately, the
	characterisation does not directly carry
	over to the quantum setting. In this subsection, we show how an
	alternative approach based on quantum coupling can be used to describe
	uniformity in quantum systems.  Let $\mathcal{H}$ be a Hilbert space
	and $\mathcal{B}=\{|i\rangle\}$ be an orthonormal basis of
	$\mathcal{H}$. For each $i$, we write $M_i=|i\rangle\langle i|$. Then
	the measurement in basis $\mathcal{B}$ is defined as
	$\cM_\mathcal{B}=\{M_i\}$.
	
	\begin{defn}
		A density operator $\rho$ in $\mathcal{H}$ with $d=\dim\mathcal{H}$ is called
		\emph{uniform} in basis $\mathcal{B}$ if the outcome of measurement
		$\cM_\mathcal{B}$ on $\rho$ is uniformly distributed; i.e. for every $i$,
		$$p_i=\mathit{tr}(M_i\rho)=\langle i|\rho|i\rangle=\frac{1}{d}.$$
	\end{defn}
	
	The following proposition gives a characterization of uniformity of a program's outputs in terms of quantum coupling.
	
	\begin{prop}[Uniformity by coupling]\label{uniformity}
		Let $P$ be a quantum program, $\mathcal{B}=\{|i\rangle\}$ be an orthonormal
		basis of $\mathcal{H}_P$ and $d=\dim\mathcal{H}_P$. Then the following three
		statements are equivalent:
		\begin{enumerate}
			\item for any input density operator $\rho$ in $\mathcal{H}_P$, output $\llbracket P\rrbracket (\rho)$ is uniform in basis $\mathcal{B}$;
			\item for any basis state $|i\rangle$ in $\mathcal{B}$, \begin{equation}\label{e-unif}\models P\sim P: \frac{I\otimes I}{d}\Rightarrow |i\rangle\langle i|\otimes I\end{equation}
			where $I$ is the identity operator in $\mathcal{H}_P$;
			\item for any basis state $|i\rangle$ in
			$\mathcal{B}$, \begin{equation}\label{eq-unif-1}[\mathit{var}(P\langle
				1\rangle),\mathit{var}(P\langle 2\rangle)]\models P\sim P \mathrel{:}
				=^e_\mathcal{C}\ \Rightarrow |i\rangle\langle i|\otimes
				I\end{equation} where $\mathcal{C}=\{|j\rangle\}$ is an arbitrary
			orthonormal basis of $\mathcal{H}_P$, and the
			equality operator $\mathit{=}^e_\mathcal{C}$   is defined to be
			$|\Phi\rangle\langle \Phi|$, where $\Phi$ is the maximally entangled state
			$|\Phi\rangle=\frac{1}{\sqrt{d}}\sum_j |jj\rangle.$
			More precisely, $\mathit{=}^e_\mathcal{C}$ is (the projection onto)
			the one-dimensional subspace spanned by the maximally entangle state
			$|\Phi\rangle$. It is interesting to see that in judgment
			(\ref{eq-unif-1}), a separability condition is enforced on inputs,
			but entanglement appears in the precondition
			$\mathit{=}^e_\mathcal{C}$. This suggests that entanglement
			cannot be avoid in such a characterisation of uniformity.
	\end{enumerate}\end{prop}
	
	Now, we consider uniformity for a concrete quantum protocol. The Bernoulli factory (BF) \cite{CBF} is a protocol for random number generation. It uses a coin with an unknown probability $p$ of heads to simulate a new coin that has probability $f(p)$ of heads for a given function $f:[0,1]\to[0,1]$. The Quantum Bernoulli factory (QBF) \cite{DJR15} also generates classical randomness (e.g., a biased coin with probability $f(p)$), but it uses quantum coins instead of classical coins. Interestingly, QBF can simulate a strictly larger class of functions $f$ than those simulated by BF. As a direct application of the above proposition, we can verify a simplified version of quantum Bernoulli factory in our logic.
	
	\begin{exam}[Simplified Quantum Bernoulli Factory]\label{Bernoulli}
		Suppose we have a two-qubit system with state Hilbert space $\mathcal{H}_{q_x}\otimes\mathcal{H}_{q_y}$ and an initial state $|0\rangle_{q_x}|0\rangle_{q_y}$. We are allowed to perform projective measurement $\cM = \{M_0,M_1\}$:
		\begin{align*}
			M_0 = |0\>_{q_x}\<0|\otimes|1\>_{q_y}\<1|+|1\>_{q_x}\<1|\otimes|0\>_{q_y}\<0|,\quad 
			M_1 = |0\>_{q_x}\<0|\otimes|0\>_{q_y}\<0|+|1\>_{q_x}\<1|\otimes|1\>_{q_y}\<1|
		\end{align*}
		and apply a given---but unknown---one-qubit unitary transformation $U$ such that $0<|\<0|U|0\>|<1$ on system $x$ or $y$.\footnote{In the classical BF, this condition means that the coin must be non-trivial---it cannot always return $0$ or always return $1$.} How can we produce the uniform state $\frac{1}{2}I_{q_x}$? The following quantum program accomplishes this task:
		\begin{align*}
			{\rm QBF}\equiv\ &q_x:=|0\rangle;\ q_y:=|0\rangle;\ \mathbf{while}\ \cM[q_x,q_y]=1\ \mathbf{do} \  q_x:=U[q_x];\ q_y:=U[q_y] \  \mathbf{od}; \ \mathbf{Tr}[q_y]
		\end{align*}
		where ${\bf Tr}$ stands for the partial trace over system $q_y$.
	\end{exam}
	
	Note that state $\frac{1}{2}I_{q_x}$ is the only density operators being uniform in any orthonormal basis $\mathcal{B}$. With Proposition \ref{uniformity}, QBF can be verified by proving that for any $|\psi\>\in \mathcal{H}_{q_x}$:
	\begin{equation}
		\label{ju-uniform}
		\models {\rm QBF}\sim{\rm QBF}: \frac{1}{2}I_{q_x}\otimes I_{q_y}\otimes I_{q_x^\prime}\otimes I_{q_y^\prime} \Rightarrow |\psi\>_{q_x}\<\psi|\otimes I_{q_x^\prime}.
	\end{equation}
	Since this judgment is valid for all $|\psi\>$, the output is uniform in all basis so the output state must be $\frac{1}{2}I_{q_x}$.
	
	It is worth pointing out that rule (LP1) plays an essential role in the proof. All registers are initialised before the loop and therefore, we are able to run two of the same QBF in a lockstep manner. Thus, rule (LP) is too weak to derive judgment (\ref{ju-uniform}).
	Rule (SO) is also needed in the verification of (\ref{ju-uniform}) because $\mathbf{Tr}[q_y]$ appears at the end of QBF. Indeed, if we do not trace out system $q_y$ at the end, then QBF outputs the Bell state $\frac{1}{\sqrt{2}}(|0\>_{q_x}|1\>_{q_y}+|1\>_{q_x}|0\>_{q_y}).$ This fact can also be realized in our logic. Moreover, it implies that our program QBF is not a trivial generalisation of classical Bernoulli factory because it is capable of producing the maximally entangled state.
	
	\subsection{Quantum Teleportation}
	
	Now we consider a more sophisticated example. Quantum teleportation \cite{Tel93} is arguably the most famous quantum communication protocol. With it, quantum information (e.g. the exact state of an atom or a photon) can be transmitted from one location to another, only through classical communication, but with the help of previously shared entanglement between the sender and the receiver. The correctness of quantum teleportation has been formally verified by several different methods in the literature, e.g. using categorical formalism of quantum mechanics \cite{AB04}. Our logic provides a new way for verifying the correctness of quantum teleportation; more importantly, it can be used to verify the reliability of quantum teleportation against various kinds of quantum noise. To the best of our knowledge, this is the first formal verification of its reliability.
	
	\begin{exam}Suppose that Alice possesses two qubits $p,q$ and Bob possesses qubit $r$, and there is entanglement, i.e. the EPR (Einstein-Podolsky-Rosen) pair: $|\beta_{00}\rangle=\frac{1}{\sqrt{2}}(|00\rangle+|11\rangle)$ between $q$ and $r$. Then Alice can send an arbitrary qubit state $|\psi\rangle=\alpha|0\rangle+\beta|1\rangle$ to Bob, i.e. from $p$ to $r$, by two-bit classical communication (for detailed description, see \cite{NC00}, Section 1.3.7). If we regard $p$ as the input state and $r$ the output state, then this protocol can be modeled by a quantum program:
		\begin{align*}
			{\rm QTEL}\equiv\ &q:=|0\rangle;\ r:=|0\rangle;\ q:=H[q];\ q,r:= {\rm CNOT}[q,r];\ p,q:=\ {\rm CNOT}[p,q];\ 
			p:=H[p];\\
			&\mathbf{if}\ (\cM[q]= 0\rightarrow \mathbf{skip}\ \square\ 1\rightarrow r:=X[r])\ \mathbf{fi};\\ 
			&\mathbf{if}\ (\cM[p]=0\rightarrow \mathbf{skip}\ \square\ 1\rightarrow r:=Z[r])\ \mathbf{fi}
		\end{align*}
		where $H$ is the Hadamard gate, $X$ and $Z$ are the Pauli gates, ${\rm CNOT}$ is the controlled-NOT:
		$${\rm CNOT}=\left(\begin{array}{cc}I & 0\\ 0 & X\end{array}\right),$$ $I, 0$ are the $2\times 2$ unit and zero matrices, respectively, and $\cM$ is the measurement in the computational basis, i.e. $\cM=\{M_0,M_1\}$, where $M_0=|0\rangle\langle 0|$, $M_1=|1\rangle\langle 1|$.\end{exam}
	
	\subsubsection{Correctness of Quantum Teleportation}
	\label{sec-correct-QTEL}
	In this subsection, we show how our logic can be used to verify the correctness of quantum teleportation. The correctness of QTEL can be described as the judgment: \begin{equation}\label{ju-tele}\models{\rm QTEL\sim \bf{skip}}: (=_\mathcal{B}) \Rightarrow (=_\mathcal{B}),\end{equation} where $\mathcal{B}=\{|\psi\rangle,|\phi\rangle\}$ is an arbitrary othornormal basis of the state Hilbert space of a qubit,  and $(=_\mathcal{B}) = |\psi\rangle|\psi\rangle\langle\psi|\langle\psi| + |\phi\rangle|\phi\rangle\langle\phi|\langle\phi|$ is the projector onto the subspace $\spans\{|\psi\rangle|\psi\rangle,|\phi\rangle|\phi\rangle\}$ [see Example \ref{exam-lifting} (2)].
	Indeed, for any input states $\rho$, there always exists an orthonormal basis $\mathcal{B}$ such that $\rho(=_{\mathcal{B}})^\#\rho$, and we assume that a witness for this lifting is $\sigma$. From judgment (\ref{ju-tele}), we know that there exists a coupling $\sigma^\prime$ for $\cp{\sm{ {\rm QTEL}}(\rho)}{\sm{ {\bf skip}}(\rho)}$ such that
	$\mathit{tr}(=_\mathcal{B}\sigma^\prime)\ge\mathit{tr}(=_\mathcal{B}\sigma) = 1.$
	So, $\sigma^\prime$ is a witness of lifting: $\llbracket {\rm QTEL}\rrbracket(\rho)(=_{\mathcal{B}})^\#\llbracket {\bf skip}\rrbracket(\rho),$ which, together with Proposition \ref{prop-equal}, implies $\llbracket {\rm QTEL}\rrbracket(\rho) = \llbracket {\bf skip}\rrbracket(\rho) = \rho.$
	
	Interestingly, the correctness of QTEL can also be described as the following judgment: \begin{equation}\label{ju-tele1}\models{\rm QTEL\sim \bf{skip}}: (=_\mathit{sym}) \Rightarrow (=_\mathit{sym})\end{equation} using a different equality $=_\mathit{sym}$, that is, the projector onto the symmetric subspace [see Example \ref{exam-lifting} (3)]. A similar argument shows that for any input $\rho$, we have $\llbracket {\rm QTEL}\rrbracket(\rho) = \llbracket {\bf skip}\rrbracket(\rho) = \rho$.
	
	The proof of these two judgments are somewhat easy. Unlike the previous two examples, the basic construct-specific rule (IF-L) is enough to derive the results. 
	
	\subsubsection{Reliability of Quantum Teleportation}\label{tele-reli}
	
	In this subsection, we further show that our logic can be used to deduce not only correctness but also reliability of quantum teleportation when its actual implementation suffers certain physical noise.
	
	Quantum noise are usually modelled by super-operators, a more general class of quantum operations than unitary transformations.
	
	\begin{exam}[Noise of Qubits, \cite{NC00}, Section 8.3]
		\begin{enumerate}
			\item[]
			\item The bit flip noise flips the state of a qubit from $|0\rangle$ to $|1\rangle$ and vice versa with probability $1-p$, and can be modelled by super-operator:
			\begin{equation}\label{flips}\mathcal{E}_{\it BF}(\rho)=E_0\rho E_0+E_1\rho E_1\end{equation} for all $\rho$, where
			\begin{align*}
				E_0 = \sqrt{p}I=\sqrt{p}\left(\begin{array}{cc}1 &0\\ 0 &1\end{array}\right) \quad
				E_1 = \sqrt{1-p}X=\sqrt{1-p}\left(\begin{array}{cc}0 &1\\ 1 &0\end{array}\right) .
			\end{align*}
			\item The phase flip noise can be modelled by the super-operator $\mathcal{E}_{\it PF}$ with
			\begin{align*}
				E_0 = \sqrt{p}I=\sqrt{p}\left(\begin{array}{cc}1 &0\\ 0 &1\end{array}\right) \quad
				E_1 = \sqrt{1-p}Z=\sqrt{1-p}\left(\begin{array}{cc}1 &0\\ 0 & -1\end{array}\right) .
			\end{align*}
			\item The bit-phase flip noise is modelled by the super-operator $\mathcal{E}_{\it BPF}$ with
			\begin{align*}
				E_0 = \sqrt{p}I=\sqrt{p}\left(\begin{array}{cc}1 &0\\ 0 &1\end{array}\right) \quad
				E_1 = \sqrt{1-p}Y=\sqrt{1-p}\left(\begin{array}{cc}0 &-i\\ i & 0\end{array}\right) .
			\end{align*}
			where $X,Y, Z$ are Pauli matrices.
	\end{enumerate}\end{exam}
	
	We sometimes write the bit flip super-operators as $\mathcal{E}_{\it BF}(p)$ in order to explicitly specify the flip probability $p$. The same convention is applied to the phase flip $\mathcal{E}_{\it PF}$ and bit-phase flip $\mathcal{E}_{\it BPF}$.
	
	\begin{exam} If the bit flip noise occurs after the Hadamard gates on both qubit $p$ and $q$, then the teleportation programs becomes:
		\begin{align*}
			{\rm QTEL}_\mathit{BF}\equiv\ &q:=|0\rangle;\ r:=|0\rangle;\ 
			q:=H[q]; q:=\mathcal{E}_{\it BF}[q];\ 
			q,r:=\ {\rm CNOT}[q,r];\\&p,q:=\ {\rm CNOT}[p,q];\ 
			p:=H[p]; p:=\mathcal{E}_{\it BF}[p];\\
			&\mathbf{if}\ (\cM[q]=0\rightarrow \mathbf{skip}\ \square\ 1\rightarrow r:=X[r])\ \mathbf{fi};\\
			&\mathbf{if}\ (\cM[p]=0\rightarrow \mathbf{skip}\ \square\ 1\rightarrow r:=Z[r])\ \mathbf{fi}
		\end{align*}
		where $\mathcal{E}_{\rm BF}$ is the bit flip super-operator.
		Moreover, we write ${\rm QTEL}_\mathit{PF}$ and ${\rm QTEL}_\mathit{BPF}$ for the phase flip and  bit-phase flip occurring at the same positions.
	\end{exam}
	
	Now the reliability of QTEL with the different noises---bit flip, phase flip and bit-phase flip---is modelled by the judgments:
	\begin{align}\label{ju-bit}&\models{\rm QTEL}_\mathit{BF} \sim {\rm QTEL}:
		\mathcal{E}_{\it PF}(p)(|\psi\rangle_p\langle\psi|)\otimes|\psi\rangle_{p^\prime}\langle\psi| \Rightarrow |\psi\rangle_r\langle\psi|\otimes|\psi\rangle_{r^\prime}\langle\psi|,\\
		\label{ju-phase}&\models{\rm QTEL}_\mathit{PF} \sim {\rm QTEL}:
		\mathcal{E}_{\it PF}(p)(|\psi\rangle_p\langle\psi|)\otimes|\psi\rangle_{p^\prime}\langle\psi| \Rightarrow |\psi\rangle_r\langle\psi|\otimes|\psi\rangle_{r^\prime}\langle\psi|,\\
		\label{ju-bp}&\models{\rm QTEL}_\mathit{BPF} \sim {\rm QTEL}: \mathcal{E}_{\it PF}(p^2+(1-p)^2)(|\psi\rangle_p\langle\psi|)\otimes|\psi\rangle_{p^\prime}\langle\psi|
		\Rightarrow |\psi\rangle_r\langle\psi|\otimes|\psi\rangle_{r^\prime}\langle\psi|.\end{align}
	
	To understand these judgments better, let us choose pure state $|\psi\>$ as the input of both QTEL$_{BF}$ and QTEL as an example. The correctness of QTEL has been verified and therefore, $\llbracket {\rm QTEL} \rrbracket(|\psi\>\<\psi|) = |\psi\>\<\psi|.$ We assume that $\llbracket {\rm QTEL}_\mathit{BF} \rrbracket(|\psi\>\<\psi|) = \rho$. There exists a unique coupling $\rho\otimes|\psi\>\<\psi|$ for the outputs $\<\rho,|\psi\>\<\psi|\>$, and according to the judgment (\ref{ju-bit}), we know that:
	\begin{align*}
		\tr(\mathcal{E}_{\it PF}(p)(|\psi\rangle_p\langle\psi|)\otimes|\psi\rangle_{p^\prime}\langle\psi|\cdot |\psi\>_p\<\psi|\otimes|\psi\rangle_{p^\prime}\langle\psi|)
		\le \tr(|\psi\rangle_r\langle\psi|\otimes|\psi\rangle_{r^\prime}\langle\psi|\cdot\rho\otimes|\psi\>_{r^\prime}\<\psi|);
	\end{align*}
	that is, $\<\psi|\rho|\psi\>\ge p+(1-p)|\<\psi|Z|\psi\>|^2.$ Whenever $p$ is close to 1, then $\rho$ is also close to $|\psi\>\<\psi|$, and this is what reliability actually means.

	Judgemens (\ref{ju-bit}), (\ref{ju-phase}) and (\ref{ju-bp}) can be verified in our logic rqPD using rule (SO-L) for general quantum operations and rule (IF) for all pairwise comparisons.
	
	\subsection{Quantum One-Time Pad}\label{sec-QOTP}
	
	In this subsection, we show that our logic can be used to specify and verify correctness and security of a basic quantum encryption scheme, namely the quantum one-time pad (QOTP) \cite{MTW00, BR03}. Similar to the classical one-time pad, it uses a one-time pre-shared secret key to encrypt and decrypt the quantum data, providing the information-theoretic security. We first consider the simplest case, for protecting one-qubit data.  
	
	\begin{exam}\label{exam-01}The QOTP scheme includes three parts: key generation $\mathbf{KeyGen}$, encryption $\mathbf{Enc}$ and decryption $\mathbf{Dec}$, which can be written as programs:
		\begin{align*}
			\mathbf{KeyGen} \equiv\ &a := |0\>; b := |0\>; \ 
			 a := H[a];  b := H[b];  \\
			& \mathbf{if}\ \cM[a,b]=00\rightarrow \mathbf{skip}\ \square\ \ \! 01\rightarrow \mathbf{skip}\ \\
			& \qquad\qquad\  \square\ \ \!  10\rightarrow \mathbf{skip}\ \square\ \ \!  11\rightarrow \mathbf{skip}\\ 
			& \mathbf{fi} \\
			\mathbf{Enc} \equiv \mathbf{Dec} \equiv\ & \mathbf{if}\ \cM[a,b]=00\rightarrow \mathbf{skip}\ \ \ \ \ \ \ \ \;\!\square \ \ \!  01\rightarrow p = Z[p]\\
			&\qquad\qquad\ \square\ \ \!  10\rightarrow p = X[p]\ \square\ \ \!  11\rightarrow p = Z[p];p = X[p]\\ &\mathbf{fi} \\
			\mathbf{DisKey} \equiv\ & {\bf Tr}[a];{\bf Tr}[b]
		\end{align*}
		Here, registers $a$ and $b$ are used as the secret key, and measurement
		$$\cM = \{M_{00} = |00\>_{ab}\<00|, M_{01} = |01\>_{ab}\<01|, M_{10} = |10\>_{ab}\<10|, M_{11} = |11\>_{ab}\<11|\}$$
		is introduced to detect the value of secret key, which has two-bit classical outcome. Register $p$ is the input quantum data which we want to protect. $H$ is the Hadamard gate and $X,Z$ are Pauli gates as usual. As the secret key is not considered when analysing the correctness and security of the protocol, we further introduce $\mathbf{DisKey}$ to discard the key.\end{exam}
	
	\subsubsection{Correctness of Quantum One-Time Pad} The correctness of QOTP can be formulated as the following judgment:
	\begin{align}
		\label{QOTP-1-Correct}
		\vdash \mathbf{KeyGen};\mathbf{Enc};\mathbf{Dec};\mathbf{DisKey}\ \sim\ {\bf skip}:  (=_{sym}) \Rightarrow (=_{sym}).
	\end{align}
	where $=_{sym}$ represents the projector onto the symmetric space.
	By an argument similar to that for the correctness of quantum teleportation in Section \ref{sec-correct-QTEL}, we can show that if judgment (\ref{QOTP-1-Correct}) is valid, then for any possible input $\rho$ of register $p$, $$\sm{\mathbf{KeyGen};\mathbf{Enc};\mathbf{Dec};\mathbf{DisKey}}(\rho) = \sm{\bf skip}(\rho) = \rho;$$
	that is, the input state is the same as the output after QOTP.
	
	Judgment (\ref{QOTP-1-Correct}) can be verified mainly with rule (IF-L). But note that the initialisations of registers $a$ and $b$ are regarded as the creation of new local qubits. So, rules (SO-L) and (SO-R) are needed here instead of (Init-L).
	
	\subsubsection{Security of Quantum One-Time Pad} Using the characterisation of uniformity given in Section \ref{sec-uniform}, we may verify the security of QOTP by proving that for any $|\psi\>\in\h_{p}$:
	\begin{align}
		\label{QOTP-1-Secure}
		\vdash \mathbf{KeyGen};\mathbf{Enc};\mathbf{DisKey}\ \sim\ \mathbf{KeyGen};\mathbf{Enc};\mathbf{DisKey}: \frac{I_p\otimes I_{p^\prime}}{2}\Rightarrow |\psi\>_p\<\psi|\otimes I_{p^\prime}.
	\end{align}
	In fact, if judgment (\ref{QOTP-1-Secure}) is valid for all $|\psi\>$, then the output is uniform in all bases; that is, the output of register $p$ is $\frac{1}{2}I_p$, which is actually the maximally mixed state and nothing can be inferred from it. To derive this  judgment, we need rule (SO) and (IF-w).
	
	\subsubsection{General Quantum One-Time Pad} Now let us generalise Example \ref{exam-01} to the general case for protecting $n$-qubit data. In this case, QOTP can be written as the following quantum program:
	\begin{align*}
		\mathbf{KeyGen}(n) \equiv\ &a_1 := |0\>; \cdots; a_n := |0\>;  b_1 := |0\>; \cdots; b_n := |0\>; \\
		&a_1 := H[a_1]; \cdots; a_n := H[a_n]; b_1 := H[b_1]; \cdots; b_n := H[b_n];  \\
		& \mathbf{if}\ (\square x_1z_1\cdot \cM[a_1,b_1]=x_1z_1\rightarrow \mathbf{skip})\ \mathbf{fi}; \\[-0.2cm]
		&\vdots \\
		& \mathbf{if}\ (\square x_nz_n\cdot \cM[a_n,b_n]=x_nz_n\rightarrow \mathbf{skip})\ \mathbf{fi} \\
		\mathbf{Enc}(n) \equiv \mathbf{Dec}(n) \equiv\ & \mathbf{if}\ (\square x_1z_1\cdot \cM[a_1,b_1]=x_1z_1\rightarrow p_1 = Z^{z_1}[p_1];\ p_1 = X^{x_1}[p_1])\ \mathbf{fi}; \\[-0.2cm]
		&\vdots \\
		& \mathbf{if}\ (\square x_nz_n\cdot \cM[a_n,b_n]=x_nz_n\rightarrow p_n = Z^{z_n}[p_n];\ p_n = X^{x_n}[p_n])\ \mathbf{fi} \\
		\mathbf{DisKey}(n) \equiv\ & {\bf Tr}[a_1];\cdots; {\bf Tr}[a_n];{\bf Tr}[b_1]\cdots; {\bf Tr}[b_n]
	\end{align*}
	Again, if we regard register $\bar{p} = p_1,\cdots,p_n$ as the input and output of QOTP and consider the trivial program $\bf skip$ with the duplicated register $\bar{p}^\prime = p_1^\prime,\cdots,p_n^\prime$, then the judgment 
	\begin{align}
		\label{QOTP-n-Correct}
		\vdash \mathbf{KeyGen}(n);\mathbf{Enc}(n);\mathbf{Dec}(n);\mathbf{DisKey}(n)\ \sim\ {\bf skip}: (=_{sym}) \Rightarrow\ (=_{sym}).
	\end{align}
	is derivable using the basic rules of logic rqPD, where $=_{sym}$ is the projector onto the symmetric space between $\bar{p}$ and $\bar{p}^\prime$. Indeed, this judgment implies the correctness of QOTP; that is, the input and output quantum data on register $\bar{p} = p_1,\cdots,p_n$ (might be entangled) are conserved. Similarly, we can also prove that for any $|\psi\>\in\h_{\bar{p}}$:
	\begin{align}
		\label{QOTP-n-Secure}
		\models \mathbf{KeyGen}(n);\mathbf{Enc}(n);\mathbf{DisKey}(n) \sim \mathbf{KeyGen}(n);\mathbf{Enc}(n);\mathbf{DisKey}(n): \frac{I_{\bar{p}}\otimes I_{\bar{p}^\prime}}{2^n}\Rightarrow |\psi\>_{\bar{p}}\<\psi|\otimes I_{\bar{p}^\prime}.
	\end{align}
	The above judgment actually implies the output after the encryption is the maximally mixed state and it is impossible for the eavesdropper to obtain any information about the quantum data.
	
	\section{Reasoning about Projective Predicates}\label{sec-projective}
	
	The logic rqPD was developed for reasoning about the equivalence between quantum programs with respect to general preconditions and postconditions represented by Hermitian operators. But in some applications, it is more convenient to use a simplified version of rqPD with preconditions and postconditions being projective predicates (equivalently, subspaces of the state Hilbert spaces). In this section, we present such a simplified version of rqPD and give an example to show its utility. As one may expect, a price for this simplification is a weaker expressive power of the logic. The coefficients $\frac{1}{d}$ and $\frac{1}{2}$ in the preconditions of judgments (\ref{e-unif}) and (\ref{ju-uniform}) are not expressible in rqPD with projective predicates, indicating that the expressive power of rqPD with projective predicates is \textit{strictly} weaker than that of full rqPD.
	
	\subsection{Inference Rules}
	
	In this subsection, we develop inference rules for judgments with projective preconditions and postconditions. We consider judgments of the form: \begin{equation}\label{judge-projector}P_1\sim P_2: A\Rightarrow B
	\end{equation} where $A,B$ are two projections in (or equivalently, subspaces of) $\mathcal{H}_{P_1\langle 1\rangle}\otimes\mathcal{H}_{P_2\langle 2\rangle}$.
	
	\begin{defn}\label{valid-ju-projector}
		Judgment (\ref{judge-projector}) is projectively valid, written: $$\models_P P_1\sim P_2: A\Rightarrow B,$$ if for any $\rho_1\in\mathcal{D}^\le\left(\mathcal{H}_{P_1\langle 1\rangle}\right)$ and $\rho_2\in\mathcal{D}^\le\left(\mathcal{H}_{P_2\langle 2\rangle}\right)$ such that $\rho_1 \mathrel{A^\#} \rho_2$, there exists a lifting of $B$ relating the output quantum states:
		$ \llbracket P_1\rrbracket(\rho_1) \mathrel{B^\#} \llbracket P_2\rrbracket(\rho_2). $
	\end{defn}
	
	The following proposition clarifies the relationship between projective validity and the notion of validity introduced in Definition \ref{valid-ju}.
	
	\begin{prop}\label{prop-comp} For any two program $P_1$ and $P_2$, and projective predicates $A$ and $B$:
		\begin{equation*}
			\begin{array}{cll}
				(1) &\models P_1\sim P_2: A\Rightarrow B\ \Rrightarrow\ \models_P P_1\sim P_2: A\Rightarrow B;\\
				(2) &\models P_1\sim P_2: A\Rightarrow B\ \not\Lleftarrow\ \models_P P_1\sim P_2: A\Rightarrow B. \end{array}
		\end{equation*}
	\end{prop}
	
	To present rules for proving projectively valid judgments, we need the following modifications of Definitions \ref{judgment-measurment} and \ref{judgment-one-side-measurment}.
	
	\begin{defn}\label{proj-measure-condition}Let $\cM_1=\{M_{1m}\}$ and $\cM_2=\{M_{2m}\}$ be two measurements with the same set $\{m\}$ of possible outcomes in $\mathcal{H}_{P_1}$ and $\mathcal{H}_{P_2}$, and let $A$ and $B_m$ be projective predicates in $\mathcal{H}_{P_1\langle 1\rangle}\otimes \mathcal{H}_{P_2\langle 2\rangle}.$ Then the assertion
		\begin{equation}
			\models_P \cM_1\approx \cM_2: A\Rightarrow \{B_m\}
		\end{equation}
		holds if for any $\rho_1\in\mathcal{D}^\le\left(\mathcal{H}_{P_1\langle 1\rangle}\right)$ and $\rho_2\in\mathcal{D}^\le\left(\mathcal{H}_{P_2\langle 2\rangle}\right)$ such that $\rho_1 \mathrel{A^\#} \rho_2$, there exists a sequence of lifting of $B_m$ relating the post-measurement states with the same outcomes: for all $m$,
		$$(M_{1m}\rho_1M_{1m}^\dag) \mathrel{B_m^\#} (M_{2m}\rho_2M_{2m}^\dag).$$
	\end{defn}
	
	%A one-side variant of the above definition is also useful.
	
	\begin{defn}Let $\cM_1=\{M_{1m}\}$ be a measurements in $\mathcal{H}_{P_1}$, and let $A$ and $B_m$ be projective predicates in $\mathcal{H}_{P_1\langle 1\rangle}\otimes \mathcal{H}_{P_2\langle 2\rangle}.$ Then the assertion
		\begin{equation}
			\models_P \cM_1\approx I_2: A\Rightarrow \{B_m\}
		\end{equation}
		holds if for any $\rho_1\in\mathcal{D}^\le\left(\mathcal{H}_{P_1\langle 1\rangle}\right)$ and $\rho_2\in\mathcal{D}^\le\left(\mathcal{H}_{P_2\langle 2\rangle}\right)$ such that $\rho_1 \mathrel{A^\#} \rho_2$, there exist $\rho_{2m}$ such that $\sum_m\rho_{2m} = \rho_2$ and for all $m$,
		$$(M_{1m}\rho_1M_{1m}^\dag) \mathrel{B_m^\#} \rho_{2m}.$$
	\end{defn}
	
	Now the proof system for judgments with projective preconditions and postconditions consists of rules (Skip), (UT), (SC), (UT-L/R), (Conseq), (Equiv) and (Frame) in Figs. \ref{fig 4.4-0}, \ref{fig 4.5-0} and \ref{fig 4.6} with $\vdash$, $\models$ being replaced by $\vdash_P$ and $\models_P$, respectively, together with the rules given in Fig. \ref{fig pro_1}.
	\begin{figure}[t]\centering
		\begin{equation*}\begin{split}
				&({\rm Init\text{-}P})\ \vdash_P q_1:=|0\rangle\sim q_2:=|0\rangle: A\Rightarrow |0\>_{q_1\<1\>}\<0|\otimes |0\>_{q_2\<2\>}\<0|\otimes \proj(\mathit{tr}_{\mathcal{H}_{q_1\<1\>}\otimes\mathcal{H}_{q_2\<2\>}}(A))\\
				&(\mathrm{Init\text{-}P\text{-}L}) \ \vdash_P q_1:=|0\rangle\sim \mathbf{skip}: A \Rightarrow |0\rangle_{q_1\langle 1\rangle}\langle 0| \otimes \proj( \mathit{tr}_{\mathcal{H}_{q_1\langle 1\rangle}} (A))\\
				&({\rm IF\text{-}P})\ \
				\frac{\models_P \cM_1\approx \cM_2: A\Rightarrow\{B_m\}\quad  \vdash_P P_{1m}\sim P_{2m}: B_m\Rightarrow C\ {\rm for\ every}\ m}{\vdash_P \mathbf{if}\ (\square m\cdot
					\cM_1[\overline{q}]=m\rightarrow P_{1m})\ \mathbf{fi}\sim\mathbf{if}\  (\square m\cdot
					\cM_2[\overline{q}]=m\rightarrow P_{2m})\ \mathbf{fi}:
					A\Rightarrow C}\\
				&({\rm IF\text{-}P\text{-}L})\ \ \
				\frac{ \models_P \cM_1\approx I_2:A\Rightarrow\{B_m\} \quad  \vdash_P P_{1m}\sim P: B_m\Rightarrow C\ {\rm for\ every}\ m}{\vdash_P \mathbf{if}\ (\square m\cdot
					\cM_1[\overline{q}]=m\rightarrow P_{1m})\ \mathbf{fi} \sim P:
					A\Rightarrow C}\\
				&({\rm LP\text{-}P}) \
				\frac{\models_P\cM_1\approx \cM_2: A\Rightarrow\{B_0,B_1\}\quad \vdash_P P_{1}\sim P_{2}: B_1\Rightarrow A}{\vdash_P \mathbf{while}\
					\cM_1[\overline{q}]=1\ \mathbf{do}\ P_1\ \mathbf{od}\sim    \mathbf{while}\
					\cM_2[\overline{q}]=1\ \mathbf{do}\ P_2\ \mathbf{od}:  A\Rightarrow B_0} \\
				&({\rm LP\text{-}P\text{-}L})\ \ \
				\frac{
					\begin{array}{c}
						\models \mathbf{while}\ \cM_1[\overline{q}]=1\ \mathbf{do}\ P_1\ \mathbf{od} \text{\ lossless} \\
						\models_P \cM_1\approx I_2: A\Rightarrow\{B_0,B_1\} \quad \vdash_P P_1\sim \mathbf{skip}: B_1\Rightarrow A
					\end{array}
				}{\vdash_P \mathbf{while}\
					\cM_1[\overline{q}]=1\ \mathbf{do}\ P_1\ \mathbf{od} \sim \mathbf{skip}:  A\Rightarrow B_0} \\
				&({\rm SO\text{-}P})\ \vdash_P \overline{q_1}:=\mathcal{E}_1\left[\overline{q_1}\right]\sim \overline{q_2}:=\mathcal{E}_2\left[\overline{q_2}\right]: A \Rightarrow \proj(\left(\mathcal{E}_1\otimes \mathcal{E}_2\right)(A))\\
				&({\rm SO\text{-}P\text{-}L})\ \vdash_P \overline{q_1}:=\mathcal{E}_1\left[\overline{q_1}\right]\sim \mathbf{skip}: A \Rightarrow \proj(\mathcal{E}_1(A))
		\end{split}\end{equation*}
		\caption{Rules for Projective Predicates. For any positive operator $A$ on Hilbert space $\mathcal{H}$, we write $\proj(A)$ for the projection onto $\supp(A)$, the subspace spanned by the eigenvectors of $A$ with nonzero eigenvalues. The quantum operations appeared in (SO-P) and (SO-P-L) are all trace-preserving. We omit the right-side counterparts of (Init-P-L), (IF-P-L), (LP-P-L) and (SO-P-L).}\label{fig pro_1}
	\end{figure}
	
	Let us carefully compare this simplified proof system for projective predicates with the original rqPD for general predicates of Hermitian operators:\begin{itemize} \item In rules (Init-P), (Init-P-L), (Init-P-R), (SO-P), (SO-P-L) and (SO-P-R), we have to use the operation $\proj(\cdot)$ because the operators in its operand there are not necessarily projective.
		
		\item Rule (Case) has no counterpart for projective predicates because probabilistic combination $\sum_i p_iA_i$ of a family of projective predicates $A_i$ is usually not projective.
		
		\item The main simplification occurs in the rules for control-flow constructs (i.e. conditionals and loops). We only consider rule (IF-P); the same explanation applies to other control-flow rules. First, the measurement condition $\models_P \cM_1\approx \cM_2: A\Rightarrow\{B_m\}$ in the premise of (IF-P) is weaker than the measurement condition $\cM_1\approx \cM_2\models A\Rightarrow\{B_m\}$ in the premise of (IF). Second, the measurement condition $\cM_1\approx \cM_2$ is the conclusion of (IF) is removed in (IF-P). As already pointed out in the Introduction, this is biggest reward of the projective simplification of our logic.
	\end{itemize}
	
	\begin{prop}
		\label{sound-proj}
		The proof system for judgments with projective preconditions and postconditions are sound.
	\end{prop}
	
	As will be seen in the next subsection, this simplified proof system, in particular, the simplified rules for control-flow constructs, whenever they are applicable, can significantly ease the verification of relational properties of quantum programs. On the other hand, some relational properties of quantum programs, e.g. judgments (\ref{e-unif}), (\ref{ju-uniform}) and (\ref{ju-bit}-\ref{ju-bp}), can be verified by the original rqPD but not by this simplified system. Even for the same quantum programs, the original rqPD usually can prove stronger relational properties in the case where the weakest preconditions or strongest postconditions are not projective.
	
	\subsection{Example: Quantum Walks}
	In this subsection, we present an example to show the effectiveness of
	the inference rules given in the previous subsection. Quantum (random)
	walks~\cite{kempe2003quantum,venegas2012quantum} are quantum analogues
	of random walks, and have been widely used in the design of quantum
	algorithms, including quantum search and quantum simulation. There are two key ideas in defining a quantum walk that are fundamentally different from that of a classical random walk: (1) a \textquotedblleft quantum coin\textquotedblright\ is introduced to govern the movement of the walker, which allows the walker to move to two different directions, say left and right, simultaneously; (2) an absorbing boundary is realised by a quantum measurement.
	Here, we show how our logic can be applied to verify a certain equivalence of
	two one-dimensional quantum walks with absorbing boundaries: the quantum coins used in these two quantum walks are different, but they terminate at the same position.
	
	\begin{exam}
		\label{q-walker}
		Let $\mathcal{H}_c$ be the coin space, the $2$-dimensional Hilbert space with orthonormal basis state $|L\rangle$ (or $|0\rangle_c$) and $|R\rangle$ (or $|1\rangle_c$), indicating directions Left and Right, respectively. Let $\mathcal{H}_p$ be the $(n+1)$-dimensional Hilbert space with orthonormal basis states $|0\rangle, |1\rangle, ..., |n-1\rangle, |n\rangle$, where vector $|i\rangle$ denotes position $i$ for each $0\leq i\leq n$; in particular, positions $0$ and $n$ are the absorbing boundaries. The state space of the walk is then $\mathcal{H}=\mathcal{H}_c\otimes \mathcal{H}_p$. Each step of the walk consists of:
		\begin{enumerate}
			\item Measure the position of the system to see whether it is $0$ or $n$. If the outcome is \textquotedblleft yes\textquotedblright, then the walk terminates; otherwise, it continues. The measurement can be described as $\cM=\{M_\mathit{yes}, M_\mathit{no}\}$, where the measurement operators are: 
			$\,
				M_\mathit{yes}=|0\rangle\langle 0| + |n\rangle\langle n|$, $M_\mathit{no}=I_p-M_{yes}=\sum_{i=1}^{n-1} |i\rangle\langle i| $
			, and $I_p$ is the identity in position space $\mathcal{H}_p$;
			\item Apply a \textquotedblleft coin-tossing\textquotedblright\ operator $C$ is in the coin space $\mathcal{H}_c$.
			\item Apply a shift operator %\vspace{-0.1cm}\\[-0.1cm]
			$S=\sum_{i=1}^{n-1}|L\rangle\langle L|\otimes|i-1\rangle\langle i|+\sum_{i=1}^{n-1} |R\rangle\langle R|\otimes|i+ 1\rangle\langle i| $ in the space $\mathcal{H}$. Intuitively, operator $S$ moves the position one step to the left or to the right according to the direction state.
		\end{enumerate}
		A major difference between a quantum walk and a classical random walk is that a
		superposition of movement to the left and a movement to the right can happen in
		the quantum case.  The quantum walk can be written as a quantum program with
		the initial state in $\mathcal{H}_c\otimes \mathcal{H}_p$ as the input:
		\begin{align}\label{qw-0}
			\mathbf{while}\ \cM[p]=\mathit{no}\ \mathbf{do}\ c:=C[c];\ c,p:=S[c,p]\ \mathbf{od}
		\end{align}
	\end{exam}
	
	We consider two frequently used \textquotedblleft coin-tossing\textquotedblright\ operators here: the Hadamard operator $H$ and the balanced operator: 
	$
	Y = \frac{1}{\sqrt{2}}\left(\begin{array}{cc} 1& {\bm i}\\ {\bm i} & 1\end{array}\right).
	$
	We use $\mathbf{while}(H)$ and $\mathbf{while}(Y)$ to denote program (\ref{qw-0}) with $C=H$ or $Y$, respectively.
	What interests us is: with what kind of initial states do the quantum walks with different \textquotedblleft coin-tossing\textquotedblright\ operators $H$ and $Y$ produce the same output position?
	To this end, we add a measurement to determine the exact position of the walks after their termination, and discard the coin. So, programs $\mathbf{while}(H)$ and $\mathbf{while}(Y)$ are modified to:
	\begin{align*}
		QW(H) \equiv &\ \mathbf{while}(H);\ \mathbf{if}\ (\square i\cdot\cM^\prime[p]=i\rightarrow \mathbf{skip})\ \mathbf{fi};\ {\bf Tr}[c],\\
		QW(Y) \equiv &\ \mathbf{while}(Y);\ \mathbf{if}\ (\square i\cdot\cM^\prime[p]=i\rightarrow \mathbf{skip})\ \mathbf{fi};\ {\bf Tr}[c]
	\end{align*}
	where measurement $\cM^\prime = \{M_i^\prime\}$ with $M_i^\prime = |i\>\<i|$ for $i=0,1,\cdots,n$.
	
	Before formulating our result in our logic, let us fix the notations. Whenever comparing programs $QW(H)$ and $QW(Y)$ and using, say $x$, to denote a variable in the former, then we shall use $x^\prime$ for the same variable in the latter. For simplicity, we  use $|d, i\rangle_{c,p}$ as an abbreviation of $|d\rangle_c|i\rangle_p$, $I_{c,p}$ is the identity over the whole space $\mathcal{H}_c\otimes\mathcal{H}_p$, and $S_{c,p;c^\prime,p^\prime}$ is the SWAP operator between two systems $\mathcal{H}_c\otimes\mathcal{H}_p$ and $\mathcal{H}_{c^\prime}\otimes\mathcal{H}_{p^\prime}$. Furthermore, we introduce the following unitary operator $U$ and projective predicates $(=_\mathit{sym})$ and $(=_{sym}^p)$:
	\vspace{-0.1cm}
	\begin{align*}
		&U: |d,i\rangle_{c,p}\mapsto (-1)^{\frac{i+d+3}{2}} |d,i\rangle_{c,p} \qquad(=_\mathit{sym}) = \frac{1}{2}(I_{c,p}\otimes I_{c^\prime,p^\prime}+ S_{c,p;c^\prime,p^\prime})\\
		&(=_{sym}^p) = \frac{1}{2}\Big(\sum_{i,i^\prime=0,n}|i\rangle_p\langle i|\otimes |i^\prime\rangle_{p^\prime}\langle i^\prime| + \sum_{i,i^\prime=0,n}|i\rangle_p\langle i^\prime|\otimes |i\rangle_{p^\prime}\langle i^\prime|\Big), \\[-0.7cm]\nonumber
	\end{align*}
	In \cite{fullversion}, we show how to derive the following judgment in the projective version of rqPD:
	\vspace{-0.1cm}
	\begin{align}\label{two-qw}
		\models_P QW(H)\sim QW(Y): U_{c^\prime,p^\prime} (=_\mathit{sym})U^\dag_{c^\prime,p^\prime} \Rightarrow (=_{sym}^p). \\[-0.7cm]\nonumber
	\end{align}
	This judgment means that if walks $QW(H)$ and $QW(Y)$ start from states $\rho_1$ and $\rho_2=U\rho_1U^\dag$, respectively, then they terminate at exactly the same position.
	
	\section{Related work}\label{sec:related}
	The formal verification of quantum programs is an active area of
	research, and many expressive formalisms have been proposed in the
	literature~\cite{CMS,DP06,FDJY07,Kaku09,Ying11,Ying16}. However, previous work largely
	considers single program executions. 
	Security of quantum one-time pad (our Example \ref{exam-01}) was verified in \cite{Unruh19} using a variant of quantum Hoare logic rather than relational logic. 
	Other formalisms explicitly target
	equivalence of quantum programs~\cite{ALGN13,FengY15,KubotaKKKS13}.
	However, these works are based on bisimulations and symbolic methods,
	which have a more limited scope and are less powerful than general
	relational program logics. Finally, some works develop specialized
	methods for proving concrete properties of quantum programs; for
	instance, Hung et al~\cite{HungHZYHW18} reason about quantitative
	robustness of quantum programs. It would be interesting to cast the
	latter into our more general framework.
	This seems possible although may not be straightforward; indeed, in Subsection \ref{tele-reli}, we showed that our logic can be used to reason about the reliability of quantum teleportation against several kinds of quantum noise.
	
	Our work is most closely related to the quantum relational Hoare logic
	recently proposed by \citet{Unruh18,LU19}. Both works are inspired by
	probabilistic relational Hoare logic~\cite{BartheGZ09} and share the
	long-term objective of providing a convenient framework for formal
	verification of quantum cryptography. However, the two works explore
	different points in the design space of relational logics for quantum
	programs. There are several fundamental differences between our logic and Unruh's one, including expressiveness, entanglement in defining the validity of judgments and inference rules.  
	A careful comparison of them is given in \cite{fullversion}.
	
	\section{Conclusion}
	We have introduced a relational program logic for a core quantum
	programming language; our logic is based on a quantum analogue of
	probabilistic couplings, and is able to verify several non-trivial
	examples of relational properties for quantum programs. 
	
	There are several promising directions for future work. First, we would like to
	further develop the theory of quantum couplings, and in particular to
	define a quantum version of approximate couplings. An extension apRHL of probabilistic relational Hoare logic pRHL was defined in \cite{Barthe13} for verification of differential privacy. A surprising connection between quantum differential privacy and gentle measurements recently observed by \citet{Aaronson19} presents a further possible application of a quantum counterpart of apRHL in quantum physics. Second, we would
	like to explore variants and applications of our logic to other areas,
	including the convergence of quantum Markov chains, quantum cryptography,
	and translation validation of quantum programs; in particular, the correctness of optimising quantum compilers for NISQ (Noisy Intermediate Quantum) devices.

	%%% Acknowledgments
	\begin{acks}                            %% acks environment is optional
	                                        %% contents suppressed with 'anonymous'
	  %% Commands \grantsponsor{<sponsorID>}{<name>}{<url>} and
	  %% \grantnum[<url>]{<sponsorID>}{<number>} should be used to
	  %% acknowledge financial support and will be used by metadata
	  %% extraction tools.
    This work is partially supported by the University of Wisconsin, a
    Facebook TAV award, the Australian Research Council (Grant No: DE180100156 and DP180100691), the National Key R\&D Program of China (Grant No: 2018YFA0306701), and the National Natural Science Foundation of China (Grant No: 61832015). We are grateful to the Max Planck Institute for Software Systems for hosting some of the authors.
	\end{acks}

	%% Bibliography
	\bibliographystyle{ACM-Reference-Format}
	\bibliography{rqPD-final}


\begin{thebibliography}{50}

%%% the Web address.

\ifx \showCODEN    \undefined \def \showCODEN     #1{\unskip}     \fi
\ifx \showDOI      \undefined \def \showDOI       #1{#1}\fi
\ifx \showISBNx    \undefined \def \showISBNx     #1{\unskip}     \fi
\ifx \showISBNxiii \undefined \def \showISBNxiii  #1{\unskip}     \fi
\ifx \showISSN     \undefined \def \showISSN      #1{\unskip}     \fi
\ifx \showLCCN     \undefined \def \showLCCN      #1{\unskip}     \fi
\ifx \shownote     \undefined \def \shownote      #1{#1}          \fi
\ifx \showarticletitle \undefined \def \showarticletitle #1{#1}   \fi
\ifx \showURL      \undefined \def \showURL       {\relax}        \fi
% invisible to TeX
\providecommand\bibfield[2]{#2}
\providecommand\bibinfo[2]{#2}
\providecommand\natexlab[1]{#1}
\providecommand\showeprint[2][]{arXiv:#2}

\bibitem[\protect\citeauthoryear{Aaronson and Rothblum}{Aaronson and
  Rothblum}{2019}]%
        {Aaronson19}
\bibfield{author}{\bibinfo{person}{Scott Aaronson} {and}
  \bibinfo{person}{Guy~N. Rothblum}.} \bibinfo{year}{2019}\natexlab{}.
\newblock \showarticletitle{Gentle Measurement of Quantum States and
  Differential Privacy}. In \bibinfo{booktitle}{\emph{Proceedings of the 51st
  Annual ACM SIGACT Symposium on Theory of Computing}}
  \emph{(\bibinfo{series}{STOC 2019})}. \bibinfo{publisher}{ACM},
  \bibinfo{address}{New York, NY, USA}, \bibinfo{pages}{322--333}.
\newblock
\showISBNx{978-1-4503-6705-9}
\urldef\tempurl%
\url{https://doi.org/10.1145/3313276.3316378}
\showDOI{\tempurl}


\bibitem[\protect\citeauthoryear{Abramsky and Coecke}{Abramsky and
  Coecke}{2004}]%
        {AB04}
\bibfield{author}{\bibinfo{person}{Samson Abramsky} {and} \bibinfo{person}{Bob
  Coecke}.} \bibinfo{year}{2004}\natexlab{}.
\newblock \showarticletitle{A Categorical Semantics of Quantum Protocols}. In
  \bibinfo{booktitle}{\emph{19th {IEEE} Symposium on Logic in Computer Science
  {(LICS} 2004), 14-17 July 2004, Turku, Finland, Proceedings}}.
  \bibinfo{pages}{415--425}.
\newblock
\urldef\tempurl%
\url{https://doi.org/10.1109/LICS.2004.1319636}
\showDOI{\tempurl}


\bibitem[\protect\citeauthoryear{Anil~Kumar and Ramesh}{Anil~Kumar and
  Ramesh}{2001}]%
        {AKR01}
\bibfield{author}{\bibinfo{person}{V.S. Anil~Kumar} {and} \bibinfo{person}{H.
  Ramesh}.} \bibinfo{year}{2001}\natexlab{}.
\newblock \showarticletitle{Coupling vs. conductance for the Jerrum–Sinclair
  chain}.
\newblock \bibinfo{journal}{\emph{Random Structures \& Algorithms}}
  \bibinfo{volume}{18}, \bibinfo{number}{1} (\bibinfo{year}{2001}),
  \bibinfo{pages}{1--17}.
\newblock
\urldef\tempurl%
\url{https://doi.org/10.1002/1098-2418(200101)18:1<1::AID-RSA1>3.0.CO;2-7}
\showDOI{\tempurl}


\bibitem[\protect\citeauthoryear{Ardeshir{-}Larijani, Gay, and
  Nagarajan}{Ardeshir{-}Larijani et~al\mbox{.}}{2013}]%
        {ALGN13}
\bibfield{author}{\bibinfo{person}{Ebrahim Ardeshir{-}Larijani},
  \bibinfo{person}{Simon~J. Gay}, {and} \bibinfo{person}{Rajagopal Nagarajan}.}
  \bibinfo{year}{2013}\natexlab{}.
\newblock \showarticletitle{Equivalence Checking of Quantum Protocols}. In
  \bibinfo{booktitle}{\emph{Tools and Algorithms for the Construction and
  Analysis of Systems - 19th International Conference, {TACAS} 2013, Held as
  Part of the European Joint Conferences on Theory and Practice of Software,
  {ETAPS} 2013, Rome, Italy, March 16-24, 2013. Proceedings}}
  \emph{(\bibinfo{series}{Lecture Notes in Computer Science})},
  \bibfield{editor}{\bibinfo{person}{Nir Piterman} {and}
  \bibinfo{person}{Scott~A. Smolka}} (Eds.), Vol.~\bibinfo{volume}{7795}.
  \bibinfo{publisher}{Springer}, \bibinfo{pages}{478--492}.
\newblock
\urldef\tempurl%
\url{https://doi.org/10.1007/978-3-642-36742-7\_33}
\showDOI{\tempurl}


\bibitem[\protect\citeauthoryear{Barthe, Espitau, Gr{\'{e}}goire, Hsu,
  Stefanesco, and Strub}{Barthe et~al\mbox{.}}{2015}]%
        {BartheEGHSS15}
\bibfield{author}{\bibinfo{person}{Gilles Barthe}, \bibinfo{person}{Thomas
  Espitau}, \bibinfo{person}{Benjamin Gr{\'{e}}goire}, \bibinfo{person}{Justin
  Hsu}, \bibinfo{person}{L{\'{e}}o Stefanesco}, {and}
  \bibinfo{person}{Pierre{-}Yves Strub}.} \bibinfo{year}{2015}\natexlab{}.
\newblock \showarticletitle{Relational Reasoning via Probabilistic Coupling}.
  In \bibinfo{booktitle}{\emph{Logic for Programming, Artificial Intelligence,
  and Reasoning - 20th International Conference, {LPAR-20} 2015, Suva, Fiji,
  November 24-28, 2015, Proceedings}} \emph{(\bibinfo{series}{Lecture Notes in
  Computer Science})}, \bibfield{editor}{\bibinfo{person}{Martin Davis},
  \bibinfo{person}{Ansgar Fehnker}, \bibinfo{person}{Annabelle McIver}, {and}
  \bibinfo{person}{Andrei Voronkov}} (Eds.), Vol.~\bibinfo{volume}{9450}.
  \bibinfo{publisher}{Springer}, \bibinfo{pages}{387--401}.
\newblock
\urldef\tempurl%
\url{https://doi.org/10.1007/978-3-662-48899-7\_27}
\showDOI{\tempurl}


\bibitem[\protect\citeauthoryear{Barthe, Espitau, Gr{\'{e}}goire, Hsu, and
  Strub}{Barthe et~al\mbox{.}}{2017}]%
        {BartheEGHS17}
\bibfield{author}{\bibinfo{person}{Gilles Barthe}, \bibinfo{person}{Thomas
  Espitau}, \bibinfo{person}{Benjamin Gr{\'{e}}goire}, \bibinfo{person}{Justin
  Hsu}, {and} \bibinfo{person}{Pierre{-}Yves Strub}.}
  \bibinfo{year}{2017}\natexlab{}.
\newblock \showarticletitle{Proving uniformity and independence by
  self-composition and coupling}. In \bibinfo{booktitle}{\emph{LPAR-21, 21st
  International Conference on Logic for Programming, Artificial Intelligence
  and Reasoning, Maun, Botswana, 7-12th May 2017}} \emph{(\bibinfo{series}{EPiC
  Series})}, \bibfield{editor}{\bibinfo{person}{Thomas Eiter} {and}
  \bibinfo{person}{David Sands}} (Eds.), Vol.~\bibinfo{volume}{46}.
  \bibinfo{publisher}{EasyChair}, \bibinfo{pages}{385--403}.
\newblock
\urldef\tempurl%
\url{http://www.easychair.org/publications/paper/340344}
\showURL{%
\tempurl}


\bibitem[\protect\citeauthoryear{Barthe, Espitau, Gr{\'{e}}goire, Hsu, and
  Strub}{Barthe et~al\mbox{.}}{2018}]%
        {BartheEGHS18}
\bibfield{author}{\bibinfo{person}{Gilles Barthe}, \bibinfo{person}{Thomas
  Espitau}, \bibinfo{person}{Benjamin Gr{\'{e}}goire}, \bibinfo{person}{Justin
  Hsu}, {and} \bibinfo{person}{Pierre{-}Yves Strub}.}
  \bibinfo{year}{2018}\natexlab{}.
\newblock \showarticletitle{Proving expected sensitivity of probabilistic
  programs}.
\newblock \bibinfo{journal}{\emph{{PACMPL}}} \bibinfo{volume}{2},
  \bibinfo{number}{{POPL}} (\bibinfo{year}{2018}),
  \bibinfo{pages}{57:1--57:29}.
\newblock
\urldef\tempurl%
\url{https://doi.org/10.1145/3158145}
\showDOI{\tempurl}


\bibitem[\protect\citeauthoryear{Barthe, Gaboardi, Gr{\'{e}}goire, Hsu, and
  Strub}{Barthe et~al\mbox{.}}{2016}]%
        {BartheGGHS16}
\bibfield{author}{\bibinfo{person}{Gilles Barthe}, \bibinfo{person}{Marco
  Gaboardi}, \bibinfo{person}{Benjamin Gr{\'{e}}goire}, \bibinfo{person}{Justin
  Hsu}, {and} \bibinfo{person}{Pierre{-}Yves Strub}.}
  \bibinfo{year}{2016}\natexlab{}.
\newblock \showarticletitle{Proving Differential Privacy via Probabilistic
  Couplings}. In \bibinfo{booktitle}{\emph{Proceedings of the 31st Annual
  {ACM/IEEE} Symposium on Logic in Computer Science, {LICS} '16, New York, NY,
  USA, July 5-8, 2016}}, \bibfield{editor}{\bibinfo{person}{Martin Grohe},
  \bibinfo{person}{Eric Koskinen}, {and} \bibinfo{person}{Natarajan Shankar}}
  (Eds.). \bibinfo{publisher}{{ACM}}, \bibinfo{pages}{749--758}.
\newblock
\urldef\tempurl%
\url{https://doi.org/10.1145/2933575.2934554}
\showDOI{\tempurl}


\bibitem[\protect\citeauthoryear{Barthe, Gr{\'{e}}goire, and
  B{\'{e}}guelin}{Barthe et~al\mbox{.}}{2009}]%
        {BartheGZ09}
\bibfield{author}{\bibinfo{person}{Gilles Barthe}, \bibinfo{person}{Benjamin
  Gr{\'{e}}goire}, {and} \bibinfo{person}{Santiago~Zanella B{\'{e}}guelin}.}
  \bibinfo{year}{2009}\natexlab{}.
\newblock \showarticletitle{Formal certification of code-based cryptographic
  proofs}. In \bibinfo{booktitle}{\emph{Proceedings of the 36th {ACM}
  {SIGPLAN-SIGACT} Symposium on Principles of Programming Languages, {POPL}
  2009, Savannah, GA, USA, January 21-23, 2009}},
  \bibfield{editor}{\bibinfo{person}{Zhong Shao} {and}
  \bibinfo{person}{Benjamin~C. Pierce}} (Eds.). \bibinfo{publisher}{{ACM}},
  \bibinfo{pages}{90--101}.
\newblock
\urldef\tempurl%
\url{https://doi.org/10.1145/1480881.1480894}
\showDOI{\tempurl}


\bibitem[\protect\citeauthoryear{Barthe, Hsu, Ying, Yu, and Zhou}{Barthe
  et~al\mbox{.}}{2019}]%
        {fullversion}
\bibfield{author}{\bibinfo{person}{Gilles Barthe}, \bibinfo{person}{Justin
  Hsu}, \bibinfo{person}{Mingsheng Ying}, \bibinfo{person}{Nengkun Yu}, {and}
  \bibinfo{person}{Li Zhou}.} \bibinfo{year}{2019}\natexlab{}.
\newblock \showarticletitle{Relational Proofs for Quantum Programs (Extended
  Version)}.
\newblock \bibinfo{journal}{\emph{CoRR}}  \bibinfo{volume}{abs/1901.05184}
  (\bibinfo{year}{2019}).
\newblock
\showeprint[arxiv]{1901.05184}
\urldef\tempurl%
\url{http://arxiv.org/abs/1901.05184}
\showURL{%
\tempurl}


\bibitem[\protect\citeauthoryear{Barthe, K{\"{o}}pf, Olmedo, and
  B{\'{e}}guelin}{Barthe et~al\mbox{.}}{2012}]%
        {BartheKOZ12}
\bibfield{author}{\bibinfo{person}{Gilles Barthe}, \bibinfo{person}{Boris
  K{\"{o}}pf}, \bibinfo{person}{Federico Olmedo}, {and}
  \bibinfo{person}{Santiago~Zanella B{\'{e}}guelin}.}
  \bibinfo{year}{2012}\natexlab{}.
\newblock \showarticletitle{Probabilistic relational reasoning for differential
  privacy}. In \bibinfo{booktitle}{\emph{Proceedings of the 39th {ACM}
  {SIGPLAN-SIGACT} Symposium on Principles of Programming Languages, {POPL}
  2012, Philadelphia, Pennsylvania, USA, January 22-28, 2012}},
  \bibfield{editor}{\bibinfo{person}{John Field} {and} \bibinfo{person}{Michael
  Hicks}} (Eds.). \bibinfo{publisher}{{ACM}}, \bibinfo{pages}{97--110}.
\newblock
\urldef\tempurl%
\url{https://doi.org/10.1145/2103656.2103670}
\showDOI{\tempurl}


\bibitem[\protect\citeauthoryear{Barthe, K\"{o}pf, Olmedo, and
  Zanella-B{\'e}guelin}{Barthe et~al\mbox{.}}{2013}]%
        {Barthe13}
\bibfield{author}{\bibinfo{person}{Gilles Barthe}, \bibinfo{person}{Boris
  K\"{o}pf}, \bibinfo{person}{Federico Olmedo}, {and} \bibinfo{person}{Santiago
  Zanella-B{\'e}guelin}.} \bibinfo{year}{2013}\natexlab{}.
\newblock \showarticletitle{Probabilistic Relational Reasoning for Differential
  Privacy}.
\newblock \bibinfo{journal}{\emph{ACM Trans. Program. Lang. Syst.}}
  \bibinfo{volume}{35}, \bibinfo{number}{3}, Article \bibinfo{articleno}{9}
  (\bibinfo{date}{nov} \bibinfo{year}{2013}), \bibinfo{numpages}{49}~pages.
\newblock
\showISSN{0164-0925}
\urldef\tempurl%
\url{https://doi.org/10.1145/2492061}
\showDOI{\tempurl}


\bibitem[\protect\citeauthoryear{Bennett, Brassard, Cr\'epeau, Jozsa, Peres,
  and Wootters}{Bennett et~al\mbox{.}}{1993}]%
        {Tel93}
\bibfield{author}{\bibinfo{person}{Charles~H. Bennett}, \bibinfo{person}{Gilles
  Brassard}, \bibinfo{person}{Claude Cr\'epeau}, \bibinfo{person}{Richard
  Jozsa}, \bibinfo{person}{Asher Peres}, {and} \bibinfo{person}{William~K.
  Wootters}.} \bibinfo{year}{1993}\natexlab{}.
\newblock \showarticletitle{Teleporting an unknown quantum state via dual
  classical and Einstein-Podolsky-Rosen channels}.
\newblock \bibinfo{journal}{\emph{Phys. Rev. Lett.}}  \bibinfo{volume}{70}
  (\bibinfo{date}{Mar} \bibinfo{year}{1993}), \bibinfo{pages}{1895--1899}.
\newblock
Issue 13.
\urldef\tempurl%
\url{https://doi.org/10.1103/PhysRevLett.70.1895}
\showDOI{\tempurl}


\bibitem[\protect\citeauthoryear{Boykin and Roychowdhury}{Boykin and
  Roychowdhury}{2003}]%
        {BR03}
\bibfield{author}{\bibinfo{person}{P.~Oscar Boykin} {and}
  \bibinfo{person}{Vwani Roychowdhury}.} \bibinfo{year}{2003}\natexlab{}.
\newblock \showarticletitle{Optimal encryption of quantum bits}.
\newblock \bibinfo{journal}{\emph{Phys. Rev. A}}  \bibinfo{volume}{67}
  (\bibinfo{date}{Apr} \bibinfo{year}{2003}), \bibinfo{pages}{042317}.
\newblock
Issue 4.
\urldef\tempurl%
\url{https://doi.org/10.1103/PhysRevA.67.042317}
\showDOI{\tempurl}


\bibitem[\protect\citeauthoryear{Chadha, Mateus, and Sernadas}{Chadha
  et~al\mbox{.}}{2006}]%
        {CMS}
\bibfield{author}{\bibinfo{person}{Rohit Chadha}, \bibinfo{person}{Paulo
  Mateus}, {and} \bibinfo{person}{Am{\'{\i}}lcar Sernadas}.}
  \bibinfo{year}{2006}\natexlab{}.
\newblock \showarticletitle{Reasoning About Imperative Quantum Programs}.
\newblock \bibinfo{journal}{\emph{Electr. Notes Theor. Comput. Sci.}}
  \bibinfo{volume}{158} (\bibinfo{year}{2006}), \bibinfo{pages}{19--39}.
\newblock
\urldef\tempurl%
\url{https://doi.org/10.1016/j.entcs.2006.04.003}
\showDOI{\tempurl}


\bibitem[\protect\citeauthoryear{Dale, Jennings, and Rudolph}{Dale
  et~al\mbox{.}}{2015}]%
        {DJR15}
\bibfield{author}{\bibinfo{person}{Howard Dale}, \bibinfo{person}{David
  Jennings}, {and} \bibinfo{person}{Terry Rudolph}.}
  \bibinfo{year}{2015}\natexlab{}.
\newblock \showarticletitle{Provable quantum advantage in randomness
  processing}.
\newblock \bibinfo{journal}{\emph{Nature communications}}  \bibinfo{volume}{6}
  (\bibinfo{year}{2015}), \bibinfo{pages}{8203}.
\newblock


\bibitem[\protect\citeauthoryear{D'Hondt and Panangaden}{D'Hondt and
  Panangaden}{2006}]%
        {DP06}
\bibfield{author}{\bibinfo{person}{Ellie D'Hondt} {and}
  \bibinfo{person}{Prakash Panangaden}.} \bibinfo{year}{2006}\natexlab{}.
\newblock \showarticletitle{Quantum weakest preconditions}.
\newblock \bibinfo{journal}{\emph{Mathematical Structures in Computer Science}}
  \bibinfo{volume}{16}, \bibinfo{number}{3} (\bibinfo{year}{2006}),
  \bibinfo{pages}{429--451}.
\newblock
\urldef\tempurl%
\url{https://doi.org/10.1017/S0960129506005251}
\showDOI{\tempurl}


\bibitem[\protect\citeauthoryear{Feng, Duan, Ji, and Ying}{Feng
  et~al\mbox{.}}{2007}]%
        {FDJY07}
\bibfield{author}{\bibinfo{person}{Yuan Feng}, \bibinfo{person}{Runyao Duan},
  \bibinfo{person}{Zheng{-}Feng Ji}, {and} \bibinfo{person}{Mingsheng Ying}.}
  \bibinfo{year}{2007}\natexlab{}.
\newblock \showarticletitle{Proof rules for the correctness of quantum
  programs}.
\newblock \bibinfo{journal}{\emph{Theor. Comput. Sci.}} \bibinfo{volume}{386},
  \bibinfo{number}{1-2} (\bibinfo{year}{2007}), \bibinfo{pages}{151--166}.
\newblock
\urldef\tempurl%
\url{https://doi.org/10.1016/j.tcs.2007.06.011}
\showDOI{\tempurl}


\bibitem[\protect\citeauthoryear{Feng and Ying}{Feng and Ying}{2015}]%
        {FengY15}
\bibfield{author}{\bibinfo{person}{Yuan Feng} {and} \bibinfo{person}{Mingsheng
  Ying}.} \bibinfo{year}{2015}\natexlab{}.
\newblock \showarticletitle{Toward Automatic Verification of Quantum
  Cryptographic Protocols}. In \bibinfo{booktitle}{\emph{26th International
  Conference on Concurrency Theory, {CONCUR} 2015, Madrid, Spain, September
  1.4, 2015}} \emph{(\bibinfo{series}{LIPIcs})},
  \bibfield{editor}{\bibinfo{person}{Luca Aceto} {and} \bibinfo{person}{David
  de~Frutos{-}Escrig}} (Eds.), Vol.~\bibinfo{volume}{42}.
  \bibinfo{publisher}{Schloss Dagstuhl - Leibniz-Zentrum fuer Informatik},
  \bibinfo{pages}{441--455}.
\newblock
\urldef\tempurl%
\url{https://doi.org/10.4230/LIPIcs.CONCUR.2015.441}
\showDOI{\tempurl}


\bibitem[\protect\citeauthoryear{Gharibian}{Gharibian}{2010}]%
        {Gha10}
\bibfield{author}{\bibinfo{person}{Sevag Gharibian}.}
  \bibinfo{year}{2010}\natexlab{}.
\newblock \showarticletitle{Strong NP-hardness of the Quantum Separability
  Problem}.
\newblock \bibinfo{journal}{\emph{Quantum Info. Comput.}} \bibinfo{volume}{10},
  \bibinfo{number}{3} (\bibinfo{date}{March} \bibinfo{year}{2010}),
  \bibinfo{pages}{343--360}.
\newblock
\showISSN{1533-7146}
\urldef\tempurl%
\url{http://dl.acm.org/citation.cfm?id=2011350.2011361}
\showURL{%
\tempurl}


\bibitem[\protect\citeauthoryear{Gurvits}{Gurvits}{2003}]%
        {Gur03}
\bibfield{author}{\bibinfo{person}{Leonid Gurvits}.}
  \bibinfo{year}{2003}\natexlab{}.
\newblock \showarticletitle{Classical Deterministic Complexity of Edmonds'
  Problem and Quantum Entanglement}. In \bibinfo{booktitle}{\emph{Proceedings
  of the Thirty-fifth Annual ACM Symposium on Theory of Computing}}
  \emph{(\bibinfo{series}{STOC '03})}. \bibinfo{publisher}{ACM},
  \bibinfo{address}{New York, NY, USA}, \bibinfo{pages}{10--19}.
\newblock
\showISBNx{1-58113-674-9}
\urldef\tempurl%
\url{https://doi.org/10.1145/780542.780545}
\showDOI{\tempurl}


\bibitem[\protect\citeauthoryear{Hsu}{Hsu}{2017}]%
        {Hsu17}
\bibfield{author}{\bibinfo{person}{Justin Hsu}.}
  \bibinfo{year}{2017}\natexlab{}.
\newblock \showarticletitle{Probabilistic Couplings for Probabilistic
  Reasoning}.
\newblock \bibinfo{journal}{\emph{CoRR}}  \bibinfo{volume}{abs/1710.09951}
  (\bibinfo{year}{2017}).
\newblock
\showeprint[arxiv]{1710.09951}
\urldef\tempurl%
\url{http://arxiv.org/abs/1710.09951}
\showURL{%
\tempurl}


\bibitem[\protect\citeauthoryear{Hung, Hietala, Zhu, Ying, Hicks, and Wu}{Hung
  et~al\mbox{.}}{2018}]%
        {HungHZYHW18}
\bibfield{author}{\bibinfo{person}{Shih{-}Han Hung}, \bibinfo{person}{Kesha
  Hietala}, \bibinfo{person}{Shaopeng Zhu}, \bibinfo{person}{Mingsheng Ying},
  \bibinfo{person}{Michael Hicks}, {and} \bibinfo{person}{Xiaodi Wu}.}
  \bibinfo{year}{2018}\natexlab{}.
\newblock \showarticletitle{Quantitative Robustness Analysis of Quantum
  Programs (Extended Version)}.
\newblock \bibinfo{journal}{\emph{CoRR}}  \bibinfo{volume}{abs/1811.03585}
  (\bibinfo{year}{2018}).
\newblock
\showeprint[arxiv]{1811.03585}
\urldef\tempurl%
\url{http://arxiv.org/abs/1811.03585}
\showURL{%
\tempurl}
\newblock
\shownote{To appear at POPL'19.}


\bibitem[\protect\citeauthoryear{Kakutani}{Kakutani}{2009}]%
        {Kaku09}
\bibfield{author}{\bibinfo{person}{Yoshihiko Kakutani}.}
  \bibinfo{year}{2009}\natexlab{}.
\newblock \showarticletitle{A Logic for Formal Verification of Quantum
  Programs}. In \bibinfo{booktitle}{\emph{Advances in Computer Science -
  {ASIAN} 2009. Information Security and Privacy, 13th Asian Computing Science
  Conference, Seoul, Korea, December 14-16, 2009. Proceedings}}
  \emph{(\bibinfo{series}{Lecture Notes in Computer Science})},
  \bibfield{editor}{\bibinfo{person}{Anupam Datta}} (Ed.),
  Vol.~\bibinfo{volume}{5913}. \bibinfo{publisher}{Springer},
  \bibinfo{pages}{79--93}.
\newblock
\urldef\tempurl%
\url{https://doi.org/10.1007/978-3-642-10622-4\_7}
\showDOI{\tempurl}


\bibitem[\protect\citeauthoryear{Keane and O'Brien}{Keane and O'Brien}{1994}]%
        {CBF}
\bibfield{author}{\bibinfo{person}{MS Keane} {and} \bibinfo{person}{George~L
  O'Brien}.} \bibinfo{year}{1994}\natexlab{}.
\newblock \showarticletitle{A Bernoulli factory}.
\newblock \bibinfo{journal}{\emph{ACM Transactions on Modeling and Computer
  Simulation (TOMACS)}} \bibinfo{volume}{4}, \bibinfo{number}{2}
  (\bibinfo{year}{1994}), \bibinfo{pages}{213--219}.
\newblock


\bibitem[\protect\citeauthoryear{Kempe}{Kempe}{2003}]%
        {kempe2003quantum}
\bibfield{author}{\bibinfo{person}{Julia Kempe}.}
  \bibinfo{year}{2003}\natexlab{}.
\newblock \showarticletitle{Quantum random walks: an introductory overview}.
\newblock \bibinfo{journal}{\emph{Contemporary Physics}} \bibinfo{volume}{44},
  \bibinfo{number}{4} (\bibinfo{year}{2003}), \bibinfo{pages}{307--327}.
\newblock


\bibitem[\protect\citeauthoryear{Kubota, Kakutani, Kato, Kawano, and
  Sakurada}{Kubota et~al\mbox{.}}{2013}]%
        {KubotaKKKS13}
\bibfield{author}{\bibinfo{person}{Takahiro Kubota}, \bibinfo{person}{Yoshihiko
  Kakutani}, \bibinfo{person}{Go Kato}, \bibinfo{person}{Yasuhito Kawano},
  {and} \bibinfo{person}{Hideki Sakurada}.} \bibinfo{year}{2013}\natexlab{}.
\newblock \showarticletitle{Automated Verification of Equivalence on Quantum
  Cryptographic Protocols}. In \bibinfo{booktitle}{\emph{5th International
  Symposium on Symbolic Computation in Software Science, {SCSS} 2013, Castle of
  Hagenberg, Austria}} \emph{(\bibinfo{series}{EPiC Series in Computing})},
  \bibfield{editor}{\bibinfo{person}{Laura Kov{\'{a}}cs} {and}
  \bibinfo{person}{Temur Kutsia}} (Eds.), Vol.~\bibinfo{volume}{15}.
  \bibinfo{publisher}{EasyChair}, \bibinfo{pages}{64--69}.
\newblock
\urldef\tempurl%
\url{http://www.easychair.org/publications/paper/143661}
\showURL{%
\tempurl}


\bibitem[\protect\citeauthoryear{K\"{u}mmerer and Schwieger}{K\"{u}mmerer and
  Schwieger}{2016}]%
        {KS16}
\bibfield{author}{\bibinfo{person}{Burkhard K\"{u}mmerer} {and}
  \bibinfo{person}{Kay Schwieger}.} \bibinfo{year}{2016}\natexlab{}.
\newblock \showarticletitle{Diagonal couplings of quantum Markov chains}.
\newblock \bibinfo{journal}{\emph{Infinite Dimensional Analysis, Quantum
  Probability and Related Topics}} \bibinfo{volume}{19}, \bibinfo{number}{2}
  (\bibinfo{year}{2016}), \bibinfo{pages}{1650012}.
\newblock


\bibitem[\protect\citeauthoryear{Li and Unruh}{Li and Unruh}{2019}]%
        {LU19}
\bibfield{author}{\bibinfo{person}{Yangjia Li} {and} \bibinfo{person}{Dominique
  Unruh}.} \bibinfo{year}{2019}\natexlab{}.
\newblock \showarticletitle{Quantum Relational Hoare Logic with Expectations}.
\newblock \bibinfo{journal}{\emph{CoRR}}  \bibinfo{volume}{abs/1903.08357}
  (\bibinfo{year}{2019}).
\newblock
\showeprint[arxiv]{1903.08357}
\urldef\tempurl%
\url{http://arxiv.org/abs/1903.08357}
\showURL{%
\tempurl}


\bibitem[\protect\citeauthoryear{Lindvall}{Lindvall}{2002}]%
        {Lindvall02}
\bibfield{author}{\bibinfo{person}{Torgny Lindvall}.}
  \bibinfo{year}{2002}\natexlab{}.
\newblock \bibinfo{booktitle}{\emph{Lectures on the coupling method}}.
\newblock \bibinfo{publisher}{Courier Corporation}.
\newblock


\bibitem[\protect\citeauthoryear{Mosca, Tapp, and de~Wolf}{Mosca
  et~al\mbox{.}}{2000}]%
        {MTW00}
\bibfield{author}{\bibinfo{person}{Michele Mosca}, \bibinfo{person}{Alain
  Tapp}, {and} \bibinfo{person}{Ronald de Wolf}.}
  \bibinfo{year}{2000}\natexlab{}.
\newblock \showarticletitle{Private quantum channels and the cost of
  randomizing quantum information}.
\newblock \bibinfo{journal}{\emph{arXiv preprint quant-ph/0003101}}
  (\bibinfo{year}{2000}).
\newblock
\urldef\tempurl%
\url{https://arxiv.org/abs/quant-ph/0003101}
\showURL{%
\tempurl}


\bibitem[\protect\citeauthoryear{Nielsen and Chuang}{Nielsen and
  Chuang}{2002}]%
        {NC00}
\bibfield{author}{\bibinfo{person}{Michael~A Nielsen} {and}
  \bibinfo{person}{Isaac Chuang}.} \bibinfo{year}{2002}\natexlab{}.
\newblock \bibinfo{booktitle}{\emph{Quantum computation and quantum
  information}}.
\newblock \bibinfo{publisher}{Cambridge University Press}.
\newblock


\bibitem[\protect\citeauthoryear{Peres}{Peres}{1996}]%
        {Per96}
\bibfield{author}{\bibinfo{person}{Asher Peres}.}
  \bibinfo{year}{1996}\natexlab{}.
\newblock \showarticletitle{Separability criterion for density matrices}.
\newblock \bibinfo{journal}{\emph{Physical Review Letters}}
  \bibinfo{volume}{77}, \bibinfo{number}{8} (\bibinfo{year}{1996}),
  \bibinfo{pages}{1413}.
\newblock


\bibitem[\protect\citeauthoryear{Selinger}{Selinger}{2004a}]%
        {Se04}
\bibfield{author}{\bibinfo{person}{Peter Selinger}.}
  \bibinfo{year}{2004}\natexlab{a}.
\newblock \showarticletitle{A Brief Survey of Quantum Programming Languages}.
  In \bibinfo{booktitle}{\emph{Functional and Logic Programming, 7th
  International Symposium, {FLOPS} 2004, Nara, Japan, April 7-9, 2004,
  Proceedings}} \emph{(\bibinfo{series}{Lecture Notes in Computer Science})},
  \bibfield{editor}{\bibinfo{person}{Yukiyoshi Kameyama} {and}
  \bibinfo{person}{Peter~J. Stuckey}} (Eds.), Vol.~\bibinfo{volume}{2998}.
  \bibinfo{publisher}{Springer}, \bibinfo{pages}{1--6}.
\newblock
\urldef\tempurl%
\url{https://doi.org/10.1007/978-3-540-24754-8\_1}
\showDOI{\tempurl}


\bibitem[\protect\citeauthoryear{Selinger}{Selinger}{2004b}]%
        {Selinger04}
\bibfield{author}{\bibinfo{person}{Peter Selinger}.}
  \bibinfo{year}{2004}\natexlab{b}.
\newblock \showarticletitle{Towards a quantum programming language}.
\newblock \bibinfo{journal}{\emph{Mathematical Structures in Computer Science}}
  \bibinfo{volume}{14}, \bibinfo{number}{4} (\bibinfo{year}{2004}),
  \bibinfo{pages}{527--586}.
\newblock
\urldef\tempurl%
\url{https://doi.org/10.1017/S0960129504004256}
\showDOI{\tempurl}


\bibitem[\protect\citeauthoryear{Strassen}{Strassen}{1965}]%
        {strassen1965existence}
\bibfield{author}{\bibinfo{person}{Volker Strassen}.}
  \bibinfo{year}{1965}\natexlab{}.
\newblock \showarticletitle{The existence of probability measures with given
  marginals}.
\newblock \bibinfo{journal}{\emph{The Annals of Mathematical Statistics}}
  (\bibinfo{year}{1965}), \bibinfo{pages}{423--439}.
\newblock
\urldef\tempurl%
\url{http://projecteuclid.org/euclid.aoms/1177700153}
\showURL{%
\tempurl}


\bibitem[\protect\citeauthoryear{Thorisson}{Thorisson}{2000}]%
        {Thorisson00}
\bibfield{author}{\bibinfo{person}{Hermann Thorisson}.}
  \bibinfo{year}{2000}\natexlab{}.
\newblock \bibinfo{booktitle}{\emph{Coupling, Stationarity, and Regeneration}}.
\newblock \bibinfo{publisher}{springer}.
\newblock


\bibitem[\protect\citeauthoryear{Unruh}{Unruh}{2018}]%
        {Unruh18fullversion}
\bibfield{author}{\bibinfo{person}{Dominique Unruh}.}
  \bibinfo{year}{2018}\natexlab{}.
\newblock \showarticletitle{Quantum Relational Hoare Logic}.
\newblock  (\bibinfo{year}{2018}).
\newblock
\showeprint[arxiv]{quant-ph/1802.03188}


\bibitem[\protect\citeauthoryear{Unruh}{Unruh}{2019a}]%
        {Unruh19}
\bibfield{author}{\bibinfo{person}{Dominique Unruh}.}
  \bibinfo{year}{2019}\natexlab{a}.
\newblock \showarticletitle{Quantum Hoare Logic with Ghost Variables}. In
  \bibinfo{booktitle}{\emph{2019 34th Annual ACM/IEEE Symposium on Logic in
  Computer Science (LICS)}}. \bibinfo{pages}{1--13}.
\newblock
\urldef\tempurl%
\url{https://doi.org/10.1109/LICS.2019.8785779}
\showDOI{\tempurl}


\bibitem[\protect\citeauthoryear{Unruh}{Unruh}{2019b}]%
        {Unruh18}
\bibfield{author}{\bibinfo{person}{Dominique Unruh}.}
  \bibinfo{year}{2019}\natexlab{b}.
\newblock \showarticletitle{Quantum Relational Hoare Logic}.
\newblock \bibinfo{journal}{\emph{Proc. ACM Program. Lang.}}
  \bibinfo{volume}{3}, \bibinfo{number}{POPL}, Article \bibinfo{articleno}{33}
  (\bibinfo{date}{Jan.} \bibinfo{year}{2019}), \bibinfo{numpages}{31}~pages.
\newblock
\showISSN{2475-1421}
\urldef\tempurl%
\url{https://doi.org/10.1145/3290346}
\showDOI{\tempurl}


\bibitem[\protect\citeauthoryear{Van~der Waerden}{Van~der Waerden}{1953}]%
        {VDW53}
\bibfield{author}{\bibinfo{person}{BL Van~der Waerden}.}
  \bibinfo{year}{1953}\natexlab{}.
\newblock \bibinfo{booktitle}{\emph{Modern Algebra, Volume II}}.
\newblock \bibinfo{publisher}{Frederick Ungar}.
\newblock


\bibitem[\protect\citeauthoryear{Venegas-Andraca}{Venegas-Andraca}{2012}]%
        {venegas2012quantum}
\bibfield{author}{\bibinfo{person}{Salvador~El{\'\i}as Venegas-Andraca}.}
  \bibinfo{year}{2012}\natexlab{}.
\newblock \showarticletitle{Quantum walks: a comprehensive review}.
\newblock \bibinfo{journal}{\emph{Quantum Information Processing}}
  \bibinfo{volume}{11}, \bibinfo{number}{5} (\bibinfo{year}{2012}),
  \bibinfo{pages}{1015--1106}.
\newblock


\bibitem[\protect\citeauthoryear{Villani}{Villani}{2008}]%
        {Villani08}
\bibfield{author}{\bibinfo{person}{C{\'e}dric Villani}.}
  \bibinfo{year}{2008}\natexlab{}.
\newblock \bibinfo{booktitle}{\emph{Optimal transport: {O}ld and new}}.
\newblock \bibinfo{publisher}{springer}.
\newblock


\bibitem[\protect\citeauthoryear{Watrous}{Watrous}{2018}]%
        {Wat18}
\bibfield{author}{\bibinfo{person}{John Watrous}.}
  \bibinfo{year}{2018}\natexlab{}.
\newblock \bibinfo{booktitle}{\emph{The theory of quantum information}}.
\newblock \bibinfo{publisher}{Cambridge University Press}.
\newblock


\bibitem[\protect\citeauthoryear{Winter}{Winter}{2016}]%
        {Winter16}
\bibfield{author}{\bibinfo{person}{Andreas Winter}.}
  \bibinfo{year}{2016}\natexlab{}.
\newblock \showarticletitle{Tight uniform continuity bounds for quantum
  entropies: conditional entropy, relative entropy distance and energy
  constraints}.
\newblock \bibinfo{journal}{\emph{Communications in Mathematical Physics}}
  \bibinfo{volume}{347}, \bibinfo{number}{1} (\bibinfo{year}{2016}),
  \bibinfo{pages}{291--313}.
\newblock


\bibitem[\protect\citeauthoryear{Ying}{Ying}{2011}]%
        {Ying11}
\bibfield{author}{\bibinfo{person}{Mingsheng Ying}.}
  \bibinfo{year}{2011}\natexlab{}.
\newblock \showarticletitle{Floyd-Hoare logic for quantum programs}.
\newblock \bibinfo{journal}{\emph{{ACM} Trans. Program. Lang. Syst.}}
  \bibinfo{volume}{33}, \bibinfo{number}{6} (\bibinfo{year}{2011}),
  \bibinfo{pages}{19:1--19:49}.
\newblock
\urldef\tempurl%
\url{https://doi.org/10.1145/2049706.2049708}
\showDOI{\tempurl}


\bibitem[\protect\citeauthoryear{Ying}{Ying}{2016}]%
        {Ying16}
\bibfield{author}{\bibinfo{person}{Mingsheng Ying}.}
  \bibinfo{year}{2016}\natexlab{}.
\newblock \bibinfo{booktitle}{\emph{Foundations of Quantum Programming}}.
\newblock \bibinfo{publisher}{Morgan-Kaufmann}.
\newblock


\bibitem[\protect\citeauthoryear{Ying, Yu, Feng, and Duan}{Ying
  et~al\mbox{.}}{2013}]%
        {YYFD13}
\bibfield{author}{\bibinfo{person}{Mingsheng Ying}, \bibinfo{person}{Nengkun
  Yu}, \bibinfo{person}{Yuan Feng}, {and} \bibinfo{person}{Runyao Duan}.}
  \bibinfo{year}{2013}\natexlab{}.
\newblock \showarticletitle{Verification of quantum programs}.
\newblock \bibinfo{journal}{\emph{Sci. Comput. Program.}} \bibinfo{volume}{78},
  \bibinfo{number}{9} (\bibinfo{year}{2013}), \bibinfo{pages}{1679--1700}.
\newblock
\urldef\tempurl%
\url{https://doi.org/10.1016/j.scico.2013.03.016}
\showDOI{\tempurl}


\bibitem[\protect\citeauthoryear{Zhou, Ying, Yu, and Ying}{Zhou
  et~al\mbox{.}}{2019a}]%
        {quantumstrassen}
\bibfield{author}{\bibinfo{person}{Li Zhou}, \bibinfo{person}{Shenggang Ying},
  \bibinfo{person}{Nengkun Yu}, {and} \bibinfo{person}{Mingsheng Ying}.}
  \bibinfo{year}{2019}\natexlab{a}.
\newblock \showarticletitle{{S}trassen's theorem for quantum couplings}.
\newblock \bibinfo{journal}{\emph{Theoretical Computer Science}}
  (\bibinfo{year}{2019}).
\newblock
\showISSN{0304-3975}
\urldef\tempurl%
\url{https://doi.org/10.1016/j.tcs.2019.08.026}
\showURL{%
\tempurl}


\bibitem[\protect\citeauthoryear{Zhou, Yu, and Ying}{Zhou
  et~al\mbox{.}}{2019b}]%
        {ZYY19}
\bibfield{author}{\bibinfo{person}{Li Zhou}, \bibinfo{person}{Nengkun Yu},
  {and} \bibinfo{person}{Mingsheng Ying}.} \bibinfo{year}{2019}\natexlab{b}.
\newblock \showarticletitle{An Applied Quantum Hoare Logic}. In
  \bibinfo{booktitle}{\emph{Proceedings of the 40th ACM SIGPLAN Conference on
  Programming Language Design and Implementation}} \emph{(\bibinfo{series}{PLDI
  2019})}. \bibinfo{publisher}{ACM}, \bibinfo{address}{New York, NY, USA},
  \bibinfo{pages}{1149--1162}.
\newblock
\showISBNx{978-1-4503-6712-7}
\urldef\tempurl%
\url{https://doi.org/10.1145/3314221.3314584}
\showDOI{\tempurl}


\end{thebibliography}
	
	\newpage
	{\Large \textbf{Supplementary material and deferred proofs}}
	\renewcommand{\thesection}{A\arabic{section}}

\section{Probabilistic couplings}\label{pcpc}

In this section, we briefly recall the basics of probabilistic
couplings so that the reader will see a close and natural
correspondence as well as some essential differences between
probabilistic coupling and their quantum counterparts.

Let $\mathcal{A}$ be a countable set. Then a sub-distribution over $\mathcal{A}$ is a mapping $\mu:\mathcal{A}\rightarrow [0,1]$ such that $\sum_{a\in\mathcal{A}}\mu (a)\leq 1$. In paricular, if $\sum_{a\in\mathcal{A}}\mu (a)= 1$, then $\mu$ is called a distribution over $\mathcal{A}$. We can define: 
\begin{enumerate}
	\item[(1)] the weight of $\mu$ is $$|\mu|=\sum_{a\in\mathcal{A}}\mu(a);$$
	\item[(2)] the support of $\mu$ is $\supp(\mu)=\{a\in\mathcal{A}: \mu(a)>0\};$
	\item[(3)] the probability of an event $S\subseteq\mathcal{A}$ is $$\mu(S)=\sum_{a\in S}\mu(a).$$
\end{enumerate}
If $\mu$ be a join sub-distribution, i.e. a sub-distribution over Cartesian product $\mathcal{A}_1\times\mathcal{A}_2$, then its marginals over $\mathcal{A}_1$ is defined by
$$\pi_1(\mu)(a_1)=\sum_{a_2\in\mathcal{A}_2}\mu(a_1,a_2)\ {\rm for\ every}\ a_1\in\mathcal{A}_1.$$
Similarly, we can define its marginals $\pi_2(\mu)$ over $\mathcal{A}_2$.

\begin{defn}[Coupling]Let $\mu_1, \mu_2$ be sub-distributions over $\mathcal{A}_1,\mathcal{A}_2$, respectively. Then a sub-distribution over $\mathcal{A}_1\times\mathcal{A}_2$ is called a coupling for $\cp{\mu_1}{\mu_2}$ if $\pi_1(\mu)=\mu_1$ and $\pi_2(\mu)=\mu_2$.
\end{defn}

\begin{exam}\label{exam-triv} For any distributions $\mu_1,\mu_2$ over $\mathcal{A}_1,\mathcal{A}_2$, respectively, the independent or trivial coupling is:
	$\mu_\times (a_1,a_2)=\mu_1(a_1)\cdot\mu_2(a_2).$
\end{exam}

\begin{exam}\label{exam-bit} Let $\mathbf{Flip}$ be the uniform distribution over booleans, i.e. $$\mathbf{Flip}(0)=\mathbf{Flip}(1)=\frac{1}{2}.$$ Then the following are two couplings for $\cp{\mathbf{Flip}}{\mathbf{Flip}}$:\begin{enumerate}
		\item Identity coupling: $\mu_{\rm id}(a_1,a_2)=\begin{cases}\frac{1}{2}\ &{\rm if}\ a_1=a_2,\\ 0\ &{\rm otherwise}.\end{cases}$
		\item Negation coupling: $\mu_\neg(a_1,a_2)=\begin{cases}\frac{1}{2}\ &{\rm if}\ \neg a_1=a_2,\\ 0\ &{\rm otherwise}.\end{cases}$
	\end{enumerate}
\end{exam}

\begin{exam}\label{exam-unif} As a generalisation of Example \ref{exam-bit}, let $\mathbf{Unif}_\mathcal{A}$ be the uniform distribution over a finite nonempty set $\mathcal{A}$, i.e. $$\mathbf{Unif}_\mathcal{A}(a)=\frac{1}{|\mathcal{A}|}$$ for every $a\in\mathcal{A}$. Then each bijection $f:\mathcal{A}\rightarrow\mathcal{A}$ yields a coupling for $\cp{\mathbf{Unif}_\mathcal{A}}{\mathbf{Unif}_\mathcal{A}}$:
	$$\mu_f(a_1,a_2)=\begin{cases}\frac{1}{|\mathcal{A}|}\ &{\rm if}\ f(a_1)=a_2,\\ 0\ &{\rm otherwise}.\end{cases}$$
\end{exam}

\begin{exam}\label{exam-id} For any sub-distribution $\mu$ over $\mathcal{A}$, the identity coupling for $\cp{\mu}{\mu}$ is: $$\mu_{\rm id}(a_1,a_2)=\begin{cases}\mu(a)\ &{\rm if}\ a_1=a_2 =a,\\ 0\ &{\rm otherwise}.\end{cases}$$
\end{exam}

\begin{defn}[Lifting] Let $\mu_1,\mu_2$ be sub-distributions over $\mathcal{A}_1,\mathcal{A}_2$, respectively, and let $\mathcal{A}\subseteq\mathcal{A}_1\times\mathcal{A}_2$ be a relation. Then a sub-distribution $\mu$ over $\mathcal{A}_1\times\mathcal{A}_2$ is called a witness for the $\mathcal{R}$-lifting of $\cp{\mu_1}{\mu_2}$ if: \begin{enumerate}\item $\mu$ is a coupling for $\cp{\mu_1}{\mu_2}$; \item $\supp(\mu)\subseteq \mathcal{R}$.
	\end{enumerate} Whenever a witness exists, we say that $\mu_1$ and $\mu_2$ are related by the $\mathcal{R}$-lifting and write $\mu_1\mathcal{R}^{\#}\mu_2$.
\end{defn}

\begin{exam}\begin{enumerate}\item Coupling $\mu_f$ in Example \ref{exam-unif} is a witness for the lifting: $$\mathbf{Unif}_\mathcal{A}\{(a_1,a_2)|f(a_1)=a_2\}^\#\mathbf{Unif}_\mathcal{A}.$$
		\item Coupling $\mu_{\rm id}$ in Example \ref{exam-id} is a witness for the lifting $\mu =^\# \mu.$
		\item Coupling $\mu_\times$ in Example \ref{exam-triv} is a witness for the lifting $\mu_1 T^\# \mu_2$, where $T=\mathcal{A}_1\times\mathcal{A}_2.$
\end{enumerate}\end{exam}

\begin{prop}Let $\mu_1,\mu_2$ be sub-distributions over $\mathcal{A}_1, \mathcal{A}$, respectively. If there exists a coupling for $\cp{\mu_1}{\mu_2}$, then their weight are equal: $|\mu_1|=|\mu_2|.$\end{prop}

\begin{prop} Let $\mu_1,\mu_2$ be sub-distributions over the same $\mathcal{A}$. Then $\mu_1=\mu_2$ if and only if $\mu_1 =^{\#} \mu_2$.
\end{prop}

\section{Hoare logic for quantum programs}\label{QHL-app}
In this section, we review the Hoare-like logic for quantum
programs developed in \cite{Ying11, Ying16}.
\subsection{Partial Correctness and Total Correctness}
D'Hondt and Panangaden \cite{DP06} suggested to use effects as quantum
predicates. Then a correctness formula (or a Hoare triple) is a
statement of the form $\{A\}P\{B\}$, where $P$ is a quantum program,
and both $A, B$ are quantum predicates in $\mathcal{H}_P$, called the
precondition and postcondition, respectively.

\begin{defn}[Partial and Total Correctness]\label{correctness-interpretation}
	\begin{enumerate}\item The correctness formula $\{A\}P\{B\}$ is true in
		the sense of total correctness, written $\models_{\mathit{tot}}\{A\}P\{B\},$ if for all
		$\rho\in\mathcal{D}^\le(\mathcal{H}_P)$ we have: $$\tr(A\rho)\leq \tr(B\sm{P} (\rho)).$$
		
		\item The correctness formula $\{A\}P\{B\}$ is true in
		the sense of partial correctness, written $\models_{\mathit{par}}\{A\}P\{B\},$ if for all
		$\rho\in\mathcal{D}^\le(\mathcal{H}_P)$ we have: $$\tr(A\rho)\leq \tr(B\sm{P} (\rho))+
		[\tr(\rho)-\tr(\sm{P} (\rho))].$$
	\end{enumerate}
\end{defn}

The defining inequalities of total and partial correctness can be easily understood by observing that the interpretation of $\tr(A\rho)$ in physics is the expectation (i.e. average value) of observable $A$ in state $\rho$, and $\tr(\rho)-\tr(\sm{P}(\rho))$ is indeed the probability that with input $\rho$ program $P$ does not terminate.

\subsection{Proof System}
The proof system qPD for partial correctness consists of the axioms and inference rules presented in Figure \ref{fig 3.2}.
\begin{figure}[h]\centering
	\begin{equation*}\begin{split}
	&({\rm Ax.Sk})\ \{A\}\mathbf{Skip}\{A\}\ \ \ \ \ \ \ \ \ \ \ \ \ \ \ \ \ \ \ \ \ \ \ \ \ \ \\  &({\rm Ax.Init}) \left\{\sum_{i}|i\rangle_q\langle 0|A|0\rangle_q\langle
	i|\right\}q:=|0\rangle\{A\}\\
	&({\rm Ax.UT})\
	\{U^{\dag}AU\}\overline{q}:=U\left[\overline{q}\right]\{A\}\ \ \ \ \ \ \ \\ &({\rm R.SC})\ \ \
	\frac{\{A\}P_1\{B\}\ \ \ \ \ \ \{B\}P_2\{C\}}{\{A\}P_1;P_2\{C\}}\\
	&({\rm R.IF})\
	\frac{\{A_m\}P_m\{B\}\ {\rm for\ every}\ m}{\left\{\sum_m
		M_m^{\dag}A_mM_m\right\}\mathbf{if}\ (\square m\cdot
		\cM[\overline{q}]=m\rightarrow P_m)\ \mathbf{fi}\{B\}}\\
	&({\rm R.LP})\
	\frac{\{B\}P\left\{M_0^{\dag}AM_0+M_1^{\dag}BM_1\right\}}{\{M_0^{\dag}AM_0+M_1^{\dag}BM_1\}\mathbf{while}\
		\cM[\overline{q}]=1\ \mathbf{do}\ P\ \mathbf{od}\{A\}}\\
	&({\rm R.Or})\ \frac{A\sqsubseteq
		A^{\prime}\ \ \ \ \{A^{\prime}\}P\{B^{\prime}\}\ \ \ \
		B^{\prime}\sqsubseteq B}{\{A\}P\{B\}}
	\end{split}\end{equation*}
	\caption{Proof System qPD.}\label{fig 3.2}
\end{figure}
It was proved in \cite{Ying11} to be sound and (relatively) complete.
\begin{thm}[Soundness and Completeness]\label{sound-complete} For any quantum program $P$, and for any quantum predicates $A,B$: $$\models_\mathit{par}\{A\}P\{B\}\Leftrightarrow\ \vdash_\mathit{qPD}\{A\}P\{B\}.$$
\end{thm}

A sound and (relatively) complete proof system for total correctness was also developed in \cite{Ying11, Ying16}.

\section{Deferred Proofs and Verifications}

\subsection{Proof of Proposition \ref{prop-equal}}

\begin{proof} (1$\Rightarrow$2) Obvious.
	
	{\vskip 4pt}
	
	(2$\Rightarrow$3) It is trivial if we realize the fact $=_{\mathcal{B}}\ \subseteq\ =_{sym}$ for any othornormal basis $\mathcal{B}$.
	
	{\vskip 4pt}
	
	(3$\Rightarrow$1) Suppose that $\sigma$ is a witness of $\rho_1 =_{sym}^\# \rho_2$. Then we have $$(=_{sym})\sigma = \sigma\ {\rm and}\ \sigma (=_{sym}) = \sigma,$$ which imply $S\sigma=\sigma$ and $\sigma S = \sigma$. Therefore $\sigma = S\sigma S$, and moreover,
	\begin{align*}
	\rho_1 = \tr_2(\sigma) = \tr_2(S\sigma S) = \tr_1(\sigma ) = \rho_2
	\end{align*}
\end{proof}

\subsection{Proof of Proposition \ref{entangled-witness}}\label{en}

For two Hilbert space $\mathcal{H}_1$ and $\mathcal{H}_2$, we write $\mathcal{S}(\mathcal{H}_1\otimes\mathcal{H}_2)$ for the set of separable states $\rho\in\mathcal{D}^\le(\mathcal{H}_1\otimes\mathcal{H}_2)$. It is clear from the definition of separability that $\mathcal{S}(\mathcal{H}_1\otimes\mathcal{H}_2)$ is a convex set. For any $\rho_1\in\mathcal{D}^\le(\mathcal{H}_1)$ and $\rho_2\in\mathcal{D}^\le(\mathcal{H}_2)$ with the same trace, let $\mathcal{C}(\rho_1,\rho_2)$ be the set of all couplings for $\cp{\rho_1}{\rho_2}$. It is easy to see that $$\mathcal{C}(\rho_1,\rho_2)\cap\mathcal{S}(\mathcal{H}_1\otimes\mathcal{H}_2)$$ is non-empty.

Proposition \ref{entangled-witness} is an immediate corollary of the following:

\begin{fact}\label{separable-state}
	There are  states $\rho_1$ in $\mathcal{H}_1$, $\rho_2$ in $\mathcal{H}_2$ and separable quantum predicate $A$ over $\mathcal{H}_1\otimes\mathcal{H}_2$ such that
	\begin{equation*}
	\max_{\rho\in\mathcal{C}(\rho_1,\rho_2)\cap\mathcal{S}(\mathcal{H}_1\otimes\mathcal{H}_2)}\tr(A\rho)\ <\ \max_{\rho\in\mathcal{C}(\rho_1,\rho_2)}\tr(A\rho).
	\end{equation*}
\end{fact}

\begin{proof} Let both $\mathcal{H}_1$ and $\mathcal{H}_2$ be the two dimensional Hilbert space, and
	\begin{align*}
	&A = \frac{1}{3}\left[
	\begin{array}{cccc}
	2 & 0 & 0 & 1 \\
	0 & 1 & 0 & 0 \\
	0 & 0 & 1 & 0 \\
	1 & 0 & 0 & 2 \\
	\end{array}
	\right],\ B = \frac{1}{2}\left[
	\begin{array}{cccc}
	1 & 0 & 0 & 1 \\
	0 & 0 & 0 & 0 \\
	0 & 0 & 0 & 0 \\
	1 & 0 & 0 & 1 \\
	\end{array}
	\right], \ \rho_1=\rho_2=\frac{1}{2}\left[
	\begin{array}{cc}
	1 & 0\\
	0 & 1\\
	\end{array}
	\right],\\
	&\rho_s = \frac{1}{4}\left[
	\begin{array}{cccc}
	1 & 0 & 0 & 1 \\
	0 & 1 & 0 & 0 \\
	0 & 0 & 1 & 0 \\
	1 & 0 & 0 & 1 \\
	\end{array}
	\right],\
	\rho_e = \frac{1}{2}\left[
	\begin{array}{cccc}
	1 & 0 & 0 & 1 \\
	0 & 0 & 0 & 0 \\
	0 & 0 & 0 & 0 \\
	1 & 0 & 0 & 1 \\
	\end{array}
	\right].
	\end{align*}
	%Separability Criterion for Density Matrices Asher Peres*
	In the $2\times2$ case, the Peres-Horodecki criterion (also called the PPT criterion, for Positive Partial Transpose) is a necessary and sufficient condition for the separability of a density matrix (see \cite{Per96} for more details). We can check that both $A$ and $\rho_s$ are separable using this criterion. Moreover, $A$ is a quantum predicate as $0\sqsubseteq A\sqsubseteq I$, and $\rho_s, \rho_e$ are both the couplings for $\cp{\rho_1}{\rho_2}$. We are going to show that
	\begin{align*}
	\max_{\rho\in\mathcal{C}(\rho_1,\rho_2)\cap\mathcal{S}(\mathcal{H}_1,\mathcal{H}_2)}\tr(A\rho) = \tr(A\rho_s) = \frac{2}{3}
	< \max_{\rho\in\mathcal{C}(\rho_1,\rho_2)} \tr(A\rho) = \tr(A\rho_e) = 1.\end{align*}
	The right hand side is obvious as $\tr(A\rho_e)$ already achieves the maximum 1. What remains to show is that for any separable state $\sigma$ that couples $(\rho_1,\rho_2)$, $\tr(A\sigma)\le\frac{2}{3}$. Write $\sigma$ and its partial transpose (with respect to $\mathcal{H}_2$) in the matrix form using the fact that it is a coupling for $\cp{\rho_1}{\rho_2}$:
	$$\sigma = \left[
	\begin{array}{cccc}
	x & \cdot & \cdot & y \\
	\cdot & \frac{1}{2}-x & \cdot & \cdot \\
	\cdot & \cdot & \frac{1}{2}-x & \cdot \\
	y^* & \cdot & \cdot & x \\
	\end{array}\right],\
	\sigma^{T_2} = \left[
	\begin{array}{cccc}
	x & \cdot & \cdot & \cdot \\
	\cdot & \frac{1}{2}-x & y & \cdot \\
	\cdot & y^* & \frac{1}{2}-x & \cdot \\
	\cdot & \cdot & \cdot & x \\
	\end{array}
	\right]$$
	with real parameter $x$ and complex number $y$ while `$\cdot$' represents the parameters we are not interested in. Note that $\sigma$ is non-negative so that $0\le x\le \frac{1}{2}$ and $x^2\ge yy^*$. On the other hand, $\sigma$ is separable, so $\sigma^{T_2}$ is non-negative, which implies $(\frac{1}{2}-x)^2\ge yy^*$. Thus,
	\begin{align*}
	\tr(A\sigma) = \frac{1}{3}(1+2x+y+y^*)\le\frac{1}{3}(1+2x+2|y|)\le\frac{1}{3}\left(1+2x+2\min\left\{x,\frac{1}{2}-x\right\}\right)\le\frac{2}{3}.
	\end{align*}
\end{proof}

\subsection{Proof of Lemma \ref{prop-finite}}

For any two Hilbert spaces $\mathcal{X}$ and $\mathcal{Y}$. We write $\mathrm{L}(\mathcal{X},\mathcal{Y})$ for the set of linear operators from $\mathcal{X}$ to $\mathcal{Y}$.
Then we can introduce an operator-vector correspondence; i.e. the mapping:
$$\mathrm{vec}: \mathrm{L}(\mathcal{X},\mathcal{Y})\mapsto \mathcal{Y}\otimes\mathcal{X}$$ defined by $$\mathrm{vec}(|b\rangle\langle a|) = |b\rangle|a\rangle$$ for the standard basis elements $|a\rangle$ and $|b\rangle$. Intuitively, it represents a change of bases from the standard basis of $\mathrm{L}(\mathcal{X},\mathcal{Y})$ to the standard basis of $\mathcal{Y}\otimes\mathcal{X}$ (see \cite{Wat18}, Chapter 2).
This mapping is a bijection, and indeed an isometry, in the sense that
$$\langle M, N\rangle = \langle \mathrm{vec}(M), \mathrm{vec}(N) \rangle = \tr(M^\dag N) = \mathrm{vec}(M)^\dag\mathrm{vec}(N)$$
for all $M,N\in\mathrm{L}(\mathcal{X},\mathcal{Y})$.
Moreover, for any linear map $A: \mathcal{D}^\le(\mathcal{H})\mapsto\mathcal{D}^\le(\mathcal{H})$ with $d=\dim\mathcal{H}$, we can define a new linear map by a $d^2\times d^2$ matrix $\hat{A}$ in the following way:
$$A(\rho) = \sigma\ \Longleftrightarrow\ \hat{A}\cdot\mathrm{vec}(\rho) = \mathrm{vec}(\sigma).$$
Now, let us write $Y_1 = M_{10}M_{10}^\dag$ and $Y_2 = M_{20}M_{20}^\dag$, and define linear maps $\hat{A}_1$ and $\hat{A}_2$ from the respective linear maps $A_1=\sm{P_1}\circ\mathcal{E}_{11}$ and $A_2=\sm{P_2}\circ\mathcal{E}_{21}$.
Then the following facts are trivial:
\begin{align*}\tr(\mathcal{E}_{10}\circ(\sm{P_1}\circ\mathcal{E}_{11})^n(\rho_1)) &= \tr(Y_1^\dag A_1^n(\rho_1)) = \mathrm{vec}(Y_1)^\dag\hat{A}^n_1\mathrm{vec}(\rho_1) \\
\tr(\mathcal{E}_{20}\circ(\sm{P_2}\circ\mathcal{E}_{21})^n(\rho_2)) &= \tr(Y_2^\dag A_2^n(\rho_2)) = \mathrm{vec}(Y_2)^\dag\hat{A}^n_2\mathrm{vec}(\rho_2).\end{align*}
Moreover, if we define two vectors $|\alpha\rangle = \mathrm{vec}(Y_1)\oplus \mathrm{vec}(Y_2)$, $|\beta\rangle = \mathrm{vec}(\rho_1) \oplus (-\mathrm{vec}(\rho_2))$ and $\hat{A} = \hat{A}_1 \oplus \hat{A}_2$, then we obtain:
\begin{align*}
Z_n &= \tr(\mathcal{E}_{10}\circ(\sm{P_1}\circ\mathcal{E}_{11})^n(\rho_1)) - \tr(\mathcal{E}_{20}\circ(\sm{P_2}\circ\mathcal{E}_{21})^n(\rho_2)) \\
&= \mathrm{vec}(Y_1)^\dag\hat{A}^n_1\mathrm{vec}(\rho_1) - \mathrm{vec}(Y_2)^\dag\hat{A}^n_2\mathrm{vec}(\rho_2) \\
&= (\mathrm{vec}(Y_1)^\dag,\mathrm{vec}(Y_2)^\dag)\left(
\begin{array}{cc}
\hat{A}^n_1 & \\
& \hat{A}^n_2
\end{array}\right)\left(
\begin{array}{c}
\mathrm{vec}(\rho_1) \\
-\mathrm{vec}(\rho_2)
\end{array}
\right) \\
&= \langle \alpha|\hat{A}^n|\beta\rangle
\end{align*}
using the fact that $\hat{A}^n = \hat{A}^n_1 \oplus \hat{A}^n_2$.
With Cayley-Hamilton theorem (see Chapter XV, \S112 in \cite{VDW53}), we know that a matrix power of any order $k$ can be written as a matrix polynomial of degree at most $N - 1$ where $N$ is the dimension of the matrix. So, for any $n\ge0$, there exists coefficients $c_0,c_1,\cdots,c_{d_1^2+d_2^2-1}$ such that
$$\hat{A}^n = \sum_{i=0}^{d_1^2+d_2^2-1}c_i \hat{A}^i,$$
noting that $\hat{A}$ is a $(d_1^2+d_2^2)\times(d_1^2+d_2^2)$ matrix.
Thus, if $Z_n = 0$ holds for any $0\le n\le d_1^2+d_2^2-1$, then for any $n\ge0$,
$$
Z_n = \langle \alpha|\hat{A}^n|\beta\rangle = \langle \alpha|\left(\sum_{i=0}^{d_1^2+d_2^2-1}c_i \hat{A}^i\right)|\beta\rangle = \sum_{i=0}^{d_1^2+d_2^2-1}c_iZ_i = 0.
$$

\subsection{Proof of Soundness Theorem (Theorem \ref{soundness-thm})}\label{sound-proof}

Let us first prove soundness of rules (SO), (SO-L) and (SO-R) because some conclusions obtained in this proof will be used in the proof of other rules. Given a quantum operation $\mathcal{E}$ with its Kraus representation $$\mathcal{E}(\rho) = \sum_iE_i\rho E_i^\dag,$$ its dual map $\mathcal{E}^*$ is defined by $$\mathcal{E}^*(A) = \sum_iE_i^\dag AE_i$$ for any operator $A$.

\begin{prop}
	For any quantum predicate $A$ and trace-preserving quantum operation $\mathcal{E}$, $\mathcal{E}^*(A)$ is still a quantum predicate. That is, $0\leq \mathcal{E}^*(A)\leq I$ is ture for any $0\leq A\leq I$.
\end{prop}

Moreover, we introduce a mathematical (rather than logical) version of Definition \ref{valid-ju}: for any quantum operations $\mathcal{E}$ and $\mathcal{F}$ in a Hilbert space $\mathcal{H}_1$, $\mathcal{H}_2$, respectively, and for any quantum predicates $A$ and $B$ in $\mathcal{H}_1\otimes\mathcal{H}_2$, we write:
$$\models\mathcal{E}\sim\mathcal{F}: A\Rightarrow B$$ if for any $\rho\in\mathcal{D}^\le\left(\mathcal{H}_1\otimes \mathcal{H}_2\right)$, there exists a coupling $\sigma\in\mathcal{D}^\le\left(\mathcal{H}_1\otimes \mathcal{H}_2\right)$ for $\cp{\mathcal{E} (\tr_{\langle 2\rangle}\rho)}{\mathcal{F}(\tr_{\langle 1\rangle}\rho)}$ such that \begin{equation}\tr(A\rho)\leq\tr(B\sigma)+\tr(\rho)-\tr(\sigma).\end{equation}

\begin{lem}\label{tech-SO} For any trace-preserving quantum operations $\mathcal{E},\mathcal{F}$, we have:
	$$\models \mathcal{E}\sim\mathcal{F}:(\mathcal{E}\otimes\mathcal{F})^*(A)\Rightarrow A
	$$
	where $(\mathcal{E}\otimes\mathcal{F})^*=\mathcal{E}^*\otimes\mathcal{F}^*$ is the dual map of $\mathcal{E}\otimes\mathcal{F}$. \end{lem}

\begin{proof} (Outline) We only need to prove that for any input $\rho\in\mathcal{D}(\mathcal{H}_{1}\otimes\mathcal{H}_{2})$, $(\mathcal{E}\otimes\mathcal{F})(\rho)$ is a coupling for $\cp{ \mathcal{E}(\tr_{\langle2\rangle}\rho)}{ \mathcal{F}(\tr_{\langle1\rangle}\rho)}$. This fact can be better understood by a physical argument. Suppose Alice and Bob share a bipartite state $\rho\in\mathcal{D}(\mathcal{H}_1\otimes\mathcal{H}_2)$. Next, Alice and Bob do the quantum operation $\mathcal{E}$ (mapping states in $\mathcal{H}_1$ to $\mathcal{H}_{1^\prime}$) and $\mathcal{F}$ (mapping states in $\mathcal{H}_2$ to $\mathcal{H}_{2^\prime}$) respectively. After that, the state they shared is $(\mathcal{E}\otimes\mathcal{F})(\rho)$. As a quantum operation can be realized by the unitary transformation over the principle system and the environment, in Alice's view, she cannot gain any information about what Bob does on his part of the state without any classical communication. So, $\tr_{\h_{2^\prime}}[(\mathcal{E}\otimes\mathcal{F})(\rho)] = \mathcal{E}(\tr_{\h_2}\rho)$.
	Therefore, $(\mathcal{E}\otimes\mathcal{F})(\rho)$ is still a coupling for $\cp{ \mathcal{E}(\tr_{\langle2\rangle}\rho)}{ \mathcal{F}(\tr_{\langle1\rangle}\rho)}$, and we have:
	\begin{equation*}
	\tr[A(\mathcal{E}\otimes\mathcal{F})(\rho)] = \tr[(\mathcal{E}\otimes\mathcal{F})^*(A)\rho].
	\end{equation*}
	
	(Details) Mathematical details for proving that $(\mathcal{E}\otimes\mathcal{F})(\rho)$ is a coupling for $\cp{ \mathcal{E}(\tr_{\langle2\rangle}\rho)}{ \mathcal{F}(\tr_{\langle1\rangle}\rho)}$ for all $\rho\in\D(\h_1\otimes\h_2)$.  As $\E$ and $\F$ are both trace-preserving quantum operations, their Kraus operators $\{E_i\}$ and $\{F_j\}$ must satisfy:
	$$\sum_iE_i^\dag E_i = I_{\h_1},\quad\sum_jF_j^\dag F_j = I_{\h_2}.$$
	On the other hand, for any $\rho\in\D(\h_1\otimes\h_2)$, due to the Schmidt decomposition theorem (see for example \cite{NC00}, Theorem 2.7), it can be written as:
	$$\rho = \sum_{l}p_l\left(\sum_{k}\lambda_{l_k}|\alpha_{l_k}\>|\beta_{l_k}\>\right)\left(\sum_{k^\prime}\lambda_{l_{k^\prime}}^\dag\<\alpha_{l_{k^\prime}}|\<\beta_{l_{k^\prime}}|\right) = \sum_{lkk^\prime}p_l\lambda_{l_k}\lambda_{l_{k^\prime}}^\dag|\alpha_{l_k}\>\<\alpha_{l_{k^\prime}}|\otimes|\beta_{l_k}\>\<\beta_{l_{k^\prime}}|$$
	where $p_l\in[0,1]$, $\lambda_{l_k}$ are complex numbers, and $\{|\alpha_{l_k}\>\}_k$ and $\{|\beta_{l_{k}}\>\}_k$ are orthonormal basis of $\h_1$ and $\h_2$ respectively for all $l$. Observe the following fact by straightforward calculations:
	\begin{align*}
	\E(\tr_{\h_2}(\rho)) &= \sum_i\left[E_i \sum_{l}p_l\left(\sum_{kk^\prime}\lambda_{l_k}\lambda_{l_{k^\prime}}^\dag|\alpha_{l_k}\>\<\alpha_{l_{k^\prime}}|\cdot\<\beta_{l_{k^\prime}}|\beta_{l_k}\>\right)  E_i^\dag\right] \\
	&= \sum_{ilk}p_l\lambda_{l_k}\lambda_{l_{k}}^\dag\left[E_i |\alpha_{l_k}\>\<\alpha_{l_{k}}|  E_i^\dag\right],
	\end{align*}
	\begin{align*}
	\tr_{\h_{2^\prime}}\left[(\E\otimes\F)(\rho)\right] 
	&= \tr_{\h_{2^\prime}}\left\{\sum_{ijlkk^\prime} p_l\lambda_{l_k}\lambda_{l_{k^\prime}}^\dag \left(E_i|\alpha_{l_k}\>\<\alpha_{l_{k^\prime}}|E_i^\dag\right)\otimes\left(F_j|\beta_{l_k}\>\<\beta_{l_{k^\prime}}|F_j^\dag\right)\right\} \\
	&= \sum_{ilkk^\prime} p_l\lambda_{l_k}\lambda_{l_{k^\prime}}^\dag \left(E_i|\alpha_{l_k}\>\<\alpha_{l_{k^\prime}}|E_i^\dag\right)\cdot\tr\left(\sum_jF_j^\dag F_j|\beta_{l_k}\>\<\beta_{l_{k^\prime}}|\right) \\
	&= \sum_{ilk}p_l\lambda_{l_k}\lambda_{l_{k}}^\dag\left[E_i |\alpha_{l_k}\>\<\alpha_{l_{k}}|  E_i^\dag\right] = \E(\tr_{\h_2}(\rho))
	\end{align*}
	and similarly, $\tr_{\h_{1^\prime}}\left[(\E\otimes\F)(\rho)\right] = \F(\tr_{\h_1}(\rho)),$ which complete the proof.
\end{proof}

\begin{rem}\begin{enumerate}\item The above lemma can help to simplify the proofs of of the validity of rules (Init) and (UT). Furthermore, we can use it to create/trace-out qubits.
		
		\item Note that the trace-preserving condition in the above lemma is necessary. So, it doesn't help when we deal with a single outcome of a measurement. But when we encounter an equivalence between two $\mathbf{if}$ statements of the form $\mathbf{if}\cdots\sim\mathbf{if}\cdots$, all of the sub-programs are all trace-preserving and have simple Kraus representations, we can regard the whole $\mathbf{if}$ statements as a quantum operation, and the above lemma applies.
\end{enumerate}\end{rem}

{\vskip 4pt}

Now soundness of rules (SO), (SO-L) and (SO-R) can be easily proved. The validity of (SO) is just a corollary of Lemma \ref{tech-SO}. Realise that the semantic function of ${\bf skip}$ is the identity operation, so (SO-L) and (SO-R) are two special cases of (SO).

In addition, the following technical lemma will be needed in the remainder of this subsection. 

\begin{lem}
	\label{part-measure}
	Consider a bipartite system $A$ and $B$ with state Hilbert space $\hs_A$ and $\hs_B$, respectively. Given two linear operators $M: \hs_A\mapsto\hs_{A^\prime}$ and $N: \hs_B\mapsto\hs_{B^\prime}$ such that $N^\dag N\le I_{\hs_B}$. Then for any $\rho\in\D(\hs_A\otimes\hs_B)$, it holds that:
	\begin{equation}
	\label{part-measure-ineq}
	\tr_{\hs_{B^\prime}} [(M\otimes N)\rho(M\otimes N)^\dag] \sqsubseteq \tr_{\hs_{B^\prime}}[(M\otimes I_{\hs_B})\rho(M\otimes I_{\hs_B})^\dag] = M\tr_{\hs_B}(\rho)M^\dag
	\end{equation}
	In particular, if linear operators $N_1,N_2,\cdots,N_n$ mapping from $\hs_B$ to $\hs_{B^\prime}$ satisfy $\sum_{i=1}^nN_i^\dag N_i \sqsubseteq I_{\hs_B}$, then:
	\begin{equation}
	\label{part-measure-eq}
	\sum_{i=1}^n\tr_{\hs_{B^\prime}} [(M\otimes N_i)\rho(M\otimes N_i)^\dag] \sqsubseteq M\tr_{\hs_B}(\rho)M^\dag.
	\end{equation}
	More generally, for any quantum operations $\E_A: \D^\le(\hs_A)\mapsto\D^\le(\hs_{A^\prime})$ and $\F_B: \D^\le(\hs_B)\mapsto\D^\le(\hs_{B^\prime})$, it holds that
	$$ \tr_{\hs_{B^\prime}}[\E_A\otimes\E_B(\rho)] \sqsubseteq \E_A(\tr_{\hs_B}(\rho)).$$
\end{lem}

\begin{proof}
	Actually, if the dimensions of $\hs_A$ and $\hs_B$ are not same, then we can always enlarge the small one to the same dimension. So, without of generality, we assume $\hs_A$ and $\hs_B$ are the same $\hs$ of dimension $d$ (finite or infinite). Then, it is possible to define the maximal entangled operator $\Phi$:
	$$\Phi = \sum_{i,j=0}^{d-1}|i\>_A\<j|\otimes|i\>_B\<j|$$
	where $\{|i\>\}$ is an orthonormal basis of $\hs$. For any pure state $|\psi\>$ in $\hs\otimes\hs$, there exists linear operator $X$ such that:
	$$
	|\psi\>\<\psi| = (X\otimes I)\Phi(X\otimes I)^\dag.
	$$
	It is straightforward to show that
	\begin{align*}
	\tr_{\hs_{B^\prime}} [(M\otimes N)|\psi\>\<\psi|(M\otimes N)^\dag] 
	=\ & \tr_{\hs_{B^\prime}} [(MX\otimes N)\Phi(MX\otimes N)^\dag] \\
	=\ & \tr_{\hs_{B^\prime}} \left[\sum_{i,j=0}^{d-1}(MX\otimes N)|i\>_A\<j|\otimes|i\>_B\<j|(MX\otimes N)^\dag\right] \\
	=\ & \sum_{i,j=0}^{d-1}\tr_{\hs_{B^\prime}} [MX|i\>_A\<j|X^\dag M^\dag \otimes N|i\>_B\<j| N^\dag] \\
	=\ & \sum_{i,j=0}^{d-1} MX|i\>\<j|X^\dag M^\dag  \cdot   \<j| N^\dag N|i\> \\
	=\ & \sum_{i,j=0}^{d-1} MX|i\>\<i|( N^\dag N)^T|j\>\<j|X^\dag M^\dag \\
	=\ & MX(N^\dag N)^TX^\dag M^\dag \\
	\sqsubseteq \ & MXX^\dag M^\dag \\
	=\ & \tr_{\hs_{B^\prime}} [(M\otimes I_\h)|\psi\>\<\psi|(M\otimes I_\h)^\dag]
	\end{align*}
	using the assumption $N^\dag N\le I_{\hs}$ which leads to $(N^\dag N)^T\le I_{\hs}$. Therefore, for any density operators, inequality (\ref{part-measure-ineq}) holds.
	
	If linear operators $N_1,N_2,\cdots,N_n$ over $\hs_B$ satisfy $\sum_{i=1}^nN_i^\dag N_i = I_{\hs_B}$, then for any pure state $|\psi\>$ in $\hs\otimes\hs$, we have:
	\begin{align*}
	\sum_{i=1}^n\tr_{\hs_{B^\prime}} [(M\otimes N_i)\rho(M\otimes N_i)^\dag]
	=\ & \sum_{i=1}^nMX(N_i^\dag N_i)^TX^\dag M^\dag \\
	\sqsubseteq\ & MXX^\dag M^\dag \\
	=\ & M\tr_{\hs_{B}}(|\psi\>\<\psi|)M^\dag.
	\end{align*}
	The linearity ensures the truth of Eqn. (\ref{part-measure-eq}).
	
	Assume that $\E_A(\cdot) = \sum_mE_m(\cdot)E_m^\dag$ and $\F_A(\cdot) = \sum_nF_n(\cdot)F_n^\dag$, then for any $\rho$:
	\begin{align*}
	\tr_{\hs_{B^\prime}}[\E_A\otimes\F_B(\rho)] &= \tr_{\hs_{B^\prime}}\left[\sum_{m,n}(E_m\otimes F_n)\rho (E_m\otimes F_n)^\dag\right] \\
	&= \sum_m\sum_n\tr_{\hs_{B^\prime}}\left[(E_m\otimes F_n)\rho (E_m\otimes F_n)^\dag\right] \\
	&\sqsubseteq \sum_mE_m\tr_{\hs_{B}}(\rho)E_m^\dag \\
	&= \E_A[\tr_{\hs_{B}}(\rho)],
	\end{align*}
	as $\sum_nF_n^\dag F_n\sqsubseteq I_\h$ ($\F$ is a quantum operation).
\end{proof}

\textit{\textbf{Now we are ready to prove the whole Theorem \ref{soundness-thm}}}. It suffices to prove validity of all axioms and inference rules in Figures \ref{fig 4.4-0}, \ref{fig 4.5-0}, \ref{fig 4.5}, \ref{fig 4.6} and \ref{fig 4.7}:

{\vskip 4pt}

$\bullet$ (Skip) Obvious.

{\vskip 4pt}

$\bullet$ (Init) For any $\rho\in\mathcal{D}^\le(\mathcal{H}_{P_1}\otimes\mathcal{H}_{P_2})$, set:
\begin{align*}
\sigma &= \sum_{ij}\left(|0\rangle_{q_1\langle1\rangle}\langle i|\otimes|0\rangle_{q_2\langle2\rangle}\langle j|\right)\rho\left(|i\rangle_{q_1\langle1\rangle}\langle 0|\otimes|j\rangle_{q_2\langle2\rangle}\langle 0|\right)= (\mathcal{E}\otimes\mathcal{F})(\rho),
\end{align*}
where $\mathcal{E}$ and $\mathcal{F}$ are two initial quantum operations with Kraus operators $\{|0\rangle_{q_1\langle1\rangle}\langle i|\}$ and $\{|0\rangle_{q_2\langle2\rangle}\langle i|\}$, respectively.
As $\mathcal{E}$ and $\mathcal{F}$ are all trace-preserving, we know that $\tr(\rho) = \tr(\sigma)$, and  $(\mathcal{E}\otimes\mathcal{F})(\rho)$ is a coupling for
\begin{align*}
&\cp{ \mathcal{E}(\tr_{\langle 2 \rangle}(\rho))}{ \mathcal{F}(\tr_{\langle 1 \rangle}(\rho)) }\\ 
=\ &\cp{ \sum_{i}\left(|0\rangle_{q_1\langle1\rangle}\langle i|\right)(\tr_{\langle 2 \rangle}(\rho))\left(|i\rangle_{q_1\langle1\rangle}\langle 0|\right)}{ \sum_{j}\left(|0\rangle_{q_2\langle2\rangle}\langle j|\right)(\tr_{\langle 1 \rangle}(\rho))\left(|j\rangle_{q_2\langle2\rangle}\langle 0|\right) } \\
=\ &\cp{ \sm{ q_1:=|0\rangle}(\tr_{\langle 2 \rangle}(\rho))}{ \sm{ q_2:=|0\rangle}(\tr_{\langle 1 \rangle}(\rho)) }
\end{align*}
Moreover, we have:
\begin{align*}
&\tr\Bigg\{ \bigg[\sum_{i,j}\left(|i\rangle_{q_1\langle 1\rangle}\langle 0|\otimes |j\rangle_{q_2\langle 2\rangle}\langle 0|\right) A\left( |0\rangle_{q_1\langle 1\rangle}\langle
i|\otimes |0\rangle_{q_2\langle 2\rangle}\langle j| \right)\bigg] \rho \Bigg\} \\
=\ &\tr\bigg[ A\sum_{i,j} \left( |0\rangle_{q_1\langle 1\rangle}\langle
i|\otimes |0\rangle_{q_2\langle 2\rangle}\langle j| \right) \rho \left(|i\rangle_{q_1\langle 1\rangle}\langle 0|\otimes |j\rangle_{q_2\langle 2\rangle}\langle 0|\right)\bigg]
= \tr( A\sigma).
\end{align*}

{\vskip 4pt}

$\bullet$ (UT) For any $\rho\in\mathcal{D}^\le(\mathcal{H}_{P_1}\otimes\mathcal{H}_{P_2})$, we set $$\sigma = (U_1\otimes U_2)\rho(U_1^\dag\otimes U_2^\dag).$$ Of course, $\tr(\rho) = \tr(\sigma)$.
Due to the unitarity of $U_1$ and $U_2$, we have:
\begin{align*}
\tr_{\langle 2 \rangle}(\sigma) &= \tr_{\langle 2 \rangle}((U_1\otimes U_2)\rho(U_1^\dag\otimes U_2^\dag)) = U_1 \tr_{\langle 2 \rangle}(\rho) U_1^\dag = \sm{\overline{q_1}:=U_1\left[\overline{q_1}\right]}(\tr_{\langle 2 \rangle}(\rho))
\end{align*}
and $$\tr_{\langle 2 \rangle}(\sigma) = \sm{\overline{q_2}:=U_2\left[\overline{q_2}\right]}(\tr_{\langle 1 \rangle}(\rho)),$$
which ensure that $\sigma$ is a coupling for
$$\cp{ \sm{ \overline{q_1}:=U_1\left[\overline{q_1}\right] }(\tr_{\langle 2 \rangle}(\rho))}{ \sm{ \overline{q_2}:=U_2\left[\overline{q_2}\right] }(\tr_{\langle 1 \rangle}(\rho))  }.$$
Moreover, it holds that
\begin{align*}
&\tr\left[ \left(U_1^\dag\otimes U_2^\dag\right) A\left( U_1 \otimes U_2\right) \rho \right]
=\ \tr\left[ A\left( U_1 \otimes U_2\right) \rho \left(U_1^\dag\otimes U_2^\dag\right)\right]
=\ \tr( A\sigma).
\end{align*}

{\vskip 4pt}

$\bullet$ (SC+) For $j = 1,2$, we write $\mathcal{H}_{\langle j\rangle}$ for the input state Hilbert space of program $P_j;P_j^\prime$. Assume that $$\text{(i)}\ \Gamma\models P_1\sim P_2: A\Rightarrow B;\ \text{(ii)}\ \Delta\models P_1^\prime\sim P_2^\prime: B\Rightarrow C;\ \text{and (iii)}\ \Gamma\stackrel{(P_1,P_2)}{\models}\Delta.$$ We are going to show that $$\Gamma\models P_1;P_2^\prime\sim P_2;P_2^\prime: A\Rightarrow C.$$
For any $\rho\in\mathcal{D}^\le(\mathcal{H}_{\langle 1\rangle}\otimes\mathcal{H}_{\langle 2\rangle})$, if $\rho\models\Gamma$, then
it follows from assumption (i) that there exists a coupling $\sigma$ for $\cp{ \sm{ P_1 } (\tr_{\langle 2\rangle}\rho)}{  \sm{ P_2 } (\tr_{\langle 1\rangle}\rho) }$ such that $$\tr(A\rho)\le\tr(B\sigma)+\tr(\rho)-\tr(\sigma).$$
Thus, by assumption (iii) and Definition \ref{def-entail} we see that $\sigma\models\Delta$. This together with assumption (ii) implies that there exists a coupling $\delta$ for $\cp{ \sm{P_1^\prime} (\tr_{\langle 2\rangle}\sigma)}{  \sm{P_2^\prime} (\tr_{\langle 1\rangle}\sigma) }$ such that $$\tr(B\sigma)\le\tr(C\delta)+\tr(\sigma)-\tr(\delta).$$ Consequently, it holds that $$\tr(A\rho)\le\tr(C\delta)+\tr(\rho)-\tr(\delta).$$
Therefore, it suffices to show that $\delta$ is a coupling for $\cp{ \sm{P_1;P_1^\prime} (\tr_{\langle 2\rangle}\rho)}{  \sm{P_2;P_2^\prime} (\tr_{\langle 1\rangle}\rho) }$. This is obvious as
$$ \tr_{\langle 2\rangle}\delta = \sm{P_1^\prime} (\tr_{\langle 2\rangle}\sigma) = \sm{P_1^\prime} (\sm{P_1} (\tr_{\langle 2\rangle}\rho)) = \sm{P_1;P_1^\prime} (\tr_{\langle 2\rangle}\rho),$$
and similarly $\tr_{\langle 1\rangle}\delta = \sm{P_2;P_2^\prime} (\tr_{\langle 1\rangle}\rho)$.

{\vskip 4pt}

$\bullet$ (IF) Suppose that $\models P_{1m}\sim P_{2n}: B_{mn}\Rightarrow C$ for every $(m,n)\in S$. Then for any input $\rho$ with reduced density matrices $\rho_1 = \tr_{\langle2\rangle}(\rho),\rho_2 = \tr_{\langle1\rangle}(\rho)$, if $(m,n)\in S$, then there exists a coupling $\sigma_{mn}$ for
$$\cp{\sm{P_{1m}}\left(\tr_{\langle 2\rangle}\big((M_{1m}\otimes M_{2n})\rho(M_{1m}\otimes M_{2n})^\dag\big)\right)}{ \sm{P_{2n}}\left(\tr_{\langle 1\rangle}\big((M_{1m}\otimes M_{2n})\rho(M_{1m}\otimes M_{2n})^\dag\big)\right)}$$ such that
$$\tr\left[B_{mn}\left((M_{1m}\otimes M_{2n})\rho(M_{1m}\otimes M_{2n})^\dag
\right)\right]\leq\tr\left(C\sigma_{mn}\right).$$
Now we set $\sigma=\sum_{(m,n)\in S}\sigma_{mn}$. Then using Proposition 3.3.1(v) in \cite{Ying16} and Lemma \ref{part-measure} we obtain: \begin{align*}\tr_{\langle 2\rangle}(\sigma)=\sum_{(m,n)\in S} \tr_{\langle 2\rangle}(\sigma_{mn}) 
&=\sum_{(m,n)\in S} \sm{P_{1m}}\left(\tr_{\langle 2\rangle}\big((M_{1m}\otimes M_{2n})\rho(M_{1m}\otimes M_{2n})^\dag\big)\right)
\\
&=\sum_m \sm{P_{1m}}\left(\sum_n\tr_{\langle 2\rangle}\big((M_{1m}\otimes M_{2n})\rho(M_{1m}\otimes M_{2n})^\dag\big)\right)
\\
&\sqsubseteq \sum_m \sm{P_{1m}}\left(M_{1m} \tr_{\langle 2\rangle}(\rho) M_{1m}^\dag \right)
\\
&=\sm{\mathbf{if}\ (\square m\cdot
	M_1[\overline{q}]=m\rightarrow P_{1m})\ \mathbf{fi}} (\rho_1),
\end{align*} 
and set the difference $\delta_1 = \sm{\mathbf{if}\ (\square m\cdot
	M_1[\overline{q}]=m\rightarrow P_{1m})\ \mathbf{fi}} (\rho_1) - \tr_{\langle 2\rangle}(\sigma)$.
Similarly, we have:
$$\tr_{\langle 1\rangle}(\sigma) \sqsubseteq \sm{\mathbf{if}\ (\square n\cdot
	M_2[\overline{q}]=n\rightarrow P_{2n})\ \mathbf{fi}} (\rho_2),$$
and set $\delta_2 = \sm{\mathbf{if}\ (\square n\cdot
	M_1[\overline{q}]=n\rightarrow P_{2n})\ \mathbf{fi}} (\rho_2) - \tr_{\langle 1\rangle}(\sigma).$  
The premises of lossless guarantee the existence of coupling for outputs, and therefore, $\delta_1$ and $\delta_2$ has the same trace. So, $\sigma + \delta_1\otimes\delta_2/tr(\delta_1)$ is a coupling for $$\cp{\sm{\mathbf{if}\ (\square m\cdot
	M_1[\overline{q}]=m\rightarrow P_{1m})\ \mathbf{fi}} (\rho_1)}{\ \sm{\mathbf{if}\ (\square m\cdot
	M_2[\overline{q}]=m\rightarrow P_{2m})\ \mathbf{fi}} (\rho_2)}.$$
Moreover, it holds that
\begin{align*}&\tr\left[\sum_{(m,n)\in S}
\left(M_{1m}^{\dag}\otimes M_{2n}^\dag\right)B_{mn}\left(M_{1m}\otimes M_{2n}\right)\rho\right]\\ 
=\ & \sum_{(m,n)\in S} \tr\left[\left(M_{1m}^{\dag}\otimes M_{2n}^\dag\right)B_{mn}\left(M_{1m}\otimes M_{2n}\right)\rho\right]\\ 
=\ & \sum_{(m,n)\in S} \tr\left[B_{mn}\left(M_{1m}\otimes M_{2n}\right)\rho\left(M_{1m}^{\dag}\otimes M_{2n}^\dag\right)\right]\\ 
\leq\ &\sum_{(m,n)\in S}
\tr\left(C\sigma_{mn}\right)=\tr\left(C\sigma \right) \le \tr\left(C(\sigma + \delta_1\otimes\delta_2/tr(\delta_1))\right).
\end{align*} Therefore, it holds that
\begin{align*}\models \mathbf{if}\ (\square m\cdot
M_1[\overline{q}]=m\rightarrow P_{1m})\ \mathbf{fi}\sim \mathbf{if}\ (\square n\cdot
M_2[\overline{q}]=n\rightarrow P_{2n})\ \mathbf{fi}:\ \ \ \ \ \ \ \ \ \ \ \ \ \ \ \ \\
\sum_m
\left(M_{1m}^{\dag}\otimes M_{2n}^\dag\right)B_{mn}\left(M_{1m}\otimes M_{2n}\right)\Rightarrow C.
\end{align*}

{\vskip 4pt}

$\bullet$ (IF1) For any $\rho\in\mathcal{D}^\le(\mathcal{H}_{P_1}\otimes\mathcal{H}_{P_2})$ satisfying the condition $\cM_1\approx \cM_2$, due to the assumption that $$\cM_1\approx \cM_2\models A\Rightarrow\{B_m\},$$ for each possible outcome $m$, there exists a coupling $\sigma_m$ for
$\cp{ M_{1m}(\tr_{\langle 2\rangle}\rho)M_{1m}^\dag}{M_{2m}(\tr_{\langle 1\rangle}\rho)M_{2m}^\dag},$
such that $\tr(\rho) = \sum_m\tr(\sigma_m)$ and $$\tr(A\rho)\le \tr\Big(\sum_mB_m\sigma_m\Big).$$
Moreover, according to the second assumption that $$\models P_{1m}\sim P_{2m}: B_m\Rightarrow C\ {\rm for\ every}\ m,$$ for each $\sigma_m$, there exists a coupling $\delta_m$ for
$\cp{ \sm{P_{1m}} (\tr_{\langle 2\rangle}\sigma_m)}{ \sm{P_{2m}} (\tr_{\langle 1\rangle}\sigma_m) }$
such that
$$\tr(B_m\sigma_m) \le \tr(C\delta_m)+\tr\sigma_m-\tr\delta_m.$$
Set $\delta = \sum_m\delta_m$. Then we obtain:
\begin{align*}
\tr(A\rho) &\le \tr\Big(\sum_mB_m\sigma_m\Big)  \le \sum_m\left[\tr(C\delta_m)+\tr\sigma_m-\tr\delta_m\right]  = \tr(C\delta) + \tr(\rho) -  \tr(\delta).
\end{align*}
Moreover, it is straightforward to see
\begin{align*}
\tr_{\langle 2\rangle}\delta &= \sum_m\tr_{\langle 2\rangle}\delta_m
= \sum_m \sm{P_{1m}} (\tr_{\langle 2\rangle}\sigma_m) = \sum_m \sm{P_{1m}} ( M_{1m}(\tr_{\langle 2\rangle}\rho)M_{1m}^\dag) \\
&= \sm{\mathbf{if}\ (\square m\cdot
	\cM_1[\overline{q}]=m\rightarrow P_{1m})\ \mathbf{fi}} (\tr_{\langle 2\rangle}\rho)
\end{align*}
and similarly $$\tr_{\langle 1\rangle}\delta = \sm{\mathbf{if}\ (\square m\cdot
	\cM_2[\overline{q}]=m\rightarrow P_{2m})\ \mathbf{fi}} (\tr_{\langle 1\rangle}\rho).$$

{\vskip 4pt}

$\bullet$ (LP) For simplicity,
let us use $Q_1$ to denote the \textbf{while}-statement: $$\mathbf{while}\ \cM_1[\overline{q}]=1\ \mathbf{do}\ P_1\ \mathbf{od}$$ and $Q_2$ to denote $$\mathbf{while}\
\cM_2[\overline{q}]=1\ \mathbf{do}\ P_2\ \mathbf{od}.$$ We also introduce an auxiliary notation: for $i = 1,2$, quantum operation $\mathcal{E}_{i0}$ and $\mathcal{E}_{i1}$ are defined by the measurement $\cM_i$
$$\mathcal{E}_{i0}(\rho) = M_{i0}\rho M_{i0}^\dag,\quad \mathcal{E}_{i1}(\rho) = M_{i1}\rho M_{i1}^\dag.$$
The lossless condition ensures that both $P_1$ and $P_2$ are terminating subprograms. For simplicity, we use notations $M_0 = M_{10}\otimes M_{20}$ and $M_1 = M_{11}\otimes M_{21}$.
For any $\rho\in\mathcal{D}(\mathcal{H}_{Q_1}\otimes\mathcal{H}_{Q_2})$, we inductively define the following sequence of states:
\begin{enumerate}
	\item $\sigma_0 = \rho$;
	\item for input $M_1\sigma_nM_1^\dag$, due to the second assumption, $\sigma_{n+1}$ is the coupling of the outputs
	$$\cp{\sm{P_1}(\tr_{\<2\>}(M_1\sigma_nM_1^\dag))}{\sm{P_2}(\tr_{\<1\>}(M_1\sigma_nM_1^\dag))}$$
	such that:
	$$
	\tr(BM_1\sigma_nM_1^\dag)\le \tr\{[M_0^\dag AM_0
	+ M_1^\dag BM_1]\sigma_{n+1}\}.
	$$
\end{enumerate}
We first prove by induction that for all $n\ge0$, it holds that
$$
\tr_{\<2\>}\sigma_{n} \sqsubseteq (\sm{P_1}\circ\E_{11})^n(\tr_{\<2\>}\rho), \quad
\tr_{\<1\>}\sigma_{n} \sqsubseteq (\sm{P_2}\circ\E_{21})^n(\tr_{\<1\>}\rho).
$$
It is trivial for $n=0$. Suppose it holds for $n = k$, then for $n=k+1$, we have
\begin{align*}
\tr_{\<2\>}\sigma_{k+1} &=
\sm{P_1}(\tr_{\<2\>}[(M_{11}\otimes M_{21})\sigma_n(M_{11}^\dag\otimes M_{21}^\dag)]) \\
&\sqsubseteq \sm{P_1} M_{11}\tr_{\<2\>}(\sigma_n)M_{11}^\dag \\
&\sqsubseteq (\sm{P_1}\circ\E_{11})[(\sm{P_1}\circ\E_{11})^k(\tr_{\<2\>}\rho)] \\
&= (\sm{P_1}\circ\E_{11})^{k+1}(\tr_{\<2\>}\rho)
\end{align*}
using Lemma \ref{part-measure}.

Actually, the assumption that $Q_1$ is lossless ensures that
$\lim_{n\rightarrow\infty}(\sm{P_1}\circ\E_{11})^n(\tr_{\<2\>}\rho) = 0$ as its trace characterizes the probability of nontermination. Therefore, $\lim_{n\rightarrow\infty}\sigma_n = 0$.
Together with the choice of $\sigma_{n}$, we have
\begin{align*}
&\tr\bigg(A \sum_{n=0}^{\infty}(M_0\sigma_nM_0^\dag)\bigg) \\
\ge\ & \sum_{n=0}^{\infty}\left[\tr(AM_0\sigma_nM_0^\dag) + \tr(BM_1\sigma_nM_1^\dag)
- \tr\{[M_0^\dag AM_0 + M_1^\dag BM_1]\sigma_{n+1}\}\right] \\
=\ &\tr(AM_0\sigma_0M_0^\dag) + \tr(BM_1\sigma_0M_1^\dag) - \lim_{n\rightarrow\infty}\tr\{[M_0^\dag AM_0 + M_1^\dag BM_1]\sigma_{n}\} \\
=\ &\tr\{[M_0^\dag AM_0 + M_1^\dag BM_1]\rho\}.
\end{align*}
Let us define $\tau = \sum_{n=0}^{\infty}(M_0\sigma_nM_0^\dag)$. According to Lemma \ref{part-measure} again, the following inequality holds:
\begin{align*}
\tr_{\<2\>}(\tau) &= \sum_{n=0}^{\infty}\tr_{\<2\>}[(M_{10}\otimes M_{20})\sigma_n(M_{10}\otimes M_{20})^\dag] \\
&\sqsubseteq \sum_{n=0}^{\infty}M_{10}\tr_{\<2\>}(\sigma_n)M_{10}^\dag \\
&\sqsubseteq \sum_{n=0}^{\infty}\E_{10}[(\sm{P_1}\circ\E_{11})^n(\tr_{\<2\>}\rho)] \\
&= \sum_{n=0}^{\infty} [\E_{10}\circ(\sm{P_1}\circ\E_{11})^n](\tr_{\<2\>}\rho) \\
&= \sm{Q_1}(\tr_{\<2\>}\rho).
\end{align*}
Similarly, $\tr_{\<1\>}(\tau)\sqsubseteq\sm{Q_2}(\tr_{\<1\>}\rho)$. Set: $$\delta_1 = \sm{Q_1}(\tr_{\<2\>}\rho) - \tr_{\<2\>}(\tau)\ {\rm and}\ 
\delta_2 = \sm{Q_2}(\tr_{\<1\>}\rho) - \tr_{\<1\>}(\tau),$$ the terminating condition of both $\bf while$ programs guarantees that $\delta_1$ and $\delta_2$ have the same trace, and therefore, $\tau + \delta_1\otimes\delta_2/\tr(\delta_1)$ is a coupling of $\cp{\sm{Q_1}(\tr_{\<2\>}\rho)}{\sm{Q_2}(\tr_{\<1\>}\rho)}$ and satisfies
$$
\tr\{[M_0^\dag AM_0 + M_1^\dag BM_1]\rho\}\le\tr(A\tau)\le \tr[A(\tau + \delta_1\otimes\delta_2/\tr(\delta_1))].
$$

{\vskip 4pt}

$\bullet$ (LP1) We use the same notations as in above proof of (LP). For any input state $\rho\in\mathcal{D}^\le(\mathcal{H}_{Q_1}\otimes\mathcal{H}_{Q_2})$ satisfying the condition $$(\cM_1,P_1)\approx (\cM_2, P_2),$$ and for any $n\ge0$, we must have:
\begin{align*}
\tr\left[\left(\mathcal{E}_{10}\circ(\sm{P_1}\circ \mathcal{E}_{11})^n\right)(\tr_{\langle 2\rangle}\rho)\right] &=
\tr\left[\left(\mathcal{E}_{20}\circ(\sm{P_2}\circ \mathcal{E}_{21})^n\right)(\tr_{\langle 1\rangle}\rho)\right],
\end{align*}
which implies that any coupling for $$\cp{ (\sm{P_1}\circ \mathcal{E}_{11})^n(\tr_{\langle 2\rangle}\rho) }{ (\sm{P_2}\circ \mathcal{E}_{21})^n(\tr_{\langle 1\rangle}\rho) }$$ must satisfies the condition $\cM_1\approx \cM_2$.
Choose $\delta_0 = \rho$ so that $\delta_0$ is a coupling for $\cp{\tr_{\langle 2\rangle}\rho}{ \tr_{\langle 1\rangle}\rho}$, the following statement holds for $n\ge0$:
\begin{align*}
&\textbf{Statement:} \\
&\quad\exists \text{\ a coupling\ } \sigma_{0n} \text{\ for\ } \cp{ \left(\mathcal{E}_{10}\circ(\sm{P_1}\circ \mathcal{E}_{11})^n\right)(\tr_{\langle 2\rangle}\rho)}{ \left(\mathcal{E}_{20}\circ(\sm{P_2}\circ \mathcal{E}_{21})^n\right)(\tr_{\langle 1\rangle}\rho) }, \\
&\quad\exists \text{\ a coupling\ } \sigma_{1n} \text{\ for\ } \cp{ \left(\mathcal{E}_{11}\circ(\sm{P_1}\circ \mathcal{E}_{11})^n\right)(\tr_{\langle 2\rangle}\rho)}{ \left(\mathcal{E}_{21}\circ(\sm{P_2}\circ \mathcal{E}_{21})^n\right)(\tr_{\langle 1\rangle}\rho) }, \\
&\quad\exists \text{\ a coupling\ } \delta_{n+1} \text{\ for\ } \cp{ (\sm{P_1}\circ \mathcal{E}_{11})^{n+1}(\tr_{\langle 2\rangle}\rho)}{ (\sm{P_2}\circ \mathcal{E}_{21})^{n+1}(\tr_{\langle 1\rangle}\rho) }, \\
&\ \ \text{such that:} \\
&\quad\quad\quad\quad\tr(\delta_n)=\tr(\sigma_{0n}) + \tr(\sigma_{1n}),\\
&\quad\quad\quad\quad\tr(A\delta_n)\le\tr(B_0\sigma_{0n}) + \tr(B_1\sigma_{1n}),\\
&\quad\quad\quad\quad\tr(B_1\sigma_{1n})\le\tr(A\delta_{n+1}) + \tr(\sigma_{1n})-\tr(\delta_{n+1}).
\end{align*}
The above statement can be proved by induction. For basis $n=0$, of course, $\rho$ satisfies $\cM_1\approx \cM_2$, due to the assumption $$\cM_1\approx \cM_2\models A\Rightarrow\{B_0,B_1\},$$ we must have two couplings $\sigma_{00}$ and $\sigma_{10}$ for $\cp{ \mathcal{E}_{10}(\tr_{\langle 2\rangle}\rho)}{ \mathcal{E}_{20}(\tr_{\langle 1\rangle}\rho) }$ and $\cp{ \mathcal{E}_{11}(\tr_{\langle 2\rangle}\rho)}{ \mathcal{E}_{21}(\tr_{\langle 1\rangle}\rho) }$ respectively, such that
\begin{align*}
&\tr(\delta_0) = \tr(\rho)=\tr(\sigma_{00}) + \tr(\sigma_{10}), \\
&\tr(A\delta_0) = \tr(A\rho)\le\tr(B_0\sigma_{00}) + \tr(B_1\sigma_{10}),
\end{align*}
as we set $\delta_0=\rho$.

According to the second assumption $\vdash P_{1}\sim P_{2}: B_1\Rightarrow A$, there exists a coupling $\delta_{1}$ for $\cp{ \sm{P_1}(\tr_{\langle 2\rangle}\sigma_{10})}{ \sm{P_2}(\tr_{\langle 1\rangle}\sigma_{10}) }$ such that
$$\tr(B_1\sigma_{10})\le\tr(A\delta_{1}) + \tr(\sigma_{10})-\tr(\delta_{1}).$$
Note that $$\sm{P_1}(\tr_{\langle 2\rangle}\sigma_{00}) = \sm{P_1}(\mathcal{E}_{11}(\tr_{\langle 2\rangle}\rho)) = (\sm{P_1}\circ\mathcal{E}_{11})(\tr_{\langle 2\rangle}\rho),$$ so $\delta_1$ is a coupling for $$\cp{ (\sm{P_1}\circ \mathcal{E}_{11})(\tr_{\langle 2\rangle}\rho)}{ (\sm{P_2}\circ \mathcal{E}_{21})(\tr_{\langle 1\rangle}\rho) }.$$
So the induction basis is true.

Suppose the statement holds for $n = k-1\ (k\ge1)$ , we will show that the statement also holds for $n = k$. From the assumption, we know that $\delta_{k}$ is a coupling for $$\cp{ (\sm{P_1}\circ \mathcal{E}_{11})^{k}(\tr_{\langle 2\rangle}\rho)}{ (\sm{P_2}\circ \mathcal{E}_{21})^{k}(\tr_{\langle 1\rangle}\rho) },$$ and it also satisfies $\cM_1\approx \cM_2$, so due to the assumption $$\cM_1\approx \cM_2\models A\Rightarrow\{B_0,B_1\},$$ we must have two couplings $\sigma_{0k}$ for
\begin{align*}\cp{ \mathcal{E}_{10}(\tr_{\langle 2\rangle}\delta_{k})}{ \mathcal{E}_{20}(\tr_{\langle 1\rangle}\delta_{k}) }
= \cp{ \big(\mathcal{E}_{10}\circ(\sm{P_1}\circ \mathcal{E}_{11})^k\big)(\tr_{\langle 2\rangle}\rho)}{ \big(\mathcal{E}_{20}\circ(\sm{P_2}\circ \mathcal{E}_{21})^k\big)(\tr_{\langle 1\rangle}\rho) }\end{align*}
and $\sigma_{1k}$ for
\begin{align*}\cp{\mathcal{E}_{11}(\tr_{\langle 2\rangle}\delta_{k})}{ \mathcal{E}_{21}(\tr_{\langle 1\rangle}\delta_{k}) }  = \cp{ \big(\mathcal{E}_{11}\circ(\sm{P_1}\circ \mathcal{E}_{11})^k\big)(\tr_{\langle 2\rangle}\rho)}{ \big(\mathcal{E}_{21}\circ(\sm{P_2}\circ \mathcal{E}_{21})^k\big)(\tr_{\langle 1\rangle}\rho) }\end{align*}
such that
\begin{align*}
&\tr(\delta_k)=\tr(\sigma_{0k}) + \tr(\sigma_{1k}), \\
&\tr(A\delta_k)\le\tr(B_0\sigma_{0k}) + \tr(B_1\sigma_{1k}).
\end{align*}
According to the second assumption $$\vdash P_{1}\sim P_{2}: B_1\Rightarrow A,$$ there exists a coupling $\delta_{k+1}$ for
\begin{align*}&\cp{ \sm{P_1}(\tr_{\langle 2\rangle}\sigma_{1k})}{ \sm{P_2}(\tr_{\langle 1\rangle}\sigma_{1k}) }  =
\cp{ (\sm{P_1}\circ \mathcal{E}_{11})^{k+1}(\tr_{\langle 2\rangle}\rho)}{ (\sm{P_2}\circ \mathcal{E}_{21})^{k+1}(\tr_{\langle 1\rangle}\rho) }
\end{align*}
such that
$$\tr(B_1\sigma_{1k})\le\tr(A\delta_{k+1}) + \tr(\sigma_{1k})-\tr(\delta_{k+1}).$$

So the statement holds for any $n\ge0$. From the statement, we have the following equations for $N\ge0$:
\begin{align*}
\tr(\rho) &= \tr(\sigma_{00}) + \tr(\sigma_{10}) + \sum_{n=1}^N\tr(\delta_n) - \sum_{n=1}^N\tr(\delta_n) \\
&= \tr(\sigma_{00}) + \tr(\sigma_{10}) + \sum_{n=1}^N\left(\tr(\sigma_{0n}) + \tr(\sigma_{1n}) \right) - \sum_{n=1}^N\tr(\delta_n) \\
&= \sum_{n=0}^N\tr(\sigma_{0n}) + \tr(\sigma_{1N}) + \sum_{n=0}^{N-1}\left(\tr(\sigma_{1n}) - \tr(\delta_{n+1})\right),
\end{align*}
and also for $N\ge0$,
\begin{align*}
\tr(A\rho) &\le \tr(B_0\sigma_{00}) + \tr(B_1\sigma_{10}) \\
&\le \tr(B_0\sigma_{00}) + \tr(A\delta_{1}) + \tr(\sigma_{10})-\tr(\delta_{1}) \\
&\le \tr(B_0\sigma_{00}) + \tr(B_0\sigma_{01}) + \tr(B_1\sigma_{11}) + \tr(\sigma_{10})-\tr(\delta_{1})\\
&\ \ \vdots\\
&\le \sum_{n=0}^N\tr(B_0\sigma_{0n}) + \tr(B_1\sigma_{1N}) + \sum_{n=0}^{N-1}\left( \tr(\sigma_{1n})-\tr(\delta_{n+1}) \right)\\
&\le \sum_{n=0}^N\tr(B_0\sigma_{0n}) + \tr(\sigma_{1N}) + \sum_{n=0}^{N-1}\left( \tr(\sigma_{1n})-\tr(\delta_{n+1}) \right) \\
&= \sum_{n=0}^N\tr(B_0\sigma_{0n}) + \tr(\rho) - \sum_{n=0}^N\tr(\sigma_{0n}) \\
&= \tr\left(B_0\sum_{n=0}^N\sigma_{0n}\right) + \tr(\rho) - \tr\left(\sum_{n=0}^N\sigma_{0n}\right).
\end{align*}
Notice that
\begin{align*}
\tr(\sum_{n=0}^N\sigma_{0n})=\sum_{n=0}^N\tr(\sigma_{0n})=\sum_{n=0}^N\tr\left[\mathcal{E}_{10}\circ(\sm{P_1}\circ \mathcal{E}_{11})^n\right](\tr_{\langle 2\rangle}\rho)\leq 1.
\end{align*}
Therefore, $\lim_{N\rightarrow\infty}\sum_{n=0}^N\sigma_{n}$ does exist.
Set $$\sigma = \lim_{N\rightarrow\infty}\sum_{n=0}^N\sigma_{0n},$$ and let $N\rightarrow \infty$ in the above inequality. Then we have:
$$\tr(A\rho) \le\tr\left(B_0\sigma\right) + \tr(\rho) - \tr\left(\sigma\right).$$
Moreover, it is straightforward that $\sigma$ is a coupling for
\begin{align*}
&\lim_{N\rightarrow\infty} \Big\langle \sum_{n=0}^N\left(\mathcal{E}_{10}\circ(\sm{P_1}\circ \mathcal{E}_{11})^n\right)(\tr_{\langle 2\rangle}\rho), \sum_{n=0}^N\left(\mathcal{E}_{20}\circ(\sm{P_2}\circ \mathcal{E}_{21})^n\right)(\tr_{\langle 1\rangle}\rho) \Big\rangle \\
=\ &\cp{ \sm{Q_1} (\tr_{\langle 2\rangle}\rho)}{ \sm{Q_2}(\tr_{\langle 1\rangle}\rho) }.
\end{align*}

{\vskip 4pt}

$\bullet$ (Init-L) For any $\rho\in\mathcal{D}^\le(\mathcal{H}_{P_1}\otimes\mathcal{H}_{P_2})$, set:
\begin{align*}
\sigma &= \sum_{i}\left(|0\rangle_{q_1\langle1\rangle}\langle i| \right)\rho\left(|i\rangle_{q_1\langle1\rangle}\langle 0|\right)
= (\mathcal{E}\otimes\mathcal{I})(\rho),
\end{align*}
where $\mathcal{E}$ is the initial quantum operation with Kraus operators $\{|0\rangle_{q_1\langle1\rangle}\langle i|\}$ and $\mathcal{I}$ is the identity operation.
As $\mathcal{E}$ and $\mathcal{I}$ are all trace-preserving, we know that $\tr(\rho) = \tr(\sigma)$, and $(\mathcal{E}\otimes\mathcal{I})(\rho)$ is a coupling for
\begin{align*}
\cp{ \mathcal{E}(\tr_{\langle 2 \rangle}(\rho))}{ \mathcal{I}(\tr_{\langle 1 \rangle}(\rho)) }
=\ &\Big\langle \sum_{i}\left(|0\rangle_{q_1\langle1\rangle}\langle i|\right)(\tr_{\langle 2 \rangle}(\rho))\left(|i\rangle_{q_1\langle1\rangle}\langle 0|\right), \tr_{\langle 1 \rangle}(\rho) \Big\rangle \\
=\ &\cp{ \sm{q_1:=|0\rangle}(\tr_{\langle 2 \rangle}(\rho))}{ \sm{\mathbf{skip}}(\tr_{\langle 1 \rangle}(\rho)) }
\end{align*}
Moreover,
\begin{align*}
&\tr\left[\sum_{i}\left(|i\rangle_{q_1\langle 1\rangle}\langle 0|\right) A\left( |0\rangle_{q_1\langle 1\rangle}\langle
i|\right)\rho\right]
=\ \tr\left[ A\sum_{i}\left( |0\rangle_{q_1\langle 1\rangle}\langle
i|\right)\rho\left(|i\rangle_{q_1\langle 1\rangle}\langle 0|\right) \right]
=\ \tr(A\sigma).
\end{align*}

{\vskip 4pt}

$\bullet$ (UT-L) Similar to (UT) if we regard $\sm{\mathbf{skip}}$ as $\sm{q := I[q]}$ where $I$ is the identity matrix of the Hilbert space of all variables.

{\vskip 4pt}

$\bullet$ (IF-L) Very similar to (IF); omitted.

{\vskip 4pt}

$\bullet$ (IF1-L) Similar to (IF1). For any $\rho\in\mathcal{D}^\le(\mathcal{H}_{P_1}\otimes\mathcal{H}_{P_2})$, due to the assumption $$\cM_1\approx I_2\models A\Rightarrow\{B_m\},$$ so there exists a serial of states $\sigma_m$ and $\rho_{2m}$ with $$\sum_m\rho_{2m} = \tr_{\langle 1\rangle}\rho,$$ such that for all $m$, $\sigma_m$ is a coupling for
$$\langle M_{1m}(\tr_{\langle 2\rangle}\rho)M_{1m}^\dag,\rho_{2m} \rangle,$$
and $$\tr(A\rho)\le \tr\Big(\sum_mB_m\sigma_m\Big).$$
Moreover, according to the second assumption $$\vdash P_{1m}\sim P: B_m\Rightarrow C\ {\rm for\ every}\ m,$$ for each $\sigma_m$, there exists a coupling $\delta_m$ for
$$\cp{ \sm{P_{1m}} (\tr_{\langle 2\rangle}\sigma_m)}{ \sm{P} (\tr_{\langle 1\rangle}\sigma_m) }$$
such that
$$\tr(B_m\sigma_m) \le \tr(C\delta_m)+\tr\sigma_m-\tr\delta_m.$$
Set $\delta = \sum_m\delta_m$, then
\begin{align*}
\tr(A\rho) &\le \tr\Big(\sum_mB_m\sigma_m\Big)  \le \sum_m\left[\tr(C\delta_m)+\tr\sigma_m-\tr\delta_m\right]  = \tr(C\delta) + \tr(\rho) -  \tr(\delta).
\end{align*}
Moreover, it is straightforward to see
\begin{align*}
\tr_{\langle 2\rangle}\delta &= \sum_m\tr_{\langle 2\rangle}\delta_m
= \sum_m \sm{P_{1m}} (\tr_{\langle 2\rangle}\sigma_m) \\
&= \sum_m \sm{P_{1m}} ( M_{1m}(\tr_{\langle 2\rangle}\rho)M_{1m}^\dag) \\
&= \sm{\mathbf{if}\ (\square m\cdot
	\cM_1[\overline{q}]=m\rightarrow P_{1m})\ \mathbf{fi}} (\tr_{\langle 2\rangle}\rho)
\end{align*}
and
\begin{align*}
\tr_{\langle 1\rangle}\delta &= \sum_m\tr_{\langle 1\rangle}\delta_m
= \sum_m \sm{P} (\tr_{\langle 1\rangle}\sigma_m) \\
&= \sum_m \sm{P} \left( \rho_{2m}\right)
= \sm{P} \left( \sum_m  \rho_{2m} \right) \\
&= \sm{P}( \tr_{\langle 1\rangle}\rho).
\end{align*}

{\vskip 4pt}

$\bullet$ (LP-L) Similar to (LP1). We use notations: $$Q_1\equiv\mathbf{while}\ \cM_1[\overline{q}]=1\ \mathbf{do}\ P_1\ \mathbf{od},\qquad\qquad Q_2 \equiv \mathbf{skip},$$ $\E_{10}(\cdot)=M_{10}(\cdot) M_{10}^\dag$ and  $\E_{11}(\cdot)=M_{11}(\cdot) M_{11}^\dag$. The lossless condition of $Q_1$ ensures that $P_1$ is a terminating subprogram. For simplicity, we use notations $M_0 = M_{10}\otimes I_2$ and $M_1 = M_{11}\otimes I_2$.
For any $\rho\in\mathcal{D}^\le(\mathcal{H}_{Q_1}\otimes\mathcal{H}_{Q_2})$, we inductively define the following sequence of states:
\begin{enumerate}
	\item $\sigma_0 = \rho$;
	\item for input $M_1\sigma_nM_1^\dag$, due to the second assumption, $\sigma_{n+1}$ is the coupling of the outputs
	$$\cp{\sm{P_1}(\tr_{\<2\>}(M_1\sigma_nM_1^\dag))}{\sm{\mathbf{skip}}(\tr_{\<1\>}(M_1\sigma_nM_1^\dag))}$$
	such that:
	$$
	\tr(BM_1\sigma_nM_1^\dag)\le \tr\{[M_0^\dag AM_0
	+ M_1^\dag BM_1]\sigma_{n+1}\}.
	$$
\end{enumerate}
We first prove by induction that for all $n\ge0$, it holds that
$$
\tr_{\<2\>}(\sigma_{n}) = (\sm{P_1}\circ\E_{11})^n(\tr_{\<2\>}\rho), \quad
\tr_{\<1\>}(\sigma_{n}) - \tr_{\<1\>}(\sigma_{n+1}) = \tr_{\<1\>}(M_0\sigma_{n}M_0^\dag).
$$
It is trivial for $n=0$ according to Lemma \ref{part-measure}. Suppose it holds for $n = k$, then for $n=k+1$, using Lemma \ref{part-measure} we have:
\begin{align*}
\tr_{\<2\>}(\sigma_{k+1}) &=
\sm{P_1}(\tr_{\<2\>}[(M_{11}\otimes I_2)\sigma_n(M_{11}^\dag\otimes I_2)]) \\
&= \sm{P_1} M_{11}\tr_{\<2\>}(\sigma_n)M_{11}^\dag \\
&= (\sm{P_1}\circ\E_{11})^{k+1}(\tr_{\<2\>}\rho)
\end{align*}
and:
\begin{align*}
&\tr_{\<1\>}(\sigma_{k}) - \tr_{\<1\>}(\sigma_{k+1}) \\
=\ & \tr_{\<1\>}[(M_{10}\otimes I_2)\sigma_k(M_{10}^\dag\otimes I_2)] +
\tr_{\<1\>}[(M_{11}\otimes I_2)\sigma_k(M_{11}^\dag\otimes I_2)] - \sm{\mathbf{skip}}(\tr_{\<1\>}(M_1\sigma_kM_1^\dag)) \\
=\ & \tr_{\<1\>}(M_0\sigma_{k}M_0^\dag).
\end{align*}

Actually, the assumption that $Q_1$ is lossless ensures that
$\lim_{n\rightarrow\infty}(\sm{P_1}\circ\E_{11})^n(\tr_{\<2\>}\rho) = 0$ as its trace characterizes the probability of nontermination. Therefore, $\lim_{n\rightarrow\infty}\sigma_n = 0$.
Together with the choice of $\sigma_{n}$, we have
\begin{align*}
&\tr\bigg(A \sum_{n=0}^{\infty}(M_0\sigma_nM_0^\dag)\bigg) \\
\ge\ & \sum_{n=0}^{\infty}\left[\tr(AM_0\sigma_nM_0^\dag) + \tr(BM_1\sigma_nM_1^\dag)
- \tr\{[M_0^\dag AM_0 + M_1^\dag BM_1]\sigma_{n+1}\}\right] \\
=\ &\tr(AM_0\sigma_0M_0^\dag) + \tr(BM_1\sigma_0M_1^\dag) - \lim_{n\rightarrow\infty}\tr\{[M_0^\dag AM_0 + M_1^\dag BM_1]\sigma_{n}\} \\
=\ &\tr\{[M_0^\dag AM_0 + M_1^\dag BM_1]\rho\}.
\end{align*}
Let us define $\tau = \sum_{n=0}^{\infty}(M_0\sigma_nM_0^\dag)$. According to Lemma \ref{part-measure} again, we can show that $\tau$ is a coupling of the outputs:
\begin{align*}
\tr_{\<2\>}(\tau) &= \sum_{n=0}^{\infty}\tr_{\<2\>}[(M_{10}\otimes I_2)\sigma_n(M_{10}\otimes I_2)^\dag] = \sum_{n=0}^{\infty}M_{10}\tr_{\<2\>}(\sigma_n)M_{10}^\dag \\
&= \sum_{n=0}^{\infty}\E_{10}[(\sm{P_1}\circ\E_{11})^n(\tr_{\<2\>}\rho)] = \sum_{n=0}^{\infty} [\E_{10}\circ(\sm{P_1}\circ\E_{11})^n](\tr_{\<2\>}\rho) \\
&= \sm{Q_1}(\tr_{\<2\>}\rho).
\end{align*}
and
\begin{align*}
\tr_{\<1\>}(\tau) &= \sum_{n=0}^{\infty}\tr_{\<1\>}[(M_0\sigma_nM_0^\dag] = \sum_{n=0}^{\infty}[\tr_{\<1\>}(\sigma_{n}) - \tr_{\<1\>}(\sigma_{n+1})] \\
&= \tr_{\<1\>}(\rho)-\lim_{n\rightarrow\infty}\tr_{\<1\>}(\sigma_{n}) \\
&= \sm{\mathbf{skip}}(\tr_{\<1\>}\rho).
\end{align*}

{\vskip 4pt}

$\bullet$ (LP1-L) Similar to (LP1). Let us use $Q_1$ to denote the left \textbf{while} program $$\mathbf{while}\ \cM_1[\overline{q}]=1\ \mathbf{do}\ P_1\ \mathbf{od}$$ and $Q_2$ to denote the right program $\mathbf{skip}$ for simplicity, and introduce an auxiliary notation: for $i = 1,2$, the quantum operation $\mathcal{E}_{0}$ and $\mathcal{E}_{1}$ are defined by the measurement $\cM_1$
$$\mathcal{E}_{0}(\rho) = M_{10}\rho M_{10}^\dag,\quad \mathcal{E}_{1}(\rho) = M_{11}\rho M_{11}^\dag.$$

For any input state $\rho\in\mathcal{D}^\le(\mathcal{H}_{Q_1}\otimes\mathcal{H}_{Q_2})$, choose $\delta_0 = \rho$ so that $\delta_0$ is a coupling for $\cp{\tr_{\langle 2\rangle}\rho}{ \tr_{\langle 1\rangle}\rho}$, and $\varrho_{1,-1} = \tr_{\langle 1\rangle}\rho$, then the following statement holds for $n\ge0$:
\begin{align*}
&\textbf{Statement:} \\
&\quad\quad\exists\ \varrho_{0n},\varrho_{1n},\sigma_{0n}, \sigma_{1n}, \delta_{n+1} \text{\ such that} \\
&\quad\quad\quad\quad \sigma_{0n} \text{\ is a coupling for\ } \cp{ \left(\mathcal{E}_{0}\circ(\sm{P_1}\circ \mathcal{E}_{1})^n\right)(\tr_{\langle 2\rangle}\rho)}{ \varrho_{0n} }, \\
&\quad\quad\quad\quad \sigma_{1n} \text{\ is a coupling for\ } \cp{ \left(\mathcal{E}_{1}\circ(\sm{P_1}\circ \mathcal{E}_{1})^n\right)(\tr_{\langle 2\rangle}\rho)}{ \varrho_{1n} }, \\
&\quad\quad\quad\quad \delta_{n+1} \text{\ is a coupling for\ } \cp{ (\sm{P_1}\circ \mathcal{E}_{1})^{n+1}(\tr_{\langle 2\rangle}\rho)}{ \varrho_{1n} }, \\
&\quad\quad\text{and:} \\
&\quad\quad\quad\quad\varrho_{0n}+\varrho_{1n}=\varrho_{1,n-1},\\
&\quad\quad\quad\quad\tr(\delta_n)=\tr(\sigma_{0n}) + \tr(\sigma_{1n}),\\
&\quad\quad\quad\quad\tr(A\delta_n)\le\tr(B_0\sigma_{0n}) + \tr(B_1\sigma_{1n}),\\
&\quad\quad\quad\quad\tr(B_1\sigma_{1n})\le\tr(A\delta_{n+1}) + \tr(\sigma_{1n})-\tr(\delta_{n+1}).
\end{align*}
The above statement can be proved by induction. For basis $n=0$, due to the assumption $\cM_1\approx I_2\models A\Rightarrow\{B_0,B_1\}$, there exist $\varrho_{00}$,  $\varrho_{10}$ and two couplings $\sigma_{00}$ for
$\cp{ \mathcal{E}_{0}(\tr_{\langle 2\rangle}\rho)}{ \varrho_{00} }$
and $\sigma_{10}$ for
$\cp{ \mathcal{E}_{1}(\tr_{\langle 2\rangle}\rho)}{ \varrho_{10} },$
such that
\begin{align*}
&\varrho_{00}+\varrho_{10} = \tr_{\langle 1\rangle}\rho = \varrho_{1,-1}, \\
&\tr(\delta_0) = \tr(\rho)=\tr(\sigma_{00}) + \tr(\sigma_{10}), \\
&\tr(A\delta_0) = \tr(A\rho)\le\tr(B_0\sigma_{00}) + \tr(B_1\sigma_{10}),
\end{align*}
as we set $\delta_0=\rho$ and $\varrho_{1,-1} = \tr_{\langle 1\rangle}\rho$.
According to the second assumption $\vdash P_{1}\sim \mathbf{skip}: B_1\Rightarrow A$, there exists a coupling $\delta_{1}$ for $\cp{ \sm{P_1}(\tr_{\langle 2\rangle}\sigma_{10})}{ \tr_{\langle 1\rangle}\sigma_{10} }$ such that
$$\tr(B_1\sigma_{10})\le\tr(A\delta_{1}) + \tr(\sigma_{10})-\tr(\delta_{1}).$$
Note that $$\sm{P_1}(\tr_{\langle 2\rangle}\sigma_{10}) = \sm{P_1}(\mathcal{E}_{1}(\tr_{\langle 2\rangle}\rho)) = (\sm{P_1}\circ\mathcal{E}_{1})(\tr_{\langle 2\rangle}\rho),$$ so $\delta_1$ is a coupling for
$\langle (\sm{P_1}\circ \mathcal{E}_{1})(\tr_{\langle 2\rangle}\rho), \varrho_{10} \rangle.$
So the basis is true.

Suppose the statement holds for $n = k-1 (k\ge1)$, we will show that the statement also holds for $n = k$. From the assumption $\cM_1\sim I_2\models A\Rightarrow\{B_0,B_1\}$, there exist $\varrho_{0k}$,  $\varrho_{1k}$ and two couplings $\sigma_{0k}$ for
\begin{align*}
\cp{ \mathcal{E}_{0}(\tr_{\langle 2\rangle}\delta_{k})}{ \varrho_{0k}       }
= \cp{ \big(\mathcal{E}_{0}\circ(\sm{P_1}\circ \mathcal{E}_{1})^k\big)(\tr_{\langle 2\rangle}\rho)}{ \varrho_{0k}  }\end{align*}
and $\sigma_{1k}$ for
\begin{align*}
\cp{ \mathcal{E}_{1}(\tr_{\langle 2\rangle}\delta_{k})}{ \varrho_{1k}      } = \cp{ \big(\mathcal{E}_{1}\circ(\sm{P_1}\circ \mathcal{E}_{1})^k\big)(\tr_{\langle 2\rangle}\rho)}{ \varrho_{1k} },\end{align*}
such that
\begin{align*}
&\varrho_{0k}+\varrho_{1k} = \tr_{\langle 1\rangle}\delta_{k} = \varrho_{1,k-1}, \\
&\tr(\delta_k)=\tr(\sigma_{0k}) + \tr(\sigma_{1k}), \\
&\tr(A\delta_k)\le\tr(B_0\sigma_{0k}) + \tr(B_1\sigma_{1k}).
\end{align*}
According to the second assumption $\vdash P_{1}\sim \mathbf{skip}: B_1\Rightarrow A$, there exists a coupling $\delta_{k+1}$ for
\begin{align*}
\cp{ \sm{P_1}(\tr_{\langle 2\rangle}\sigma_{1k})}{ \tr_{\langle 1\rangle}\sigma_{1k}) } = \cp{ (\sm{P_1}\circ \mathcal{E}_{1})^{k+1}(\tr_{\langle 2\rangle}\rho)}{ \varrho_{1k} }
\end{align*}
such that
$$\tr(B_1\sigma_{1k})\le\tr(A\delta_{k+1}) + \tr(\sigma_{1k})-\tr(\delta_{k+1}).$$

So the statement holds for any $n\ge0$. From the statement, we have the following inequality for $N\ge0$ (similar to the proof of (LP-E)):
\begin{align*}
\tr(A\rho) \le \tr\left(B_0\sum_{n=0}^N\sigma_{0n}\right) + \tr(\rho) - \tr\left(\sum_{n=0}^N\sigma_{0n}\right).
\end{align*}
We can notice that $\lim_{N\rightarrow\infty}\sum_{n=0}^N\sigma_{0n}$ does exist via similar arguments of of (LP).
Set $\sigma = \lim_{N\rightarrow\infty}\sum_{n=0}^N\sigma_{0n}$, and let $N\rightarrow \infty$ in the above inequality, we have:
$$\tr(A\rho) \le\tr\left(B_0\sigma\right) + \tr(\rho) - \tr\left(\sigma\right).$$
Moreover, we can check the following equation:
\begin{align*}
\lim_{N\rightarrow\infty} \sum_{n=0}^N \varrho_{0n}
&= \lim_{N\rightarrow\infty} \left\{ \sum_{n=0}^N (\varrho_{0n}+\varrho_{1n}) - \sum_{n=0}^N \varrho_{1n} \right\} \\
&= \lim_{N\rightarrow\infty} \left\{ \sum_{n=0}^N \varrho_{1,n-1} - \sum_{n=0}^N \varrho_{1n} \right\} \\
&= \lim_{N\rightarrow\infty} \left\{ \varrho_{1,-1} - \varrho_{1N} \right\} \\
&= \tr_{\langle 1\rangle}\rho - \lim_{N\rightarrow\infty}\varrho_{1N} \\
&= \tr_{\langle 1\rangle}\rho.
\end{align*}
The last equality holds because
\begin{align*}
\lim_{N\rightarrow\infty}\tr(\varrho_{1N}) &= \tr(\tr_{\langle 1\rangle}\rho) -  \lim_{N\rightarrow\infty}\tr\left(\sum_{n=0}^{N}\varrho_{0n}\right) \\
&= \tr(\rho) - \lim_{N\rightarrow\infty}\tr\left[\sum_{n=0}^{N}\left(\mathcal{E}_{0}\circ(\sm{P_1}\circ \mathcal{E}_{1})^n\right)(\tr_{\langle 2\rangle}\rho)\right] \\
&= \tr(\rho) - \tr[\sm{\mathbf{while}\ \cM_1[\overline{q}]=1\ \mathbf{do}\ P_1\ \mathbf{od}} (\tr_{\langle 2\rangle}\rho)] \\
&= \tr(\rho) - \tr(\tr_{\langle 2\rangle}\rho) \\
&= 0,
\end{align*}
followed from the assumption $\models \mathbf{while}\ \cM_1[\overline{q}]=1\ \mathbf{do}\ P_1\ \mathbf{od} \text{\ lossless}$, therefore, $\lim_{N\rightarrow\infty}\varrho_{1N} = 0$.
Now, it is straightforward that $\sigma$ is a coupling for
\begin{align*}
&\lim_{N\rightarrow\infty} \Big\langle \sum_{n=0}^N\left(\mathcal{E}_{0}\circ(\sm{P_1}\circ \mathcal{E}_{1})^n\right)(\tr_{\langle 2\rangle}\rho), \sum_{n=0}^N \varrho_{0n} \Big\rangle \\
=\ &\cp{ \sm{\mathbf{while}\ \cM_1[\overline{q}]=1\ \mathbf{do}\ P_1\ \mathbf{od}} (\tr_{\langle 2\rangle}\rho)}{ \sm{\mathbf{skip}} (\tr_{\langle 1\rangle}\rho) }.
\end{align*}

{\vskip 4pt}

%\begin{figure}[h]\centering
%\begin{equation*}\begin{split}
%&({\rm Conseq})\ \ \
%\frac{\vdash P_1\sim P_2: A^\prime\Rightarrow B^\prime\ \ \ \ \ \ A\sqsubseteq A^\prime\ \ \ \ \ \ B^\prime\sqsubseteq B}{\vdash P_1\sim P_2: A\Rightarrow B}\\
%&({\rm Equiv})\ \ \
%\frac{\vdash P_{1}^\prime\sim P_{2}^\prime: A\Rightarrow B\ \ \ \ \ \  P_1\equiv P_1^\prime\ \ \ \ \ \  P_2\equiv P_2^\prime}{\vdash P_1\sim P_2:A\Rightarrow B}\\
%&{\color{blue}({\rm Case})\ \ \
%\frac{\vdash P_{1}\sim P_{2}: A_i\Rightarrow B\ (i=1,...,n)}{\vdash P_1\sim P_2: \sum_{i=1}^n p_iA_i\Rightarrow B \quad\{p_i\} \text{\ forms\ a\ distribution}}}\\
%&({\rm Frame})\ \ \
%\frac{\vdash P_1\sim P_2: A\Rightarrow B}{V, \mathit{var}(P_1,P_2) \vdash P_1\sim P_2:A\otimes C\Rightarrow B\otimes C}
%\end{split}\end{equation*}
%\caption{Structural rqPD rules.\ \ \ \ In (Conseq), $\sqsubseteq$ stands for the L\"{o}wner order between operators. In (Equiv), $\equiv$ stands for equivalence between quantum programs. In (Frame), $V\cap\mathit{var}(P_1,P_2)=\emptyset$ and $C$ is a quantum predicate in $\mathcal{H}_V$.}
%\end{figure}

$\bullet$ (Conseq) It follows immediately from the fact that $A\sqsubseteq B$ implies $\tr(A\rho)\leq\tr(B\rho)$ for all $\rho$.

{\vskip 4pt}

$\bullet$ (Weaken) Obvious.

{\vskip 4pt}

$\bullet$ (Case) For any $\rho\in\mathcal{D}^\le(\mathcal{H}_{P_1}\otimes\mathcal{H}_{P_2})$, if $\rho\models\Gamma$, then it follows from the assumption that for each $i$, there exists a coupling $\sigma_i$ for $\cp{ \sm{P_1}(\tr_{2}(\rho))}{ \sm{P_2}(\tr_{1}(\rho))}$ such that $$\tr((A_i)\rho)\le\tr(B\sigma_i)+\tr(\rho)-\tr(\sigma_i).$$ We set $\sigma = \sum_{i=1}^np_i\sigma_i$. Then
\begin{align*}
&\tr\left[\left(\sum_{i=1}^np_iA_i\right)\rho\right] = \sum_{i=1}^np_i\tr\left(A_i\rho\right)  \le\sum_{i=1}^n p_i\left[\tr(B\sigma_i)+\tr(\rho)-\tr(\sigma_i)\right]  =\tr(B\sigma)+\tr(\rho)-\tr(\sigma),\\
&\tr_1(\sigma)=\sum_{i=1}^n\tr_1(p_i\sigma_i) = p_i\sum_{i=1}^n\sm{P_2}(\tr_1(\rho)) = \sm{P_2}(\tr_1(\rho)),
\end{align*}
and $\tr_2(\sigma)= \sm{P_1}(\tr_2(\rho))$. Therefore, it holds that $\Gamma\models P_1\sim P_2: \sum_{i=1}^n p_iA_i\Rightarrow B$.

{\vskip 4pt}

$\bullet$ (Frame) Assume that $\Gamma\models P_1\sim P_2:A\Rightarrow B$. Moreover,
suppose $V = V_1\cup V_2$ where $V_1$ represents the extended variables of $P_1$ and $V_2$ of $P_2$. Of course, $V_1\cap V_2=\emptyset$ and $\mathcal{H}_V = \mathcal{H}_{V_1}\otimes\mathcal{H}_{V_2}$. We prove $\Gamma\cup\left\{[V,\mathit{var}(P_1,P_2)]\right\}\models P_1\sim P_2:A\otimes C\Rightarrow B\otimes C$; that is, for any separable state $\rho$ between $\mathcal{H}_{P_1}\otimes\mathcal{H}_{P_2}$ and $\mathcal{H}_{V}$, if $\rho\models$, then there exists a coupling $\rho^\prime$ in $\mathcal{H}_{P_1}\otimes\mathcal{H}_{P_2}\otimes\mathcal{H}_{V}$ for
$$\cp{ \sm{P_1} (\tr_{\mathcal{H}_{P_2}\otimes\mathcal{H}_{V_2}}(\rho))}{ \sm{P_2} (\tr_{\mathcal{H}_{P_1}\otimes\mathcal{H}_{V_1}}(\rho)) }$$
such that
$$\tr((A\otimes C)\rho)\le \tr((B\otimes C)\rho^\prime) + \tr(\rho)-\tr(\rho^\prime).$$
First of all, by separability of $\rho$, we can write:
$$\rho = \sum_i p_i (\rho_i\otimes \sigma_i)$$
where $\rho_i\in\mathcal{D}^\le(\mathcal{H}_{P_1}\otimes\mathcal{H}_{P_2})$, $\sigma_i\in\mathcal{D}^\le(\mathcal{H}_{V})$ and $p_i\ge0$. Since $V\cap \mathrm{var}(P_1,P_2)=\emptyset$, it holds that
\begin{align*}
&\sm{P_1} (\tr_{\mathcal{H}_{P_2}\otimes\mathcal{H}_{V_2}}(\rho)) = \sum_i p_i[\sm{P_1} (\tr_{\mathcal{H}_{P_2}}(\rho_i)) \otimes \tr_{\mathcal{H}_{V_2}}(\sigma_i) ], \\
&\sm{P_2} (\tr_{\mathcal{H}_{P_1}\otimes\mathcal{H}_{V_1}}(\rho)) = \sum_i p_i[\sm{P_2} (\tr_{\mathcal{H}_{P_1}}(\rho_i)) \otimes \tr_{\mathcal{H}_{V_1}}(\sigma_i) ].
\end{align*}
For each $i$, since $\rho\models\Gamma$ and $\Gamma\models P_1\sim P_2: A\Rightarrow B$, there exists a coupling $\rho_i^\prime$ for $$\cp{ \sm{P_1} (\tr_{\mathcal{H}_{P_2}}(\rho_i))}{ \sm{P_2} (\tr_{\mathcal{H}_{P_1}}(\rho_i)) }$$ such that $\tr(A\rho_i)\le\tr(B\rho_i^\prime) + \tr(\rho_i) - \tr(\rho_i^\prime)$. We set:
$$\rho^\prime = \sum_ip_i(\rho_i^\prime\otimes\sigma_i).$$
Then we can check that
\begin{align*}
\tr_{\langle 2\rangle}(\rho^\prime) &= \tr_{\mathcal{H}_{P_2}\otimes\mathcal{H}_{V_2}}(\rho^\prime)
= \sum_ip_i[\tr_{\mathcal{H}_{P_2}}(\rho_i^\prime) \otimes \tr_{\mathcal{H}_{V_2}}(\sigma_i)] \\
&= \sum_ip_i[\sm{P_1} (\tr_{\mathcal{H}_{P_2}}(\rho_i)) \otimes \tr_{\mathcal{H}_{V_2}}(\sigma_i)]
= \sm{P_1} (\tr_{\mathcal{H}_{P_2}\otimes\mathcal{H}_{V_2}}(\rho))
\end{align*}
and $\tr_{\langle 1\rangle}(\rho^\prime) = \sm{P_2} (\tr_{\mathcal{H}_{P_1}\otimes\mathcal{H}_{V_1}}(\rho))$. Therefore, $\rho^\prime$ is a coupling for
$$\cp{ \sm{P_1} (\tr_{\mathcal{H}_{P_2}\otimes\mathcal{H}_{V_2}}(\rho))}{ \sm{P_2} (\tr_{\mathcal{H}_{P_1}\otimes\mathcal{H}_{V_1}}(\rho)) }.$$ Furthermore, we have:
\begin{align*}
\tr((A\otimes C)\rho) &= \tr\Big((A\otimes C)\sum_i p_i (\rho_i\otimes \sigma_i)\Big) \\
&= \sum_i p_i\tr(A\rho_i)\cdot\tr(C\sigma_i) \\
&\le \sum_i p_i\tr(B\rho_i^\prime)\cdot\tr(C\sigma_i) + \sum_i p_i\Big[\tr(\rho_i) - \tr(\rho_i^\prime)\Big]\cdot\tr(C\sigma_i) \\
&\le \tr\Big((B\otimes C)\sum_i p_i(\rho_i^\prime\otimes\sigma_i)\Big) + \sum_i p_i\Big[\tr(\rho_i) - \tr(\rho_i^\prime)\Big]\cdot\tr(\sigma_i) \\
&= \tr((B\otimes C)\rho^\prime) + \tr\Big(\sum_ip_i\rho_i\otimes\sigma_i\Big) - \tr\Big(\sum_ip_i\rho^\prime_i\otimes\sigma_i\Big) \\
&= \tr((B\otimes C)\rho^\prime) + \tr(\rho) - \tr(\rho^\prime).
\end{align*}\qed

\subsection{Verification of Example \ref{exam-1}}

First, the applications of axiom (UT-R) and (UT) yields:
\begin{align}\label{ex-1.1}(\mathrm{UT\text{-}R})\ \ \ \vdash\mathbf{skip}\sim q:=H[q]:
\frac{1}{2}[I\otimes I+(I\otimes H)S(I\otimes H)]\Rightarrow\ =_{sym}\end{align}
\begin{align}\label{ex-1.2}(\mathrm{UT})\ \ \ \vdash q:=X[q]\sim q:=Z[q]:
A_{00}\Rightarrow\frac{1}{2}[I\otimes I+(I\otimes H)S(I\otimes H)]\end{align}
\begin{align}\label{ex-1.3}(\mathrm{UT})\ \ \ \vdash q:=H[q]\sim q:=H[q]:
A_{11}\Rightarrow\frac{1}{2}[I\otimes I+(I\otimes H)S(I\otimes H)]\end{align}
where:
\begin{align*}&A_{00}=\frac{1}{2}[I\otimes I+(X\otimes ZH)S(X\otimes HZ)],\\ \ \ \ \ &A_{11}=\frac{1}{2}[I\otimes I+(H\otimes HH)S(H\otimes HH)].\end{align*}

Now we need the following technical lemma:

\begin{lem} Let $B_0 =|0\>\<0|\otimes I$, $B_1 = |1\>\<1|\otimes I,$ and $B=B_0+B_1=I\otimes I.$
	Then we have:
	\begin{equation}\label{ex-1.4}\cM\approx \cM^\prime \models B\Rightarrow \{A_{00},A_{11}\}.\end{equation}\end{lem}
\begin{proof} For any $\rho\models \cM\approx \cM^\prime$,  it must holds that:
	$$\<0|\tr_2(\rho)|0\> = \<+|\tr_1(\rho)|+\>, \quad \<1|\tr_2(\rho)|1\> = \<-|\tr_1(\rho)|-\>.$$
	Set $\sigma_0 = \<0|\tr_2(\rho)|0\>\cdot |0\>\<0|\otimes |+\>\<+|$ and $\sigma_1 = \<1|\tr_2(\rho)|1\>\cdot |1\>\<1|\otimes |-\>\<-|$, which are exactly the coupling of 
	$\cp{ M_0\tr_2(\rho)M_0}{ M_0^\prime\tr_1(\rho)M_0^\prime}$ and $\cp{ M_1\tr_2(\rho)M_1}{ M_1^\prime\tr_1(\rho)M_1^\prime}$, respectively.
	So, $\tr(A_{00}\sigma_0) = \<0|\tr_2(\rho)|0\>$ and $\tr(A_{11}\sigma_1) = \<1|\tr_2(\rho)|1\>$.
	On the other hand,
	\begin{align*}\tr(B_0\rho) = \tr(|0\rangle\langle 0|\otimes I\rho)
	= \tr(|0\rangle\langle 0|\tr_2(\rho)) = \<0|\tr_2(\rho)|0\> = \tr(A_{00}\sigma_0).\end{align*}
	Similarly, we have $\tr(B_1\rho)=\tr(A_{11}\sigma_1)$. Therefore, $\tr(B\rho)=\tr(A_{00}\sigma_0)+\tr(A_{11}\sigma_1).$
\end{proof}

Combining (\ref{ex-1.2}), (\ref{ex-1.3}) and (\ref{ex-1.4}), we obtain:
\begin{equation}\label{ex-1.5}(\mathrm{IF1})\ \ \ \cM\approx \cM^\prime\vdash Q_1\sim Q_2: I\otimes I\Rightarrow\frac{1}{2}[I\otimes I+(I\otimes H)S(I\otimes H)].\end{equation}
Furthermore, we have:
\begin{equation}\label{ex-1.6}(\mathrm{UT\text{-}L})\ \ \ \vdash q:=H[q]\sim \mathbf{skip}: I\otimes I\Rightarrow I\otimes I.\end{equation}
It holds that \begin{align}(\mathrm{Init})&\ \ \ \vdash q:=|0\rangle\sim q:=|0\rangle: I\otimes I\Rightarrow I\otimes I \nonumber \\
(\mathrm{Conseq})&\ \ \ \vdash q:=|0\rangle\sim q:=|0\rangle: (|00\>\<00|+|11\>\<11|)\Rightarrow I\otimes I
\label{ex-1.7}
\end{align}
because \begin{align*}&\sum_{i,j=0}^1(|i\rangle\langle 0|\otimes |j\rangle\langle 0|)(I\otimes I)(|0\rangle\langle i|\otimes |0\rangle\langle j|)
= I\otimes I, \\
&(|00\>\<00|+|11\>\<11|) \sqsubseteq I\otimes I.
\end{align*}
Then (\ref{ex-1.0}) is proved by combining (\ref{ex-1.1}) and (\ref{ex-1.5}) to (\ref{ex-1.7}) using rule (SC+) by checking the entailment between side-conditions:
$$\emptyset\stackrel{(q:=|0\>,q:=|0\>)}{\models}\cM^\prime\approx\cM^\prime\quad\text{and}\quad \cM^\prime\approx\cM^\prime\stackrel{(q:=H[q],{\bf skip})}{\models}\cM\approx\cM^\prime.$$

\noindent\textbf{Deduction without (IF1)} Now, we use rule (IF) to derive the precondition of two ${\bf if}$ statements.
\begin{align*}
&(\mathrm{UT})\ \ \ \vdash q:=X[q]\sim q:=Z[q]:
A_{00}\Rightarrow\frac{1}{2}[I\otimes I+(I\otimes H)S(I\otimes H)] \\
&(\mathrm{UT})\ \ \ \vdash q:=H[q]\sim q:=H[q]:
A_{11}\Rightarrow\frac{1}{2}[I\otimes I+(I\otimes H)S(I\otimes H)] \\
&(\mathrm{UT})\ \ \ \vdash q:=X[q]\sim q:=H[q]:
A_{01}\Rightarrow\frac{1}{2}[I\otimes I+(I\otimes H)S(I\otimes H)] \\
&(\mathrm{UT})\ \ \ \vdash q:=H[q]\sim q:=Z[q]:
A_{10}\Rightarrow\frac{1}{2}[I\otimes I+(I\otimes H)S(I\otimes H)]
\end{align*}
where: 
\begin{align*}A_{00}=\frac{1}{2}[I\otimes I+(X\otimes ZH)S(X\otimes HZ)],\quad \ A_{11}=\frac{1}{2}[I\otimes I+(H\otimes HH)S(H\otimes HH)], \\
A_{01}=\frac{1}{2}[I\otimes I+(X\otimes HH)S(X\otimes HH)],\quad \ A_{10}=\frac{1}{2}[I\otimes I+(H\otimes ZH)S(H\otimes HZ)].
\end{align*}
Applying rule (IF) directly we obtain:
$$
\vdash Q_1 \sim Q_2:\ A \Rightarrow \frac{1}{2}[I\otimes I+(I\otimes H)S(I\otimes H)],
$$
where $$A = \sum_{i,j=0}^1(M_i\otimes M_j^\prime)^\dag A_{ij}(M_i\otimes M_j^\prime) = \left(
\begin{array}{cccc}
7/8 & 1/8 & 0 & 0 \\
1/8 & 7/8 & 0 & 0 \\
0 & 0 & 7/8 & -1/8 \\
0 & 0 & -1/8 & 7/8 \\
\end{array}
\right).$$
Similarly, using rule (UT-L) and (Init), we are only able to derive 
$$\vdash P_1\sim P_2:\ \frac{7}{8}I\otimes I\Rightarrow\ =_{sym}.$$
However, $=_\B\ \not\sqsubseteq \frac{7}{8}I\otimes I$, so the rule of (IF) is too weak to derive the judgment we desired.

\subsection{Proof of Proposition \ref{uniformity}}

\begin{proof} For simplicity, we write $$E=\frac{I\otimes I}{d}.$$ First of all, we notice that for any $\rho\in\D^\le(\h_P\otimes\h_P)$ and any basis state $|i\>$ in $\mathcal{B}$,
	$$\tr((|i\rangle\langle i|\otimes I)\rho)=\langle i|\tr_2(\rho)|i\rangle,\quad\tr(E\rho)=\frac{\tr(\rho)}{d}.$$
	
	(1) $\Rightarrow$ (2): If (1) is valid, then
	$$\langle i|\sm{P}(\rho_1)|i\rangle=\frac{1}{d},\quad
	\tr(\sm{P}(\rho_1)) = \sum_i\langle i|\sm{P}(\rho_1)|i\rangle = 1 = tr(\rho_1)$$
	for any basis state $|i\>$ in $\mathcal{B}$ and any $\rho_1\in\D(\h_P)$ with $tr(\rho_1)=1$. Therefore, $\sm{P}$ is a terminating quantum program which ensures the existence of the coupling for the outputs.
	
	On the other hand, for any $\rho$ with unit trace, suppose $\sigma$ is a coupling for the outputs $$\cp{ \sm{P}(\tr_2(\rho))}{\sm{P}(\tr_1(\rho)) },$$ then it holds that $$\tr(E\rho)=\frac{\tr(\rho)}{d}=\frac{1}{d}=\langle i|\sm{P}(\tr_2(\rho))|i\rangle = \tr((|i\rangle\langle i|\otimes I)\sigma).$$
	
	(2) $\Rightarrow$ (1): If equation (\ref{e-unif}) is true for every $i$, then for any $\rho_1,\rho_2\in\D(\h_P)$ with unit trace, there exist couplings $\rho$ and $\sigma$ for $\cp{\rho_1}{\rho_2}$ and  $\cp{\sm{P}(\rho_1)}{\sm{P}(\rho_2)}$ respectively such that $$\frac{1}{d}=\tr(E\rho)\leq\tr((|i\rangle\langle i|\otimes I)\sigma)=\langle i|\sm{P}(\rho_1)|i\rangle.$$
	On the other hand, we have: $$\sum_i\langle i|\sm{P}(\rho_1)|i\rangle=\tr(\sm{P}(\rho_1))\leq 1.$$ Consequently, we obtain: $$\langle i|\sm{P}(\rho_1)|i\rangle=\frac{1}{d}$$ for every $i$; that is, $\sm{P} (\rho_1)$ is uniform in $\mathcal{B}=\{|i\rangle\}$ for any $\rho_1$.
	
	(1) $\Rightarrow$ (3): Suppose (1) is valid. Similar to the above proof for (1) $\Rightarrow$ (2), it is sufficient to show that, for any $\rho_1,\rho_2\in\D(\h_P)$ with unit trace, by linearity we have
	$$\tr(=_{\mathcal{C}}^e(\rho_1\otimes\rho_2))\leq \sup_{|\alpha\rangle,|\beta\rangle\in\mathcal{H}_P}\tr[=_{\mathcal{C}}^e(|\alpha\rangle\langle \alpha|\otimes|\beta\rangle\langle \beta|)]=\frac{\tr[(|\alpha\rangle\langle \alpha|)^T|\beta\rangle\langle \beta|]}{d}=\frac{|\langle \bar{\alpha}|\beta\rangle|^2}{d}\leq \frac{1}{d}.$$
	
	(3) $\Rightarrow$ (1): If equation (\ref{eq-unif-1}) is true for every $i$, then for any pure state $\rho_1=|\alpha\rangle\langle \alpha|\in\D(\mathcal{H}_P)$, we set $\rho_2=\rho_1^T$. There must exist a coupling $\sigma$ for $\cp{\sm{P}(\rho_1)}{\sm{P}(\rho_2)}$ such that $$\frac{1}{d}=\tr(=_{\mathcal{C}}^e(\rho_1\otimes\rho_2))\leq\tr((|i\rangle\langle i|\otimes I)\sigma)=\langle i|\sm{P}(\rho_1)|i\rangle.$$
	On the other hand, we have: $$\sum_i\langle i|\sm{P}(\rho_1)|i\rangle=\tr(\sm{P}(\rho_1))\leq 1,$$
	which implies $$\langle i|\sm{P}(\rho_1)|i\rangle=\frac{1}{d}$$ for every $i$; that is, $\sm{P} (\rho_1)$ is uniform in $\mathcal{B}=\{|i\rangle\}$ for any pure state $\rho_1$. By linearity, we know that $\sm{P} (\rho)$ is uniform in $\mathcal{B}=\{|i\rangle\}$ for any state $\rho$.
	
	\begin{rem} Note that the maximally entangled state $|\Phi\rangle$ and thus $=_{\mathcal{C}}^e$ are determined by the given orthonormal basis $\mathcal{C}$. From the above proof, we see that if clause (3) in the above proposition is valid for some orthonormal basis $\mathcal{C}$, then it must be valid for any other orthonormal basis $\mathcal{C}^\prime$.\end{rem}
\end{proof}

\subsection{Verification of Example \ref{Bernoulli}}

For simplicity of the presentation, a judgment $\Gamma \vdash P_1\sim P_2:A\Rightarrow B$ using an inference rule $R$ in rqPD is displaced as
\begin{equation*}\begin{split}
&\{A\}\{{\rm SC:}\ \Gamma\} \\
&\bullet\ P_1\sim P_2 \quad\quad (R)\\
&\{B\}\end{split}\end{equation*}
A formal proof of judgment (\ref{ju-uniform}) using the inference rules of rqPD is presented in Figure \ref{qRHL-uniform} together with the following parameters:
\begin{figure}[b]
	\centering
	\begin{align*}
	& \left\{D = I_{q_x}\otimes I_{q_y}\otimes I_{q_x^\prime}\otimes I_{q_y^\prime}[{}_{q_x}\<0|{}_{q_y}\<0|{}_{q_x^\prime}\<0|{}_{q_y^\prime}\<0|A|0\>_{q_x}|0\>_{q_y}|0\>_{q_x^\prime}|0\>_{q_y^\prime}]\right\}\\
	&\bullet\quad q_x:=|0\rangle;q_y:=|0\rangle;\ \sim\ q_x^\prime:=|0\rangle;q_y^\prime:=|0\rangle; \quad\quad{\rm(Init)}\\
	& \{A\}\{{\rm SC:}\ (\cM,P)\approx(\cM,P)\}\\
	&\bullet\quad \mathbf{while}\ \cM[q_x,q_y]=1\ \mathbf{do}\ \sim\ \mathbf{while}\ \cM[q_x^\prime,q_y^\prime]=1\ \mathbf{do} \quad\quad{\rm(LP1)}\\
	& \quad\quad \left\{B_1 = (U^\dag_{q_x}\otimes U^\dag_{q_y}\otimes U^\dag_{q_x^\prime}\otimes U^\dag_{q_y^\prime})A(U_{q_x}\otimes U_{q_y}\otimes U_{q_x^\prime}\otimes U_{q_y^\prime})\right\}\\
	& \quad\quad \bullet\quad\quad q_x:=U[q_x];\quad\quad \sim\ \quad\quad q_x^\prime:=U[q_x^\prime]; \quad\quad{\rm(UT)}\\
	& \quad\quad  \left\{(U^\dag_{q_y}\otimes U^\dag_{q_y^\prime})A(U_{q_y}\otimes U_{q_y^\prime})\right\} \\
	&\quad\quad  \bullet\quad\quad q_y:=U[q_y];\quad\quad \sim\ \quad\quad q_y^\prime:=U[q_y^\prime]; \quad\quad{\rm(UT)}\\
	& \quad\quad  \{A\}\\
	&\bullet\quad \mathbf{od};\quad\quad\quad\quad\quad\quad\quad\quad\ \ \ \sim\ \mathbf{od};  \\
	& \left\{B_0 = |\psi\>_{q_x}\<\psi|\otimes I_{q_y}\otimes I_{q_x^\prime}\otimes I_{q_y^\prime}\right\}\\
	&\bullet\quad \mathbf{Tr}[q_y];\ \sim\ \mathbf{Tr}[q_y^\prime]; \quad\quad{\rm(SO)}\\
	& \left\{C = |\psi\>_{q_x}\<\psi|\otimes I_{q_x^\prime}\right\}
	\end{align*}
	\caption{Verification of ${\rm MQBF}\sim{\rm MQBF}$ in rqPD.
		Note that two programs are the same and independent of inputs, so the condition of equal probability of the while loop is trivially holds. Moreover, we are able to check $\cM\approx \cM\models A \Rightarrow \{B_0,B_1\}$ which is discussed in the context. Therefore, followed from the pre-condition $A$, we have the post-condition $B_0$ after the while loop.
	}\label{qRHL-uniform}
\end{figure}
\begingroup
\allowdisplaybreaks
\begin{align*}
A &= \left(\begin{array}{cccc}
\frac{1}{2} & & & \\
& |\<0|\psi\>|^2 & & \\
& & |\<1|\psi\>|^2 & \\
& & & \frac{1}{2} \\
\end{array}\right)\otimes I_{q_x^\prime}\otimes I_{q_y^\prime} = A_{\<1\>}\otimes I_{q_x^\prime}\otimes I_{q_y^\prime}, \\ \\
D &= I_{q_x}\otimes I_{q_y}\otimes I_{q_x^\prime}\otimes I_{q_y^\prime}[{}_{q_x}\<0|{}_{q_y}\<0|{}_{q_x^\prime}\<0|{}_{q_y^\prime}\<0|A|0\>_{q_x}|0\>_{q_y}|0\>_{q_x^\prime}|0\>_{q_y^\prime}] = \frac{1}{2}I_{q_x}\otimes I_{q_y}\otimes I_{q_x^\prime}\otimes I_{q_y^\prime} \\ \\
B_1 &= (U^\dag_{q_x}\otimes U^\dag_{q_y}\otimes U^\dag_{q_x^\prime}\otimes U^\dag_{q_y^\prime})A(U_{q_x}\otimes U_{q_y}\otimes U_{q_x^\prime}\otimes U_{q_y^\prime})\\ &= \left(\begin{array}{cccc}
\frac{1}{2} & \cdot & \cdot & 0\\
\cdot & \cdot & \cdot & \cdot \\
\cdot & \cdot & \cdot & \cdot \\
0 & \cdot & \cdot & \frac{1}{2} \\
\end{array}\right)\otimes I_{q_x^\prime}\otimes I_{q_y^\prime}
= B_{1\<1\>}\otimes I_{q_x^\prime}\otimes I_{q_y^\prime}.
\end{align*}
\endgroup
Before using the (LP1) rules, we need first to check the validity of
$$
\cM\approx \cM\models A \Rightarrow \{B_0,B_1\}.
$$
To see this, for any $\rho$ satisfies $\cM\approx \cM$, we choose the coupling $\sigma_0$ and $\sigma_1$
$$\sigma_0 = \frac{(M_0\tr_{\<2\>}(\rho)M_0^\dag)\otimes(M_0\tr_{\<1\>}(\rho)M_0^\dag)}{\tr(M_0\tr_{\<1\>}(\rho)M_0^\dag)},\quad
\sigma_1 = \frac{(M_1\tr_{\<2\>}(\rho)M_1^\dag)\otimes(M_1\tr_{\<1\>}(\rho)M_1^\dag)}{\tr(M_1\tr_{\<1\>}(\rho)M_1^\dag)}$$
for $\cp{M_0\tr_{\<2\>}(\rho)M_0^\dag}{\ M_0\tr_{\<1\>}(\rho)M_0^\dag)}$ and $\cp{M_1\tr_{\<2\>}(\rho)M_1^\dag}{\ M_1\tr_{\<1\>}(\rho)M_1^\dag)}$ respectively, and suppose that the matrix form of $\rho_1 = \tr_{\<2\>}(\rho)$ is
$$
\rho_1 = \tr_{\<2\>}(\rho) = \left(\begin{array}{cccc}
\lambda_{00} & \cdot & \cdot & \cdot\\
\cdot & \lambda_{11} & \cdot & \cdot \\
\cdot & \cdot & \lambda_{22} & \cdot \\
\cdot & \cdot & \cdot & \lambda_{33} \\
\end{array}\right).$$
Then we can calculate the following equations carefully:
\begin{align*}
\tr(A\rho) &= \tr(A_{\<1\>}\otimes I_{q_x^\prime}\otimes I_{q_y^\prime}\rho) = \tr(A_{\<1\>}\tr_{\<2\>}(\rho)) = \frac{1}{2}\lambda_{00} + |\<0|\psi\>|^2\lambda_{11} + |\<1|\psi\>|^2\lambda_{22} + \frac{1}{2}\lambda_{33}, \\ \\[-0.2cm]
\tr(B_0\sigma_0) &= \tr(|\psi\>_{q_x}\<\psi|\otimes I_{q_y}\otimes I_{q_x^\prime}\otimes I_{q_y^\prime}\sigma_0) = \tr(|\psi\>_{q_x}\<\psi|\otimes I_{q_y}\tr_{\<2\>}(\sigma_0)) \\
&= \tr(|\psi\>\<\psi|\tr_{q_y}(M_0\rho_1M_0^\dag)) = \tr(|\psi\>\<\psi| (\lambda_{11}|0\>\<0|+\lambda_{22}|1\>\<1|)   ) \\
&= \lambda_{11}|\<0|\psi\>|^2+\lambda_{22}|\<1|\psi\>|^2, \\ \\[-0.2cm]
\tr(B_1\sigma_1) &= \tr(B_{1\<1\>}\otimes I_{q_x^\prime}\otimes I_{q_y^\prime}\sigma_1) = \tr(B_{1\<1\>}M_1\rho_1M_1^\dag) = \left(\begin{array}{cccc}
\frac{1}{2} & \cdot & \cdot & 0\\
\cdot & \cdot & \cdot & \cdot \\
\cdot & \cdot & \cdot & \cdot \\
0 & \cdot & \cdot & \frac{1}{2} \\
\end{array}\right)
\left(\begin{array}{cccc}
\lambda_{00} & 0 & 0 & \cdot\\
0 & 0 & 0 & 0 \\
0 & 0 & 0 & 0 \\
\cdot & 0 & 0 & \lambda_{33} \\
\end{array}\right) \\
&= \frac{1}{2}\lambda_{00} +  \frac{1}{2}\lambda_{33},\\ \\[-0.2cm]
\tr(A\rho) &= \tr(B_0\sigma_0) + \tr(B_1\sigma_1).
\end{align*}
The side condition of equal probability $(\cM,P)\approx(\cM,P)$ (where $P\equiv q_x:=U[q_x];\ q_y:=U[q_y];$ is the loop body) is trivially holds as the two programs are the same and independent of inputs; that is,
$$\emptyset\stackrel{(q_x:=|0\>;q_y:=|0\>,\ q_x:=|0\>;q_y:=|0\>)}{\models}(\cM,P)\approx(\cM,P)$$
Therefore, we can regard $A$ as the pre-condition and $B_0$ as the post-condition of the while loop, and together with the ${\rm(SC+)}$ rule we conclude that judgment (\ref{ju-uniform}) is valid.

{\vskip 4pt}
\noindent\textbf{Impossibility of Using (LP) to derive judgment (\ref{ju-uniform})}

{\vskip 4pt} 
As shown in Fig. \ref{qRHL-uniform}, if $D = \frac{1}{2}I_{q_x}\otimes I_{q_y} \otimes I_{q_x^\prime}\otimes I_{q_y^\prime}$ is derivable, then the precondition $A$ before the ${\bf while}$ loops must satisfy:
\begin{equation}
\label{QBF-contri}
{}_{q_x}\<0|{}_{q_y}\<0|{}_{q_x^\prime}\<0|{}_{q_y^\prime}\<0|A|0\>_{q_x}|0\>_{q_y}|0\>_{q_x^\prime}|0\>_{q_y^\prime}\ge \frac{1}{2}.
\end{equation}

To simplify the calculation, let us choose a special case to see that above inequality does not holds for all possible $|\psi\>$ and non-trivial $U$. We set:
$$U = H = \frac{1}{\sqrt{2}}\left(\begin{array}{cc}1&1\\1&-1\end{array}\right),\quad |\psi\> = |0\>.$$

Suppose $A$ is derivable using rule (LP) for two ${\bf while}$ programs with postcondition $B_0$, we must have the following proof outline:
\begin{align*}
& \left\{A = (M_0\otimes M_0)^\dag B_0 (M_0\otimes M_0) +  (M_1\otimes M_1)^\dag B_1 (M_1\otimes M_1)\right\}\\
&\bullet\quad \mathbf{while}\ \cM[q_x,q_y]=1\ \mathbf{do}\ \sim\ \mathbf{while}\ \cM[q_x^\prime,q_y^\prime]=1\ \mathbf{do} \quad\quad{\rm(LP)}\\
& \quad\quad \left\{B_1 = (U^\dag_{q_x}\otimes U^\dag_{q_y}\otimes U^\dag_{q_x^\prime}\otimes U^\dag_{q_y^\prime})A(U_{q_x}\otimes U_{q_y}\otimes U_{q_x^\prime}\otimes U_{q_y^\prime})\right\}\\
& \quad\quad \bullet\quad\quad q_x:=U[q_x];\quad\quad \sim\ \quad\quad q_x^\prime:=U[q_x^\prime]; \quad\quad{\rm(UT)}\\
& \quad\quad  \left\{(U^\dag_{q_y}\otimes U^\dag_{q_y^\prime})A(U_{q_y}\otimes U_{q_y^\prime})\right\} \\
&\quad\quad  \bullet\quad\quad q_y:=U[q_y];\quad\quad \sim\ \quad\quad q_y^\prime:=U[q_y^\prime]; \quad\quad{\rm(UT)}\\
& \quad\quad  \left\{A = (M_0\otimes M_0)^\dag B_0 (M_0\otimes M_0) +  (M_1\otimes M_1)^\dag B_1 (M_1\otimes M_1)\right\}\\
&\bullet\quad \mathbf{od};\quad\quad\quad\quad\quad\quad\quad\quad\ \ \ \sim\ \mathbf{od};  \\
& B_0 = |\psi\>_{q_x}\<\psi|\otimes I_{q_y}\otimes I_{q_x^\prime}\otimes I_{q_y^\prime}
\end{align*}

It is straightforward to observe:
\begin{align*}
{}_{q_x}\<0|{}_{q_y}\<0|{}_{q_x^\prime}\<0|{}_{q_y^\prime}\<0|[(M_0\otimes M_0)^\dag B_0 (M_0\otimes M_0)]|0\>_{q_x}|0\>_{q_y}|0\>_{q_x^\prime}|0\>_{q_y^\prime} = 0
\end{align*}
and
\begin{align*}
B_1&= (H\otimes H\otimes H\otimes H) A (H\otimes H\otimes H\otimes H) \\
&\sqsubseteq (H\otimes H\otimes H\otimes H) \big[(M_0\otimes M_0)^\dag (|0\>\<0|\otimes I\otimes I\otimes I) (M_0\otimes M_0) + \\
&\qquad\qquad\qquad  (M_1\otimes M_1)^\dag (I\otimes I\otimes I\otimes I) (M_1\otimes M_1)\big] (H\otimes H\otimes H\otimes H) \\
&= (|+\>\<+|\otimes|-\>\<-|)\otimes(|+\>\<+|\otimes|-\>\<-| + |-\>\<-|\otimes|+\>\<+|) \\
&\quad + (|+\>\<+|\otimes|+\>\<+| + |-\>\<-|\otimes|-\>\<-|)\otimes (|+\>\<+|\otimes|+\>\<+| + |-\>\<-|\otimes|-\>\<-|)
\end{align*}
which implies that:
\begin{align*}
{}_{q_x}\<0|{}_{q_y}\<0|{}_{q_x^\prime}\<0|{}_{q_y^\prime}\<0|[(M_1\otimes M_1)^\dag B_1 (M_1\otimes M_1)]|0\>_{q_x}|0\>_{q_y}|0\>_{q_x^\prime}|0\>_{q_y^\prime} \le
3/8.
\end{align*}
Therefore, 
$${}_{q_x}\<0|{}_{q_y}\<0|{}_{q_x^\prime}\<0|{}_{q_y^\prime}\<0|A|0\>_{q_x}|0\>_{q_y}|0\>_{q_x^\prime}|0\>_{q_y^\prime} \le 0 + 3/8 = 3/8$$
which violates Eqn. (\ref{QBF-contri}).
So we conclude that judgment (\ref{ju-uniform}) is not derivable using (LP) and the use of (LP1) is needed.

\subsection{Verification of Judgment (\ref{ju-tele})}

This subsection gives a verification of the first formulation of correctness of quantum teleportation. The verification is presented in Figure \ref{qRHL-tel}.
\begin{figure}[b]
	\centering
	\begin{align*}
	&\left\{D = {}_q\langle0|{}_r\langle0|H_q{\rm CNOT}_{q,r}{\rm CNOT}_{p,q} H_pCH_p{\rm CNOT}_{p,q}{\rm CNOT}_{q,r}H_q|0\rangle_q|0\rangle_r\right\} \\
	&\bullet\quad q:=|0\rangle;\ \sim\ {\bf skip;} \quad\quad{\rm(Init\text{-}L)}\\
	&\left\{ {}_r\langle0|H_q{\rm CNOT}_{q,r}{\rm CNOT}_{p,q}H_pCH_p{\rm CNOT}_{p,q}{\rm CNOT}_{q,r}H_q|0\rangle_r\right\} \\
	&\bullet\quad r:=|0\rangle;\ \sim\ {\bf skip;} \quad\quad{\rm(Init\text{-}L)}\\
	&\left\{H_q{\rm CNOT}_{q,r}{\rm CNOT}_{p,q}H_pCH_p{\rm CNOT}_{p,q}{\rm CNOT}_{q,r}H_q\right\} \\
	&\bullet\quad q:=H[q];\ \sim\ {\bf skip;} \quad\quad{\rm(UT\text{-}L)}\\
	&\left\{{\rm CNOT}_{q,r}{\rm CNOT}_{p,q}H_pCH_p{\rm CNOT}_{p,q}{\rm CNOT}_{q,r}\right\}\\
	&\bullet\quad q,r:=\ {\rm CNOT}[q,r];\ \sim\ {\bf skip;} \quad\quad{\rm(UT\text{-}L)}\\
	&\left\{{\rm CNOT}_{p,q}H_pCH_p{\rm CNOT}_{p,q}\right\} \\
	&\bullet\quad p,q:=\ {\rm CNOT}[p,q];\ \sim\ {\bf skip;} \quad\quad{\rm(UT\text{-}L)}\\
	&\left\{H_pCH_p\right\} \\
	&\bullet\quad p:=H[p];\ \sim\ {\bf skip;} \quad\quad{\rm(UT\text{-}L)}\\
	&\left\{C = |0\rangle_q\langle0|\otimes B + |1\rangle_q\langle1|\otimes (X_rBX_r)\right\} \\
	&\bullet\ \ \begin{array}{l}
	\mathbf{if}\ \cM[q]=0\rightarrow \mathbf{skip}\\ \square\, \ \ \ \ \ \ \ \ \ \ \quad 1\rightarrow r:=X[r]\ \mathbf{fi};
	\end{array}\ \sim\ {\bf skip;} \quad\quad{\rm(IF\text{-}L)}\\
	&\left\{B = |0\rangle_p\langle0|\otimes A + |1\rangle_p\langle1|\otimes (Z_rAZ_r)\right\} \\
	&\bullet\ \ \begin{array}{l}
	\mathbf{if}\ \cM[p]=0\rightarrow \mathbf{skip}\\ \square\, \ \ \ \ \ \ \ \ \ \ \quad 1\rightarrow r:=Z[r]\ \mathbf{fi};
	\end{array}\ \sim\ {\bf skip;} \quad\quad{\rm(IF\text{-}L)}\\
	&\left\{A = |\psi\rangle_r\langle\psi|\otimes|\psi\rangle_s\langle\psi|+|\phi\rangle_r\langle\phi|\otimes|\phi\rangle_s\langle\phi|\right\}
	\end{align*}
	\caption{Verification of QTEL$\sim$\textbf{skip} in rqPD. Each line marked by the bullet presents two programs $P_1\sim P_2$, and using the proof rule indicated in the bracket, we can show the validity of the judgment $\vdash P_1\sim P_2: \mathrm{Pre}\Rightarrow\mathrm{Post}$ if we choose the quantum predicate in the previous line as the precondition $\mathrm{Pre}$ and the quantum predicate in the next line as the postcondition $\mathrm{Post}$. The (SC) rule ensures the validity of the whole programs $\vdash \mathrm{QTEL}\sim \mathbf{skip}: D\Rightarrow A$.}\label{qRHL-tel}
\end{figure}
Let us carefully explain the notations and several key steps in the proof. 
%First of all, we need following technical lemma in the steps using rule (IF-L):
%
%\begin{lem}\label{lem-one-side-measure-annihilation}
%	Let $\cM_1=\{M_{1m}=\langle\psi_m|\}$ be a measurement with $\{|\psi_m\rangle\}$ being an orthonormal basis of $\mathcal{H}_{\cM_1}$. Then
%	$$\cM_1\approx I_2\models \sum_m|\psi_m\rangle\langle\psi_m|\otimes B_m\Rightarrow\{B_m\},$$
%	where $I_2$ stands for the identity operator in $\mathcal{H}_{P_2}$, the input Hilbert space of $P_1$ is $\mathcal{H}_{P_1} = \mathcal{H}_{\cM_1} \otimes \mathcal{H}_{1}$ and the output Hilbert space of $P_1$ is $\mathcal{H}_{1}$ as the measurement discards the qubits of $\mathcal{H}_{\cM_1}$, and $B_m$ are quantum predicates in $\mathcal{H}_{1}\otimes\mathcal{H}_{P_2}$ for all $m$.
%\end{lem}
%
%\begin{proof} Straightforward.\end{proof}
%
%Secondly, 
We use $s$ to denote the state space of the second program $\bf{skip}$. Therefore we can write: $$A = |\psi\rangle_r\langle\psi|\otimes|\psi\rangle_s\langle\psi|+|\phi\rangle_r\langle\phi|\otimes|\phi\rangle_s\langle\phi|= (=_\mathcal{B}).$$ Note that
\begin{equation}
\label{defE}
E = {}_q\langle0|{}_r\langle0|H_q{\rm CNOT}_{q,r}{\rm CNOT}_{p,q}H_p.
\end{equation}
So, our goal is to prove: $$D = ECE^\dag = (=_\mathcal{B}) = |\psi\rangle_p\langle\psi|\otimes|\psi\rangle_s\langle\psi| + |\phi\rangle_p\langle\phi|\otimes|\phi\rangle_s\langle\phi|.$$
For any vectors $|\alpha\rangle$ and $|\beta\rangle$, we have the following:
\begin{align}
\label{propE}
&E|0\rangle_p|0\rangle_q|\alpha\rangle_r|\beta\rangle_s = \frac{1}{2}|\alpha\rangle_p|\beta\rangle_s, \qquad
&&E|1\rangle_p|0\rangle_qZ_r|\alpha\rangle_r|\beta\rangle_s = \frac{1}{2}|\alpha\rangle_p|\beta\rangle_s, \nonumber\\
&E|0\rangle_p|1\rangle_qX_r|\alpha\rangle_r|\beta\rangle_s = \frac{1}{2}|\alpha\rangle_p|\beta\rangle_s,\qquad
&&E|1\rangle_p|1\rangle_qX_rZ_r|\alpha\rangle_r|\beta\rangle_s = \frac{1}{2}|\alpha\rangle_p|\beta\rangle_s.
\end{align}
So, we can calculate $D$ directly:
\begin{align*}
D &= ECE^\dag = E(|0\rangle_q\langle0|\otimes B + |1\rangle_q\langle1|\otimes (X_rBX_r))E^\dag \\
&= E\left[|0\rangle_p\langle0|\otimes|0\rangle_q\langle0|\otimes A + |1\rangle_p\langle1|\otimes|0\rangle_q\langle0|\otimes (Z_rAZ_r)\right] + |0\rangle_p\langle0|\otimes|1\rangle_q\langle1|\otimes (X_rAX_r) \\&\qquad  + |1\rangle_p\langle1|\otimes|1\rangle_q\langle1|\otimes (X_rZ_rAZ_rX_r)]E^\dag \\
&=  (E|0\rangle_p|0\rangle_q|\psi\rangle_r|\psi\rangle_s)({}_p\langle0|_q\langle0|_r\langle\psi|_s\langle\psi|E^\dag) + (E|0\rangle_p|0\rangle_q|\phi\rangle_r|\phi\rangle_s)({}_p\langle0|_q\langle0|_r\langle\phi|_s\langle\phi|E^\dag) \\
&\qquad + (E|1\rangle_p|0\rangle_qZ_r|\psi\rangle_r|\psi\rangle_s)({}_p\langle1|_q\langle0|_r\langle\psi|Z_r{}_s\langle\psi|E^\dag) + (E|1\rangle_p|0\rangle_qZ_r|\phi\rangle_r|\phi\rangle_s)({}_p\langle1|_q\langle0|_r\langle\phi|Z_r{}_s\langle\phi|E^\dag) \\
&\qquad + (E|0\rangle_p|1\rangle_qX_r|\psi\rangle_r|\psi\rangle_s)({}_p\langle0|_q\langle1|_r\langle\psi|X_r{}_s\langle\psi|E^\dag) + (E|0\rangle_p|1\rangle_qX_r|\phi\rangle_r|\phi\rangle_s)({}_p\langle0|_q\langle1|_r\langle\phi|X_r{}_s\langle\phi|E^\dag) \\
&\qquad + (E|1\rangle_p|1\rangle_qX_rZ_r|\psi\rangle_r|\psi\rangle_s)({}_p\langle1|_q\langle1|_r\langle\psi|Z_rX_r{}_s\langle\psi|E^\dag)\\ &\qquad + (E|1\rangle_p|1\rangle_qX_rZ_r|\phi\rangle_r|\phi\rangle_s)({}_p\langle1|_q\langle1|_r\langle\phi|Z_rX_r{}_s\langle\phi|E^\dag) \\
&= |\psi\rangle_p\langle\psi|\otimes|\psi\rangle_s\langle\psi| + |\phi\rangle_p\langle\phi|\otimes|\phi\rangle_s\langle\phi|.
\end{align*}
Therefore, the judgment is valid.

\subsection{Verification of Judgment (\ref{ju-tele1})}

This subsection gives a verification of the second formulation of correctness of quantum teleportation.
A formal proof of judgment (\ref{ju-tele1}) using the inference rules of rqPD is presented in Figure \ref{qRHL-tel1}.
\begin{figure}[b]
	\centering
	\begin{align*}
	&\left\{D = {}_q\langle0|{}_r\langle0|H_q{\rm CNOT}_{q,r}{\rm CNOT}_{p,q}H_pCH_p{\rm CNOT}_{p,q}{\rm CNOT}_{q,r}H_q|0\rangle_q|0\rangle_r\right\} \\
	&\bullet\quad q:=|0\rangle;\ \sim\ {\bf skip;} \quad\quad{\rm(Init\text{-}L)}\\
	&\left\{{}_r\langle0|H_q{\rm CNOT}_{q,r}{\rm CNOT}_{p,q}H_pCH_p{\rm CNOT}_{p,q}{\rm CNOT}_{q,r}H_q|0\rangle_r\right\} \\
	&\bullet\quad r:=|0\rangle;\ \sim\ {\bf skip;} \quad\quad{\rm(Init\text{-}L)}\\
	&\left\{H_q{\rm CNOT}_{q,r}{\rm CNOT}_{p,q}H_pCH_p{\rm CNOT}_{p,q}{\rm CNOT}_{q,r}H_q\right\} \\
	&\bullet\quad q:=H[q];\ \sim\ {\bf skip;} \quad\quad{\rm(UT\text{-}L)}\\
	&\left\{{\rm CNOT}_{q,r}{\rm CNOT}_{p,q}H_pCH_p{\rm CNOT}_{p,q}{\rm CNOT}_{q,r}\right\}\\
	&\bullet\quad q,r:=\ {\rm CNOT}[q,r];\ \sim\ {\bf skip;} \quad\quad{\rm(UT\text{-}L)}\\
	&\left\{{\rm CNOT}_{p,q}H_pCH_p{\rm CNOT}_{p,q}\right\} \\
	&\bullet\quad p,q:=\ {\rm CNOT}[p,q];\ \sim\ {\bf skip;} \quad\quad{\rm(UT\text{-}L)}\\
	&\left\{H_pCH_p\right\} \\
	&\bullet\quad p:=H[p];\ \sim\ {\bf skip;} \quad\quad{\rm(UT\text{-}L)}\\
	&\left\{C = |0\rangle_q\langle0|\otimes B + |1\rangle_q\langle1|\otimes (X_rBX_r)\right\} \\
	&\bullet\ \ \begin{array}{l}
	\mathbf{if}\ \cM[q]=0\rightarrow \mathbf{skip}\\ \ \ \ \ \ \square\, \ \ \ \ \ \ \ \ 1\rightarrow r:=X[r]\ \mathbf{fi};
	\end{array}\ \sim\ {\bf skip;} \quad\quad{\rm(IF\text{-}L)}\\
	&\left\{B = |0\rangle_p\langle0|\otimes A + |1\rangle_p\langle1|\otimes (Z_rAZ_r)\right\} \\
	&\bullet\ \ \begin{array}{l}
	\mathbf{if}\ \cM[p]=0\rightarrow \mathbf{skip}\\ \ \ \ \ \ \square\, \ \ \ \ \ \ \ \ 1\rightarrow r:=Z[r]\ \mathbf{fi};
	\end{array}\ \sim\ {\bf skip;} \quad\quad{\rm(IF\text{-}L)}\\
	&\left\{A = \frac{1}{2}\left(I_{rs}+\sum_{ij=0}^1|i\>_r\<j|\otimes|j\>_s\<i|\right)\right\}
	\end{align*}
	\caption{Alternative verification of QTEL$\sim$\textbf{skip} in rqPD.}\label{qRHL-tel1}
\end{figure}
Similarly to the proof of judgment (\ref{ju-tele}) above, we use $s$ to denote the state space of the second program $\bf{skip}$ and introduce a different symbol $A$:
\begin{align*}&A = \frac{1}{2}\left(I_{rs}+\sum_{ij=0}^1|i\>_r\<j|\otimes|j\>_s\<i|\right) = \frac{1}{2}\sum_{ij=0}^1( |i\>_r\<i|\otimes|j\>_s\<j|  +  |i\>_r\<j|\otimes|j\>_s\<i|)\equiv (=_{sym}).\end{align*}
Using the same definition of $E$ in Eqn. (\ref{defE}), our goal is to prove:
$$D = ECE^\dag = (=_{sym}) = \frac{1}{2}\left(I_{ps}+\sum_{ij=0}^1|i\>_p\<j|\otimes|j\>_s\<i|\right).$$
We can calculate $D$ directly from $C$ according to the facts in Eqns. (\ref{propE}):
\begingroup
\allowdisplaybreaks
\begin{align*}
C &= \sum_{m,n=0}^1|m\>_p\<m|\otimes |n\>_q\<n| \otimes X_r^nZ_r^mAZ_r^mX_r^n \\
D &= ECE^\dag = \frac{1}{2}E\Big(\sum_{m,n=0}^1|m\>_p\<m|\otimes |n\>_q\<n| \otimes X_r^nZ_r^m \frac{1}{2}\sum_{ij=0}^1|i\>_r\<i|\otimes|j\>_s\<j|Z_r^mX_r^n  \Big)E^\dag \\
&\quad + \frac{1}{2}E\Big(\sum_{m,n=0}^1|m\>_p\<m|\otimes |n\>_q\<n| \otimes X_r^nZ_r^m \frac{1}{2}\sum_{ij=0}^1|i\>_r\<j|\otimes|j\>_s\<i|Z_r^mX_r^n  \Big)E^\dag \\
&= \frac{1}{2}\sum_{m,n,i,j=0}^1 (E|m\>_p|n\>_qX_r^nZ_r^m|i\>_r|j\>_s)({}_p\<m|_q\<n|_r\<i|_s\<j|Z_r^mX_r^nE^\dag) \\
&\quad + \frac{1}{2}\sum_{m,n,i,j=0}^1 (E|m\>_p|n\>_qX_r^nZ_r^m|i\>_r|j\>_s)({}_p\<m|_q\<n|_r\<j|_s\<i|Z_r^mX_r^nE^\dag) \\
&= \frac{1}{2}\sum_{m,n,i,j=0}^1 \left(\frac{1}{2}|i\>_p|j\>_s\right)\left(\frac{1}{2}{}_p\<i|_s\<j|\right) + \frac{1}{2}\sum_{m,n,i,j=0}^1 \left(\frac{1}{2}|i\>_p|j\>_s\right)\left(\frac{1}{2}{}_p\<j|_s\<i|\right) \\
&= \frac{1}{2}I_{ps} + \frac{1}{2}\sum_{i,j=0}^1 (|i\>_p\<j|\otimes|j\>_s\<i|) \\
&= (=_{sym})
\end{align*}
\endgroup
Therefore, the judgment is valid.

\subsection{Verification of Judgments (\ref{ju-bit}), (\ref{ju-phase}) and (\ref{ju-bp})}

In this subsection, we verify the reliability of quantum teleportation against several models of quantum noise. 

%Principally, any {\bf if} statement can also be written as a quantum operation. More precisely, we can use the following two quantum operations $\mathcal{E}_{\cM_q}$ and $\mathcal{E}_{\cM_p}$ with theirs Kraus representation to replace the last two {\bf if} statement in QTEL respectively.
%\begin{align*}
%	\mathcal{E}_{\cM_q} = \{\mathcal{E}_{0\cM_q}=|0\rangle_q\otimes I_r, \mathcal{E}_{1\cM_q}=|1\rangle_q\otimes X_r\},\quad \mathcal{E}_{\cM_p} = \{\mathcal{E}_{0\cM_p}=|0\rangle_p\otimes I_r, \mathcal{E}_{1\cM_p}=|1\rangle_p\otimes Z_r\}
%\end{align*}

\begin{figure}
	\centering
	\begin{align*}
	&\left\{S = {}_q\langle0|{}_{q^\prime}\langle0|{}_r\langle0|{}_{r^\prime}\langle0|(H_q\otimes H_{q^\prime})R(H_q\otimes H_{q^\prime})|0\rangle_q|0\rangle_{q^\prime}|0\rangle_r|0\rangle_{r^\prime}\right\}
	\\
	&\bullet\quad q:=|0\rangle;\ \sim\ q^\prime:=|0\rangle; \quad\quad{\rm(Init)}\\
	&\left\{{}_r\langle0|{}_{r^\prime}\langle0|(H_q\otimes H_{q^\prime})R(H_q\otimes H_{q^\prime})|0\rangle_r|0\rangle_{r^\prime}\right\} \\
	&\bullet\quad r:=|0\rangle;\ \sim\ r^\prime:=|0\rangle; \quad\quad{\rm(Init)}\\
	&\left\{(H_q\otimes H_{q^\prime})R(H_q\otimes H_{q^\prime})\right\} \\
	&\bullet\quad q:=H[q];\ \sim\ q^\prime:=H[q^\prime]; \quad\quad{\rm(UT)}\\
	&\,\big\{R = \mathcal{E}^*_{q} (Q)\big\}\\
	&\bullet\quad q:=\mathcal{E}[q];\ \sim\ {\bf skip;} \quad\quad{\rm(SO\text{-}L)}\\
	&\left\{Q = ({\rm CNOT}_{q,r}\otimes {\rm CNOT}_{q^\prime,r^\prime})P({\rm CNOT}_{q,r} \otimes {\rm CNOT}_{q^\prime,r^\prime} )\right\}\\
	&\bullet\quad q,r:=\ {\rm CNOT}[q,r];\ \sim\ q^\prime,r^\prime:=\ {\rm CNOT}[q^\prime,r^\prime]; \quad\quad{\rm(UT)}\\
	&\left\{P = ({\rm CNOT}_{p,q}H_p\otimes{\rm CNOT}_{p^\prime,q^\prime}H_{p^\prime})   D
	(H_p{\rm CNOT}_{p,q}\otimes H_{p^\prime}{\rm CNOT}_{p^\prime,q^\prime})\right\} \\
	&\bullet\quad p,q:=\ {\rm CNOT}[p,q];\ \sim\ p^\prime,q^\prime:=\ {\rm CNOT}[p^\prime,q^\prime]; \quad\quad{\rm(UT)}\\
	&\left\{(H_p\otimes H_{p^\prime})D(H_p\otimes H_{p^\prime})\right\} \\
	&\bullet\quad p:=H[p];\ \sim\ p^\prime:=H[p^\prime]; \quad\quad{\rm(UT)}\\
	&\,\big\{D = \mathcal{E}^*_{p} (C)\big\} \\
	&\bullet\quad p:=\mathcal{E}[p];\ \sim\ {\bf skip;} \quad\quad{\rm(SO\text{-}L)}\\
	&\,\bigg\{C =  \sum_{j,j^\prime=0,1} |j\>_q\<j|\otimes |j^\prime\>_{q^\prime}\<j^\prime| \otimes (X_r^j\otimes X_{r^\prime}^{j^\prime})B (X_r^j\otimes X_{r^\prime}^{j^\prime})\bigg\} \\
	&\bullet\ \ \begin{array}{l}
	\mathbf{if}\ \cM[q]=0\rightarrow \mathbf{skip}\\ \square\, 1\rightarrow r:=X[r]\ \mathbf{fi};
	\end{array}\ \sim\ \begin{array}{l}
	\mathbf{if}\ \cM[q^\prime]=0\rightarrow \mathbf{skip}\\ \square\, 1\rightarrow r^\prime:=X[r^\prime]\ \mathbf{fi};
	\end{array}\ \quad\quad{\rm(IF)}\\
	&\,\bigg\{B = \sum_{i,i^\prime=0,1} |i\>_p\<i|\otimes |i^\prime\>_{p^\prime}\<i^\prime| \otimes (Z_r^i\otimes Z_{r^\prime}^{i^\prime})A (Z_r^i\otimes Z_{r^\prime}^{i^\prime})\bigg\} \\
	&\bullet\ \ \begin{array}{l}
	\mathbf{if}\ \cM[p]=0\rightarrow \mathbf{skip}\\ \square\, 1\rightarrow r:=Z[r]\ \mathbf{fi};
	\end{array}\ \sim\ \begin{array}{l}
	\mathbf{if}\ \cM[p^\prime]=0\rightarrow \mathbf{skip}\\  \square\, 1\rightarrow r^\prime:=Z[r^\prime]\ \mathbf{fi};
	\end{array}\ \quad\quad{\rm(IF)}\\
	&\left\{A = |\psi\rangle_r\langle\psi|\otimes|\psi\rangle_{r^\prime}\langle\psi|\right\}
	\end{align*}
	\caption{Verification of ${\rm QTEL_{\it noise}}\sim {\rm QTEL}$ in rqPD.}\label{qRHL-tel-noise}
\end{figure}

Now, using the inference rules in Figure \ref{fig 4.4-0}, \ref{fig 4.5-0} and \ref{fig 5.1}, we present the formal proof in Figure \ref{qRHL-tel-noise}. We can calculate the following equations directly. For simplicity, $[\cdot]$ is a copy of the context between the previous $[$ and $]$.
\begingroup
\allowdisplaybreaks
\begin{align*}
&C = \sum_{i,j,i^\prime,j^\prime = 0}^1 \left[|i\rangle_p|j\rangle_qX^jZ^i|\psi\rangle_r|i^\prime\rangle_{p^\prime}|j^\prime\rangle_{q^\prime}X^{j^\prime} Z^{i^\prime}|\psi\rangle_{r^\prime}\right][\cdot]^\dag \\
&F = H_q{\rm CNOT}_{q,r}{\rm CNOT}_{p,q}H_p, \\
&F^\prime = H_{q^\prime}{\rm CNOT}_{q^\prime,r^\prime}{\rm CNOT}_{p^\prime,q^\prime}H_{p^\prime} \\
&D = \mathcal{E}^*_{p}(C) = \sum_{k = 0}^1 E^\dag_{kp}(C)E_{kp} \\
&R = \mathcal{E}^*_{q}(Q) = \mathcal{E}^*_{q}((H_qH_q)\otimes (H_{q^\prime}H_{q^\prime})Q(H_qH_q)\otimes (H_{q^\prime}H_{q^\prime}))  \\
&\ \ = \mathcal{E}^*_{q}((H_qF)\otimes (H_{q^\prime}F^\prime) D(F^\dag H_q)\otimes (F^{\prime\dag}H_{q^\prime})) \\
&\ \ = \sum_{k,l=0}^1(E^\dag_{lq}H_qFE^\dag_{kp}) \otimes (H_{q^\prime}F^\prime) (C) (E_{kp}H_qF^\dag E_{lq})\otimes  (F^{\prime\dag}H_{q^\prime})
\end{align*}
\endgroup
Our aim is to calculate what is S:
\begin{align*}
S &= {}_q\langle0|{}_{q^\prime}\langle0|{}_r\langle0|{}_{r^\prime}\langle0|(H_q\otimes H_{q^\prime})R(H_q\otimes H_{q^\prime})|0\rangle_q|0\rangle_{q^\prime}|0\rangle_r|0\rangle_{r^\prime} \\
&= \sum_{k,l=0}^1 ({}_q\langle0|{}_r\langle0|H_qE^\dag_{lq}H_qFE^\dag_{kp})\otimes ({}_{q^\prime}\langle0|{}_{r^\prime}\langle0|F^\prime) (C) (E_{kp}F^\dag H_qE_{lq}H_q|0\rangle_q|0\rangle_r)\otimes (F^{\prime\dag}|0\rangle_{q^\prime}|0\rangle_{r^\prime})  \\
&= \left\{\sum_{k,l=0}^1 \sum_{i,j=0}^1 \left[ {}_q\langle0|{}_r\langle0|H_qE^\dag_{lq}H_qFE^\dag_{kp}|i\rangle_p|j\rangle_qX^jZ^i|\psi\rangle_r   \right][\cdot]^\dag\right\} \otimes |\psi\rangle_{p^\prime}\langle\psi| \\
&= S_1\otimes |\psi\rangle_{p^\prime}\langle\psi|
\end{align*}
using the following fact (of course, it also holds with superscript primes over each index):
\begin{align*}
{}_{q}\langle0|{}_{r}\langle0|F|i\rangle_{p}|j\rangle_{q}X^{j} Z^{i}|\psi\rangle_{r} = \frac{1}{2}|\psi\rangle_{p},\quad\quad
{}_{q}\langle1|{}_{r}\langle0|F|i\rangle_{p}|j\rangle_{q}X^{j} Z^{i}|\psi\rangle_{r} = \frac{1}{2}(-)^jZ|\psi\rangle_{p}.
\end{align*}

{\vskip 4pt}

$\bullet$ ${\rm QTEL_{\it BF}}\sim {\rm QTEL}$

{\vskip 4pt}

For the bif flip noise, the following fact can be easily realized:
\begin{align*}
E^\dag_{1p}|i\rangle_p|j\rangle_qX^jZ^i|\psi\rangle_r &= \sqrt{1-p}|1\oplus i\rangle_p|j\rangle_qX^jZ^i|\psi\rangle_r = \sqrt{1-p}|1\oplus i\rangle_p|j\rangle_qX^jZ^{(1\oplus i)}Z|\psi\rangle_r \\
{}_q\langle 0|H_qE^\dag_{1q}H_q &= \sqrt{1-p} {}_q\langle 0|
\end{align*}
Thus,
\begin{align*}
S_1 &= \sum_{i,j=0}^1\bigg\{\left[ {}_q\langle0|{}_r\langle0|H_qE^\dag_{0q}H_qFE^\dag_{0p}|i\rangle_p|j\rangle_qX^jZ^i|\psi\rangle_r   \right][\cdot]^\dag + \left[ {}_q\langle0|{}_r\langle0|H_qE^\dag_{0q}H_qFE^\dag_{1p}|i\rangle_p|j\rangle_qX^jZ^i|\psi\rangle_r   \right][\cdot]^\dag \\
&\qquad + \left[ {}_q\langle0|{}_r\langle0|H_qE^\dag_{1q}H_qFE^\dag_{0p}|i\rangle_p|j\rangle_qX^jZ^i|\psi\rangle_r   \right][\cdot]^\dag + \left[ {}_q\langle0|{}_r\langle0|H_qE^\dag_{1q}H_qFE^\dag_{1p}|i\rangle_p|j\rangle_qX^jZ^i|\psi\rangle_r   \right][\cdot]^\dag\bigg\} \\
&= p^2 \sum_{i,j=0}^1\left[ {}_q\langle0|{}_r\langle0|F|i\rangle_p|j\rangle_qX^jZ^i|\psi\rangle_r   \right][\cdot]^\dag  + p(1-p)\sum_{i,j=0}^1\left[ {}_q\langle0|{}_r\langle0|F|i\rangle_p|j\rangle_qX^jZ^iZ|\psi\rangle_r   \right][\cdot]^\dag \\
&\quad + p(1-p)\sum_{i,j=0}^1\left[ {}_q\langle0|{}_r\langle0|F|i\rangle_p|j\rangle_qX^jZ^i|\psi\rangle_r   \right][\cdot]^\dag  + (1-p)^2\sum_{i,j=0}^1\left[ {}_q\langle0|{}_r\langle0|F|i\rangle_p|j\rangle_qX^jZ^iZ|\psi\rangle_r   \right][\cdot]^\dag \\
&= p^2|\psi\rangle_p\langle\psi| + p(1-p)Z|\psi\rangle_p\langle\psi|Z + p(1-p)|\psi\rangle_p\langle\psi| + (1-p)^2Z|\psi\rangle_p\langle\psi|Z \\
&= p|\psi\rangle_p\langle\psi| + (1-p)Z|\psi\rangle_p\langle\psi|Z \\
&= \mathcal{E}_{\it PF}(p)(|\psi\rangle_p\langle\psi|).
\end{align*}
Therefore, the precondition $S$ is:
$$S = S_1\otimes |\psi\rangle_{p^\prime}\langle\psi| =  \mathcal{E}_{\it PF}(p)(|\psi\rangle_p\langle\psi|)\otimes |\psi\rangle_{p^\prime}\langle\psi|.$$

{\vskip 4pt}

$\bullet$ ${\rm QTEL_{\it PF}}\sim {\rm QTEL}$

{\vskip 4pt}

For the phase flip noise, the following fact can be easily realized:
\begin{align*}
E^\dag_{1p}|i\rangle_p|j\rangle_qX^jZ^i|\psi\rangle_r &= (-)^i\sqrt{1-p}|i\rangle_p|j\rangle_qX^jZ^i|\psi\rangle_r \\
{}_q\langle 0|H_qE^\dag_{1q}H_q &= \sqrt{1-p} {}_q\langle 1|
\end{align*}
Thus,
\begin{align*}
S_1 &= \sum_{i,j=0}^1\bigg\{\left[ {}_q\langle0|{}_r\langle0|H_qE^\dag_{0q}H_qFE^\dag_{0p}|i\rangle_p|j\rangle_qX^jZ^i|\psi\rangle_r   \right][\cdot]^\dag + \left[ {}_q\langle0|{}_r\langle0|H_qE^\dag_{0q}H_qFE^\dag_{1p}|i\rangle_p|j\rangle_qX^jZ^i|\psi\rangle_r   \right][\cdot]^\dag \\
&\qquad + \left[ {}_q\langle0|{}_r\langle0|H_qE^\dag_{1q}H_qFE^\dag_{0p}|i\rangle_p|j\rangle_qX^jZ^i|\psi\rangle_r   \right][\cdot]^\dag + \left[ {}_q\langle0|{}_r\langle0|H_qE^\dag_{1q}H_qFE^\dag_{1p}|i\rangle_p|j\rangle_qX^jZ^i|\psi\rangle_r   \right][\cdot]^\dag\bigg\} \\
&= p^2 \sum_{i,j=0}^1\left[ {}_q\langle0|{}_r\langle0|F|i\rangle_p|j\rangle_qX^jZ^i|\psi\rangle_r   \right][\cdot]^\dag  + p(1-p)\sum_{i,j=0}^1\left[ {}_q\langle0|{}_r\langle0|F|i\rangle_p|j\rangle_qX^jZ^i|\psi\rangle_r   \right][\cdot]^\dag \\
&\quad + p(1-p)\sum_{i,j=0}^1\left[ {}_q\langle1|{}_r\langle0|F|i\rangle_p|j\rangle_qX^jZ^i|\psi\rangle_r   \right][\cdot]^\dag  + (1-p)^2\sum_{i,j=0}^1\left[ {}_q\langle1|{}_r\langle0|F|i\rangle_p|j\rangle_qX^jZ^i|\psi\rangle_r   \right][\cdot]^\dag \\
&= p^2|\psi\rangle_p\langle\psi| + p(1-p)|\psi\rangle_p\langle\psi| + p(1-p)Z|\psi\rangle_p\langle\psi|Z + (1-p)^2Z|\psi\rangle_p\langle\psi|Z \\
&= p|\psi\rangle_p\langle\psi| + (1-p)Z|\psi\rangle_p\langle\psi|Z \\
&= \mathcal{E}_{\it PF}(p)(|\psi\rangle_p\langle\psi|).
\end{align*}
Therefore, the precondition $S$ is:
$$S = S_1\otimes |\psi\rangle_{p^\prime}\langle\psi| =  \mathcal{E}_{\it PF}(p)(|\psi\rangle_p\langle\psi|)\otimes |\psi\rangle_{p^\prime}\langle\psi|.$$

{\vskip 4pt}

$\bullet$ ${\rm QTEL_{\it BPF}}\sim {\rm QTEL}$

{\vskip 4pt}

For the phase flip noise, the following fact can be easily realized:
\begin{align*}
E^\dag_{1p}|i\rangle_p|j\rangle_qX^jZ^i|\psi\rangle_r &= \sqrt{1-p}(-)^i{\bm i}|1\oplus i\rangle_p|j\rangle_qX^jZ^i|\psi\rangle_r = \sqrt{1-p}(-)^i{\bm i}|1\oplus i\rangle_p|j\rangle_qX^jZ^{(1\oplus i)}Z|\psi\rangle_r \\
{}_q\langle 0|H_qE^\dag_{1q}H_q &= \sqrt{1-p}{\bm i} {}_q\langle 1|
\end{align*}
Thus,
\begin{align*}
S_1
&= \sum_{i,j=0}^1\bigg\{\left[ {}_q\langle0|{}_r\langle0|H_qE^\dag_{0q}H_qFE^\dag_{0p}|i\rangle_p|j\rangle_qX^jZ^i|\psi\rangle_r   \right][\cdot]^\dag + \left[ {}_q\langle0|{}_r\langle0|H_qE^\dag_{0q}H_qFE^\dag_{1p}|i\rangle_p|j\rangle_qX^jZ^i|\psi\rangle_r   \right][\cdot]^\dag \\
&\qquad + \left[ {}_q\langle0|{}_r\langle0|H_qE^\dag_{1q}H_qFE^\dag_{0p}|i\rangle_p|j\rangle_qX^jZ^i|\psi\rangle_r   \right][\cdot]^\dag + \left[ {}_q\langle0|{}_r\langle0|H_qE^\dag_{1q}H_qFE^\dag_{1p}|i\rangle_p|j\rangle_qX^jZ^i|\psi\rangle_r   \right][\cdot]^\dag\bigg\} \\
&= p^2 \sum_{i,j=0}^1\left[ {}_q\langle0|{}_r\langle0|F|i\rangle_p|j\rangle_qX^jZ^i|\psi\rangle_r   \right][\cdot]^\dag  + p(1-p)\sum_{i,j=0}^1\left[ {}_q\langle0|{}_r\langle0|F|i\rangle_p|j\rangle_qX^jZ^iZ|\psi\rangle_r   \right][\cdot]^\dag \\
& + p(1-p)\sum_{i,j=0}^1\left[ {}_q\langle1|{}_r\langle0|F|i\rangle_p|j\rangle_qX^jZ^i|\psi\rangle_r   \right][\cdot]^\dag  + (1-p)^2\sum_{i,j=0}^1\left[ {}_q\langle1|{}_r\langle0|F|i\rangle_p|j\rangle_qX^jZ^iZ|\psi\rangle_r   \right][\cdot]^\dag \\
&= p^2|\psi\rangle_p\langle\psi| + p(1-p)Z|\psi\rangle_p\langle\psi|Z + p(1-p)Z|\psi\rangle_p\langle\psi|Z + (1-p)^2ZZ|\psi\rangle_p\langle\psi|ZZ \\
&= (p^2+(1-p)^2)|\psi\rangle_p\langle\psi| + 2p(1-p)Z|\psi\rangle_p\langle\psi|Z \\
&= \mathcal{E}_{\it PF}(p^2+(1-p)^2)(|\psi\rangle_p\langle\psi|).
\end{align*}
Therefore, the precondition $S$ is:
$$S = S_1\otimes |\psi\rangle_{p^\prime}\langle\psi| =  \mathcal{E}_{\it PF}(p^2+(1-p)^2)(|\psi\rangle_p\langle\psi|)\otimes |\psi\rangle_{p^\prime}\langle\psi|.$$

\subsection{Proof of Judgment (\ref{QOTP-1-Correct}), (\ref{QOTP-1-Secure}), (\ref{QOTP-n-Correct}) and (\ref{QOTP-n-Secure})}

In this subsection, we verify the correctness and security of quantum one-time pad. 
The proof outlines for judgments (\ref{QOTP-1-Correct}), (\ref{QOTP-1-Secure}), (\ref{QOTP-n-Correct}) and (\ref{QOTP-n-Secure}) are shown in Figs. \ref{fig-QOTP-1-Correct}, \ref{fig-QOTP-1-Secure}, \ref{fig-QOTP-n-Correct} and \ref{fig-QOTP-n-Secure}, respectively.

It is worth noting that the initializations of registers $a$ (or $a_i$) and $b$ (or $b_i$) are regarded as the creation of new local qubits, so the rules (SO-L) / (SO) are used instead of (Init-L) / (Init). All the calculations are straightforward, except the following equation:
$$
I_{\bar{p}} = \sum_{\forall\ i\in[n]: x_i,z_i\in\{0,1\}}\prod_{i}^n (Z_{p_i}^{z_i}X_{p_i}^{x_i})|\psi\>_{\bar{p}}\<\psi|\prod_{i}^n (X_{p_i}^{x_i}Z_{p_i}^{z_i})$$
for all possible pure state $|\psi\>\in\h_{\bar{p}}$. This fact has already been proved in \cite{MTW00} (see Theorem 4.1 there) if we realize the relationship between Pauli matrices $ZX = -iY$.

\begin{figure}
	\footnotesize
	\begin{align*}
	& \left\{(=_{sym}) = \frac{1}{2}(I_{pp^\prime}+S_{p;p^\prime})\right\}   \\
	&\bullet\ a := |0\>; b := |0\>; \ \sim \ {\bf skip} \quad\quad{\rm(SO\text{-}L)} \\
	& \left\{I_a \otimes I_b \otimes \frac{1}{2}(I_{pp^\prime}+S_{p;p^\prime})\right\} \\
	&\bullet\ a := H[a]; b := H[b]; \ \sim \ {\bf skip}   \quad\quad{\rm(UT\text{-}L)}\\
	& \left\{I_a \otimes I_b \otimes \frac{1}{2}(I_{pp^\prime}+S_{p;p^\prime})\right\} \\
	&\bullet\ \ \begin{array}{l}\mathbf{if}\ \cM[a,b]=00\rightarrow \mathbf{skip}\\
	\square\qquad\qquad\quad 01\rightarrow \mathbf{skip}\\
	\square\qquad\qquad\quad 10\rightarrow \mathbf{skip}\\
	\square\qquad\qquad\quad 11\rightarrow \mathbf{skip}\\
	\mathbf{fi}
	\end{array} \ \sim \ {\bf skip}   \quad\quad{\rm(IF\text{-}L)}\\
	& \bigg\{\{I_a \otimes I_b \otimes (I_{pp^\prime}+S_{p;p^\prime}) = |00\>_{ab}\<00|\otimes \frac{1}{2}(I_{pp^\prime}+S_{p;p^\prime}) + |01\>_{ab}\<01|\otimes \frac{1}{2}(I_{pp^\prime}+S_{p;p^\prime}) \\
	& \qquad\qquad\qquad\qquad\qquad\qquad + |10\>_{ab}\<10|\otimes \frac{1}{2}(I_{pp^\prime}+S_{p;p^\prime}) + |11\>_{ab}\<11|\otimes \frac{1}{2}(I_{pp^\prime}+S_{p;p^\prime})\bigg\} \\
	&\bullet\ \ \begin{array}{l}\mathbf{if}\ \cM[a,b]=00\rightarrow \mathbf{skip}\\
	\square\qquad\qquad\quad 01\rightarrow p = Z[p]\\
	\square\qquad\qquad\quad 10\rightarrow p = X[p\\
	\square\qquad\qquad\quad 11\rightarrow p = Z[p];p = X[p]\\
	\mathbf{fi}
	\end{array} \ \sim \ {\bf skip}   \quad\quad{\rm(IF\text{-}L)}\\
	& \bigg\{|00\>_{ab}\<00|\otimes \frac{1}{2}(I_{pp^\prime}+S_{p;p^\prime}) + |01\>_{ab}\<01|\otimes \frac{1}{2}(I_{pp^\prime}+Z_pS_{p;p^\prime}Z_p) \\
	&\qquad+ |10\>_{ab}\<10|\otimes \frac{1}{2}(I_{pp^\prime}+X_pS_{p;p^\prime}X_p) + |11\>_{ab}\<11|\otimes \frac{1}{2}(I_{pp^\prime}+Z_pX_pS_{p;p^\prime}X_pZ_p)\bigg\} \\
	&\bullet\ \ \begin{array}{l}\mathbf{if}\ \cM[a,b]=00\rightarrow \mathbf{skip}\\
	\square\qquad\qquad\quad 01\rightarrow p = Z[p]\\
	\square\qquad\qquad\quad 10\rightarrow p = X[p\\
	\square\qquad\qquad\quad 11\rightarrow p = Z[p];p = X[p]\\
	\mathbf{fi}
	\end{array} \ \sim \ {\bf skip}  \quad\quad{\rm(IF\text{-}L)}\\
	&\left\{(=_{sym}) = I_a \otimes I_b \otimes \frac{1}{2}(I_{pp^\prime}+S_{p;p^\prime})\right\} \\
	&\bullet\ \ {\bf Tr}[a];{\bf Tr}[b]\ \sim\ {\bf skip} \quad\quad{\rm(SO\text{-}L)}\\
	&\left\{(=_{sym}) = \frac{1}{2}(I_{pp^\prime}+S_{p;p^\prime})\right\}
	\end{align*}
	\caption{Verification of correctness for QOTP. The proof outline for Judgment (\ref{QOTP-1-Correct}).}
	\label{fig-QOTP-1-Correct}
\end{figure}

\begin{figure}
	\footnotesize
	\begin{align*}
	& \bigg\{\frac{I_p\otimes I_{p^\prime}}{2} =
	\frac{1}{4}|\psi\>_p\<\psi|\otimes I_{p^\prime} + \frac{1}{4}Z_p|\psi\>_p\<\psi|Z_p\otimes I_{p^\prime} + \frac{1}{4}X_p|\psi\>_p\<\psi|X_p\otimes I_{p^\prime} + \frac{1}{4}Z_pX_p|\psi\>_p\<\psi|X_pZ_p\otimes I_{p^\prime}\bigg\} \\
	&\bullet\ a := |0\>; b := |0\>; \ \sim \ a^\prime := |0\>; b^\prime := |0\>; \quad\quad{\rm(SO)}\\
	& \Big\{|++\>_{ab}\<++|\otimes|++\>_{a^\prime b^\prime}\<++|\otimes |\psi\>_p\<\psi|\otimes I_{p^\prime} \\
	&\quad + |+-\>_{ab}\<+-|\otimes|+-\>_{a^\prime b^\prime}\<+-|\otimes Z_p|\psi\>_p\<\psi|Z_p\otimes I_{p^\prime} \\
	&\quad + |-+\>_{ab}\<-+|\otimes|-+\>_{a^\prime b^\prime}\<-+|\otimes X_p|\psi\>_p\<\psi|X_p\otimes I_{p^\prime} \\
	&\quad + |++\>_{ab}\<++|\otimes|++\>_{a^\prime b^\prime}\<++|\otimes Z_pX_p|\psi\>_p\<\psi|X_pZ_p\otimes I_{p^\prime}\Big\} \\
	&\bullet\ a := H[a]; b := H[b]; \ \sim \ a^\prime := H[a^\prime]; b^\prime := H[b^\prime];   \quad\quad{\rm(UT)}\\
	& \Big\{|00\>_{ab}\<00|\otimes|00\>_{a^\prime b^\prime}\<00|\otimes |\psi\>_p\<\psi|\otimes I_{p^\prime}
	+ |01\>_{ab}\<01|\otimes|01\>_{a^\prime b^\prime}\<01|\otimes Z_p|\psi\>_p\<\psi|Z_p\otimes I_{p^\prime} \\
	&\ \ + |10\>_{ab}\<10|\otimes|10\>_{a^\prime b^\prime}\<10|\otimes X_p|\psi\>_p\<\psi|X_p\otimes I_{p^\prime}
	+ |11\>_{ab}\<11|\otimes|11\>_{a^\prime b^\prime}\<11|\otimes Z_pX_p|\psi\>_p\<\psi|X_pZ_p\otimes I_{p^\prime}\Big\} \\
	&\bullet\ \ \begin{array}{l}\mathbf{if}\ \cM[a,b]=00\rightarrow \mathbf{skip}\\
	\square\qquad\qquad\quad 01\rightarrow \mathbf{skip}\\
	\square\qquad\qquad\quad 10\rightarrow \mathbf{skip}\\
	\square\qquad\qquad\quad 11\rightarrow \mathbf{skip}\\
	\mathbf{fi}
	\end{array} \ \sim \ \begin{array}{l}\mathbf{if}\ \cM[a,b]=00\rightarrow \mathbf{skip}\\
	\square\qquad\qquad\quad 01\rightarrow \mathbf{skip}\\
	\square\qquad\qquad\quad 10\rightarrow \mathbf{skip}\\
	\square\qquad\qquad\quad 11\rightarrow \mathbf{skip}\\
	\mathbf{fi}
	\end{array}   \quad\quad{\rm(IF\text{-}w)}\\
	& \Big\{|00\>_{ab}\<00|\otimes|00\>_{a^\prime b^\prime}\<00|\otimes |\psi\>_p\<\psi|\otimes I_{p^\prime}
	+ |01\>_{ab}\<01|\otimes|01\>_{a^\prime b^\prime}\<01|\otimes Z_p|\psi\>_p\<\psi|Z_p\otimes I_{p^\prime} \\
	&\ \ + |10\>_{ab}\<10|\otimes|10\>_{a^\prime b^\prime}\<10|\otimes X_p|\psi\>_p\<\psi|X_p\otimes I_{p^\prime}
	+ |11\>_{ab}\<11|\otimes|11\>_{a^\prime b^\prime}\<11|\otimes Z_pX_p|\psi\>_p\<\psi|X_pZ_p\otimes I_{p^\prime}\Big\} \\
	&\bullet\ \ \begin{array}{l}\mathbf{if}\ \cM[a,b]=00\rightarrow \mathbf{skip}\\
	\square\qquad\qquad\quad 01\rightarrow p = Z[p]\\
	\square\qquad\qquad\quad 10\rightarrow p = X[p\\
	\square\qquad\qquad\quad 11\rightarrow p = Z[p];p = X[p]\\
	\mathbf{fi}
	\end{array} \ \sim \ \begin{array}{l}\mathbf{if}\ \cM[a,b]=00\rightarrow \mathbf{skip}\\
	\square\qquad\qquad\quad 01\rightarrow p = Z[p]\\
	\square\qquad\qquad\quad 10\rightarrow p = X[p\\
	\square\qquad\qquad\quad 11\rightarrow p = Z[p];p = X[p]\\
	\mathbf{fi}
	\end{array}  \quad\quad{\rm(IF\text{-}w)}\\
	&\left\{I_{ab}\otimes I_{a^\prime b^\prime} \otimes |\psi\>_p\<\psi|\otimes I_{p^\prime}\right\} \\
	&\bullet\ \ {\bf Tr}[a];{\bf Tr}[b]\ \sim\ {\bf Tr}[a^\prime];{\bf Tr}[b^\prime] \quad\quad{\rm(SO)}\\
	&\left\{|\psi\>_p\<\psi|\otimes I_{p^\prime}\right\}
	\end{align*}
	\caption{Verification of the security for QOTP. The proof outline for Judgment (\ref{QOTP-1-Secure}). $|+\> = \frac{1}{\sqrt{2}}(|0\>+|1\>)$ and $|-\> = \frac{1}{\sqrt{2}}(|0\>-|1\>)$.}
	\label{fig-QOTP-1-Secure}
\end{figure}

\begin{figure}
	\footnotesize
	\begin{align*}
	& \left\{(=_{sym}) = \frac{1}{2}(I_{\bar{p}\bar{p}^\prime}+S_{\bar{p};\bar{p}^\prime})\right\}   \\
	&\bullet\ a_1 := |0\>; \cdots; a_n := |0\>; b_1 := |0\>; \cdots; b_n := |0\>;
	\ \sim \ {\bf skip} \quad\quad{\rm(SO\text{-}L)} \\
	& \left\{\bigotimes_{i=1}^n(I_{a_i} \otimes I_{b_i}) \otimes \frac{1}{2}(I_{\bar{p}\bar{p}^\prime}+S_{\bar{p};\bar{p}^\prime})\right\} \\
	&\bullet\ a_1 := H[a_1]; \cdots; a_n := H[a_n]; b_1 := H[b_1]; \cdots; b_n := H[b_n]; \ \sim \ {\bf skip}   \quad\quad{\rm(UT\text{-}L)}\\
	& \left\{\bigotimes_{i=1}^n(I_{a_i} \otimes I_{b_i}) \otimes \frac{1}{2}(I_{\bar{p}\bar{p}^\prime}+S_{\bar{p};\bar{p}^\prime})\right\} \\
	&\bullet\ \ \begin{array}{l}
	\mathbf{if}\ (\square x_1z_1\cdot\cM[a_1,b_1]=x_1z_1\rightarrow \mathbf{skip})\ \mathbf{fi}; \\
	\vdots \\
	\mathbf{if}\ (\square x_nz_n\cdot\cM[a_n,b_n]=x_nz_n\rightarrow \mathbf{skip})\ \mathbf{fi}
	\end{array} \ \sim \ {\bf skip}   \quad\quad{\rm(IF\text{-}L)}\\
	& \Bigg\{\bigotimes_{i=1}^n(I_{a_i} \otimes I_{b_i}) \otimes \frac{1}{2}(I_{\bar{p}\bar{p}^\prime}+S_{\bar{p};\bar{p}^\prime}) =
	\sum_{\forall\ i\in[n]: x_i,z_i \in\{0,1\}} \bigotimes_{i=1}^n(|x_iz_i\>_{a_ib_i}\<x_iz_i|)\otimes\frac{1}{2}\bigg(I_{\bar{p}\bar{p}^\prime}+  \\
	& \qquad\qquad\qquad\qquad\qquad\qquad\qquad\quad \prod_{i}^n (Z_{p_i}^{z_i}X_{p_i}^{x_i})\prod_{i}^n (Z_{p_i}^{z_i}X_{p_i}^{x_i})S_{\bar{p};\bar{p}^\prime}
	\prod_{i}^n (X_{p_i}^{x_i}Z_{p_i}^{z_i})\prod_{i}^n (X_{p_i}^{x_i}Z_{p_i}^{z_i})\bigg)\Bigg\} \\
	&\bullet\ \ \begin{array}{l}\mathbf{if}\ (\square x_1z_1\cdot\cM[a_1,b_1]=x_1z_1\rightarrow p_1 = Z^{z_1}[p_1];\ p_1 = X^{x_1}[p_1])\ \mathbf{fi}; \\
	\vdots \\
	\mathbf{if}\ (\square x_nz_n\cdot\cM[a_n,b_n]=x_nz_n\rightarrow p_n = Z^{z_n}[p_n];\ p_n = X^{x_n}[p_n])\ \mathbf{fi}
	\end{array} \ \sim \ {\bf skip}   \quad\quad{\rm(IF\text{-}L)}\\
	& \left\{\sum_{\forall\ i\in[n]: x_i,z_i\in\{0,1\}} \bigotimes_{i=1}^n(|x_iz_i\>_{a_ib_i}\<x_iz_i|)\otimes \frac{1}{2}\bigg(I_{\bar{p}\bar{p}^\prime}+ \prod_{i}^n (Z_{p_i}^{z_i}X_{p_i}^{x_i})S_{\bar{p};\bar{p}^\prime}
	\prod_{i}^n (X_{p_i}^{x_i}Z_{p_i}^{z_i})\bigg)\right\} \\
	&\bullet\ \ \begin{array}{l}\mathbf{if}\ (\square x_1z_1\cdot\cM[a_1,b_1]=x_1z_1\rightarrow p_1 = Z^{z_1}[p_1];\ p_1 = X^{x_1}[p_1])\ \mathbf{fi}; \\
	\vdots \\
	\mathbf{if}\ (\square x_nz_n\cdot\cM[a_n,b_n]=x_nz_n\rightarrow p_n = Z^{z_n}[p_n];\ p_n = X^{x_n}[p_n])\ \mathbf{fi}
	\end{array} \ \sim \ {\bf skip}  \quad\quad{\rm(IF\text{-}L)}\\
	&\left\{(=_{sym}) = \bigotimes_{i=1}^n(I_{a_i} \otimes I_{b_i}) \otimes \frac{1}{2}(I_{\bar{p}\bar{p}^\prime}+S_{\bar{p};\bar{p}^\prime})\right\} \\
	&\bullet\ \ {\bf Tr}[a_1];\cdots; {\bf Tr}[a_n];{\bf Tr}[b_1]\cdots; {\bf Tr}[b_n]\ \sim\ {\bf skip} \quad\quad{\rm(SO\text{-}L)}\\
	&\left\{(=_{sym}) = \frac{1}{2}(I_{\bar{p}\bar{p}^\prime}+S_{\bar{p};\bar{p}^\prime})\right\}
	\end{align*}
	\caption{Verification of correctness for general QOTP with $n$-qubit quantum data. The proof outline for Judgment (\ref{QOTP-n-Correct}).}
	\label{fig-QOTP-n-Correct}
\end{figure}

\begin{figure}
	\footnotesize
	\begin{align*}
	& \Bigg\{\frac{I_{\bar{p}}\otimes I_{{\bar{p}}^\prime}}{2^n} =
	\frac{1}{2^n} \bigg(\sum_{\forall\ i\in[n]: x_i,z_i\in\{0,1\}}\prod_{i}^n (Z_{p_i}^{z_i}X_{p_i}^{x_i})|\psi\>_{\bar{p}}\<\psi|\prod_{i}^n (X_{p_i}^{x_i}Z_{p_i}^{z_i})\bigg)\otimes I_{{\bar{p}}^\prime}\Bigg\} \\
	&\bullet\ \begin{array}{l}
	a_1 := |0\>; \cdots; a_n := |0\>; \\
	b_1 := |0\>; \cdots; b_n := |0\>;
	\end{array} \ \sim \ \begin{array}{l}
	a_1^\prime := |0\>; \cdots; a_n^\prime := |0\>; \\
	b_1^\prime := |0\>; \cdots; b_n^\prime := |0\>;
	\end{array} \quad\quad{\rm(SO)}\\
	& \Bigg\{\sum_{\forall\ i\in[n]: x_i,z_i\in\{0,1\}} \bigotimes_{i=1}^n(|f(x_i)f(z_i)\>_{a_ib_i}\<f(x_i)f(z_i)|\otimes|f(x_i)f(z_i)\>_{a_i^\prime b_i^\prime}\<f(x_i)f(z_i)|) \\
	&\qquad\qquad\qquad\qquad \otimes
	\prod_{i}^n (Z_{p_i}^{z_i}X_{p_i}^{x_i})|\psi\>_{\bar{p}}\<\psi|\prod_{i}^n (X_{p_i}^{x_i}Z_{p_i}^{z_i})\otimes I_{{\bar{p}}^\prime}\Bigg\}\\
	&\bullet\ \begin{array}{l}
	a_1 := H[a_1]; \cdots; a_n := H[a_n]; \\
	b_1 := H[b_1]; \cdots; b_n := H[b_n];
	\end{array} \ \sim \ \begin{array}{l}
	a_1^\prime := H[a_1^\prime]; \cdots; a_n^\prime := H[a_n^\prime]; \\
	b_1^\prime := H[b_1^\prime]; \cdots; b_n^\prime := H[b_n^\prime];
	\end{array}   \quad\quad{\rm(UT)}\\
	& \Bigg\{\sum_{\forall\ i\in[n]: x_i,z_i\in\{0,1\}} \bigotimes_{i=1}^n(|x_iz_i\>_{a_ib_i}\<x_iz_i|\otimes|x_iz_i\>_{a_i^\prime b_i^\prime}\<x_iz_i|)\otimes
	\prod_{i}^n (Z_{p_i}^{z_i}X_{p_i}^{x_i})|\psi\>_{\bar{p}}\<\psi|\prod_{i}^n (X_{p_i}^{x_i}Z_{p_i}^{z_i})\otimes I_{{\bar{p}}^\prime}\Bigg\}\\
	&\bullet\ \ \begin{array}{l}
	\mathbf{if}\ (\square x_1z_1\cdot\cM[a_1,b_1]=x_1z_1\rightarrow\mathbf{skip})\ \mathbf{fi}; \\
	\vdots \\
	\mathbf{if}\ (\square x_nz_n\cdot\cM[a_n,b_n]=x_nz_n\rightarrow\mathbf{skip})\ \mathbf{fi}; \\
	\end{array}
	\ \sim \ \begin{array}{l}
	\mathbf{if}\ (\square x_1z_1 \cdot\cM[a_1^\prime ,b_1^\prime ]=x_1z_1 \rightarrow\mathbf{skip})\ \mathbf{fi}; \\
	\vdots \\
	\mathbf{if}\ (\square x_nz_n \cdot\cM[a_n^\prime ,b_n^\prime ]=x_nz_n \rightarrow\mathbf{skip})\ \mathbf{fi}; \\
	\end{array}   \quad\quad{\rm(IF)}\\
	& \Bigg\{\sum_{\forall\ i\in[n]: x_i,z_i\in\{0,1\}} \bigotimes_{i=1}^n(|x_iz_i\>_{a_ib_i}\<x_iz_i|\otimes|x_iz_i\>_{a_i^\prime b_i^\prime}\<x_iz_i|)\otimes
	\prod_{i}^n (Z_{p_i}^{z_i}X_{p_i}^{x_i})|\psi\>_{\bar{p}}\<\psi|\prod_{i}^n (X_{p_i}^{x_i}Z_{p_i}^{z_i})\otimes I_{{\bar{p}}^\prime}\Bigg\}\\
	&\bullet\ \ \begin{array}{l}
	\mathbf{if}\ (\square x_1z_1\cdot\cM[a_1,b_1]=x_1z_1 \\
	\qquad\rightarrow p_1 = Z^{z_1}[p_1];\ p_1 = X^{x_1}[p_1])\ \mathbf{fi}; \\
	\vdots \\
	\mathbf{if}\ (\square x_nz_n\cdot\cM[a_n,b_n]=x_nz_n\\
	\qquad\rightarrow p_n = Z^{z_n}[p_n];\ p_n = X^{x_n}[p_n])\ \mathbf{fi}; \\
	\end{array}
	\ \sim \ \begin{array}{l}
	\mathbf{if}\ (\square x_1z_1 \cdot\cM[a_1^\prime ,b_1^\prime ]=x_1z_1  \\
	\qquad\rightarrow p_1^\prime  = Z^{z_1}[p_1^\prime];\ p_1^\prime  = X^{x_1}[p_1^\prime])\ \mathbf{fi}; \\
	\vdots \\
	\mathbf{if}\ (\square x_nz_n \cdot\cM[a_n^\prime ,b_n^\prime ]=x_nz_n \\
	\qquad\rightarrow p_n^\prime  = Z^{z_n}[p_n^\prime];\ p_n^\prime  = X^{x_n}[p_n^\prime ])\ \mathbf{fi}; \\
	\end{array}   \quad\quad{\rm(IF)}\\
	&\Bigg\{\bigotimes_{i=1}^n(I_{a_ib_i} \otimes I_{a_i^\prime b_i^\prime}) \otimes |\psi\>_{\bar{p}}\<\psi|\otimes I_{{\bar{p}}^\prime}\Bigg\} \\
	&\bullet\ \ {\bf Tr}[a_1];\cdots; {\bf Tr}[a_n];{\bf Tr}[b_1]\cdots; {\bf Tr}[b_n]\ \sim\ {\bf Tr}[a_1^\prime ];\cdots; {\bf Tr}[a_n^\prime ];{\bf Tr}[b_1^\prime ]\cdots; {\bf Tr}[b_n^\prime ] \quad\quad{\rm(SO)}\\
	&\left\{|\psi\>_{\bar{p}}\<\psi|\otimes I_{{\bar{p}}^\prime}\right\}
	\end{align*}
	\caption{Verification of the security for general QOTP with $n$-qubit quantum data. The proof outline for Judgment (\ref{QOTP-n-Secure}). The function $f: \{0,1\}\mapsto\{+,-\}$ is defined by $f(0)=+,f(1)=-$; that is, $|f(0)\> = |+\>=\frac{1}{\sqrt{2}}(|0\>+|1\>)$ and $|f(1)\> = |-\>=\frac{1}{\sqrt{2}}(|0\>-|1\>)$.}
	\label{fig-QOTP-n-Secure}
\end{figure}

\subsection{Proof of Proposition \ref{prop-comp}}
\label{proof-prop-comp}
\begin{proof}
	(1) If $\models P_1\sim P_2: A\Rightarrow B$, then for any input $\rho\in A$, there exists a coupling $\sigma$ for $\cp{ \sm{P_1} (\tr_2(\rho))}{\sm{P_2}(\tr_2(\rho)) }$ such that
	$$\tr(\rho) = \tr(A\rho) \le \tr(B\sigma) + \tr(\rho) - \tr(\sigma),$$
	which implies $\tr(B\sigma) = \tr(\sigma)$, or equivalently, $\sigma\in B$. Therefore, $\models_P P_1\sim P_2: A\Rightarrow B$.
	
	(2) We show a counterexample here. Let us consider a qubit $q$ and two programs $P_1$ and $P_2$
	$$P_1\equiv q = X[q];\quad\quad P_2\equiv\mathbf{skip};$$
	and choose projective predicates $A = B = |\Psi\>\<\Psi|$ where $|\Psi\>$ is the maximally entangled state $\frac{1}{\sqrt{2}}(|00\>+|11\>)$.
	
	If $\rho\in A$, then $\rho = \lambda|\Psi\>$ for some real number $0\le\lambda\le1$. Consequently, $$\tr_1(\rho) = \tr_2(\rho) = \frac{\lambda}{2}I$$ and $\sm{P_1} (\tr_2(\rho)) = \frac{\lambda}{2}I$. Therefore, $\rho$ is still a coupling of $\cp{ \sm{P_1}}{ (\tr_2(\rho)),\sm{P_2}(\tr_2(\rho)) }$, which shows that $$\models_P P_1\sim P_2: A\Rightarrow B.$$
	
	On the other hand, if we choose a separable input $\rho = |00\>\<00|$, then the output of two programs are $\sm{P_1} (\tr_2(\rho))=|1\>\<1|$ and $\sm{P_2}(\tr_2(\rho))  =  |0\>\<0|$.  They have a unique coupling $\sigma = |10\>\<10|$. However,
	$$\tr(A\rho) = \frac{1}{2} \ge \tr(B\sigma) = 0,$$
	which rules out that $\models P_1\sim P_2: A\Rightarrow B$.
\end{proof}

\subsection{Proof of Proposition \ref{sound-proj}}

\begin{proof}
	The validity of axioms (Skip-P), (UT-P) and rules (SC-P), (Conseq-P) and (Equiv) are trivial. We only show the validity of rules in Figure \ref{fig pro_1} together with (Frame-P) here.
	
	{\vskip 4pt}
	
	$\bullet$ (SO-P) As shown in the proof of Lemma \ref{tech-SO}, for any inputs $\rho_1\in\D^\le(\h_{P_1\<1\>})$ and $\rho_2\in\D^\le(\h_{P_2\<2\>})$ with a witness $\rho$ of the lifting $\rho_1A^\#\rho_2$, then $(\E_1\otimes\E_2)(\rho)$ is a coupling for $$\cp{\E_1(\tr_{\h_2}(\rho))}{\E_2(\tr_{\h_1}(\rho))} = \cp{\E_1(\rho_1)}{\E_2(\rho_1)}.$$
	Moreover, as $\supp(\rho)\subseteq A$, then trivially $$\supp((\E_1\otimes\E_2)(\rho))\subseteq\proj((\E_1\otimes\E_2)(A)),$$ which implies the validity of (SO-P).
	
	{\vskip 4pt}
	
	$\bullet$ (Init-P), (Init-P-L) and (SO-P-L) Special cases of (SO-P).
	
	{\vskip 4pt}
	
	$\bullet$ (IF-P) From the first assumptions we know that, for any $\rho_1\in\D^\le(\h_{P_1\<1\>})$ and $\rho_2\in\D^\le(\h_{P_2\<2\>})$ such that $\rho_1A^\#\rho_2$, there exists a sequence of lifting of $B_m$ relating the post-measurement states with the same outcomes; that is, for all $m$,
	$$(M_{1m}\rho_1M_{1m}^\dag)B_m^\#(M_{2m}\rho_2M_{2m}^\dag).$$
	Together with the second assumption, we must have:
	$$\sm{P_{1m}}(M_{1m}\rho_1M_{1m}^\dag)C^\#\sm{P_{2m}}(M_{2m}\rho_2M_{2m}^\dag).$$
	Due to the linearity of partial trace, we conclude that
	$$\Big[\sum_m\sm{P_{1m}}(M_{1m}\rho_1M_{1m}^\dag)\Big]C^\#\Big[\sum_m\sm{P_{2m}}(M_{2m}\rho_2M_{2m}^\dag)\Big],$$
	or equivalently,
	$$\sm{\mathbf{if}\ (\square m\cdot
		\cM_1[\overline{q}]=m\rightarrow P_{1m})\ \mathbf{fi}}(\rho_1)C^\# \sm{\mathbf{if}\ (\square m\cdot
		\cM_2[\overline{q}]=m\rightarrow P_{2m})\ \mathbf{fi}}(\rho_2).$$
	
	{\vskip 4pt}
	
	$\bullet$ (IF-P-L) Similar to (IF-P).
	
	{\vskip 4pt}
	
	$\bullet$ (LP-P) We first introduce an auxiliary notation: for $i = 1,2$, quantum operation $\mathcal{E}_{i0}$ and $\mathcal{E}_{i1}$ are defined by the measurement $\cM_i$
	$$\mathcal{E}_{i0}(\rho) = M_{i0}\rho M_{i0}^\dag,\quad \mathcal{E}_{i1}(\rho) = M_{i1}\rho M_{i1}^\dag.$$
	For any $\rho_1\in\D^\le(\h_{P_1\<1\>})$ and $\rho_2\in\D^\le(\h_{P_2\<2\>})$ such that $\rho_1A^\#\rho_2$, we claim that for all $n\ge0$, the following statement holds:
	$$
	{\bf statement:}\ (\sm{P_1}\circ\E_{11})^n(\rho_1)A^\#(\sm{P_2}\circ\E_{21})^n(\rho_2),\quad [\E_{10}\circ(\sm{P_1}\circ\E_{11})^n](\rho_1)B_0^\#[\E_{20}\circ(\sm{P_2}\circ\E_{21})^n](\rho_2).
	$$
	We prove it by induction on $n$. For $n=0$, $\rho_1A^\#\rho_2$ is already assumed. The first assumption of measurement ensures that $$(M_{10}\rho_1M_{10}^\dag) B_0^\# (M_{20}\rho_2M_{20}^\dag),$$ or in other words, $\E_{10}(\rho_1) B_0^\# \E_{20}(\rho_2)$.
	Suppose the statement holds for $n = k$. Then for $n=k+1$, because
	$$(\sm{P_1}\circ\E_{11})^k(\rho_1)A^\#(\sm{P_2}\circ\E_{21})^k(\rho_2),$$
	the assumption of measurement implies that
	$$[M_{11}(\sm{P_1}\circ\E_{11})^k(\rho_1)M_{11}^\dag]B_1^\#[M_{21}(\sm{P_2}\circ\E_{21})^k(\rho_2)M_{21}^\dag],$$
	and followed by the second assumption $\models_P P_1\sim P_2:B_1\Rightarrow A$, we have
	$$\sm{P_1}[M_{11}(\sm{P_1}\circ\E_{11})^k(\rho_1)M_{11}^\dag]A^\#\sm{P_2}[M_{21}(\sm{P_2}\circ\E_{21})^k(\rho_2)M_{21}^\dag],$$
	or equivalently
	$$(\sm{P_1}\circ\E_{11})^{k+1}(\rho_1)A^\#(\sm{P_2}\circ\E_{21})^{k+1}(\rho_2).$$
	Applying the assumption of measurement on the above formula, it is trivial that
	$$[\E_{10}\circ(\sm{P_1}\circ\E_{11})^{k+1}](\rho_1)B_0^\#[\E_{20}\circ(\sm{P_2}\circ\E_{21})^{k+1}](\rho_2),$$
	and this complete the proof the statement.
	
	Now, suppose that the witness of the lifting $$[\E_{10}\circ(\sm{P_1}\circ\E_{11})^n](\rho_1)B_0^\#[\E_{20}\circ(\sm{P_2}\circ\E_{21})^n](\rho_2)$$ is $\sigma_n$ for all $n\ge0$. We set $\sigma = \sum_n\sigma_n$ whose convergence is guaranteed by
	\begin{align*}
	1 &\ge \tr\left[\sm{\mathbf{while}\ \cM_1[\overline{q}]=1\ \mathbf{do}\ P_1\ \mathbf{od}}(\rho_1)\right] = \tr\Big(\sum_n[\E_{10}\circ(\sm{P_1}\circ\E_{11})^n](\rho_1)\Big) \\
	&= \sum_n\tr\big([\E_{10}\circ(\sm{P_1}\circ\E_{11})^n](\rho_1)\big) = \sum_n\tr(\sigma_n) \\
	&= \tr\Big(\sum_n\sigma_n\Big).
	\end{align*}
	Then, it is straightforward to show that $\sigma$ is a coupling of
	$$
	\cp{\sm{\mathbf{while}\ \cM_1[\overline{q}]=1\ \mathbf{do}\ P_1\ \mathbf{od}}(\rho_1)}{\ \sm{\mathbf{while}\ \cM_2[\overline{q}]=1\ \mathbf{do}\ P_2\ \mathbf{od}}(\rho_2)}
	$$
	and $\supp(\sigma)\subseteq B_0$.
	
	{\vskip 4pt}
	
	$\bullet$ (LP-P-L) Similar to (LP-P).
	
	{\vskip 4pt}
	
	$\bullet$ (Frame-P)
	Suppose $V = V_1\cup V_2$ where $V_1$ represents the extended variables of $P_1$ and $V_2$ of $P_2$. Of course, $V_1\cap V_2=\emptyset$ and $\mathcal{H}_V = \mathcal{H}_{V_1}\otimes\mathcal{H}_{V_2}$. We prove: $$[V,\mathit{var}(P_1,P_2)]\models_P P_1\sim P_2:A\otimes C\Rightarrow B\otimes C;$$ that is, for any separable state $\rho$ between $\mathcal{H}_{P_1}\otimes\mathcal{H}_{P_2}$ and $\mathcal{H}_{V}$ satisfies $\supp(\rho)\subseteq A\otimes C$,
	$$ \sm{P_1} (\tr_{\mathcal{H}_{P_2}\otimes\mathcal{H}_{V_2}}(\rho)) (B\otimes C)^\# \sm{P_2} (\tr_{\mathcal{H}_{P_1}\otimes\mathcal{H}_{V_1}}(\rho)).$$
	
	First of all, by separability of $\rho$, we can write:
	$$\rho = \sum_i p_i (\rho_i\otimes \sigma_i)$$
	where $\rho_i\in\mathcal{D}^\le(\mathcal{H}_{P_1}\otimes\mathcal{H}_{P_2})$, $\sigma_i\in\mathcal{D}^\le(\mathcal{H}_{V})$ and $p_i>0$. Since $V\cap \mathrm{var}(P_1,P_2)=\emptyset$, it holds that
	\begin{align*}
	&\sm{P_1} (\tr_{\mathcal{H}_{P_2}\otimes\mathcal{H}_{V_2}}(\rho)) = \sum_i p_i[\sm{P_1} (\tr_{\mathcal{H}_{P_2}}(\rho_i)) \otimes \tr_{\mathcal{H}_{V_2}}(\sigma_i) ], \\
	&\sm{P_2} (\tr_{\mathcal{H}_{P_1}\otimes\mathcal{H}_{V_1}}(\rho)) = \sum_i p_i[\sm{P_2} (\tr_{\mathcal{H}_{P_1}}(\rho_i)) \otimes \tr_{\mathcal{H}_{V_1}}(\sigma_i) ].
	\end{align*}
	As $\supp(\rho)\subseteq A\otimes C$, so for all $i$, $\supp(\rho_i)\subseteq A$ and $\supp(\sigma_i)\subseteq C$. For each $i$, since $\models_P P_1\sim P_2: A\Rightarrow B$, then
	$$\sm{P_1} (\tr_{\mathcal{H}_{P_2}}(\rho_i)) B^\# \sm{P_2} (\tr_{\mathcal{H}_{P_1}}(\rho_i)), $$
	and we assume $\rho_i^\prime$ is a witness.
	
	We set:
	$$\rho^\prime = \sum_ip_i(\rho_i^\prime\otimes\sigma_i).$$
	Then we can check that
	\begin{align*}
	\tr_{\langle 2\rangle}(\rho^\prime) &= \tr_{\mathcal{H}_{P_2}\otimes\mathcal{H}_{V_2}}(\rho^\prime)
	= \sum_ip_i[\tr_{\mathcal{H}_{P_2}}(\rho_i^\prime) \otimes \tr_{\mathcal{H}_{V_2}}(\sigma_i)] \\
	&= \sum_ip_i[\sm{P_1} (\tr_{\mathcal{H}_{P_2}}(\rho_i)) \otimes \tr_{\mathcal{H}_{V_2}}(\sigma_i)]
	= \sm{P_1} (\tr_{\mathcal{H}_{P_2}\otimes\mathcal{H}_{V_2}}(\rho))
	\end{align*}
	and $\tr_{\langle 1\rangle}(\rho^\prime) = \sm{P_2} (\tr_{\mathcal{H}_{P_1}\otimes\mathcal{H}_{V_1}}(\rho))$. Therefore, $\rho^\prime$ is a coupling for
	$$\cp{ \sm{P_1} (\tr_{\mathcal{H}_{P_2}\otimes\mathcal{H}_{V_2}}(\rho))}{ \sm{P_2} (\tr_{\mathcal{H}_{P_1}\otimes\mathcal{H}_{V_1}}(\rho)) }.$$ Furthermore, as $\supp(\rho_i^\prime)\subseteq B$ and $\supp(\sigma_i)\subseteq C$, we conclude that $\supp(\rho^\prime)\subseteq B\otimes C$.
\end{proof}

\subsection{Verification of Judgment (\ref{two-qw})}\label{ver-q-walker}

In this subsection, we verify the equivalence of two quantum random walkers with different coin-tossing operators. 
The formal proof is displayed in Figure \ref{qRHL-QW} together with two extra notations:
\begin{align*}
=_{sym}^B &= \frac{1}{2}\bigg(\sum_{i,i^\prime=0,n}\sum_{d,d^\prime = 0,1}|d,i\rangle_{c,p}\langle d, i|\otimes |d^\prime,i^\prime\rangle_{c^\prime,p^\prime}\langle d^\prime, i^\prime|\\ &\qquad + \sum_{i,i^\prime=0,n}\sum_{d,d^\prime = 0,1}|d,i\rangle_{c,p}\langle d^\prime,i^\prime|\otimes |d^\prime,i^\prime\rangle_{c^\prime,p^\prime}\langle d, i|\bigg) \\
=_{sym}^I &= \frac{1}{2}\bigg(\sum_{i,i^\prime=1}^{n-1}\sum_{d,d^\prime = 0}^1|d,i\rangle_{c,p}\langle d, i|\otimes |d^\prime,i^\prime\rangle_{c^\prime,p^\prime}\langle d^\prime, i^\prime| + \sum_{i,i^\prime=1}^{n-1}\sum_{d,d^\prime = 0}^1|d,i\rangle_{c,p}\langle d^\prime,i^\prime|\otimes |d^\prime,i^\prime\rangle_{c^\prime,p^\prime}\langle d, i|\bigg)
\end{align*}

\begin{figure}
	\centering
	\begin{align*}
	& \left\{A = U_{c^\prime,p^\prime}(=_{sym})U^\dag_{c^\prime,p^\prime}\right\}\\
	&\bullet \mathbf{while}\ \cM[p]=1\ \mathbf{do}\ \sim\ \mathbf{while}\ \cM[p^\prime]=1\ \mathbf{do}\quad\quad {\rm(LP\text{-}P)}\\
	&  \quad\quad\quad \left\{B_1 = U_{c^\prime,p^\prime}(=_{sym}^I)U^\dag_{c^\prime,p^\prime}\right\}\\
	&  \quad\quad\quad \bullet c:=H[c]; \sim\  c^\prime:=Y[c^\prime]; \quad\quad{\rm(UT\text{-}P)}\\
	&  \quad\quad\quad \left\{(H_c\otimes Y_{c^\prime})U_{c^\prime,p^\prime}(=_{sym}^I)U^\dag_{c^\prime,p^\prime}(H_c^\dag\otimes Y_{c^\prime}^\dag)\right\}\\
	&  \quad\quad\quad \bullet c,p:=S[c,p]; \sim\  c^\prime,p^\prime:=S[c^\prime,p^\prime]; \quad\quad{\rm(UT\text{-}P)}\\
	&   \quad\quad\quad \left\{C = (S_{c,p}H_c\otimes S_{c^\prime,p^\prime}Y_{c^\prime})U_{c^\prime,p^\prime}(=_{sym}^I)U^\dag_{c^\prime,p^\prime}(H_c^\dag S_{c,p}^\dag\otimes Y_{c^\prime}^\dag S_{c^\prime,p^\prime}^\dag)\right\}\\
	&\bullet \mathbf{od};\sim\ \mathbf{od};  \\
	&\left\{B_0 = U_{c^\prime,p^\prime}(=_{sym}^B)U^\dag_{c^\prime,p^\prime}\right\}\\
	&\bullet\quad \mathbf{if}\ (\square i\cdot\cM^\prime[p]=i\rightarrow \mathbf{skip})\ \mathbf{fi}; \sim \mathbf{if}\ (\square i\cdot\cM^\prime[p^\prime]=i\rightarrow \mathbf{skip})\ \mathbf{fi};  \quad\quad{\rm(IF\text{-}P)}\\
	&\left\{D = I_c\otimes I_{c^\prime}\otimes(=_{sym}^p)\right\} \\
	&\bullet\quad \mathbf{Tr}[c];\ \sim\ \mathbf{Tr}[c^\prime]; \quad\quad{\rm(SO\text{-}P)}\\
	&\left\{E = \proj(\tr_{c,c^\prime}(D))\right\}
	\end{align*}
	\caption{Verification of $QW(H)\sim QW(Y)$ in rqPD.
	}\label{qRHL-QW}
\end{figure}

The use of (SO-P) for the last line of the program is obvious:
\begin{align*}
E &= \proj(\tr_{c,c^\prime}(D)) = \proj(\tr_{c,c^\prime}(I_c\otimes I_{c^\prime}\otimes(=_{sym}^p))) = \proj(=_{sym}^p) = (=_{sym}^p).
\end{align*}

For the ${\bf if}$ sentence, we first show that the following assertion holds:
\begin{equation}
\label{qw-if}
\models_P \cM^\prime\approx\cM^\prime: B_0\Rightarrow\{F_i\}
\end{equation}
where $F_1=F_2=\cdots=F_{n-1}=0$ and
$$
F_0 = I_c\otimes I_{c^\prime}\otimes |0\>_p\<0|\otimes |0\>_{p^\prime}\<0|,\quad
F_n = I_c\otimes I_{c^\prime}\otimes |n\>_p\<n|\otimes |n\>_{p^\prime}\<n|.
$$
It is not difficult to realize that, for any $$\rho\in B_0 = U_{c^\prime,p^\prime}(=_{sym}^B)U^\dag_{c^\prime,p^\prime},$$ it must have following properties:
\begin{align*}
&\rho_1 = \tr_2(\rho) = U^\dag\tr_1(\rho)U = U^\dag\rho_2U; \\
&\forall\ i = 1,2,\cdots,n-1: \langle i|\rho_1|i\rangle = \langle i|\rho_2|i\rangle = 0.
\end{align*}
The second property ensures that for $i = 1,2,\cdots,n-1$, $$M_i^\prime\rho_1M_i^{\prime\dag} = M_i^\prime\rho_2M_i^{\prime\dag} = 0,$$ or equivalently to say $$(M_i^\prime\rho_1M_i^{\prime\dag}) F_i^\# (M_i^\prime\rho_2M_i^{\prime\dag})$$ as $F_i=0$. Moreover, for $i=0, n$, we note that
\begin{align*}
&M_i^\prime\rho_1M_i^{\prime\dag} = \rho_{1,c,i}\otimes|i\>\<i|, \quad M_i^\prime\rho_2M_i^{\prime\dag} = \rho_{2,c,i}\otimes|i\>\<i|; \\
&\tr(M_i^\prime\rho_1M_i^{\prime\dag}) = \tr(\langle i|\rho_1|i\rangle) = \tr(\langle i|U^\dag\rho_2U|i\rangle)= \tr(U^\dag\langle i|\rho_2|i\rangle U)= \tr(\langle i|\rho_2|i\rangle) = \tr(M_i^\prime\rho_2M_i^{\prime\dag})
\end{align*}
where $\rho_{1,c,i}, \rho_{2,c,i}\in\mathcal{D}^\le(\mathcal{H}_c)$, which implies
$$(M_i^\prime\rho_1M_i^{\prime\dag}) (I_c\otimes I_{c^\prime}\otimes |i\>_p\<i|\otimes |i\>_{p^\prime}\<i|)^\# (M_i^\prime\rho_2M_i^{\prime\dag}),$$ or equivalently $$(M_i^\prime\rho_1M_i^{\prime\dag}) F_i^\# (M_i^\prime\rho_2M_i^{\prime\dag}).$$ Therefore, assertion (\ref{qw-if}) holds.

As $F_i\sqsubseteq D$ for all $i$, using rules (Skip-P) and (Conseq-P), we have:
\begin{equation}
\label{qw-if1}
\vdash_P{\bf skip}\sim{\bf skip}:F_i\Rightarrow D.
\end{equation}
Note that all subprograms of $\bf{if}$ statement are ${\bf skip}$, and therefore with the precondition $B_0$, the precondition $D$ of the $\bf{if}$ statement is valid using rule (IF-P) directly according to assertions (\ref{qw-if}) and (\ref{qw-if1}).

Now, we focus on the proof of $C\sqsubseteq A$. For simplicity, we use bold ${\bm i}$ to denote the imaginary unit, that is $\bm i = (-1)^{\frac{1}{2}}$. First of all, the following equations are trivial:
\begin{align*}
S_{c,p}H_c|d,i\rangle &= \frac{1}{\sqrt{2}}((-)^{d+1}|0,i-1\rangle+|1,i+1\rangle),\\
S_{c,p}Y_cU_{c,p}|d,i\rangle &= \frac{1}{\sqrt{2}} {\bm i}^i   ({\bm i}(-)^{d+1}|0,i-1\rangle+|1,i+1\rangle)
\end{align*}
Secondly, we can calculate $C$ directly from the above equations: %\sum_{\substack{d,d^\prime = 0,1\\i,i^\prime=1,\cdots,n-1}}
\begingroup
\allowdisplaybreaks
\begin{align*}
C &= (S_{c,p}H_c\otimes S_{c^\prime,p^\prime}Y_{c^\prime})U_{c^\prime,p^\prime}(=_{sym}^I)U^\dag_{c^\prime,p^\prime}(H_c^\dag S_{c,p}^\dag\otimes Y_{c^\prime}^\dag S_{c^\prime,p^\prime}^\dag) \\
&= \frac{1}{2}\sum_{i,i^\prime=1}^{n-1}\sum_{d,d^\prime = 0}^1(S_{c,p}H_c\otimes S_{c^\prime,p^\prime}Y_{c^\prime})U_{c^\prime,p^\prime}|d,i\rangle_{c,p}\langle d, i|\otimes |d^\prime,i^\prime\rangle_{c^\prime,p^\prime}\langle d^\prime, i^\prime|U^\dag_{c^\prime,p^\prime}(H_c^\dag S_{c,p}^\dag\otimes Y_{c^\prime}^\dag S_{c^\prime,p^\prime}^\dag) \\
&\quad + \frac{1}{2}\sum_{i,i^\prime=1}^{n-1}\sum_{d,d^\prime = 0}^1(S_{c,p}H_c\otimes S_{c^\prime,p^\prime}Y_{c^\prime})U_{c^\prime,p^\prime}|d,i\rangle_{c,p}\langle d^\prime,i^\prime|\otimes |d^\prime,i^\prime\rangle_{c^\prime,p^\prime}\langle d, i|U^\dag_{c^\prime,p^\prime}(H_c^\dag S_{c,p}^\dag\otimes Y_{c^\prime}^\dag S_{c^\prime,p^\prime}^\dag) \\
&= \frac{1}{8}\sum_{i,i^\prime=1}^{n-1}\sum_{d,d^\prime = 0}^1
\Big[
((-)^{d+1}|0,i-1\rangle_{c,p}+|1,i+1\rangle_{c,p})
{\bm i}^{i^\prime}({\bm i}(-)^{d^\prime+1}|0,i^\prime-1\rangle_{c^\prime,p^\prime}+|1,i^\prime+1\rangle_{c^\prime,p^\prime})\cdot\\
&\qquad\qquad\qquad((-)^{d+1}{}_{c,p}\langle0,i-1|+{}_{c,p}\langle1,i+1|)
(-{\bm i})^{i^\prime}
((-{\bm i})(-)^{d^\prime+1}{}_{c^\prime,p^\prime}\langle0,i^\prime-1|+{}_{c^\prime,p^\prime}\langle1,i^\prime+1|)\Big]\\
&\quad +\frac{1}{8}\sum_{i,i^\prime=1}^{n-1}\sum_{d,d^\prime = 0}^1
\Big[
((-)^{d+1}|0,i-1\rangle_{c,p}+|1,i+1\rangle_{c,p}){\bm i}^{i^\prime}
({\bm i}(-)^{d^\prime+1}|0,i^\prime-1\rangle_{c^\prime,p^\prime}+|1,i^\prime+1\rangle_{c^\prime,p^\prime}) \cdot\\
&\qquad\qquad\qquad((-)^{d^\prime+1}{}_{c,p}\langle0,i^\prime-1|+{}_{c,p}\langle1,i^\prime+1|)
(-{\bm i})^{i}
((-{\bm i})(-)^{d+1}{}_{c^\prime,p^\prime}\langle0,i-1|+{}_{c^\prime,p^\prime}\langle1,i+1|)\Big]\\
&= \frac{1}{2}\sum_{i,i^\prime=1}^{n-1}\Big[
|0,i-1\rangle_{c,p}\langle0,i-1|\otimes |0,i^\prime-1\rangle_{c^\prime,p^\prime}\langle0,i^\prime-1| + |1,i+1\rangle_{c,p}\langle1,i+1|\otimes |1,i^\prime+1\rangle_{c^\prime,p^\prime}\langle1,i^\prime+1| \\
&\quad + |0,i-1\rangle_{c,p}\langle0,i-1|\otimes |1,i^\prime+1\rangle_{c^\prime,p^\prime}\langle1,i^\prime+1| + |1,i+1\rangle_{c,p}\langle1,i+1|\otimes |0,i^\prime-1\rangle_{c^\prime,p^\prime}\langle0,i^\prime-1| \\
&\quad +  {\bm i}^{i^\prime}(-{\bm i})^{i}|0,i-1\rangle_{c,p}\langle0,i^\prime-1|\otimes |0,i^\prime-1\rangle_{c^\prime,p^\prime}\langle0,i-1|\\ &\quad + {\bm i}^{i^\prime}(-{\bm i})^{i}|1,i+1\rangle_{c,p}\langle1,i^\prime+1|\otimes |1,i^\prime+1\rangle_{c^\prime,p^\prime}\langle1,i+1| \\
&\quad +  {\bm i}^{i^\prime}(-{\bm i})^{i}({-\bm i})|0,i-1\rangle_{c,p}\langle1,i^\prime+1|\otimes |1,i^\prime+1\rangle_{c^\prime,p^\prime}\langle0,i-1|\\ &\quad + {\bm i}^{i^\prime}(-{\bm i})^{i}({\bm i})|1,i+1\rangle_{c,p}\langle0,i^\prime-1|\otimes |0,i^\prime-1\rangle_{c^\prime,p^\prime}\langle1,i+1|
\Big] \\
\end{align*}
\endgroup
Next, we show that $C$ is invariant under the premultiplication of $U_{c^\prime,p^\prime}S_{c,p;c^\prime,p^\prime}U^\dag_{c^\prime,p^\prime}$:
\begin{align*}
& U_{c^\prime,p^\prime}S_{c,p;c^\prime,p^\prime}U^\dag_{c^\prime,p^\prime}C \\
=\ &\frac{1}{2}\sum_{i,i^\prime=1}^{n-1}\Big[
(-{\bm i})^{i^\prime}{\bm i}^i     |0,i^\prime-1\rangle_{c,p}\langle0,i-1|\otimes |0,i-1\rangle_{c^\prime,p^\prime}\langle0,i^\prime-1|\\ &\quad + (-{\bm i})^{i^\prime}{\bm i}^i |1,i^\prime+1\rangle_{c,p}\langle1,i+1|\otimes |1,i+1\rangle_{c^\prime,p^\prime}\langle1,i^\prime+1| \\
&  +(-{\bm i})^{i^\prime}{\bm i}^{i}({\bm i})|1,i^\prime+1\rangle_{c,p}\langle0,i-1|\otimes |0,i-1\rangle_{c^\prime,p^\prime}\langle1,i^\prime+1| \\ &\quad +  (-{\bm i})^{i^\prime}{\bm i}^{i}(-{\bm i})|0,i^\prime-1\rangle_{c,p}\langle1,i+1|\otimes |1,i+1\rangle_{c^\prime,p^\prime}\langle0,i^\prime-1| \\
&  +|0,i-1\rangle_{c,p}\langle0,i^\prime-1|\otimes |0,i^\prime-1\rangle_{c^\prime,p^\prime}\langle0,i-1| +   |1,i+1\rangle_{c,p}\langle1,i^\prime+1|\otimes |1,i^\prime+1\rangle_{c^\prime,p^\prime}\langle1,i+1| \\
&  +|0,i-1\rangle_{c,p}\langle1,i^\prime+1|\otimes |1,i^\prime+1\rangle_{c^\prime,p^\prime}\langle0,i-1| +  |1,i+1\rangle_{c,p}\langle0,i^\prime-1|\otimes |0,i^\prime-1\rangle_{c^\prime,p^\prime}\langle1,i+1|
\Big] \\
=\ &C
\end{align*}
Finally, let calculate $AC$:
\begin{align*}
AC &= \frac{1}{2}U_{c^\prime,p^\prime}(I_{c,p}\otimes I_{c^\prime,p^\prime}+ S_{c,p;c^\prime,p^\prime})U^\dag_{c^\prime,p^\prime}C = \frac{1}{2}C + \frac{1}{2}U_{c^\prime,p^\prime}S_{c,p;c^\prime,p^\prime}U^\dag_{c^\prime,p^\prime}C = \frac{1}{2}C + \frac{1}{2}C = C
\end{align*}
which implies that $C\sqsubseteq A$, as both $A$  and $C$ are projectors.

What remains to be shown is:
\begin{equation}
\label{exam-QW-meas}
\models_P \mathcal{M}\approx\mathcal{M}: A\Rightarrow\{B_0,B_1\}.
\end{equation}
Actually, this one is not difficult to realize as measurement $\mathcal{M}$ and $U$ are commutative, that is, $M_iU=UM_i$ for $i=0,1$. Therefore, for any $\rho \in A = U_{c^\prime,p^\prime}(=_{sym})U^\dag_{c^\prime,p^\prime}$,
\begin{align*}
& \rho_1 = \tr_2(\rho) = U^\dag\tr_1(\rho)U = U^\dag\rho_2U; \\
& M_i\rho_1M_i^\dag = M_iU^\dag\rho_2UM_i^\dag = U^\dag M_i\rho_2M_i^\dag U.
\end{align*}
Moreover, realize that $$\supp(M_0\rho_1M_0^\dag)\subseteq(I_c\otimes (|0\>_p\<0|+|n\>_p\<n|)),$$ so $$M_0\rho_1M_0^\dag(=_{sym}^B)^\#M_0\rho_1M_0^\dag,$$ which implies that $$M_0\rho_1M_0^\dag [U_{c^\prime,p^\prime}(=_{sym}^B)U^\dag_{c^\prime,p^\prime}]^\# M_0\rho_2M_0^\dag.$$ Similarly, we have: $$M_1\rho_1M_1^\dag [U_{c^\prime,p^\prime}(=_{sym}^I)U^\dag_{c^\prime,p^\prime}]^\# M_1\rho_2M_1^\dag.$$

Combining the above fact $C\sqsubseteq A$ and Eqn. (\ref{exam-QW-meas}), we conclude that with precondition $A$, the postcondition $B_0$ of the while loop is valid.

\subsection{Proof of Proposition \ref{prop-sep}}

Let us first present a technical lemma:

\begin{lem}
	\label{max_ent_pre}
	Suppose $\mathcal{H}_1$ and $\mathcal{H}_2$ are both two dimensional Hilbert spaces and $|\Phi\>$ is a maximally entangled state between $\mathcal{H}_1$ and $\mathcal{H}_2$. Then for any separable state $\rho$ between $\mathcal{H}_1$ and $\mathcal{H}_2$,
	$$
	\tr(|\Psi\>\<\Psi|\rho)\le\frac{1}{2}\tr(\rho).
	$$
\end{lem}
\begin{proof}
	Suppose $|\Phi\> = \frac{1}{\sqrt{2}}(|00\>+|11\>)$ and write $\rho$ and $\rho^{T_2}$ in the matrix forms:
	$$\rho = \left[
	\begin{array}{cccc}
	x & \cdot & \cdot & v \\
	\cdot & y & \cdot & \cdot \\
	\cdot & \cdot & z & \cdot \\
	v^* & \cdot & \cdot & w \\
	\end{array}\right],\
	\rho^{T_2} = \left[
	\begin{array}{cccc}
	x & \cdot & \cdot & \cdot \\
	\cdot & y & v & \cdot \\
	\cdot & v^* & z & \cdot \\
	\cdot & \cdot & \cdot & w \\
	\end{array}
	\right]$$
	with non-negative real numbers $x,y,z,w$ and complex number $v$ while `$\cdot$' represents the parameters we are not interested in. As $\rho$ is separable, so $\rho^{T_2}$ is non-negative which implies $yz\ge vv^*$. Therefore, we obtain:
	\begin{align*}
	\tr(|\Psi\>\<\Psi|\rho) &= \frac{1}{2}(x+w+v+v^*) \le \frac{1}{2}(x+w+2|v|) \\
	&\le \frac{1}{2}(x+w+2\sqrt{yz}) \le \frac{1}{2}(x+w+y+z) \\
	&= \frac{1}{2}\tr(\rho).
	\end{align*}
\end{proof}

Now we are ready to prove Proposition \ref{prop-sep}.

\begin{proof}
	(i) It is sufficient to use the same counterexample shown in the proof of Proposition \ref{prop-comp} (2) [see Appendix \ref{proof-prop-comp}] to prove this. Actually, that example not only rules out $\models P_1\sim P_2:A\Rightarrow B$, but also $\models_S P_1\sim P_2:A\Rightarrow B$ (see Definition \ref{def-judgment-sep}).
	
	(ii) Let us still consider two one-qubit programs $$P_1\equiv P_2 \equiv q:=|0\>;$$ and two projective predicate $A = B = |\Psi\>\<\Psi|$. According to Lemma \ref{max_ent_pre} we know that for any separable state $\rho$ between $\mathcal{H}_1$ and $\mathcal{H}_2$, $\tr(A\rho)\le\frac{1}{2}\tr(\rho)$. The output states of two programs are both $\tr(\rho)|0\>\<0|$, which have the unique coupling $\sigma = \tr(\rho)|00\>\<00|$, therefore,
	$$\tr(B\sigma) = \frac{1}{2}\tr(\rho) \ge \tr(A\rho)$$
	which concludes the validity of $\models_S P_1\sim P_2: A\Rightarrow B$.
	
	On the other hand, choosing $\rho = |\Psi\>\<\Psi|\in A$ as the input, the coupling $\sigma$ of outputs is still unique and $\sigma = |00\>\<00|$, whose support is not in $B$; that is, $\models_p P_1\sim P_2: A\Rightarrow B$ is not valid.
	
	(iii) Obvious.
	
	(iv) Let us consider two one-qubit programs
	$$P_1\equiv P_2\equiv \mathbf{skip};$$ with their two-dimensional Hilbert space $\mathcal{H}_1$ and $\mathcal{H}_2$ respectively, and the separable quantum predicates $$A = \frac{1}{3}I\otimes I + \frac{2}{3}|\Psi\>\<\Psi|$$ and $B = \frac{2}{3}I\otimes I$ over $\mathcal{H}_1\otimes\mathcal{H}_2$ where $|\Psi\>$ is the maximally entangled state $\frac{1}{\sqrt{2}}(|00\>+|11\>)$.
	
	$B$ is trivially separable and the separability of $A$ is already shown in the proof of Fact \ref{separable-state} [see Appendix \ref{entangled-witness}]. Now, it is straightforward to see that for any separable input $\rho$ the following inequality holds:
	\begin{align*}
	\tr(A\rho) &= \frac{1}{3}\tr(\rho) + \frac{2}{3}\tr(|\Psi\>\<\Psi|\rho) \le \frac{2}{3}\tr(\rho) = \tr(B\rho).
	\end{align*}
	according to Fact \ref{max_ent_pre}.
	Note that $P_1$ and $P_2$ are $\mathbf{skip}$, so $\rho$ is still a coupling for $\cp{ \sm{P_1}(\tr_2(\rho))}{  \sm{P_2}(\tr_1(\rho)) }$, which implies the validity of the judgment $\models_S P_1\sim P_2: A\Rightarrow B$.
	
	On the other hand, however, if we choose $|\Psi\>\<\Psi|$ as the input, then
	$$\tr(A|\Psi\>\<\Psi|) = 1 > \frac{2}{3} = \frac{2}{3}\tr(\sigma) = \tr(B\sigma)$$
	for any possible couplings $\sigma$ for the output states, which rules out the validity of judgment $\models P_1\sim P_2: A\Rightarrow B$.
\end{proof}

\section{A Strategy for Collecting Measurement and Separability Conditions}\label{strategy}

In this section, we propose a strategy for collecting measurement and separability conditions in order to warrant the side-conditions of the form (\ref{entail-form}) in a sequence of applications of the strong sequential rule (SC+). Separability conditions are relatively easier to deal with. So, we focus on measurement conditions.   

We first observe that for each quantum program containing \textbf{if} and \textbf{while}, it naturally includes branches in execution.
We can generate a (potentially) infinite probabilistic branching tree for any input state as follows: 
\begin{itemize}\item 
	$({\rm IF})\ \langle\mathbf{if}\ (\square m\cdot
	\cM[\overline{q}]=m\rightarrow P_m)$ is executed with input $\rho$. The current node has $M$ children nodes, where $M$ is the number of possible outcomes $m$ of the measurement $\mathcal{M}$, such that the edge connected to the $m$-th child has probability weight $\tr(M_m^{\dag}M_m\rho)$.
	\item 
	$\mathbf{while}\ \cM[\overline{q}]=1\ \mathbf{do}\ P\ \mathbf{od}$ is executed with input $\rho$. The current node has two children nodes such that the edge connected to the $i$-th child has probability weight $\tr(M_i^{\dag}M_i\rho)$ $(i=0,1)$. The second node then will have two children nodes such that the edge connected to the $i$-th child has probability weight $\tr(M_i^{\dag}M_iP(M_1\rho M_1^{\dag}))$. $\cdots$. Repeating this process, each \textbf{while} loop would generate an infinite  probabilistic branching tree.
\end{itemize}

Now we can define comparability between two quantum programs.
\begin{defn} Let $P_1$ and $P_2$ be two quantum programs and $\rho_1$ and $\rho_2$ two input states. We say that $(P_1,\rho_1)$ and $(P_2, \rho_2)$ are comparable, if the probabilistic branching tree generated from $P_1$ executing on $\rho_1$ is the same as the probabilistic branching tree generated from $P_2$ executing on $\rho_2$.
\end{defn}

The measurement conditions used in verifying relational properties between two programs $P_1$ and $P_2$ are actually the constraints on the respective inputs $\rho_1, \rho_2$ under which $(P_1,\rho_1)$ and $(P_2,\rho_2)$ are comparable. 

Although the probabilistic branching trees of $P_1$ and $P_2$ are generally infinite, we have a back tracking procedure to find sufficient constrains for comparability of them in the case of finite-dimensional state Hilbert spaces. 
\begin{lemma}
	For any two programs $P_1$ and $P_2$ in finite-dimensional Hilbert spaces, we can compute the constrains on inputs under which they become comparable within a finite number of steps.
\end{lemma}
\begin{proof} We show find linear constrains by back tracking through the structure of programs $P_,P_2$: 
	\begin{itemize}\item At the output points, the constrains are $\{(I,I)\}$, which means that the output probability should be the same.
		
		\item Now suppose that at some point, the constrains are
		$\{(A_1,B_1),\cdots,(A_k,B_k)\}$, which means that the respective inputs $\rho_1,\rho_2$ at this point should satisfy $\tr(A_i\rho_1)=\tr(B_i\rho_2)$ to make the following subprogram generate same probabilistic branching trees. Such constrains are also considered as the requirements of the outputs of the front subprograms.
		
		Our aim is to generate the constrains for the point before the current one. We only illustrate the following two cases:
		\begin{itemize}\item The subprograms of $P_1$ and $P_2$ associated with this point are: 
			\begin{align*} &({\rm IF_1})\qquad \langle\mathbf{if}\ (\square m\cdot
			\cM[\overline{q}]=m\rightarrow P_m),\\ &({\rm IF_2})\qquad \langle\mathbf{if}\ (\square m\cdot
			\cN[\overline{q}]=m\rightarrow Q_m).\end{align*}
			Then the following equations should be added into the constraints: 
			\begin{itemize}
				\item $\tr(M_i^{\dag}M_i\rho_1)=\tr(N_i^{\dag}N_i\rho_2)$ for $1\leq i\leq m$.
				\item $\tr(\sum_{i=1}^m P_i(M_i\rho_1M_i^{\dag})A_j)=(\sum_{i=1}^m P_i(M_i\rho_1M_i^{\dag})B_j)$ for all $1\leq j\leq k$.
			\end{itemize}
			It is easy to derive new constrains for the input at the points of ${\rm IF_1}$ and ${\rm IF_2}$.
			
			\item The subprograms of $P_1$ and $P_2$ associated with this point are: 
			\begin{align*}&({\rm LP_1})\qquad \mathbf{while}\ \cM[\overline{q}]=1\ \mathbf{do}\ P\ \mathbf{od},\\ &({\rm LP_2})\qquad \mathbf{while}\ \cN[\overline{q}]=1\ \mathbf{do}\ Q\ \mathbf{od}.\end{align*} Assume that the semantic functions of $P$ and $Q$ are super-operators $\E$ and $\F$, respectively. We set $\E_a=\E\circ M_1\cdot M_1^{\dag}$ and $\F_a=\F\circ N_1\cdot N_1^{\dag}$
			Then the following equations should be added into the constrains: for each $1\leq j$,
			\begin{itemize}
				\item $\tr(M_0\E_a^j(\rho_1)M_0^{\dag})=\tr(N_0\F_a^j\rho_2 N_0^{\dag})$ for $1\leq i\leq m$.
				\item $\tr(M_0\E_a^j(\rho_1)M_0^{\dag}A_i)=(N_0\F_a^j\rho_2 N_0^{\dag}B_j)$ for all $1\leq i\leq k$.
			\end{itemize}
			Note that the second requirement is very strong in the sense that rather than only requiring the outputs of the loops satisfies the latter constrains, we want each branch of them satisfies the latter constrains. This is reasonable because different branches of the loops can actually happen in different points of time.\end{itemize}\end{itemize}
	
	The key idea here is that although the above procedure may produce infinitely many constrains, a finite number of them will be enough to ensure that all of them are satisfied. Indeed, we can see that $d_a^2+d_b^2+1$ of these constraints are enough by observing that the equations given above can be regarded as linear constrains of linear recurrent series of degree $d_a^2+d_b^2$. It is well known that a $d$-degree linear recurrent series is $0$ if and only if the front $d+1$ items are all $0$.
\end{proof}

For a better understanding of the procedure given in the proof of the above lemma, it is visualised as Algorithm \ref{algo:Robust Entanglement Distribution}.  

\begin{algorithm}[h]
	
	\Input{Quantum Programs: $P_1=S_{1,1};S_{1,2};\cdots; S_{1,n}$, $P_2=S_{2,1};S_{2,2};\cdots;S_{2,n}$}	
	\tcc{$S_{1,j}$ and $S_{2,j}$ are both $\mathbf{if}$ or $\mathbf{while}$}
	$i\gets n$\;
	$C=\{(I,I)\}$\;
	$l_C=1$\;
	\While {$i>0$}{
		\If{$S_{1,i}=\ \langle\mathbf{if}\ (\square m\cdot\cM[\overline{q}]=m\rightarrow P_m\rangle$ and $S_{2,i}=\ \langle\mathbf{if}\ (\square m\cdot\cN[\overline{q}]=m\rightarrow Q_m\rangle$}
		{
			$D\gets\emptyset$\;
			\For{$1\leq j\leq m$}{
				Add $(M_j^{\dag}M_j,N_j^{\dag}N_j)$ into $D$\;
				\For{$1\leq k\leq l_C$}{
					Add $(M_j^{\dag}P_j^*(A_k)M_j,N_j^{\dag}Q_j^*(B_k)N_j)$ into $D$\;
				}
			}
			$C\gets D$\;
			$l_C\gets|C|$\;
		}
		
		\If{$S_{1,i}=\mathbf{while}\ \cM[\overline{q}]=1\ \mathbf{do}\ P\ \mathbf{od}$ and $\mathbf{while}$ and $S_{2,i}=\mathbf{while}\ \cN[\overline{q}]=1\ \mathbf{do}\ Q\ \mathbf{od}$}{
			$D\gets\emptyset$\;
			$\E_a(\cdot)\gets M_1^{\dag}P^*(\cdot)M_1$\;
			$\F_a(\cdot)\gets N_1^{\dag}Q^*(\cdot)N_1$\;
			\For{$0\leq j\leq d_1^2+d_2^2+1$}{
				Add $(\E_a^j(M_0^{\dag}M_0),\F_a^j(N_0^{\dag}N_0))$ into $D$\;
				\For{$1\leq k\leq l_C$}{
					Add $(\E_a^j(M_0^{\dag}A_kM_0),\F_a^j(N_0^{\dag}B_kN_0))$ into $D$\;
				}
			}
			
			$C\gets D$\;
			$l_C\gets|C|$\;
		}
		
		$i\gets i-1$\;
	}	
	
	\Return{\textsf{Constraints for Comparability between $P_1$ and $P_2$}}\;
	\caption{\textsf{Find Constrains for Comparability between Quantum Programs}}
	\label{algo:Robust Entanglement Distribution}
\end{algorithm}

\section{Comparison between Unruh's Logic and Ours}\label{comparison-two}

\subsection*{Expressiveness}One main consideration in~\cite{Unruh18} is to retain the
intuitive flavour of proofs in probabilistic relational Hoare logic
and to give an implementation in a general-purpose proof assistant.
Thus, \citep{Unruh18} only considers projective predicates as
pre-conditions and post-conditions. As already noted in Section \ref{sec-projective}, this restriction
can simplify the verification of quantum programs, but also limits the
expressiveness of the logic. In particular, we believe that our main
examples cannot be dealt with by Unruh's logic~\cite{Unruh18}: The symmetry in Example \ref{exam-1}, the uniformity in Proposition \ref{uniformity}, Quantum Bernoulli Factory (Example \ref{Bernoulli}) and the reliability of quantum teleportation (Section \ref{tele-reli}) cannot be specified in that logic because the preconditions in judgments (\ref{e-unif}), (\ref{ju-uniform}) and (\ref{ju-bit}-\ref{ju-bp}) are not projective (i.e. the projection onto a (closed) subspace of the state Hilbert space).
In addition, \citep{Unruh18} introduces classic variables and Boolean predicates into programming language and logic, which significantly enhances the expressiveness and practicability. 
However, introducing classic elements does not bring about a deeper understanding of the relationship between quantum programs while makes the logic itself much more complicated.
%(2) The equivalence of quantum walks cannot be verified in that logic, although the judgment there can be expressed because its precondition and postcondition are both projective.

Later on, \citep{LU19} also considers general observables as predicates and remove all classical variables. Its expressiveness is the same as ours, but the proof rules are different (see discussion below).

\subsection*{Inference Rules}

\noindent\textbf{\cite{Unruh18} vs. rqPD (Section \ref{sec-projective})}: 
Since \citep{Unruh18} introduces classical variables, the proof rules are quite different from ours. Here we mainly focus on the difference between {\bf if} and {\bf while} rules. We believe that Unruh's proof rules \textsc{JointMeasure} together with \textsc{JointIf} and \textsc{JointWhile} for reasoning about {\bf if} and {\bf while} statements with quantum measurements as guards is weaker than our rules (IF-P) and (LP-P). 
The reasons are given as follows.

The guards in {\bf if} and {\bf while} statements considered in \citep{Unruh18} are boolean expressions while {\bf measure} statement is introduced to extract information from quantum registers. Our {\bf if} and {\bf while} statements (with quantum measurements as guards) can be regarded as a combination of {\bf measure} and {\bf if}, {\bf while} statements in \cite{Unruh18}. Note that the rule \textsc{JointMeasure} (see the arXiv version \cite{Unruh18fullversion};  which is the general form of rule \textsc{JointMeasureSimple} discussed in the conference version \cite{Unruh18}) is quite restrictive; more precisely, the preconditions must satisfy: 1. Two measurements can be converted through isometries $u_1$, $u_2$ (expressed as condition $C_e$); 2. The states over the measured quantum registers need to be equal up to the basis transforms $u_1$ and $u_2$ (expressed as condition $u_1Q_1^\prime\equiv_{\textsf{quant}} u_2Q_2^\prime$). Such preconditions are significantly stronger than the preconditions given in our (IF-P) and (LP-P) rules. In fact, we only need to ensure that two measurements produces the same output distribution if two measured states can be lifted to the precondition (see Definition \ref{proj-measure-condition}).

It is worth mentioning that, Unruh's definition of the valid judgments requires that the coupling of inputs and outputs must be separable, so \textsc{Frame} rule holds directly. In contrast, our definition of validity allows the entangled couplings, so a separability condition is added to our (Frame) rule. In addition, there is no (Frame) rule in \cite{LU19}.

\vspace{0.2cm}

\textbf{\cite{LU19} vs. rqPD (Section \ref{sec:qrhl})}: The inference rules in \cite{LU19} are just those basic construct-specific rules presented in Section \ref{subsec-basic}. However, the proof of soundness is somewhat different because of the difference between the definitions of valid judgments. Moreover, we present more rules using measurement condition (see Fig. \ref{fig 4.5}), which enable us to capture more relational properties between quantum programs. For instance, although the symmetry in Example \ref{exam-1} and the uniformity of Quantum Bernoulli Factory (Example \ref{Bernoulli}) can be expressed in \cite{LU19}, their rules are not strong enough for reasoning about these judgments. In opposite, the example of Quantum Zeno Effect used in \cite{LU19} can also be verified using our logic.

\subsection*{Complexity}One of the major differences between our logic and Unruh's one~[\citealp{Unruh18}; \citealp{LU19}] comes from the treatment of entanglement in relational judgments:
\begin{enumerate}
	\item Our definition of valid judgment $P_1\sim P_2:A\Rightarrow B$
	quantifies over all inputs in $\mathcal{H}_{P_1\langle 1\rangle}
	\otimes\mathcal{H}_{P_2\langle 2\rangle}$ including entanglements
	between $P_1$ and $P_2$.  In contrast, Unruh~[\citealp{Unruh18}; \citealp{LU19}] only
	allows separable inputs between $P_1$ and $P_2$.
	\item Unruh~[\citealp{Unruh18}; \citealp{LU19}] also requires that couplings for the outputs
	of programs $P_1$ and $P_2$ must be separable. This is not required in our logic. 
\end{enumerate}

The above differences render essentially different complexities of checking the validity of a given judgment in applications. 
It is known that deciding the separability of a given quantum state is NP-hard with respect to the dimension \cite{Gur03, Gha10}; thus, we believe that deciding the existence of
separable couplings in some subspaces is also NP-hard. If so, this will prevent the
possibility of efficiently automatically checking validity defined
in [\citealp{Unruh18}; \citealp{LU19}] even if the semantic function is known. 
However, quantum Strassen theorem proved in the previous
work \cite{quantumstrassen} implies that it has only a polynomial time
complexity with respect to the dimension to deciding the existence of couplings used in our projective judgments. With the same techniques of SDP (semi-definite programming), the validity of general judgment can also check within polynomial time with respect to the dimension.

\section{Validity of Judgments: Separability versus Entanglements}
In this subsection, we will explore more details between separability and entanglements.
Let us first consider item (1). Intuitively, in our logic, to test the relationship between $P_1$ and $P_2$, we run them in parallel within an environment where entanglement between the two parties are allowed. However, in Unruh's logic, such an entanglement is not provided in the experiment. As one can imagine, our logic can be used to reveal some subtler relational properties of quantum programs. To see this formally, let us introduce the notion of validity with separable
inputs (but not necessarily with separable coupling for outputs) for
general Hermitian operators (rather than projective ones) as the
preconditions and postconditions.

\begin{defn}\label{def-judgment-sep} Judgment (\ref{syntax-judge}) is valid with separable inputs, written: $$\models_\mathit{S} P_1\sim P_2: A\Rightarrow B$$ if it is valid under separability condition $\Gamma= [\mathit{var}(P_1\langle 1\rangle),\mathit{var}(P_2\langle 2\rangle)]$; that is,  $$[\mathit{var}(P_1\langle 1\rangle),\mathit{var}(P_2\langle 2\rangle)]\models P_1\sim P_2: A\Rightarrow B.$$\end{defn}

The relationship between the general validity $\models$, projective validity $\models_P$ and validity $\models_S$ with separable inputs is clarified in the following:
\begin{prop}
	\label{prop-sep}
	For projective predicates $A$ and $B$: \begin{equation*}
	\begin{array}{cll}
	(i) &\models_P P_1\sim P_2: A\Rightarrow B\ \not\Rrightarrow\ \models_S P_1\sim P_2: A\Rightarrow B;\\
	(ii) &\models_P P_1\sim P_2: A\Rightarrow B\ \not\Lleftarrow\ \models_S P_1\sim P_2: A\Rightarrow B; \end{array}
	\end{equation*} and for general quantum predicates $A$ and $B$,
	\begin{equation*}
	\begin{array}{cll}
	(iii) &\models P_1\sim P_2: A\Rightarrow B\ \Rrightarrow\ \models_S P_1\sim P_2: A\Rightarrow B;\\
	(iv) &\models P_1\sim P_2: A\Rightarrow B\ \not\Lleftarrow\ \models_S P_1\sim P_2: A\Rightarrow B. \end{array}
	\end{equation*}
\end{prop}

Clauses (ii) and (iv) clearly indicate that our intuition is correct; that is, a judgment with entangled inputs is strictly stronger than the same judgment with separable inputs.

Now let us turn to items (2). It was clearly pointed out in \cite{Unruh18} that
the first try there is to define validity of judgments without the
requirement of separable couplings for the outputs of $P_1$ and
$P_2$. But separability for the couplings for outputs was finally imposed to the definition of validity of judgments in \cite{Unruh18}.
A nice benefit is that the rule (SC) for sequential composition can be easily achieved for such a notion of validity. Our design decision for preserving rule (SC) is to allow entanglement occurring in both inputs and the couplings for outputs. A benefit of our definition is that the entanglement in these couplings enables us to reveal subtler relational properties between quantum programs. To be more precise, let us consider two lossless programs $P_1, P_2$ and judgment $P_1\sim P_2:A\Rightarrow B$. For an (even separable) $\rho$, let $\lambda=\tr(A\rho)$, then Proposition \ref{entangled-witness} shows that there can be an entangled witness for $\lambda$-lifting $\sm{
	P_1} (\tr_{\langle 2\rangle}(\rho)) B^\#\sm{
	P_2} (\tr_{\langle 1\rangle}(\rho))$. Indeed, it is the case even for separable postcondition $B$.

\subsection*{Inference Rules}
Of course, different definitions of valid judgments imply different inference rules in our logic and Unruh's one.
We already said that without the separability requirement on the couplings for outputs, rule (SC) is not true for validity $\models_S$ with separable inputs.
A series of inference rules are derived in \cite{Unruh18} by imposing the separability requirement on the couplings for outputs.
Here, we would like to point out how this difficulty can be partially solved without this separability. Indeed, a weaker version of (SC), namely (SC-Sep) shown in Figure \ref{fig 6.1}, can be obtained by combining $\models_S$ with our general validity $\models$. 
Furthermore, other inference rules in our rqPD are sound for $\models_S$ except that the rule (Frame) should be slightly modified to (Frame-Sep) in Figure \ref{fig 6.1}.
\begin{figure}[h]\centering
	\begin{equation*}\begin{split}
	&({\rm SC\text{-}Sep})\ \ \
	\frac{\vdash_\mathit{S} P_1\sim P_2: A\Rightarrow B\ \ \ \vdash P_1^\prime\sim P_2^\prime: B\Rightarrow C}{\vdash_\mathit{S} P_1;P_1^\prime\sim P_2;P_2^\prime: A\Rightarrow C}\\
	&({\rm Frame\text{-}Sep})\ \ \
	\frac{\Gamma\vdash_\mathit{S} P_1\sim P_2: A\Rightarrow B}{\Gamma\cup\left\{[V\langle 1\rangle, V\langle 2\rangle, \mathit{var}(P_1),\mathit{var}(P_2)]\right\} \vdash_\mathit{S} P_1\sim P_2: A\otimes C\Rightarrow B\otimes C}
	\end{split}\end{equation*}
	\caption{Rules for Separable Inputs.}\label{fig 6.1}
\end{figure}

\end{document}